%% file: main.tex
\pdfoutput=1
\documentclass[12pt]{report}

\usepackage{amsmath, amsthm, amssymb}
\usepackage{hyperref}
\usepackage{cleveref}
\usepackage{graphicx}
\usepackage{bm}
\usepackage{geometry}
\usepackage{float}
\usepackage{placeins}
\usepackage{doi}
\usepackage{url}
\usepackage{cite}
\hypersetup{colorlinks=true, linkcolor=blue, citecolor=blue, urlcolor=blue}
\newcommand{\R}{\mathbb{R}}
\newcommand{\inside}{\mathrm{int}\,}
\newcommand{\bound}{\partial}
\newcommand{\bs}[1]{\bm{#1}}
\newcommand{\const}{\mathrm{const}}
\DeclareMathOperator{\diag}{diag}

\newtheorem{theorem}{Theorem}[section]
\newtheorem{proposition}[theorem]{Proposition}
\newtheorem{corollary}[theorem]{Corollary}
\newtheorem{conseq}[theorem]{Corollary}
\theoremstyle{definition}
\newtheorem{definition}[theorem]{Definition}
\newtheorem{assumpt}[theorem]{Assumption}
\theoremstyle{remark}
\newtheorem{remark}[theorem]{Remark}

\newenvironment{chapterabstract}
  {\begin{quotation}\small\noindent\textbf{Abstract.}\ }
  {\end{quotation}}
\makeatletter
\renewenvironment{thebibliography}[1]
     {\section*{References}%
      \addcontentsline{toc}{section}{References}%
      \list{\@biblabel{\@arabic\c@enumiv}}%
           {\settowidth\labelwidth{\@biblabel{#1}}%
            \leftmargin\labelwidth
            \advance\leftmargin\labelsep
            \@openbib@code
            \usecounter{enumiv}%
            \let\p@enumiv\@empty
            \renewcommand\theenumiv{\@arabic\c@enumiv}}%
      \sloppy
      \clubpenalty4000
      \@clubpenalty \clubpenalty
      \widowpenalty4000%
      \sfcode`\.\@m}
     {\def\@noitemerr
       {\@latex@warning{Empty `thebibliography' environment}}%
      \endlist}
\makeatother

\begin{document}

\begin{titlepage}
\centering
\vspace*{2.5cm}
{\LARGE\bfseries Mathematical Models of Evolution\\[6pt]
  and Replicator Systems Dynamics\par}
\vspace{1.2cm}
{\large Chapters 1--4\par}
\vspace{2cm}
{\large A.\,S.~Bratus$^{1}$, \quad S.~Drozhzhin$^{1}$, \quad T.~Yakushkina$^{2}$\par}
\vspace{0.8cm}
{\itshape
  $^{1}$Moscow Center for Fundamental and Applied Mathematics,\\
  Lomonosov Moscow State University, Moscow 119991, Russia\\[4pt]
  $^{2}$A.\,I.~Alikhanyan National Science Laboratory\\
  (Yerevan Physics Institute) Foundation,\\
  Alikhanian Brothers St.~2, Yerevan 375036, Armenia\par}
\vfill
\begin{minipage}{0.88\textwidth}\small
\noindent\textbf{About this document.}
This is an edited English version of the authors' monograph
\textit{Mathematical Models of Evolution and Replicator Systems Dynamics},
originally published in Russian (URSS, Moscow, 2022). The text has been
revised and expanded to make the results accessible to a wider international
audience. The book is posted on arXiv as a single document that grows as
chapters are completed: each new arXiv version extends the previous one, and
the changes are listed in the arXiv Comments field. The present version
contains Chapters~1--4. It supersedes the separately posted chapter preprints
arXiv:2604.05720 (Chapter~1), arXiv:2605.05385 (Chapter~2) and
arXiv:2607.04456 (Chapter~3). Each chapter is self-contained and has its own
reference list; equations, figures and theorems are numbered by chapter.
\end{minipage}\par
\vspace{1.5cm}
{\normalsize September 2026\par}
\end{titlepage}

\tableofcontents
\clearpage

\input{chapter1}

\input{chapter2}

\input{chapter3}

\input{chapter4}

\end{document}

%% file: chapter1.tex
\chapter{Introduction to Replicator Systems}
\label{ch:1}

\begin{chapterabstract}
	This chapter is an overview of foundational results in the
	mathematical theory of replicator systems. Its primary aim is to
	provide a unified framework for the mathematical formalisation of
	evolutionary processes in the spirit of generalised Darwinism ~---
	that is, for any system in which heredity, variability, and
	selection can be meaningfully defined, regardless of the specific
	biological substrate. Starting from the Kolmogorov equations for
	interacting populations, we derive the replicator equation and
	examine three canonical regimes: independent, autocatalytic, and
	hypercyclic replication. The hypercycle is shown to be permanent
	and to carry evolutionary variability intrinsically. We then survey
	the quasispecies framework~--- the Eigen and Crow--Kimura models ~--- covering global stability of equilibria, sequence space
	structure, and the error-threshold phenomenon. Throughout, the
	emphasis is on the mathematical structures that underlie these
	models rather than on biological detail, with the goal of making
	the framework applicable to abstract evolutionary dynamics beyond
	its original molecular biology context.

\end{chapterabstract}

\bigskip\hrule\bigskip

\section{Derivation of the Dynamical Equations}
\label{c1:section:1.1}

In evolutionary theory, replication is the process of multiplication or copying.  For the purposes of this research, we define a \emph{replicator} as an object capable of self-reproduction with hereditary stability. The exact definitions may vary depending on the context, e.g. ``a replicator is any
entity that causes certain environments to copy it'' \cite{c1:Deutsch1997} or an entity that is ``able to create copies of itself'' \cite{c1:Dawkins1976}. To formulate a mathematical description of a replicator, it is necessary to specify the law governing the replication rate and precision, as well as other relevant factors.

Let $N_{i}(t)$ denote the population size of species $M_{i}$,
$i = \overline{1, n}$, at time $t$, satisfying the general
Kolmogorov's forward equations for evolutionary dynamics:
\begin{equation}
  \frac{dN_{i}}{dt} = N_{i}g_{i}(\mathbf{N}),\quad
  \mathbf{N}(t) = \bigl(N_{1}(t),\ldots,N_{n}(t)\bigr),
  \label{c1:eq1.1}
\end{equation}
where $g_{i}(\mathbf{N})$ are sufficiently smooth functions that describe inter-species interactions, $g_{i}:\R_{+}^{n}\to\R$.
Here and below, we use the notation
$\R_{+}^{n}=\{\mathbf{x}\in\R^{n}:\mathbf{x}\geqslant 0\}$,
$\bound\R_{+}^{n}=\R_{+}^{n}\setminus\inside\R_{+}^{n}$,
$\inside\R_{+}^{n}=\{\mathbf{x}\in\R^{n}:\mathbf{x}>0\}$,
and the vector inequalities $\mathbf{x}\geqslant 0$, $\mathbf{x}>0$
are understood component-wise.

Moving from absolute population sizes to relative frequencies
\[
  u_{i}(t)=\frac{N_{i}(t)}{\sum_{k=1}^{n}N_{k}(t)},\qquad
  \sum_{k=1}^{n}u_{k}(t)=1,
\]
and substituting into \eqref{c1:eq1.1}, one obtains the following:
\begin{equation}
      \frac{du_{i}}{dt}=
  \frac{1}{\bigl(\sum_{k=1}^{n}N_{k}(t)\bigr)^{2}}
  \Bigl(N_{i}g_{i}(\mathbf{N})\sum_{k=1}^{n}N_{k}
       - N_{i}\sum_{k=1}^{n}N_{k}g_{k}(\mathbf{N})\Bigr),
  \quad i=\overline{1,n}.
  \label{c1:eq1.2}
\end{equation}
If $g_{i}(\mathbf{N})$ are homogeneous functions of order $s$,
i.e.,\ $g_{i}(\xi\mathbf{N})=\xi^{s}g_{i}(\mathbf{N})$,
$\xi\in\R$, then \eqref{c1:eq1.2} can be written as
\begin{equation}
  \frac{du_{i}}{dt}=
  \Bigl(\sum_{k=1}^{n}N_{k}(t)\Bigr)^{s}
  \Bigl(u_{i}g_{i}(\mathbf{u})
       - u_{i}\sum_{k=1}^{n}u_{k}g_{k}(\mathbf{u})\Bigr),
  \quad i=\overline{1,n}.
  \label{c1:eq1.3}
\end{equation}
Since $\sum_{k=1}^{n}N_{k}(t)>0$, system \eqref{c1:eq1.3} is
orbitally topologically equivalent \cite{c1:Arnold1978}
to
\begin{equation}
  \frac{du_{i}}{dt}=u_{i}\bigl(g_{i}(\mathbf{u})-f(t)\bigr),
  \quad f(t)=\sum_{k=1}^{n}g_{k}(\mathbf{u}(t))u_{k}(t),
  \label{c1:eq1.4}
\end{equation}
\[
  \sum_{k=1}^{n}u_{k}(t)=1,\quad u_{i}(0)=u_{i}^{0}, \quad \mathbf{u}(t) = \bigl(u_1(t), \ldots, u_n(t)\bigr),\quad
  i=\overline{1,n}.
\]

The equivalence of \eqref{c1:eq1.3} and \eqref{c1:eq1.4} means, in
particular, that these systems have the same number of equilibria
of the same type and that every closed trajectory of \eqref{c1:eq1.3}
corresponds to a closed trajectory of \eqref{c1:eq1.4}. Therefore, their qualitative behaviour coincides. These systems have identical phase portraits, differing only in the speeds of motion along the phase trajectories. Thus, when only the
asymptotic behaviour ($t\to\infty$) is of interest, either system
may be analysed without loss of generality.

Setting $g_{i}(\mathbf{u})=(\mathbf{Au})_{i}=\sum_{j=1}^{n}a_{ij}u_{j}$,
where $\mathbf{A} = \bigl( a_{ij}\bigr)_{i,j=1,\ldots,n}$, equation \eqref{c1:eq1.4}
becomes the \emph{replicator equation}:
\begin{equation}
  \frac{du_{i}}{dt}=u_{i}\bigl[(\mathbf{Au})_{i}-f(\mathbf{u})\bigr],
  \quad f(\mathbf{u})=\Bigl(\mathbf{Au},\mathbf{u}\Bigr),
  \quad u_{i}(0)=u_{i}^{0},\quad i=\overline{1,n},
  \label{c1:eq1.5}
\end{equation}
where solutions are confined to the simplex
\[
  S_{n}=\Bigl\{u_{i}(t)\geqslant 0,\;i=\overline{1,n},\;
               \sum_{i=1}^{n}u_{i}(t)=1\Bigr\}.
\]
Here and below, round brackets denote the scalar product in $\R^{n}$.

The quantity $(\mathbf{Au})_{i}$ is called the \emph{fitness} of
species $i$, and $f(t)$ is the \emph{mean fitness} of the population.
The entry $a_{ij}$ of $\mathbf{A}$ describes the effect of species $j$
on the population of species $i$; the matrix $\mathbf{A}$ itself determines the \emph{fitness landscape} of the replicator system.
System \eqref{c1:eq1.5} has a natural interpretation in terms of the
per-capita growth rate: $\dot{u}_{i}/u_{i}$ equals the excess of the
fitness of species $i$ over the mean population fitness. Throughout the chapter, we will use the notation $\dot{u}_i$ for derivative with respect to time for brevity.

Replicator systems of this form were first studied in the context of
evolutionary theory by M.~Eigen and P.~Schuster
\cite{c1:Eigen1971,c1:Eigen1977,c1:Eigen1982}, and independently by
V.\,A.~Ratner, R.\,A.~Poluektov, Yu.\,A.~Pykh, Yu.\,M.~Svirezhev,
and D.\,O.~Logofet
\cite{c1:Poluektov1974,c1:Svirezhev1978,c1:MaynardSmith1973,c1:Taylor1978}. 
Eigen and Schuster originally studied replicator systems in the
context of \emph{prebiotic evolution} — the evolutionary process by
which macromolecules capable of producing complex self-replicating structures, analogous to RNA molecules, could have arisen.  These
works attracted considerable interest both from biologists
\cite{c1:Lincoln2009,c1:Vaidya2012} and from mathematicians
\cite{c1:Hofbauer1998,c1:Hofbauer2003}.

\section{Asymptotic Behaviour of a General Class of Replicator Systems}
\label{c1:section:1.2}

We begin the study of replicator systems by analysing the dynamics
in three important special cases.

\begin{enumerate}
\item \textit{Independent replication.}
\begin{equation}
  \frac{du_{i}}{dt}=u_{i}\Bigl(k_{i}-f_{1}(u)\Bigr),\quad
  f_{1}(u)=\sum_{i=1}^{n}k_{i}u_{i}(t),\quad
  \mathbf{u}(t)\in S_{n},\quad i=\overline{1,n}.
  \label{c1:eq1.6}
\end{equation}

\item \textit{Autocatalytic replication.}
\begin{equation}
  \frac{du_{i}}{dt}=u_{i}\Bigl(k_{i}u_{i}-f_{2}(u)\Bigr),\quad
  f_{2}(u)=\sum_{i=1}^{n}k_{i}u_{i}^{2}(t),\quad
  \mathbf{u}(t)\in S_{n},\quad i=\overline{1,n}.
  \label{c1:eq1.7}
\end{equation}

\item \textit{Hypercyclic replication.}
\begin{equation}
  \frac{du_{i}}{dt}=u_{i}\Bigl(k_{i}u_{i-1}-f_{3}(u)\Bigr),\quad
  f_{3}(u)=\sum_{i=1}^{n}k_{i}u_{i}(t)u_{i-1}(t),\quad
  \mathbf{u}(t)\in S_{n},\quad i=\overline{1,n}.
  \label{c1:eq1.8}
\end{equation}
\end{enumerate}

In the last case, indices are taken modulo $n$, i.e., \ $u_{0}=u_{n}$. Throughout, $k_{i}>0$ for all $i=\overline{1,n}$.

Systems \eqref{c1:eq1.6}, \eqref{c1:eq1.7},  and  \eqref{c1:eq1.8}
represent two extreme cases of replication.  In the first two
systems, each species replicates using only itself.  In \eqref{c1:eq1.8},
replication of each species requires the preceding species in a
closed cycle (see Fig.~\ref{c1:fig1.1}).

\begin{figure}[ht]
  \centering
  \includegraphics[width=0.7\textwidth]{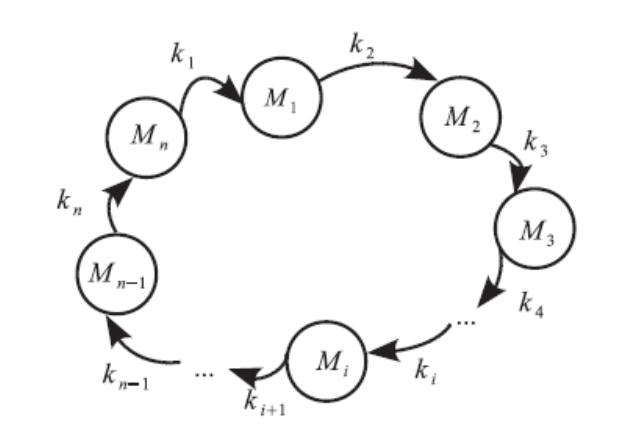}
  \caption{Graph representing hypercyclic replication.}
  \label{c1:fig1.1}
\end{figure}

If the behaviour of the first two systems can be characterised as
\emph{selfish}, then hypercyclic replication demonstrates
\emph{altruistic} behaviour: reproduction of each species constitutes
the simplest form of mutual aid, where every species — directly or
indirectly — benefits from all other species included in the cycle.

The interplay between selfish and cooperative behaviour in replicator systems arises across multiple disciplines beyond mathematical biology, including economics, game theory, and sociology; for a comprehensive review see \cite{c1:Sigmund2009}.

\medskip
\textit{Independent replication.}
The limiting behaviour is characterised by survival of the species
with the maximum Malthusian fitness coefficient $k_{i}$.

Let $k_{m}=\max\{k_{1},k_{2},\ldots,k_{n}\}$.
For any $i\neq m$,
\[
  \dot{\Bigl(\frac{u_{i}}{u_{m}}\Bigr)}
  =\Bigl(\frac{u_{i}}{u_{m}}\Bigr)(k_{i}-k_{m})<0,
\]
so $u_{i}(t)/u_{m}(t)=C_{0}e^{(k_{i}-k_{m})t}\to 0$
as $t\to+\infty, C_0 = const>0$.  Since $\mathbf{u}\in S_{n}$, this means $u_{i}(t)\to 0$ for all $i\neq m$ and $u_{m}(t)\to 1$.

The dynamics of the mean fitness $f_1$ satisfy
\[
  \frac{df_{1}}{dt}
  =\sum_{i=1}^{n}k_{i}\dot u_{i} = \sum_{i=1}^{n}k_{i}^{2}u_{i}
   -\Bigl(\sum_{i=1}^{n}k_{i}u_{i}\Bigr)^{2}\geqslant 0.
\]
Since we are working on a simplex, the right-hand side is the variance of a random variable taking values $k_1,k_2,\ldots,k_n$ with probabilities $u_1(t),u_2(t),\ldots,u_n(t)$, and is therefore 
always non-negative. Hence, the mean fitness $f_1(t)$ is monotonically non-decreasing along every trajectory of system~\eqref{c1:eq1.6}. In the simplest interpretation, this is the mathematical form of Fisher's fundamental theorem of natural selection, which asserts that ``the rate of increase in fitness of any organism at any time is equal to its genetic variance in fitness at that time'' \cite{c1:Fisher1930}.
We note that Fisher himself did not provide a rigorous mathematical formulation of this theorem, and that the precise meaning of ``genetic variance'' requires careful definition in the general case. For the purposes of this chapter, two observations suffice: in its simplest interpretation, the theorem asserts that mean fitness is non-decreasing over time — a property we have just established for independent replication; and this monotonicity generally fails for more complex replicator systems, as illustrated below.

\medskip
\textit{Autocatalytic replication.}
The equilibrium $\bar{\mathbf{u}}\in S_{n}$ is determined by
\[
  k_{1}\bar{u}_{1}=\ldots=k_{n}\bar{u}_{n}
  =\bar{f}=\sum_{i=1}^{n}k_{i}u_{i}^{2},\quad \bar{\mathbf{u}}\in S_{n}.
\]
Hence
\[
  \bar{\mathbf{u}}
  =\frac{1}{k_{i}\sum_{j=1}^{n}k_{j}^{-1}}.
\]
Introducing barycentric coordinates \cite{c1:Hofbauer1978} 
that map this equilibrium to the centroid $(n^{-1},\ldots,n^{-1})$:
\[
  v_{i}=\frac{k_{i}u_{i}}{R},\quad R=\sum_{j=1}^{n}k_{j}u_{j},
\]
system \eqref{c1:eq1.7} transforms into the equivalent system
\begin{equation}
  \frac{dv_{i}}{dt}=v_{i}\Bigl(v_{i}-\sum_{j=1}^{n}v_{j}^{2}\Bigr),
  \quad i=\overline{1,n},\quad \mathbf{v}(t)\in S_{n}.
  \label{c1:eq1.9}
\end{equation}
Since $R > 0$, we may choose $R=1$ as systems are orbitally topologically equivalent, which 
considerably simplifies the analysis of the dynamics.

All equilibria of \eqref{c1:eq1.9} are easily found.
Besides the interior equilibrium
$\bar{\mathbf{u}}_{n}^{1}=(n^{-1},\ldots,n^{-1})\in\inside S_{n}$,
the system has equilibria on the boundary of $S_{n}$ (which is a simplex of a smaller dimension $n-1$).
Writing $2^{n}-1$ fixed points explicitly:
$\bar{\mathbf{u}}_{n-1}^{j}=\bigl((n-1)^{-1},\ldots,0,\ldots,(n-1)^{-1}\bigr)$
(zero in the $j$-th position), and so on down to the vertices
$p^{s}=(0,\ldots,1,\ldots,0)$ (one in the $s$-th position).

The Jacobian matrix at the point $\bar{\mathbf{u}}^1_n$ takes the form
\[
  \mathbf{J}\!\left(\bar{\mathbf{u}}^1_n\right)
  = \frac{1}{n^2}
  \begin{pmatrix}
    n-2 & -2   & \cdots & -2   \\
    -2   & n-2 & \cdots & -2   \\
    \cdots & \cdots & \cdots & \cdots \\
    -2   & -2   & \cdots & n-2
  \end{pmatrix}.
\]
This matrix has eigenvalue $\lambda_1 = n^{-1}$ of multiplicity 
$n-1$, with eigenvectors
\[
\mathbf{e}^1 = (1,-1,0,\ldots,0),\quad
\mathbf{e}^2 = (1,0,-1,0,\ldots,0),\quad\ldots,\quad
\mathbf{e}^{n-1} = (1,0,\ldots,0,-1).
\]
The vector $\mathbf{e}^{n}$ is orthogonal to the hyperplane
$\sum u_{i}=1$ and does not belong to $S_{n}$; it corresponds to
eigenvalue $\lambda_{2}=-n^{-1}$, so the interior equilibrium is an
unstable node.

Consider  interior equilibria of the $(n-1)$-dimensional faces of $S_n:$ $\bar{\mathbf{u}}^{j}_{n-1}$, $j=\overline{1,n}$.
The stability analysis of these equilibria is entirely analogous to the
one carried out above. The Jacobian matrices have eigenvalue
$\lambda_1 = (n-1)^{-1}$ of multiplicity $n-2$ and eigenvalue
$\lambda_2 = -(n-1)^{-1}$. In contrast to the previous case, the
eigenvector corresponding to $\lambda_2$ belongs to $S_{n-1}$;
consequently, the equilibria $\bar{\mathbf{u}}^{j}_{n-1}$,
$j = \overline{1,n}$, are saddle points with a one-dimensional stable
manifold.  Similarly, the equilibria $\bar{\mathbf{u}}^{jk}_{n-2}$,
$j, k = \overline{1,n}$, are also saddle points.

The phase portrait of system~\eqref{c1:eq1.9} (which is shown in
Fig.~\ref{c1:fig1.2} for $n=3$): trajectories leave the interior equilibrium,
approach the equilibria on the boundary faces $S_{n-1}$, then
continue to those on $S_{n-2}$, and so on until reaching a vertex
$p^{j}$, $j=\overline{1,n}$. All trajectories starting in $S_n$,
except those beginning on the stable manifolds of the saddle points,
converge to one of the vertices $p^{j}$ of $S_n$. The asymptotic
behaviour of the system depends on the initial conditions: depending
on the initial state, exactly one species survives the competition,
as in the case of independent replication. This behaviour is called
\emph{adaptive} or \emph{multistable}: the initial conditions
determine which vertex is reached as $t\to\infty$.

One may say that whereas in independent replication the fittest
species always survives (i.e.\ the one with the highest fitness
coefficient), in autocatalytic replication the species with the
largest product of fitness coefficient and initial frequency
survives.

\begin{figure}[ht]
  \centering
  \includegraphics[width=0.75\textwidth]{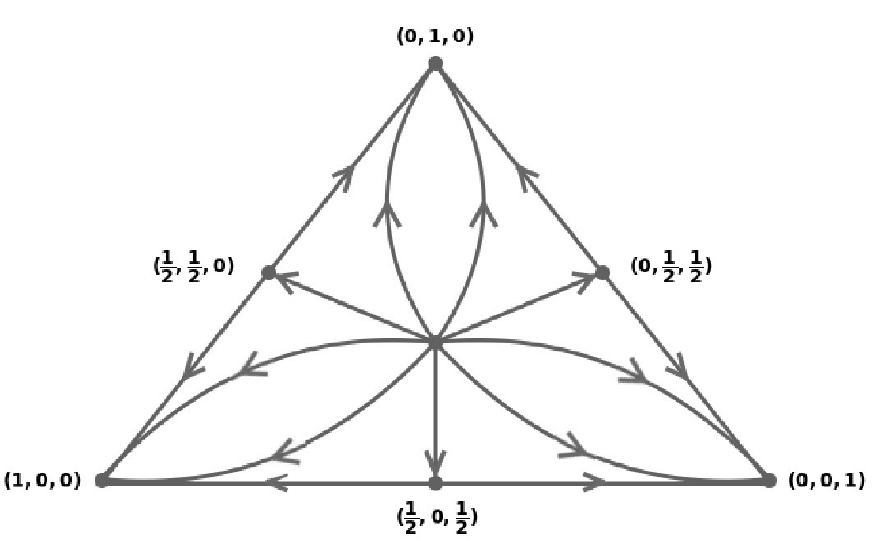}
  \caption{Phase portrait of the autocatalytic replication 
  		system~\eqref{c1:eq1.9} for $n=3$. Trajectories emanate from the 
  		unstable interior equilibrium $\bar{\mathbf{u}}^1_3 = 
  		(\tfrac{1}{3},\tfrac{1}{3},\tfrac{1}{3})$, pass through the saddle 
  		points on the boundary faces, and converge to one of the vertices 
  		$(1,0,0)$, $(0,1,0)$, or $(0,0,1)$, depending on initial conditions.}
  \label{c1:fig1.2}
\end{figure}

Consider system \eqref{c1:eq1.7}.  The mean fitness satisfies
\[
  \frac{df_{2}}{dt}
  =2\Bigl(\sum_{i=1}^{n}k_{i}^{2}u_{i}^{3}
          -\Bigl(\sum_{i=1}^{n}k_{i}u_{i}^{2}\Bigr)^{2}\Bigr)\geqslant 0.
\]
By the Cauchy--Schwarz inequality,
\[
\left(\sum_{i=1}^{n} k_i u_i^{\frac{3}{2}} u_i^{\frac{1}{2}}\right)^2
\leqslant
\sum_{i=1}^{n} k_i^2 u_i^3 \cdot \sum_{i=1}^{n} u_i
= \sum_{i=1}^{n} k_i^2 u_i^3,
\]
hence $\dot{f}_2(t) \geqslant 0$. Consequently, as in the case of
independent replication, the mean fitness of system~\eqref{c1:eq1.7}
is a monotonically non-decreasing function of time.

\medskip
\textit{Hypercyclic replication.}
The interior equilibrium of \eqref{c1:eq1.8} is
\[
  \bar{u}_{i}
  =\frac{k_{i+1}^{-1}}{\sum_{j=0}^{n-1}k_{j+1}^{-1}},\quad
  i=\overline{1,n},\quad k_{n+1}=k_{1}.
\]
As in the autocatalytic case, we introduce barycentric coordinates
\[
  v_{i}=\frac{k_{i+1}u_{i}}{R},\quad R=\sum_{j=0}^{n-1}k_{j+1}u_{j},
\]
that bring the equilibrium to $(n^{-1},\ldots,n^{-1})$, and \eqref{c1:eq1.8}
becomes the orbitally equivalent system
\begin{equation}
  \frac{dv_{k}}{dt}=v_{k}\Bigl(v_{k-1}-\sum_{j=1}^{n}v_{j}v_{j-1}\Bigr),
  \quad v_{0}(t)=v_{n}(t),\quad k=\overline{1,n},\quad
  \mathbf{v}(t)\in S_{n}.
  \label{c1:eq1.10}
\end{equation}

\begin{proposition}
The eigenvalues of the Jacobian of system \eqref{c1:eq1.10} at the
equilibrium \[\bar{\mathbf{v}}=(n^{-1},\ldots,n^{-1})\in\inside S_{n}\]
may be expressed as 
\[
  \lambda_{j}=\frac{1}{n}\exp\!\Biggl(\frac{2\pi j}{n}i\Biggr),
  \quad j=\overline{0,n-1},
\]
where $i$ is the imaginary unit.
\end{proposition}

\begin{proof}
If $j\neq k-1$, system \eqref{c1:eq1.10} gives at equilibrium
$\bar{\mathbf{v}}$:
\[
  \frac{\partial\dot{v}_{k}}{\partial v_{j}}=-\frac{2}{n^{2}},
\]
\[
  \frac{\partial\dot{v}_{k}}{\partial v_{j}}
  =\frac{n-2}{n^{2}},\quad\text{if }j=k-1.
\]
The Jacobian is therefore
\[
  \mathbf{J}(\bar{\mathbf{v}})=\frac{1}{n}
  \begin{pmatrix}
    -2 & -2 & \cdots & -2 & n-2\\
    n-2 & -2 & \cdots & -2 & -2\\
    \vdots & \vdots & \ddots & \vdots & \vdots\\
    -2 & -2 & \cdots & n-2 & -2\\
  \end{pmatrix}.
\]
This is a \emph{circulant matrix}, whose eigenvalues are given by
the formula \cite{c1:Bellman1960}:
\begin{equation}
  \lambda_{j}
  =-\frac{2}{n^{2}}\sum_{k=0}^{n-1}\eta^{kj}
  +\frac{1}{n}\eta^{(n-1)j}
  =-\frac{\eta^{j}}{n},\quad j=\overline{0,n-1},
  \label{c1:eq1.11}
\end{equation}
where $\eta=\exp\!\biggl(\frac{2\pi j}{n}i\biggr)$. 

\end{proof}

When $j=0$, $\lambda_{0}=-n^{-1}$ with eigenvector $(1,1,\ldots,1)$,
which is orthogonal to the simplex $S_{n}$ and hence excluded from
the stability analysis.  From \eqref{c1:eq1.11}: the equilibrium
$\bar{\mathbf{v}}$ is asymptotically stable for $n=2,3$ and unstable
for $n\geqslant 5$, since in the latter case there are always
eigenvalues with positive real part.  For $n=4$ one has
$\lambda_{1,2}=\pm\frac{i}{4}$, $\lambda_{3}=-\frac{1}{4}$, and
linear analysis is inconclusive.  In this case we use the Lyapunov function
\[
\Phi(\mathbf{v}) = \bigl(v_1+v_2+v_3+v_4\bigr)^2 - 4f
= \bigl[(v_1+v_3)-(v_2+v_4)\bigr]^2,
\]
where $f = v_1v_4 + v_2v_1 + v_3v_2 + v_4v_3$, 
whose time derivative along trajectories of \eqref{c1:eq1.10} satisfies
$\dot{\Phi}(\mathbf{v})\leqslant 0$.  The zero set of $\dot\Phi$ lies
in $Z=\{\mathbf{v}\in S_{n}:v_{1}+v_{3}=v_{2}+v_{4}\}$.
By LaSalle's invariance principle \cite{c1:LaSalle1961}, 
every trajectory of $S_{4}$ converges to the largest invariant subset
$M$ of $Z$, which, from the additional condition
\[\frac{d}{dt}(v_{1}+v_{3})=\frac{d}{dt}(v_{2}+v_{4}).\]
It follows that
\[
v_{1}v_{4}+v_{3}v_{2}-(v_{1}+v_{3})f
= v_{2}v_{1}+v_{4}v_{3}-(v_{2}+v_{4})f,
\qquad\text{if }(v_{1}-v_{3})(v_{4}-v_{2})=0.
\]
This means that the set $M$ is contained in the set $v_{1}=v_{3}$ or $v_{2}=v_{4}$. Hence, $M$  consists only of the equilibrium $\bar{\mathbf{v}}_{4}\in S_{4}$, and the equilibrium is stable for $n=4$.
\medskip
The preceding analysis was specific to the hypercycle. We now turn to the general replicator system~\eqref{c1:eq1.5} and consider the general case of arbitrary fitness $(\mathbf{Au})_i$.
Let $\bar{\mathbf{u}}\in\operatorname{int}S_{n}$ be an equilibrium of
System~\eqref{c1:eq1.5}, whose existence we assume. Then the following
equalities hold:
\begin{equation}
	\mathbf{A}\bar{\mathbf{u}}=\bar{f}\,\mathbf{I},\quad
	\bar{f}=\bigl(\mathbf{A}\bar{\mathbf{u}},\bar{\mathbf{u}}\bigr),\quad
	\bigl(\mathbf{u},\mathbf{I}\bigr)=1,\quad
	\mathbf{I}=(1,1,\ldots,1)\in\R^{n}.
	\label{c1:eq1.12}
\end{equation}
For this general case, we introduce the Lyapunov function
\begin{equation}
  V(\mathbf{u})
  =\sum_{i=1}^{n}\Bigl[(u_{i}-\bar{u}_{i})
   -\bar{u}_{i}\ln\!\Bigl(\frac{u_{i}}{\bar{u}_{i}}\Bigr)\Bigr],
  \label{c1:eq1.13}
\end{equation}
which is positive and vanishes only at $\mathbf{u}=\bar{\mathbf{u}}$.
Its time derivative along trajectories of \eqref{c1:eq1.5} is
\begin{equation}
  \dot{V}(\mathbf{u})=(\mathbf{Au},\mathbf{u}-\bar{\mathbf{u}}),
  \quad \mathbf{u}\in S_{n}.
  \label{c1:eq1.14}
\end{equation}
Since $u_i - \bar{u}_i \geqslant \bar{u}_i \ln\!\left(\dfrac{u_i}{\bar{u}_i}\right)$
for all $u_i\bar{u}_i > 0$, the function $V(\mathbf{u})$ is positive
and goes to zero only at $\mathbf{u}=\bar{\mathbf{u}}$, and is therefore
a Lyapunov function candidate.

Denote $\xi=\mathbf{u}-\bar{\mathbf{u}}$. Decomposing
$\mathbf{A}=\mathbf{B}+\mathbf{C}$, where
$\mathbf{B}=(\mathbf{A}+\mathbf{A}^{\top})/2$ is symmetric and
$\mathbf{C}=(\mathbf{A}-\mathbf{A}^{\top})/2$ is skew-symmetric
(so $(\mathbf{C}\xi,\xi)=0$), and taking into account  $\Big({\bf u, I}\Big) = 1$, so $\Big(\xi, {\bf I}\Big) = 0$, we obtain
\[\dot{V}(\mathbf{u})=\Big(\mathbf{B}\xi,\xi\Big)\].
The stability condition for an interior equilibrium
$\bar{\mathbf{u}}\in\inside S_{n}$ therefore reduces to
\begin{equation}
  (\mathbf{B}\xi,\xi)\leqslant 0
  \label{c1:eq1.15}
\end{equation}
for all $\xi\in\R^{n}$ satisfying
\begin{equation}
  (\xi,\mathbf{I})=0,\quad \mathbf{I}=(1,1,\ldots,1)\in\R^{n}.
  \label{c1:eq1.16}
\end{equation}
That is, all eigenvalues of the symmetric matrix $\mathbf{B}$
restricted to the $(n-1)$-dimensional subspace defined by
\eqref{c1:eq1.16} must be non-positive.

If $\bar{\mathbf{u}}\in\partial S_{n}$, for example $\bar{u}_{1}=0$,
and $\bar{\mathbf{u}}'=(\bar{u}_{2},\bar{u}_{3},\ldots,\bar{u}_{n})$
is an interior point of the corresponding simplex
$S_{n-1}=\bigl\{\mathbf{u}: \sum_{i=2}^{n}u_{i}=1\bigr\}$,
then, applying the function $V$ with $i=\overline{2,n}$, one can obtain
a stability condition for this equilibrium analogous to
\eqref{c1:eq1.15}--\eqref{c1:eq1.16}.
We note that in many cases it is more important to verify that the
boundary equilibrium $\bar{\mathbf{u}}\in\partial S_{n}$ is unstable.

For circulant matrices, eigenvalue $\lambda_{1}$ always exists with
eigenvector $(1,1,\ldots,1)$, orthogonal to all eigenvectors in $S_{n}$;
hence stability is determined by the signs of the remaining eigenvalues
$\lambda_{2},\ldots,\lambda_{n}$.  The method of finding eigenvalues on
the constrained subspace \eqref{c1:eq1.16} was proposed by
M.\,G.~Krein and can be found in \cite{c1:Shilov1969}. 

As an illustration, consider
\[
  \mathbf{A}=
  \begin{pmatrix}
    0 & a_1 & a_2 & a_3\\
    a_3 & 0 & a_1 & a_2\\
    a_2 & a_3 & 0 & a_1\\
    a_1 & a_2 & a_3 & 0\\
  \end{pmatrix},\quad
  \mathbf{B}=\frac{\mathbf{A}+\mathbf{A}^{\top}}{2}=
  \begin{pmatrix}
    0 & \alpha & \beta & \alpha\\
    \alpha & 0 & \alpha & \beta\\
    \beta & \alpha & 0 & \alpha\\
    \alpha & \beta & \alpha & 0\\
  \end{pmatrix},
\]
where $\alpha=(a_1+a_3)/2$, $\beta=a_2$. 
The eigenvector corresponding to the first eigenvalue has the form
$(1,1,\ldots,1)$ and is orthogonal to the simplex $S_n$.
The eigenvalues of $\mathbf{B}$ are
$\lambda_{1}=2\alpha+\beta$, $\lambda_{2}=-\beta$,
$\lambda_{3}=\beta-2\alpha$.
If $\beta>0$ and $2\alpha>\beta$, the interior equilibrium is
asymptotically stable.

\section{Darwin's Evolutionary Postulates and Properties of the Hypercycle}
\label{c1:section:1.3}

The theory of biological evolution was proposed by Charles Darwin
\cite{c1:Darwin1859}.  In that work, the fundamental triad of the
evolutionary process was formulated: \textit{heredity —
variability — natural selection}.  Together with other supplementary
principles, these postulates form the foundation of modern evolutionary
theory, notwithstanding all the fundamental discoveries of recent
centuries.  Since we consider mathematical models of evolutionary
processes, it is necessary to specify precisely what serves as the
mathematical formalisation of heredity, variability, and natural
selection in our models.

Heredity, as should be clear from the preceding discussion, is
formalised by the general form of the replicator equation:
\[
  \frac{\dot{u}_{i}}{u_{i}}=g_{i}(\mathbf{u})-f(t),
\]
where the right-hand side is the excess fitness of species $i$ over
the mean population fitness.

The frequently used term \emph{fitness} is the mathematical
formalisation of natural selection in our models.  In the simplest
case of independent replication, fitness is constant; in the general
case, it is a complex function of the population structure.

Variability is often specified in terms of explicit parameters
describing the probability of transition from species $i$ to species $j$.
These parameters are usually called \emph{mutation parameters} or
the \emph{mutation landscape}.  In other cases, variability is an intrinsic property of the
model.  In particular, as we will show, variability is implicitly
built into the hypercycle model.

In addition to Darwin's three postulates, we will also be interested
in models that, in an evolutionary sense, we may call
\emph{permanent} (or \emph{non-degenerate}).  By this we mean
replicator systems in which no species becomes extinct over time
(in some sense, the system sustains its own complexity).
In the population dynamics literature, the terms \textit{permanent} and
\textit{persistent} are also used \cite{c1:Hofbauer1978} (where \textit{permanent} is the stronger condition). 
The mathematical formulation is as follows.

\begin{definition}\label{c1:def1.1}
A replicator system \eqref{c1:eq1.5} is called \emph{permanent}
(\emph{non-degenerate}) if for any initial data
$u_{i}^{0},\;i=\overline{1,n}$,
$\mathbf{u}_{0}\in\inside S_{n}$, there exists $\delta_{0}>0$ such that
\[
  \liminf_{t\to+\infty}u_{i}(t)\geqslant\delta_{0}>0,\quad
  i=\overline{1,n}.
\]

\end{definition}

Informally, a system is permanent if the boundary of the simplex
``repels'' all trajectories starting in the interior.

While in the course of evolution many species have gone extinct, the
permanence condition cannot be regarded as absolutely necessary for
biological systems.  On the other hand, the extinction of species
diminishes biodiversity, which is also undesirable.  For these
reasons, in many problems we will require permanence of our
replicator equations in the sense of the definition above.

A remarkable fact is that the hypercycle system is permanent.
To prove this, we use the following result \cite{c1:Hofbauer2003}.

\begin{theorem}\label{c1:t1.1}
A replicator system \eqref{c1:eq1.5} is permanent if and only if there
exists a vector $\mathbf{p}\in\inside S_{n}$ such that
\[
  (\mathbf{p},\mathbf{A}\bar{\mathbf{u}})
  >(\bar{\mathbf{u}},\mathbf{A}\bar{\mathbf{u}})
\]
for all fixed points $\bar{\mathbf{u}}\in\bound S_{n}$
of system \eqref{c1:eq1.5}.
\end{theorem}

\begin{proof}
The proof of this theorem is based on the following reasoning.
If on the boundary of the simplex there are fixed points of the system
characterised by the presence of trajectories entering those points,
the system will be degenerate.  Consider an arbitrary point
$\mathbf{p}\in\inside S_{n}$ and the function
$v(t)=(\mathbf{u}(t),\mathbf{p})$.  One can show that 
\begin{equation*}
	\dot{\mathbf v} = \bigl(\dot{\mathbf{u}}(t), \mathbf{p}\bigr)
	= \sum_{i=1}^{n} \bigl(\mathbf{Au}\bigr)_{\!i} p_{i} - f(t)
	= \lvert\dot{\mathbf{u}}(t)\rvert\,\lvert\mathbf{p}\rvert
	\cos\Bigl(\widehat{\dot{\mathbf{u}}(t),\, \mathbf{p}}\Bigr)
	 = \bigl(\mathbf{Au}, \mathbf{u}\bigr).
\end{equation*}

If a fixed point $\bar{\mathbf{u}}\in\bound S_{n}$ satisfies
$\dot{\mathbf v}(\bar{\mathbf{u}})\leqslant 0$, then the phase trajectory
forms an obtuse angle with $\mathbf{p}\in\inside S_{n}$, and
the trajectory enters the point.  In the opposite case, the motion
proceeds into the interior of $S_{n}$, making $\bar{\mathbf{u}}$
a repeller.
The precise proof of Theorem~\ref{c1:t1.1} can be found in
\cite{c1:Hofbauer2003}. For further results on permanence and fitness optimisation in replicator systems, see \cite{c1:Drozhzhin2021}.

We now apply Theorem~\ref{c1:t1.1} to prove permanence of the hypercycle
system.  For this purpose we use the equivalent system \eqref{c1:eq1.10}.

Let $\mathbf{p}=(n^{-1},\ldots,n^{-1})$.  All equilibria of
\eqref{c1:eq1.10} on the boundary satisfy
$$\bar{u}_{i}(\bar{u}_{i-1}-\bar{f}(\bar{\mathbf{u}}))=0, \quad i = \overline{1, n}, \quad \bar{u}_{0} = \bar{u}_{n}, 
\quad \bar{f}({\bf \bar{u}}) = \sum\limits_{i = 1}^{n}\bar{u}_{i}\bar{u}_{i - 1}.
$$

If, for example, $\bar{u}_{1}=0$ but $\bar{u}_{2}\neq 0$, then
$\bar{f}(\bar{\mathbf{u}})=0$ and the equilibrium lies on
$\bound S_{n}$.  
This means that at least one component of such a vector must be
non-zero. The interaction matrix of \eqref{c1:eq1.10} is
\[
  \mathbf{A}=
  \begin{pmatrix}
    0 & 0 & \cdots & 0 & 1\\
    1 & 0 & \cdots & 0 & 0\\
    \vdots &   &  \ddots& & \vdots\\
    0 & 0 & \cdots & 1 & 0\\
  \end{pmatrix},
\]
so $\mathbf{A\bar{u}}=(\bar{u}_{n},\bar{u}_{1},\ldots,\bar{u}_{n-1})$
and $(\mathbf{A\bar{u}},\mathbf{p})>0$ while $f(\bar{\mathbf{u}})=0$,
establishing the condition of Theorem~\ref{c1:t1.1}.
\end{proof}

\begin{corollary}
Let $\mathbf{u}(t)=(u_{1}(t),\ldots,u_{n}(t))$ be a solution of
\eqref{c1:eq1.8} with $u_{i}(0)=u_{i}^{0}>0$, $i=\overline{1,n}$.
Then the time-averaged frequencies
\[
  \bar{u}_{i}
  =\lim_{T\to+\infty}\frac{1}{T}\int_{0}^{T}u_{i}(t)\,dt
\]
are the coordinates of the interior equilibrium.
\end{corollary}

\begin{proof}
Write \eqref{c1:eq1.8} as
$$
\frac{\dot{u}_{i}}{u_{i}}=(\mathbf{Au})_{i}-f(t),$$
integrate from $0$ to $T$, dividing by $T$, and use permanence
to conclude that
$$\lim_{T\to+\infty}\frac{\ln u_{i}(T)-\ln u_{i}^{0}}{T}=0.$$
Therefore
$$
\Big({\bf A\bar{u}}\Big)_{i} = \lim\limits_{T \to +\infty}\frac{1}{T}\int\limits_{0}^{T}f(t)dt = \Big({\bf A\bar{u}, \bar{u}}\Big) = \bar{f}, \quad i = \overline{1, n}.
$$
Hence $(\mathbf{A\bar{u}})_{i}=\bar{f}$ for all $i$, which is
precisely the system determining the interior equilibrium.
\end{proof}

In fact, the hypercycle system possesses an even stronger property.
For $n\geqslant 5$, system \eqref{c1:eq1.8} admits a stable limit cycle
— a closed trajectory around which all other trajectories accumulate
as $t\to+\infty$ \cite{c1:Hofbauer1998,c1:Hofbauer1991}.  The proof relies on a more
general result \cite{c1:MalletParet1990} 
concerning systems of the form
$\dot{u}_{i}=f_{i}(u_{i},u_{i-1})$.
The Poincar\'e--Bendixson conditions for the existence of a limit
cycle are in general difficult to verify; for the hypercycle system
on the simplex, however, these conditions are satisfied for all
$n \geqslant 5$.

We also note that the hypercycle system possesses the property of
\emph{evolutionary variability}.

\begin{definition}
Row $i$ of matrix $\mathbf{A}$ is said to be \emph{strictly
dominated} by row $j$ if
$(\mathbf{Au})_{i}<(\mathbf{Au})_{j}$ for all $\mathbf{u}\in S_{n}$.
\end{definition}

\begin{proposition}
If row $i$ is strictly dominated by row $j$ in replicator system
\eqref{c1:eq1.5}, then $u_{i}(t)\to 0$ as $t\to+\infty$.
\end{proposition}

\begin{proof}
Multiply the $i$-th equation of \eqref{c1:eq1.5} by $u_{j}$ and
subtract the $j$-th equation multiplied by $u_{i}$:
\[
\frac{d}{dt}\!\left(\frac{u_i}{u_j}\right)
= \left(\frac{u_i}{u_j}\right)\bigl((\mathbf{Au})_i - (\mathbf{Au})_j\bigr).
\]
If $(\mathbf{Au})_{i}<(\mathbf{Au})_{j}$ for all $\mathbf{u}\in S_{n}$,
then $u_{i}(t)/u_{j}(t)\to 0$, hence $u_{i}(t)\to 0$.
\end{proof}
\medskip
Consider a hypercycle whose interaction graph is shown in 
Fig.~\ref{c1:fig1.3}. Unlike the standard hypercycle of 
Fig.~\ref{c1:fig1.1}, this system contains $n+1$ species: species~$1$ 
is catalysed by both species~$n$ and species~$n+1$, with rate 
coefficients $k_1,\,k_n$ and $\bar{k}_1,\,\bar{k}_n$ respectively.

\begin{figure}[ht]
  \centering
  \includegraphics[width=0.65\textwidth]{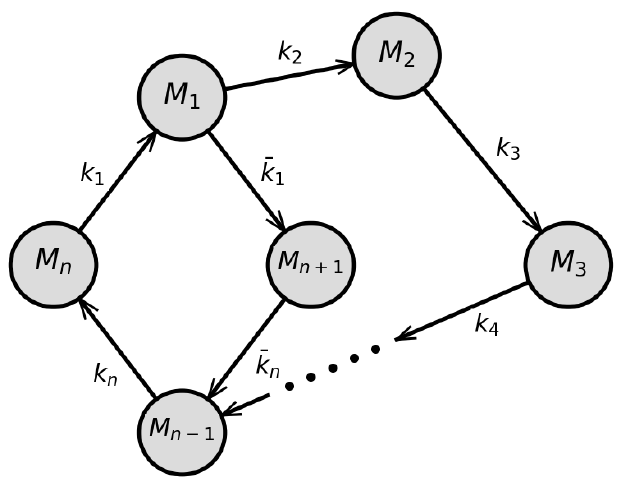}
  \caption{Graph of the modified hypercyclic replication with matrix $\mathbf{A}$
           \eqref{c1:eq1.17} with $n+1$ species.
           	Species~$1$ receives catalytic input from both species~$n$ 
           	and the additional species~$n+1$.}
  \label{c1:fig1.3}
\end{figure}

The interaction matrix is
\begin{equation}
  \mathbf{A}=
  \begin{pmatrix}
    0 & 0 & \cdots & 0 & k_1 & \bar{k}_1\\
    k_2 & 0 & \cdots & 0 & 0 & 0\\
    0 & k_3 & \cdots & 0 & 0 & 0\\
    \vdots & & \ddots & & & \vdots\\
    0 & 0 & \cdots & k_n & 0 & 0\\
    0 & 0 & \cdots & \bar{k}_n & 0 & 0\\
  \end{pmatrix}.
  \label{c1:eq1.17}
\end{equation}

Depending on which of the coefficients $k_{n}$ or $\bar{k}_{n}$ is
larger, one of the last two rows of matrix $\mathbf{A}$ will be
dominated. Species~$n$ goes extinct if $k_{n}<\bar{k}_{n}$;
conversely, species~$n+1$ goes extinct and species~$n$ survives if
$k_{n}>\bar{k}_{n}$. In either case, survival is independent of the
ratio $k_{1}/\bar{k}_{1}$.

Thus, if in the course of evolution a species with ``better''
properties appears ($\bar{k}_{n}>k_{n}$), the hypercycle with two
catalytic branches selects exactly one of them. This reflects a
capacity for evolutionary change: species with ``better'' properties
can be incorporated into the hypercycle while those with ``worse''
properties are eliminated. It is readily seen that such a process can
only increase the mean fitness of the system.
\medskip

For two competing hypercycles sharing common vertices
(Fig.~\ref{c1:fig1.4}), with matrix
\begin{equation}
  \mathbf{A}=
  \begin{pmatrix}
    0 & 0 & k_1 & 0 & 0\\
    k_2 & 0 & 0 & 0 & \bar{k}_2\\
    0 & k_3 & 0 & 0 & 0\\
    0 & k_4 & 0 & 0 & 0\\
    0 & 0 & 0 & k_5 & 0\\
  \end{pmatrix},
  \label{c1:eq1.18}
\end{equation}
The row-dominance proposition shows that the behaviour of the system
depends on the ratio of $k_{3}$ and $k_{4}$. If $k_{3}>k_{4}$, the
fourth species goes extinct, which in turn causes the extinction of
all remaining species of hypercycle $2$--$4$--$5$. In the opposite
case, hypercycle $1$--$2$--$3$ perishes.
Thus, of the two competing hypercycles, only one survives.

\begin{figure}[ht]
  \centering
  \includegraphics[width=0.65\textwidth]{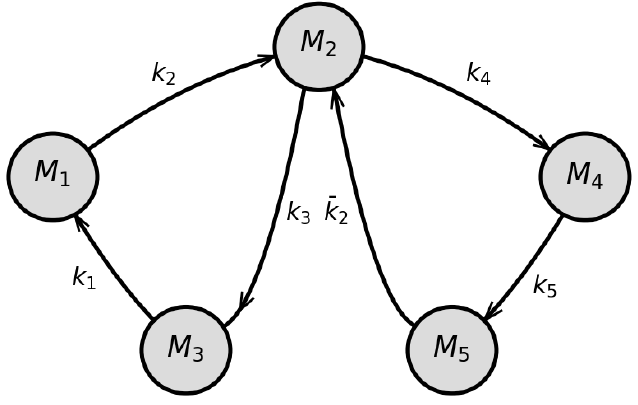}
  \caption{Graph of hypercyclic replication with matrix $\mathbf{A}$
           \eqref{c1:eq1.18}.}
  \label{c1:fig1.4}
\end{figure}

This conclusion extends to any number of hypercycles without common
species and with the same mean fitness. The argument rests on the
observation that interactions among several such hypercycles can be
described by an autocatalytic replicator system, with each species
representing one hypercycle. Consequently, depending on initial
conditions, at most one hypercycle survives \cite{c1:Hofbauer2002}: 
two distinct hypercycles cannot stably coexist. One hypercycle may,
however, supersede another if they share species in common,
inheriting those species from its predecessor. This unbranched mode
of evolution is consistent with the hypothesis of prebiotic
evolution, in which a common ancestral molecule could have developed
sequentially into a complex self-replicating system such as an RNA
molecule. For a detailed analysis of hypercycle evolution in this framework, see also \cite{c1:Bratus2018}.

An essential shortcoming of the hypercycle system, however, is its
vulnerability to parasitic species.  If a species is introduced that
exploits the resources of the system but contributes nothing in
return (an egoist), the system typically collapses
(Fig.~\ref{c1:fig1.5}).

\begin{figure}[ht]
  \centering
  \includegraphics[width=0.5\textwidth]{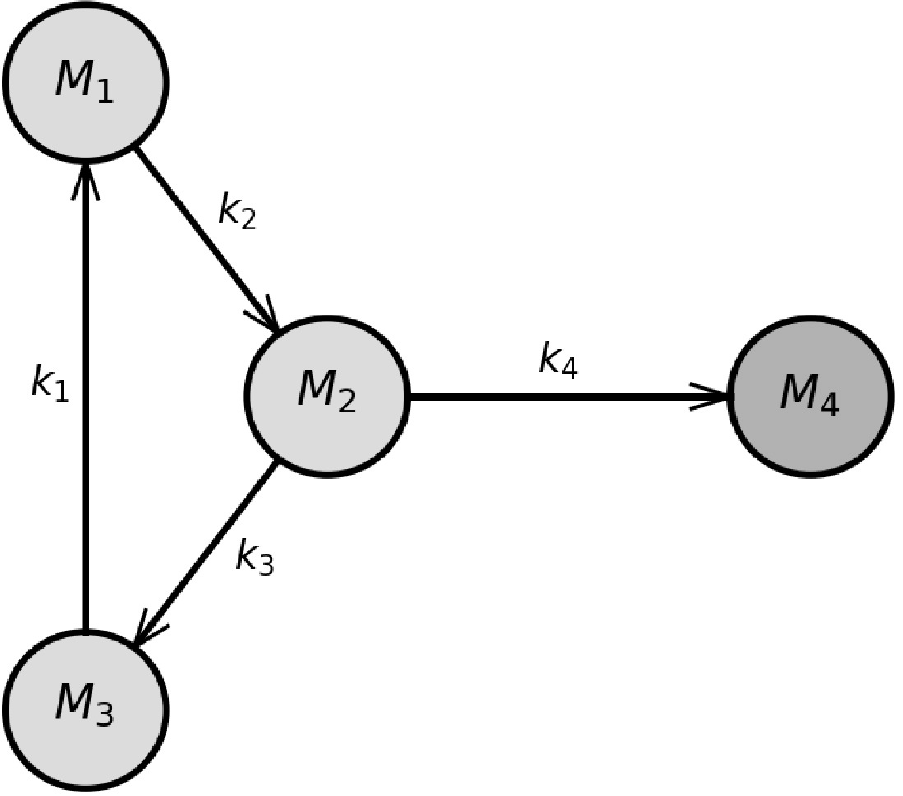}
\caption{Interaction graph of a hypercycle $1$--$2$--$3$ with a 
	parasitic species~$4$. }
  \label{c1:fig1.5}
\end{figure}

The matrix $\mathbf{A}$ is singular, which precludes the existence 
of an interior equilibrium $\bar{\mathbf{u}}\in\inside S_{n}$. 
Consequently, all four species cannot stably coexist. Species~$1$ 
is catalysed by species~$3$ (rate $k_1$), species~$2$ by 
species~$1$ (rate $k_2$), and species~$3$ by species~$2$ (rate 
$k_3$); species~$4$ is catalysed by species~$2$ (rate $k_4$) but 
contributes nothing to the cycle. If $k_3>k_4$, the parasite goes 
extinct and hypercycle $1$--$2$--$3$ survives; if $k_3<k_4$, 
species~$3$ is eliminated, bringing down the entire hypercycle 
with it.

This vulnerability can be remedied by allowing evolutionary 
modification of the entries of matrix $\mathbf{A}$, as shown in 
later chapters.

\section{Hypercycles of Higher Order and Other Replicator Systems}
\label{c1:section:1.4}

In hypercyclic systems of order $s$, the catalysis of species $i$
is carried out by species $i-1, i-2, \ldots, i-s$. Such systems
generalise the standard hypercycle, which corresponds to $s=1$.
We refer to them as \emph{hypercycles of higher order}; the term
reflecting the niche nature of this generalisation. The study of
such systems is motivated by real biochemical processes.

Consider a hypercyclic system of order two.  The dynamical equation is
\begin{equation}
  \dot{u}_{i}=u_{i}(k_{i}k_{i-1}u_{i-1}u_{i-2}-f(t)),
  \quad \mathbf{u}\in S_{n},
  \label{c1:eq1.19}
\end{equation}
\[
  f(t)=\sum_{i=1}^{n}k_{i}k_{i-1}u_{i}u_{i-1}u_{i-2},\quad
  u_{0}=u_{n},\quad u_{-1}=u_{n-1},\quad k_{0}=k_{n},\quad
  i=\overline{1,n}.
\]

Introducing barycentric coordinates analogously to the standard
hypercycle, system \eqref{c1:eq1.19} reduces to the equivalent system
\begin{equation}
  \dot{v}_{i}=v_{i}(v_{i-1}v_{i-2}-f(t)),\quad \mathbf{v}\in S_{n},
  \label{c1:eq1.20}
\end{equation}
\[
  f(t)=\sum_{i=1}^{n}v_{i}v_{i-1}v_{i-2},\quad v_{0}=v_{n},\quad
  v_{-1}=v_{n-1},\quad i=\overline{1,n}.
\]

\begin{proposition}
For odd $n\geqslant 5$, the second-order hypercycle system has a unique
equilibrium $u_{i}=\frac{1}{n}$, $i=\overline{1,n}$, which is
asymptotically stable for $n=5$ and unstable for $n>5$.
\end{proposition}

\begin{proof}
Interior equilibria are determined by
$\bar{u}_{i-1}\bar{u}_{i-2}=\bar{f}$.
For odd $n$, the unique solution is
$\bar{u}_{1}=\ldots=\bar{u}_{n}=\frac{1}{n}$.
The Jacobian at this point is
\[
  \mathbf{J}(\bar{\mathbf{u}})=-\frac{1}{n^{3}}
  \begin{pmatrix}
    3 & 3 & \cdots & 3-n & 3-n\\
    3-n & 3 & \cdots & 3 & 3-n\\
    \vdots & & \ddots & & \vdots\\
    3 & 3 & \cdots & 3-n & 3\\
  \end{pmatrix},
\]
which is again a circulant.  Its eigenvalues are
\[
  \lambda_{k}=-\frac{1}{n^{3}}
  \bigl(3+3r_{k}+\cdots+(3-n)r_{k}^{n-2}+(3-n)r_{k}^{n-1}\bigr),
\]
where $r_{k}=\exp\!\Bigl(\dfrac{2\pi}{n}ki\Bigr)$, 
$k=\overline{0,n-1}$, and $i$ is the imaginary unit.
Direct calculations show that for $n=5$ all eigenvalues have strictly
negative real parts, while for $n>5$ there are always eigenvalues with
positive real parts.  
\end{proof}

\begin{remark}
	For the standard hypercycle it has been proved that a stable limit 
	cycle exists for $n>4$ \cite{c1:Hofbauer1998}. This result does not 
	directly apply to second-order hypercycles; however, numerical 
	simulations suggest that a stable limit cycle may also exist for 
	$n>5$ in that case.
\end{remark}

The second-order hypercycle possesses an evolutionary variability
property, proven by the row-dominance theorem as for the standard
hypercycle.

We next consider the replicator system that may be described 
figuratively as an ``anthill'' or ``beehive'', in reference to 
the character of species interactions. Species~$0, 1, \ldots, n-1$ 
form a hypercycle, each additionally catalysed by species~$n$, 
which plays the role of the \emph{queen}. In turn, species~$n$ 
is catalysed by all the remaining members of the hypercycle.
The interaction graph is shown in Fig.~\ref{c1:fig1.6}.

\begin{figure}[ht]
  \centering
  \includegraphics[width=0.6\textwidth]{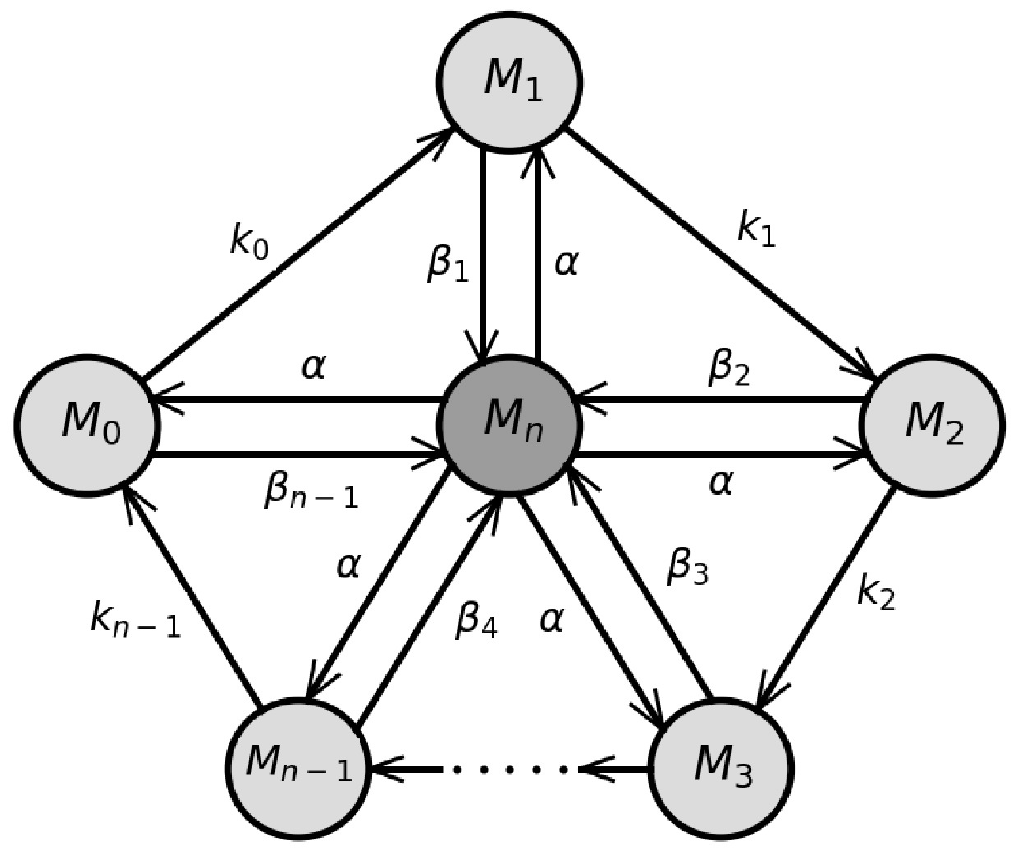}
\caption{Interaction graph of the ``anthill'' replicator system.}
  \label{c1:fig1.6}
\end{figure}

The state equations are
\begin{equation}
  \dot{u}_{i}=u_{i}\bigl(\alpha u_{n}+k_{i}u_{i-1}-f(t)\bigr),
  \quad i=\overline{0,n-1},  \label{c1:eq1.21}
\end{equation}
\begin{equation*}
  \dot{u}_{n}=u_{n}\Bigl(\sum_{i=0}^{n-1}\beta_{i}u_{i}-f(t)\Bigr),
  \quad \mathbf{u}\in S_{n+1},
\end{equation*}
\[
  f(t)=\alpha u_{n}\sum_{i=0}^{n-1}u_{i}
       +\sum_{i=0}^{n-1}k_{i}u_{i}u_{i-1}
       +u_{n}\sum_{i=0}^{n-1}\beta_{i}u_{i},\]
       \[
  \alpha,\beta_{i},k_{i}>0,\quad k_{0}=k_{n}, u_{-1} = u_{n-1}.
\]

The system has a unique interior equilibrium
$\bar{\mathbf{u}}\in\inside S_{n}$, a necessary condition for
permanence \cite{c1:Hofbauer2003} (where $n>2$).  

From the equilibrium equations \eqref{c1:eq1.21} it follows that
\[
k_{i}\bar{u}_{i-1} = \bar{f} - \alpha\bar{u}_{n},
\quad i=\overline{0,n-1}.
\]
Therefore
\[
k_{1}\bar{u}_{0}=k_{2}\bar{u}_{1}=\cdots=k_{0}\bar{u}_{n-1},
\]

\[
  \bar{u}_{i}=\frac{k_{1}}{k_{i+1}}\bar{u}_{0},\quad
  i=\overline{1,n-1}, k_n = k_0.
\]
Hence,
$$
\alpha \bar{u}_{n} + k_{1}\bar{u}_{0} = \beta_{0}\bar{u}_{0} + \bar{u}_{0}\sum\limits_{j = 1}^{n - 1}\beta_{j}\dfrac{k_{1}}{k_{j + 1}}, \quad k_{n} = k_{0}.
$$
Since $\bar{u}_{n} = 1 - \sum\limits_{j = 0}^{n - 1}\bar{u}_{j}$, we have
\[
  \bar{u}_{0}=\alpha\Biggl[
    \Bigl((\alpha+1)\sum_{j=1}^{n-1}\frac{1}{k_{j+1}}-1\Bigr)k_{1}
    +\alpha+\beta_{0}\Biggr]^{-1}>0, \quad i=\overline{1,n-1}, k_n = k_0.
\]

\begin{proposition}
Let $k_{m}=\min\{k_{0},\ldots,k_{n-1}\}$,
$k_{M}=\max\{k_{0},\ldots,k_{n-1}\}$,
$\beta_{m}=\min\{\beta_{0},\ldots,\beta_{n-1}\}$,
$\beta_{M}=\max\{\beta_{0},\ldots,\beta_{n-1}\}$.
If the conditions
\begin{equation}
  k_{M}<\beta_{m},\quad
  \alpha+\beta_{m}>\frac{k_{m}}{n}>\beta_{M},\quad
  n=3,4,\ldots,N
  \label{c1:eq1.22}
\end{equation}
are satisfied, then system \eqref{c1:eq1.21} is permanent.
\end{proposition}

\begin{proof}
Consider the function
\[
  \Phi(\mathbf{u})
  =\ln\!\prod_{i=0}^{n-1}\Bigl(u_{i}(t)\Bigr)^{\frac{1}{n}}-\ln u_{n}(t).
\]
Then
\[
  \dot{\Phi}(\mathbf{u})
  =\alpha u_{n}+\frac{1}{n}\sum_{i=0}^{n-1}k_{i}u_{i-1}
   -\sum_{i=0}^{n-1}\beta_{i}u_{i}.
\]
Using the bounds
\begin{equation}\label{c1:eq1.23}
	\sum\limits_{i = 0}^{n - 1}k_{i}u_{i - 1} \geqslant k_{m}\sum\limits_{i = 0}^{n - 1}u_{i} = k_{m}(1 - u_{n}),
\end{equation}
\begin{equation*}
	\sum\limits_{i = 0}^{n - 1}\beta_{i}u_{i} \leqslant \beta_{M}\sum\limits_{i = 0}^{n - 1}u_{i} = \beta_{M}(1 - u_{n}). \\
\end{equation*}
If \eqref{c1:eq1.22} holds then
$\dot{\Phi}(\mathbf{u})\geqslant\delta_{0}>0$
where $\delta_{0}=k_{m}/n-\beta_{M}>0$, so
\begin{equation}
  \prod_{i=0}^{n-1}\Bigl(u_{i}(t)\Bigr)^{\frac{1}{n}}
  \geqslant Ce^{\delta_{0}t}u_{n}(t).
  \label{c1:eq1.24}
\end{equation}
Consider the function $S(t)=\sum_{i=0}^{n-1}u_{i}(t)=1-u_{n}(t)$.
Its time derivative satisfies
\begin{multline*}
	\dot{S}(t) = \alpha S(1-S) - \alpha(1-S)S^{2}
	+ (1-S)\sum_{i=0}^{n-1}k_{i}u_{i}u_{i-1}
	- S(1-S)\sum_{i=0}^{n-1}\beta_{i}u_{i}.
\end{multline*}
Using the bounds \eqref{c1:eq1.23} together with
$\sum_{i=0}^{n-1}k_{i}u_{i}u_{i-1}\leqslant k_{M}S$, we obtain
\[
\dot{S}(t)\leqslant(\alpha+\beta_{m})S(S-1)(S-r^{2}),
\]
where $r^{2}=(k_{M}+\alpha)/(\beta_{m}+\alpha)<1$ by condition
\eqref{c1:eq1.22}. By the comparison theorem \cite{c1:Tihonov1985},
\[
S(t)\leqslant\max\{r^{2},\phi^{2}\},
\]
where $\phi^{2}=S(0)=\sum_{i=0}^{n-1}u_{i}(0)<1$, so
$u_{n}(0)=1-\phi^{2}>\varepsilon_{0}>0$.
Therefore
\[
u_{n}(t)=1-S(t)\geqslant\min\{1-r^{2},\,1-\phi^{2}\}>0.
\]

\end{proof}

Replicator systems are studied not only by theoreticians --- 
mathematicians and biologists --- but have also been realised 
in laboratory experiments with genuine biochemical reactions. 
In \cite{c1:Lincoln2009}, a two-element hypercyclic reaction was 
constructed experimentally. In \cite{c1:Vaidya2012}, a replicator 
system of six RNA macromolecule species was demonstrated; its 
interaction graph is shown in Fig.~\ref{c1:fig1.7}.

\begin{figure}[ht]
  \centering
  \includegraphics[width=0.7\textwidth]{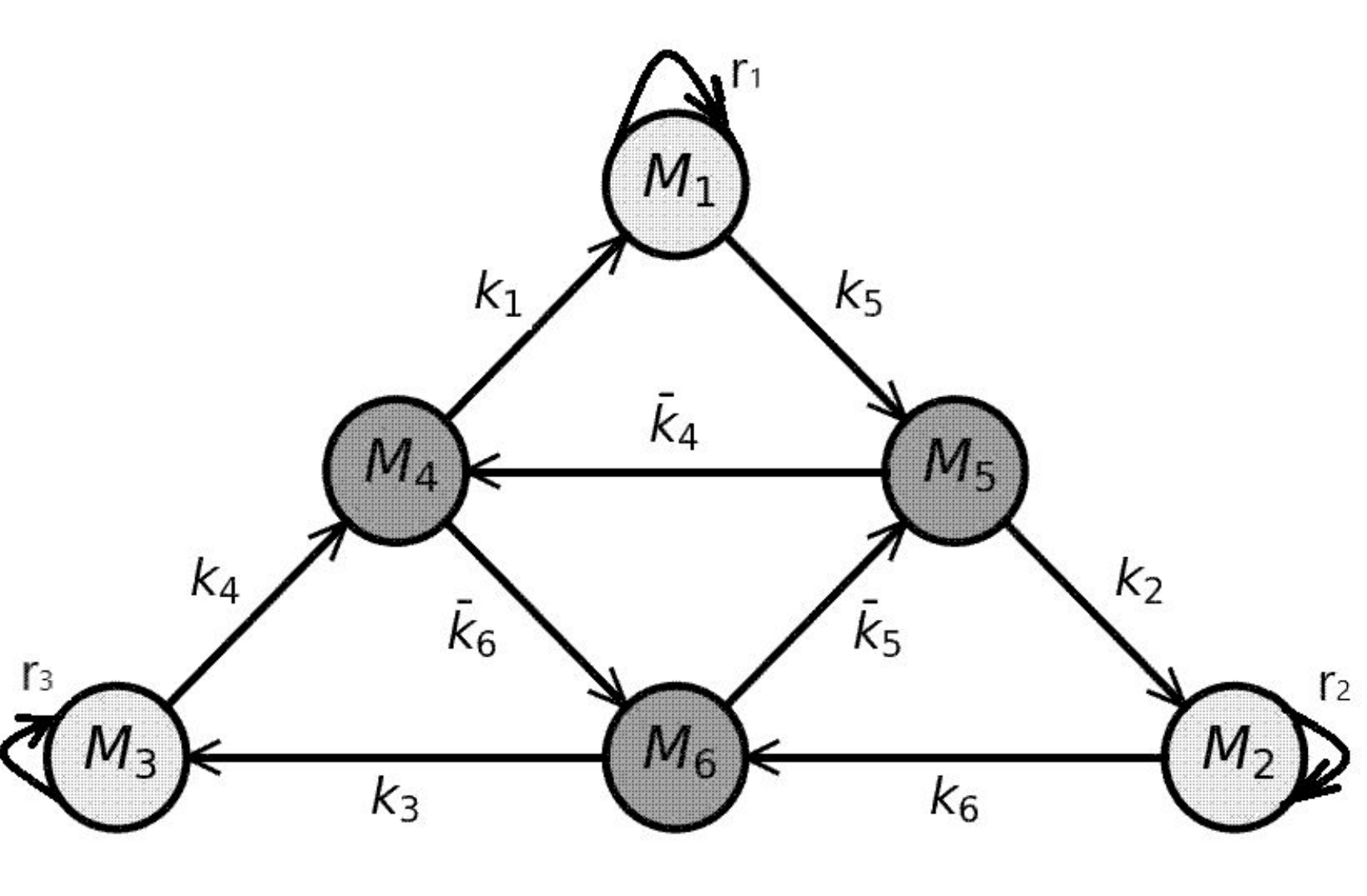}
  \caption{Graph of interactions among six RNA molecules.}
  \label{c1:fig1.7}
\end{figure}

Here elements $4$–$5$–$6$ form a hypercycle, while elements $1$–$2$–$3$,
in addition to participating in the hypercycle, also possess
autocatalytic replication properties.  The state equations are
\begin{equation}
  \begin{aligned}
    &\dot{u}_{1}=u_{1}(r_1 u_1+k_1 u_4-f(t)),\\
    &\dot{u}_{2}=u_{2}(r_2 u_2+k_2 u_5-f(t)),\\
    &\dot{u}_{3}=u_{3}(r_3 u_3+k_3 u_6-f(t)),\\
    &\dot{u}_{4}=u_{4}(k_4 u_3+\bar{k}_4 u_5-f(t)),\\
    &\dot{u}_{5}=u_{5}(k_5 u_1+\bar{k}_5 u_6-f(t)),\\
    &\dot{u}_{6}=u_{6}(k_6 u_2+\bar{k}_6 u_4-f(t)),\\
    &\mathbf{u}\in S_{6},\quad r_{i},k_{i},\bar{k}_{i}>0.
  \end{aligned}
  \label{c1:eq1.25}
\end{equation}

The mean fitness of the system is
\[
  f(t)=\sum_{i=1}^{3}r_i u_i^2
       +k_1 u_1 u_4+k_2 u_2 u_5+k_3 u_3 u_6
       +\bar{k}_4 u_4 u_5+\bar{k}_5 u_5 u_6+\bar{k}_6 u_6 u_4.
\]

Experiments reported in \cite{c1:Vaidya2012} confirmed that the dynamics
of this system are analogous to those of a permanent replicator system.
The phase portrait of system \eqref{c1:eq1.25} is shown in
Fig.~\ref{c1:fig1.8} for parameter values
$r_1=r_2=r_3=-0.3$, $k_1=k_2=k_3=0.1$,
$k_4=k_5=k_6=0.4$, $\bar{k}_4=\bar{k}_5=\bar{k}_6=0.05$.

\begin{figure}[ht]
  \centering
  \includegraphics[width=1.0\textwidth]{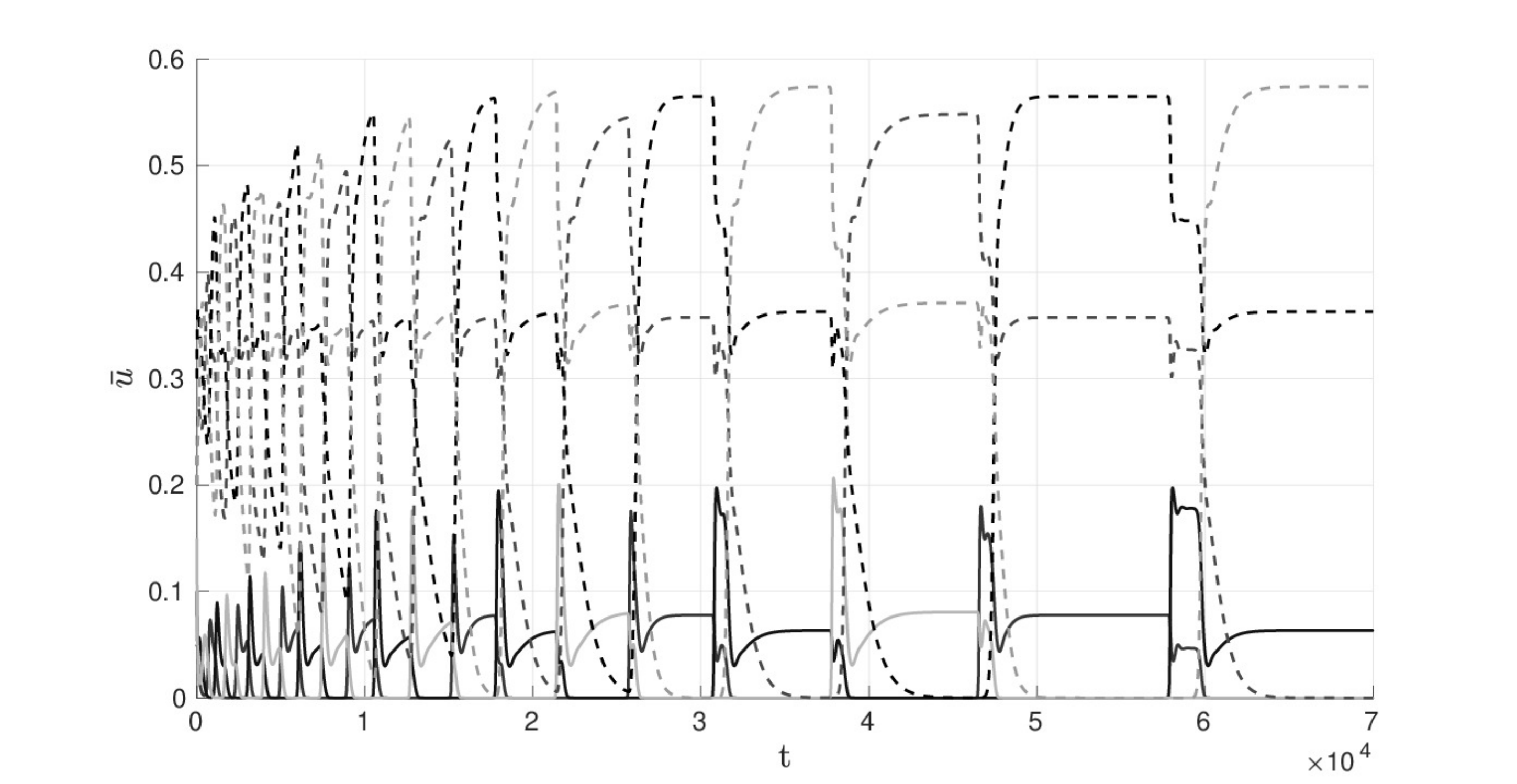}
\caption{Phase portrait of system~\eqref{c1:eq1.25} (six interacting 
	RNA molecules: species~$1$--$3$ undergo both autocatalytic 
	replication and hypercyclic catalysis by species~$4$--$6$, which 
	form a pure hypercycle).}
  \label{c1:fig1.8}
\end{figure}

\section{Eigen and Crow--Kimura Replicator Models}
\label{c1:section:1.5}

The Darwinian property of natural selection is described
mathematically through \emph{fitness coefficients}
(recall that the entire collection of such coefficients is called
the \emph{fitness landscape}).  Heredity is represented explicitly
by the general replicator equations.  It was shown that the hypercycle
equation inherently satisfies the variability property.  On the other
hand, it is important to incorporate variability explicitly in
mathematical models.  This is usually done by means of
\emph{mutation probabilities} or \emph{mutation intensities}.
The general framework yields the so-called \emph{quasispecies model}
(the origin of this term will be explained below), which appears in
two closely related but formally distinct forms.

When both natural selection and mutation occur simultaneously, the
sequence of events must be described carefully.  In the simplest case,
time is assumed to be discrete and generations non-overlapping —
that is, all individuals in the population reproduce simultaneously
and die immediately afterwards.

Suppose our population has $l$ distinct \emph{types} of individuals.
Denote the count of each type at time $t$ by $n_i(t)$,
$i=1,\ldots,l$, and the fitness of each type by $w_i\geqslant 0$.
These fitnesses are called \emph{Wrightian fitnesses}.

Since time is discrete and generations non-overlapping, population
growth follows the linear equations
\begin{equation}
  n_i(t+1)=w_i n_i(t),\quad i=1,\ldots,l.
  \label{c1:eq1.26}
\end{equation}
These are precisely the independent-replication equations in discrete
time.  Denoting
$\bs{n}(t)=(n_1(t),\ldots,n_l(t))^\top$,
$\bs{w}=(w_1,\ldots,w_l)^\top$,
$\bs{W}=\mathrm{diag}(w_1,\ldots,w_l)$, and passing to frequencies

$$\bs{p}(t)=\frac{\bs{n}(t)}{\sum_{i=1}^l n_i(t)},$$ one obtains
\begin{equation}
  \bs{p}(t+1)=\frac{\bs{W}\bs{p}(t)}{\bar{w}(t)},
  \label{c1:eq1.27}
\end{equation}
where $\bar{w}(t)=(\bs{w},\bs{p}(t))$ is the \emph{mean fitness}:
$$
\bar{w}(t)=\sum\limits_{i=1}^l w_ip_i(t),
$$ 

The dynamics of \eqref{c1:eq1.27} is simple: the species with the highest
fitness tends to frequency 1 while all others die out, following
\[
  \frac{p_i(t+1)}{p_j(t+1)}
  =\frac{w_i}{w_j}\frac{p_i(t)}{p_j(t)}.
\]
The mean fitness of system~\eqref{c1:eq1.27} is also non-decreasing:
\[
\bar{w}(t+1)-\bar{w}(t)
=\frac{\sum_{i=1}^{l}(w_i-\bar{w}(t))^2p_i(t)}{\bar{w}(t)}
=\frac{\mathrm{Var}_t(\bs{w})}{\bar{w}(t)}\geqslant 0,
\]
where $\mathrm{Var}_t(\bs{w})$ is the variance of a random variable
taking values $w_i$ with probabilities $p_i(t)$.

Now suppose replication occurs with errors.  Let $q_{ij}\in[0,1]$
denote the probability that an individual of type $j$ produces an
offspring of type $i$.  Then, $q_{ii}=1-\sum_{\substack{j=1\\j\neq i}}^{l}q_{ij}$ is the probability of error-free replication. Accounting for possible
replication errors, the equations for absolute population counts
become 
\[
  n_i(t+1)=\sum_{j=1}^{l}w_j q_{ij}n_j(t),
\]
or in matrix form $$\bs{n}(t+1)=\bs{QW}\bs{n}(t),$$ 
where
$\bs{Q}=(q_{ij})$ is a stochastic mutation matrix.  Moving to frequencies:
\begin{equation}
  \bs{p}(t+1)=\frac{\bs{QW}\bs{p}(t)}{\bar{w}(t)},
  \label{c1:eq1.28}
\end{equation}
where $\bar{w}(t)$ is the mean fitness.
Equation \eqref{c1:eq1.28} is the \emph{quasispecies model in discrete time}.

In most real populations generations overlap, so a continuous-time 
analogue of~\eqref{c1:eq1.28} is needed. Deriving it correctly is 
less straightforward than it might appear: a direct attempt to 
describe replication with mutations in continuous time quickly runs 
into difficulties, since obtaining ordinary differential equations 
requires assuming that at most one elementary event occurs in any 
sufficiently short time interval.

To circumvent this difficulty, we separate
replication (treated as error-free) from mutation (occurring at random
moments during an individual's lifetime).  For absolute population counts:
\[
  \bs{\dot{n}}(t)=(\bs{M}+\bs{\mathcal{M}})\bs{n}(t),
\]
where $\bs{M}=\mathrm{diag}(m_1,\ldots,m_l)$ is the
\emph{Malthusian fitness landscape} (each $m_i$ is a replication rate, not
an absolute quantity), and $\bs{\mathcal{M}}=(\mu_{ij})$ is the
matrix of mutation rates with $\mu_{ii}=-\sum_{j\neq i}\mu_{ij}$.

Indeed, assuming that the probability of producing offspring in time 
$\Delta t$ equals $m_j\Delta t$, the probability of mutation to 
type $i$ equals $\mu_{ij}\Delta t$, and that at most one elementary 
event occurs in the interval $\Delta t$, we obtain
\begin{multline*}
	n_j(t+\Delta t) = m_j\Delta t\, n_j(t)
	+ \sum_{\substack{i=1\\i\neq j}}^{l}\mu_{ji}\Delta t\, n_i(t)
	+ \left(1-\sum_{\substack{i=1\\i\neq j}}^{l}
	\mu_{ij}\Delta t\right)n_j(t)
	+ o(\Delta t^{2}).
\end{multline*}
Dividing by $\Delta t$ and passing to the limit $\Delta t\to 0$
yields the required equation.

\medskip
Analogously to the frequency equation in discrete time, in
continuous time we obtain
\begin{equation}
  \dot{\bs{p}}(t)
  =\bigl(\bs{M}-\bar{m}(t)\bs{E}\bigr)\bs{p}(t)+\bs{\mathcal{M}}\bs{p}(t),
  \label{c1:eq1.29}
\end{equation}
where $\bar{m}(t)=(\bs{m},\bs{p}(t))$ is the mean Malthusian fitness, and  $\bs{E}$ is the 
identity matrix. 
 Model \eqref{c1:eq1.29} is
called the \emph{Crow--Kimura model}, since it was thoroughly analysed
in the textbook on theoretical population genetics by Crow and Kimura
\cite{c1:Crow1970}.
Models \eqref{c1:eq1.28} and \eqref{c1:eq1.29} are related. In 
particular, the mutation probabilities and intensities are connected 
by
\[
q_{ij}=\delta_{ij}+\mu_{ij}\Delta t, \quad \Delta t\to 0,
\]
where $\delta_{ij}$ is the Kronecker delta. Moreover, 
system~\eqref{c1:eq1.29} can be obtained as the limit of 
system~\eqref{c1:eq1.28} under the assumption of short non-overlapping 
generation times and weak mutations \cite{c1:Hofbauer1985}, with 
Wrightian and Malthusian fitnesses related by
\[
w_i=e^{m_i\Delta t}\approx 1+m_i\Delta t, \quad \Delta t\to 0.
\]

Both \eqref{c1:eq1.28} and \eqref{c1:eq1.29} are referred to in the 
modern literature as \textit{quasispecies models}. Note, however, 
that Eigen's original quasispecies model \cite{c1:Eigen1971} was 
written in a different form. Eigen considered a system of ordinary 
differential equations of the form
\begin{equation}
  \dot{\bs{p}}(t)=\bs{QW}\bs{p}(t)-\bar{w}(t)\bs{p}(t),
  \label{c1:eq1.30}
\end{equation}
a form that is difficult to derive rigorously from first principles;
to our knowledge, no such derivation from elementary processes
exists in the literature.
Its equilibrium $\hat{\bs{p}}$ satisfies
\[
  \bs{QW}\hat{\bs{p}}=\bar{w}\hat{\bs{p}}.
\]
The quantity $\bar{w}(t)$ is determined from the condition
$\bigl(\bs{\dot{p}}(t),\bs{I}\bigr)=0$, where
$\bs{I}=(1,1,\ldots,1)\in\R^{l}$. Using $q_{ii}=1-\sum_{\substack{j=1\\j\neq i}}^{l}q_{ij}$
we obtain
\[
\bar{w}(t)=\sum_{i=1}^{l}w_i p_i(t)=\bigl(\bs{w},\bs{p}(t)\bigr).
\]
For the Crow--Kimura case, the equilibrium satisfies
\begin{equation}
	(\bs{M}+\bs{\mathcal{M}})\hat{\bs{p}}=\bar{m}\hat{\bs{p}}.
	\label{c1:eq1.31}
\end{equation}

Before stating the main result of this chapter, we note a simple but
useful property of quasispecies systems. System~\eqref{c1:eq1.28} is
unchanged if all Wrightian fitnesses are multiplied by the same
positive constant, and system~\eqref{c1:eq1.29} is unchanged if the
same constant is added to all Malthusian fitnesses. This property
usually allows one to normalise the fitnesses in the most convenient
form for analysis. For instance, if in~\eqref{c1:eq1.29} the fitness
landscape vector $\bs{m}$ has the form
$(m_1,m_0,m_0,\ldots,m_0)^\top$, it is convenient to replace it by
$(m_1-m_0,0,0,\ldots,0)^\top$. Similarly, in~\eqref{c1:eq1.28} the
Wrightian fitnesses are usually scaled so that either the largest or
the smallest equals one.

An elementary yet fundamental fact is that the equilibria of these
systems almost always exist (that is, belong to the simplex $S_l$)
and are globally stable.

Recall that a matrix $\bs{A}$ is called \textit{positive} if all its
entries are positive, \textit{non-negative} if all entries are
non-negative, \textit{irreducible} if the corresponding directed
graph (with a directed edge from vertex $i$ to vertex $j$ whenever
$a_{ij}>0$) is strongly connected (i.e.\ for any pair of vertices
there exists a path connecting them), and \textit{primitive} if
there exists an integer $k$ such that $\bs{A}^k$ is positive. Every
positive matrix is primitive, every primitive matrix is irreducible,
and every irreducible matrix is non-negative; the converses do not
hold. The key property of primitive matrices is given by the
well-known \textit{Perron--Frobenius theorem}, which asserts, in
particular, that every primitive matrix has a dominant eigenvalue
$\lambda>0$ satisfying $\lambda>|\lambda_j|$ for all other
eigenvalues $\lambda_j$. Moreover, the algebraic and geometric
multiplicity of the dominant eigenvalue is one, and the
corresponding eigenvector can be chosen with all positive
components. Furthermore, for a primitive matrix $\bs{A}$,
\[
\lim_{k\to\infty}\frac{1}{\lambda^k}\bs{A}^k=\bs{v}\bs{w}^\top,
\]
where $\bs{v}$ and $\bs{w}$ are the right and left eigenvectors of
$\bs{A}$ corresponding to the dominant eigenvalue.

\begin{theorem}\label{c1:t1.2}
Suppose the matrices $\bs{QW}$ and $\bs{M}+\bs{\mathcal{M}}$ are
primitive (the latter possibly after adding the same positive constant
to all diagonal elements).  Then quasispecies systems \eqref{c1:eq1.28},
\eqref{c1:eq1.29}, \eqref{c1:eq1.30} always have a unique strictly positive
equilibrium $\hat{\bs{p}}$, which is globally stable in the simplex
$S_l$.  This equilibrium is the normalised positive eigenvector of
$\bs{QW}$ and $\bs{M}+\bs{\mathcal{M}}$ corresponding to the dominant
eigenvalue; the mean fitness at equilibrium equals this dominant
eigenvalue.
\end{theorem}

\begin{proof}
For \eqref{c1:eq1.28}, the absolute-count equation is linear:
$$\bs{n}(t)=(\bs{QW})^t\bs{n}(0).$$
  Primitivity of $\bs{QW}$
with dominant eigenvalue $\lambda$ and corresponding right/left
positive eigenvectors $\hat{\bs{p}}$, $\hat{\bs{q}}$ gives
\[
  \bs{p}(t)=\frac{(\bs{QW})^t\bs{n}(0)}{|(\bs{QW})^t\bs{n}(0)|_1}
  \to\frac{\lambda^t\hat{\bs{p}}\hat{\bs{q}}^{\top}\bs{n}(0)}
         {|\lambda^t\hat{\bs{p}}\hat{\bs{q}}^{\top}\bs{n}(0)|_1}
  =\hat{\bs{p}},
\]
where $\hat{\bs{q}}^{\top}\bs{n}(0)>0$.
Hence $\hat{\bs{p}}$ satisfies $\bs{QW}\hat{\bs{p}}=\bar{w}\hat{\bs{p}}$,
so $\bar{w}=\lambda$ at equilibrium.
For \eqref{c1:eq1.30}, primitivity of $\bs{M}+\bs{\mathcal{M}}$
is required. 
\end{proof}

\medskip
We now explain the origin of the term \emph{quasispecies}.  Let $\bs{A}=\bs{QW}$ and
suppose $\bs{A}$ has only simple real eigenvalues, so there exists
$\bs{T}$ with $\bs{\Lambda}=\bs{TAT}^{-1}$ diagonal.
Substituting $\bs{q}=\bs{Tp}$ into \eqref{c1:eq1.30} gives
\[
  \dot{q}_i=(\lambda_i-\bar{w}(t))q_i,\quad i=1,\ldots,l, 
\]
where $\bar{w}(t)$ remained unchanged as $\bs T^{-1}\bar{w}(t)\bs I\bs T=\bar{w}(t)\bs I$ and constant sum of $q_i$: $\bar{w}(t)=\sum\limits_{j=1}^l\lambda_jq_j(t)$.

This is formally the independent-replication equation, in which only
the ``species'' $q_i$ with the largest $\lambda_i$ survives.
However, $q_i$ here is not a species but a \emph{linear combination
of frequencies} $p_i$.  Such a cloud of representatives of different
individual types was termed a \emph{quasispecies} by Eigen.
In the mathematical model of independent replication with mutations,
selection therefore acts not between individual types but between
different quasispecies; the unit of selection is not a unique type
but rather their ensemble.

We note that the mere existence of a globally stable equilibrium in
the quasispecies model gives no quantitative information that could
be used to compare model predictions with real data.  For applications,
one needs methods to find, for given matrices $\bs{QW}$ and
$\bs{M}+\bs{\mathcal{M}}$, the leading eigenvalue and the
corresponding eigenvector.  The difficulty, however, is that these
matrices typically have very large dimension, which prevents effective
numerical computation even on modern computers.

\section{Sequence Space and the Error Threshold}
\label{c1:section:1.6}

In \S\ref{c1:section:1.5}, the quasispecies models were formulated in
full generality; their analysis reduces to finding the leading
eigenvalue and eigenvector of the problems
\begin{equation}
  \bs{QWp}=\bar{w}\bs{p},\quad \bar{w}=\sum_{i=1}^{l}w_ip_i,
  \label{c1:eq1.32}
\end{equation}
for discrete time with Wrightian fitnesses (the classical Eigen model),
or
\begin{equation}
  (\bs{M}+\bs{Q}_N)\bs{p}=\bar{m}\bs{p},\quad
  \bar{m}=\sum_{i=1}^{l}m_ip_i,
  \label{c1:eq1.33}
\end{equation}
for continuous time with Malthusian fitnesses (the Crow--Kimura model).
Here $\bs{W}$ and $\bs{M}$ are diagonal with entries
$w_1,\ldots,w_l$ and $m_1,\ldots,m_l$ respectively (these matrices
define the fitness landscapes, which we also denote $\bs{w}$ and
$\bs{m}$); $\bs{Q}$ is stochastic with row sums equal to one; the
off-diagonal entries of $\bs{Q}_N$ are the mutation intensities,
and each row also sums to zero.

In full generality, the problem is formulated as follows: given
matrices $\bs{Q}$ and $\bs{W}$ (or $\bs{M}$ and $\bs{Q}_N$), find
$\bar{w}$ and $\bs{p}$ (or $\bar{m}$ and $\bs{p}$). In this
generality the problem is too unconstrained to yield concrete
results; progress requires specifying the structure of these
matrices. Eigen's key contribution was to propose a specific
structure for $\bs{Q}$ and $\bs{Q}_N$, determined by the
biological observation that the individuals of the population are
\emph{sequences} of fixed length $N$.

The biological motivation is straightforward.
Eigen originally formulated the quasispecies model in the context of
the origin of life.  One of the most plausible molecules that could
have driven this process is RNA, which is essentially a sequence
(chain) of four nucleotides.  For simplicity we assume a two-letter
alphabet ($0$ and $1$), though all subsequent results generalise to
an arbitrary alphabet size.

We begin with the Eigen model \eqref{c1:eq1.32}.
Suppose individuals are binary sequences of fixed length $N$,
so the number of distinct types is $l=2^N$.
For $N=3$, for example, the population consists of eight types:
\begin{align*}
  \sigma_0=[000],\;\sigma_1=[001],\;\sigma_2=[010],\;\sigma_3=[011],\;
  \sigma_4=[100],\;\sigma_5=[101],\;\sigma_6=[110],\;\sigma_7=[111].
\end{align*}

The set of all sequences of length $N$ can be endowed with a metric
structure by defining the Hamming distance $H(\sigma_i,\sigma_j)=H_{ij}$:
\[
  H(\sigma_i,\sigma_j)=\sum_{k=1}^{N}|\sigma_i(k)-\sigma_j(k)|.
\]
Geometrically, the set of binary sequences of length $N$ equipped with
the Hamming distance forms an $N$-dimensional hypercube
(Fig.~\ref{c1:fig1.9}).

\begin{figure}[ht]
  \centering
  \includegraphics[width=1.0\textwidth]{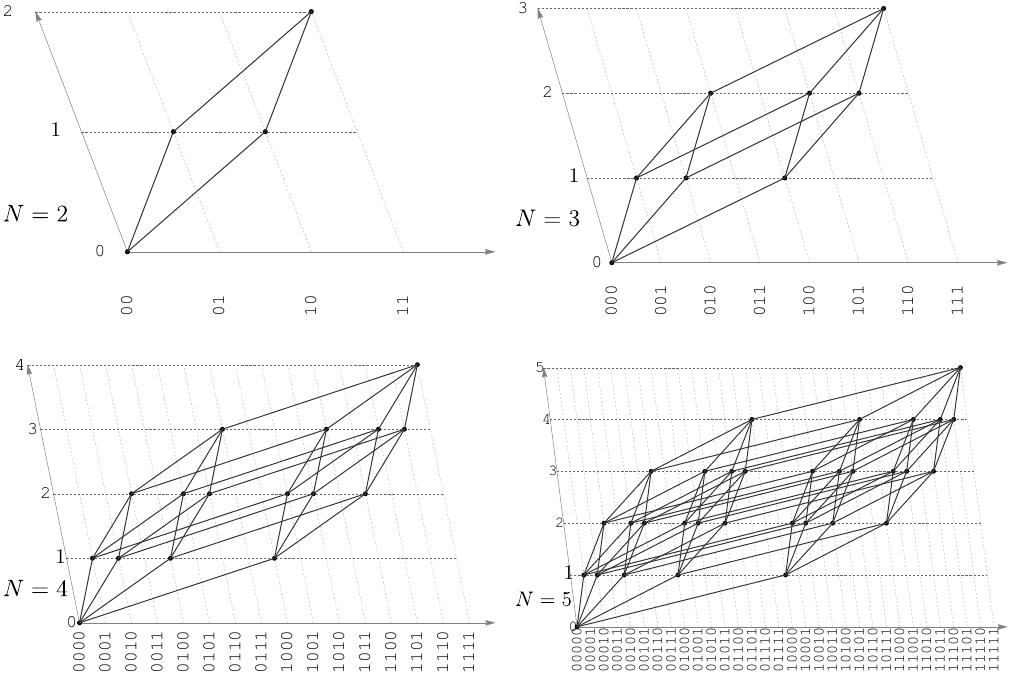}
  \caption{Sequence spaces for binary sequences of length $N=2,3,4,5$.
           Vertices of the hypercubes are the sequences, enumerated
           along the horizontal axis.  Hamming distance between two
           sequences equals the minimum number of edges connecting them.}
  \label{c1:fig1.9}
\end{figure}

Assuming mutations at each position occur independently and with equal
probability $q\in[0,1]$, the entry $q_{ij}$ of the mutation matrix
$\bs{Q}$ is
\begin{equation}
  q_{ij}=(1-q)^{N-H_{ij}}q^{H_{ij}},\quad i,j=0,\ldots,2^{N-1}.
  \label{c1:eq1.34}
\end{equation}
Indeed, for $\sigma_j$ to mutate into $\sigma_i$, exactly $H_{ij}$
positions must mutate (probability $q^{H_{ij}}$) and the remaining
$N-H_{ij}$ must not (probability $(1-q)^{N-H_{ij}}$).
One checks that $\bs{Q}$ is stochastic.

With this structure, $\bs{Q}$ is a function of a single scalar
parameter $q$, making analysis of \eqref{c1:eq1.32} considerably more
tractable: for a given fitness landscape $\bs{W}$, one seeks the
functions $q\mapsto\bar{w}(q)$ and $q\mapsto\bs{p}(q)$.

The dependence of $\bar{w}$ on $q$ for system \eqref{c1:eq1.32} was
investigated in \cite{c1:bratus2014}.
A typical graph is shown in Fig.~\ref{c1:fig1.10} for $N=3$,
$\bs{W}=\mathrm{diag}(10,3,3,2,3,2,2,1)$.

\begin{figure}[ht]
  \centering
  \includegraphics[width=0.6\textwidth]{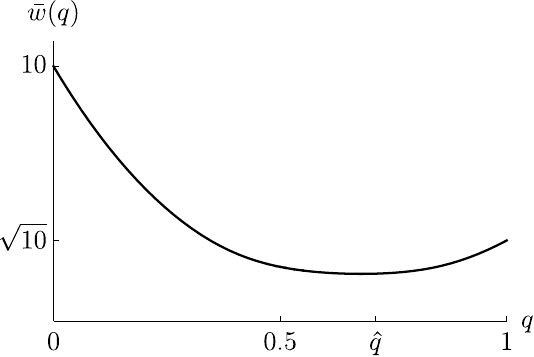}
  \caption{Typical dependence of mean fitness $\bar{w}$ on the
           single-site mutation probability $q$ in system \eqref{c1:eq1.32},
           for $N=3$, $\bs{W}=\mathrm{diag}(10,3,3,2,3,2,2,1)$.
           Here $\bar{w}(0)=10$, $\bar{w}(1)=\sqrt{10}$.}
  \label{c1:fig1.10}
\end{figure}

An analogous sequence space is introduced for the Crow--Kimura model.
Since time is continuous, only one elementary event can occur in a
sufficiently short time interval; hence only single-position mutations
(transitions to neighbouring vertices of the hypercube) are possible.
This gives
\begin{equation}
  \mu_{ij}=\begin{cases}
    \mu, & H_{ij}=1,\\
    -N\mu, & H_{ij}=0,\\
    0, & H_{ij}>1,  \end{cases}
  \label{c1:eq1.35}
\end{equation}
where $\mu$ is the mutation intensity per position per unit time.
Analogous results hold for the Crow--Kimura model \cite{c1:bratus2014}.

A natural question arises: can one find exact solutions of
\eqref{c1:eq1.32} and \eqref{c1:eq1.33} at least for some fixed fitness
landscapes on the given sequence space?
  The answer is positive:
the only known non-trivial example (i.e.\ with a fitness landscape
distinct from a constant) was obtained in \cite{c1:bratus2014}.
In most cases one must rely on numerical computations, which are
hampered by the exponentially large dimension: even the simplest
viruses have sequences of several thousand nucleotides, so for
$N=1000$ the problems \eqref{c1:eq1.32}–\eqref{c1:eq1.33} have dimension
$l=2^{1000}$, precluding any numerical approach.

A partial solution was proposed in \cite{c1:Swetina1982} through the
concept of \emph{single-peak fitness landscapes} (also called
\emph{permutation-invariant} landscapes): the fitness of a sequence
depends only on the Hamming distance from a reference sequence
$\sigma_0$ (the ``master sequence''), not on its precise composition.
Sequences can thus be grouped into \emph{classes}, where class $k$
consists of all sequences at Hamming distance $k$ from $\sigma_0$.
There are $C_N^k=\binom{N}{k}$ types in class $k$ and $N+1$ classes
in total, reducing the problem dimension from $2^N$ to $N+1$.

Another approach is the transition to a continuum of
types \cite{c1:Saakian2004,c1:Saakian2006}.

Under the permutation-invariance assumption, the mutation matrix for
the Crow--Kimura model takes the tridiagonal form
\begin{equation}
  \nu_{ij}=\begin{cases}
    (N-j)\mu, & i=j+1,\\
    j\mu, & i=j-1,\\
    -N\mu, & i=j,\\
    0, & |i-j|>1,
  \end{cases}
  \label{c1:eq1.36}
\end{equation}
and the Crow--Kimura problem \eqref{c1:eq1.33} reduces to
\begin{equation}
  (\bs{M}+\bs{Q}_N)\bs{p}=\bar{m}\bs{p},
  \label{c1:eq1.37}
\end{equation}
where $\bs{M}=\mathrm{diag}(m_0,\ldots,m_N)$ and
\begin{equation}
  \bs{Q}_N=\mu
  \begin{bmatrix}
    -N & 1 & 0 & \cdots & 0\\
    N & -N & 2 & \cdots & 0\\
    0 & N-1 & -N & \cdots & 0\\
    \vdots & & \ddots & \ddots & \vdots\\
    0 & \cdots & 2 & -N & N\\
    0 & \cdots & 0 & 1 & -N\\
  \end{bmatrix}.
  \label{c1:eq1.38}
\end{equation}

For the Eigen model \eqref{c1:eq1.32} the reduction to classes is
analogous, and the problem becomes
\begin{equation}
  \bs{WR}\bs{p}=\bar{w}\bs{p},
  \label{c1:eq1.39}
\end{equation}
where $\bs{W}=\mathrm{diag}(w_0,\ldots,w_N)$ and
$\bs{R}=(r_{ij})$ is the transition-probability matrix between classes
\cite{c1:Nowak1989}:
\begin{equation}
  r_{ij}=\sum_{a=j+i-N}^{\min\{i,j\}}
  \binom{j}{a}\binom{N-j}{i-a}q^N\!\left(\frac{1-q}{q}\right)^{i+j-2a},
  \quad i,j=0,\ldots,N.
  \label{c1:eq1.40}
\end{equation}

For numerical experiments, the \emph{single-peak} fitness landscape
\[
  \bs{W}=\mathrm{diag}(1+s,\,1,\ldots,1),\quad s>0
\]
is the simplest nontrivial case.
Figure~\ref{c1:fig1.11} shows numerical results for the Eigen model with
$\bs{W}=\mathrm{diag}(10,1,\ldots,1)$ and $N=5,10,50,100$.

\begin{figure}[ht]
  \centering
  \includegraphics[width=1.0\textwidth]{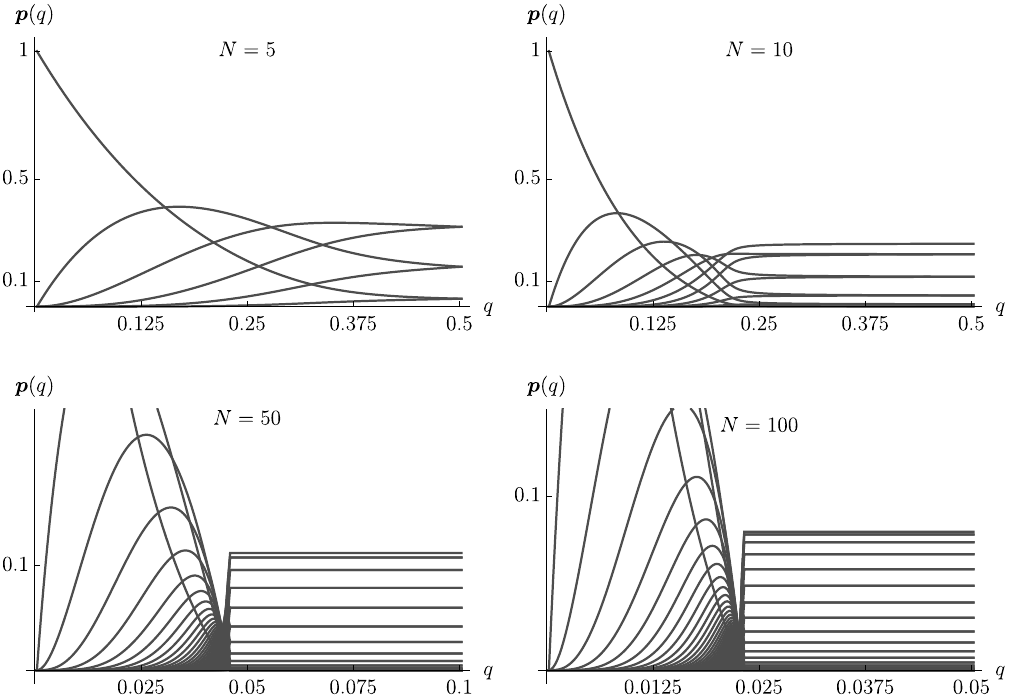}
  \caption{Equilibrium quasispecies distribution (class frequencies)
           as a function of single-site mutation probability $q$ in
           model \eqref{c1:eq1.32} for various sequence lengths.
           Fitness landscape: $\bs{W}=\mathrm{diag}(10,1,\ldots,1)$.
           Note the different scales on both axes for different $N$.}
  \label{c1:fig1.11}
\end{figure}

As $N$ increases, a striking qualitative phenomenon is observed:
the quasispecies distribution ceases to change after some critical
value of $q$.  Moreover, this fixed distribution closely approximates
the binomial distribution, implying that the type distribution becomes
nearly uniform — the population ``loses memory'' of the master sequence.
This phenomenon was named the \emph{error catastrophe} (or
\emph{error threshold}) by Eigen and Schuster, attracting enormous
interest in the quasispecies literature.  On the graphs the error
threshold appears as a sharp transition, especially pronounced for
large $N$.

The term ``error threshold'' reflects the fact that the stationary
quasispecies, after exceeding a critical value of $q$, ceases to carry
information about the fittest type and is no longer subject to natural
selection; thus the system effectively stops evolving.  This
corresponds to the practical cessation of evolution of the system
and may also be called the \emph{evolutionary threshold}.

An analogous phenomenon in physics is known as a \emph{phase
transition}: above a critical parameter value, the system undergoes
a change from chaotic to ordered behaviour, or vice versa.

\section{Stabilisation of the Leading Eigenvalue in the Crow--Kimura Model}
\label{c1:section:1.7}

To analyse this phenomenon we need additional information on the
spectrum and eigenvectors of the mutation matrix $\bs{Q}_N$
defined by \eqref{c1:eq1.38} \cite{c1:bratus2014_2}.

\begin{enumerate}
\item The eigenvalues of $\bs{Q}_N(\mu=1)$ (in decreasing order) are:
\[
  0;\;-2;\;-4;\;\ldots;\;-2N.
\]

\item Let $\bs{v}^k=(v_{0,k},v_{1,k},\ldots,v_{N,k})^{\top}$ be the
eigenvector corresponding to eigenvalue $q_k=-2k$, normalised so
that $v_{0,k}=1$, $k=0,1,\ldots,N$.

The entries of $\bs{v}^k$ have the following properties:
\begin{enumerate}
\item All $v_{i,k}$ ($i=0,1,\ldots,N$) are integers.

\item Symmetry:
\[
  v_{N-i,k}=(-1)^k v_{i,k},\quad
  v_{i,N-k}=(-1)^i v_{i,k},\quad
  i,k=0,1,\ldots,N.
\]

\item The first column of the matrix $\bs{V}$ formed from the
vectors $\bs{v}^k$, $k=0,1,\ldots,N$, consists of binomial
coefficients: $v_{i,0}=\binom{N}{i}$, $i=0,\ldots,N$.
The generating function for the $k$-th column is
\[
  p_k(t)=\sum_{i=0}^{N}v_{i,k}t^i=(1-t)^k(1+t)^{N-k},\quad
  k=0,\ldots,N.
\]
For example, for $N=6$:
\[
  \bs{V}=
  \begin{bmatrix}
    1 & 1 & 1 & 1 & 1 & 1 & 1\\
    6 & 4 & 2 & 0 & -2 & -4 & -6\\
    15 & 5 & -1 & -3 & -1 & 5 & 15\\
    20 & 0 & -4 & 0 & 4 & 0 & -20\\
    15 & -5 & -1 & 3 & -1 & -5 & 15\\
    6 & -4 & 2 & 0 & -2 & 4 & -6\\
    1 & -1 & 1 & -1 & 1 & -1 & 1\\
  \end{bmatrix}.
\]

\item The determinant and inverse of $\bs{V}$ are
\[
  \det\bs{V}=(-2)^{N(N+1)/2},\quad
  \bs{V}^{-1}=2^{-N}\bs{V},\quad
  \bs{VV}^{-1}=2^N\bs{E}.
\]
\end{enumerate}
\end{enumerate}

Consider equation \eqref{c1:eq1.37} and substitute $\bs{p}=\bs{Vx}$,
$\bs{x}=(x_0,x_1,\ldots,x_N)\in\R^{N+1}$.
Multiplying by $\bs{V}^{-1}$ and using
$\bs{V}^{-1}\bs{Q}_N\bs{V}=-2\bs{D}_N$,
$\bs{D}_N=\mathrm{diag}(0,1,2,\ldots,N)$, one obtains
\begin{equation}
  2^{-N}\bigl(\bs{VMV}-2\mu\bs{D}_N\bigr)\bs{x}=\bar{m}\bs{x}.
  \label{c1:eq1.41}
\end{equation}

This equation has a non-trivial solution if and only if
\begin{equation}
  \det\!\Bigl(2^{-N}(\bs{VMV})-2\mu\bs{D}_N-\bar{m}\bs{E}\Bigr)=0.
  \label{c1:eq1.42}
\end{equation}

As $\mu$ varies, the components of $\bs{x}$ and the eigenvalue
$\bar{m}$ are smooth functions of $\mu$ (by perturbation theory
for simple eigenvalues \cite{c1:Kato1976}).  Equation \eqref{c1:eq1.42}
defines a curve in the $(\mu,\bar{m})$ plane called the
\emph{characteristic curve}.

\begin{definition}
The smooth characteristic curve $\bar{m}(\mu)$ defined by
\eqref{c1:eq1.41} is said to admit \emph{limiting stabilisation} as
$\mu\to+\infty$ if there exists a constant $\bar{m}^*$ such that
\[
  \lim_{\mu\to+\infty}\bar{m}(\mu)=\bar{m}^*,\qquad
  \lim_{\mu\to+\infty}\bar{m}'(\mu)=0.
\]
\end{definition}

\begin{theorem}
If $\bar{m}(\mu)$ is a simple eigenvalue of \eqref{c1:eq1.37} for
$\mu>0$, then
\begin{equation}
  \bar{m}'(\mu)=-2(\bs{D}_N\bs{x}(\mu),\bs{y}(\mu)),
  \label{c1:eq1.43}
\end{equation}
where $\bs{x}(\mu)$ is the eigenvector of \eqref{c1:eq1.41} and
$\bs{y}(\mu)$ is the eigenvector of the adjoint problem
\begin{equation}
  \Bigl(2^{-N}(\bs{VMV}^{\top})-2\mu\bs{D}_N\Bigr)\bs{y}(\mu)
  =\bar{m}(\mu)\bs{y}(\mu),
  \label{c1:eq1.44}
\end{equation}
normalised by $(\bs{x}(\mu),\bs{y}(\mu))=1$.
\end{theorem}

\begin{proof}
Differentiability of $\bar{m}(\mu)$ follows from the simplicity of
the eigenvalue and perturbation theory \cite{c1:Kato1976}.
Perturbing $\mu\to\mu+\varepsilon\Delta\mu$ and expanding
\begin{equation}
  \bar{m}(\mu_\varepsilon)=\bar{m}(\mu)+\varepsilon\bar{m}'(\mu)\Delta\mu+o(\varepsilon),\quad
  \bs{x}(\mu_\varepsilon)=\bs{x}(\mu)+\varepsilon\bs{x}'(\mu)\Delta\mu+o(\varepsilon),
  \label{c1:eq1.46}
\end{equation}
substituting into \eqref{c1:eq1.41}, isolating linear terms, and
multiplying scalarly by $\bs{y}(\mu)$ (using
$(\bs{x}(\mu),\bs{y}(\mu))=1$) yields \eqref{c1:eq1.43}.
\hfill$\blacksquare$
\end{proof}

\begin{theorem}
For any matrix $\bs{M}=\mathrm{diag}(m_0,\ldots,m_N)$,
\begin{equation}
  \lim_{\mu\to+\infty}\bs{x}(\mu)=\bs{x}_\infty=2^{-N}(1,0,0,\ldots,0).
  \label{c1:eq1.48}
\end{equation}
\end{theorem}

A full proof is given in \cite{c1:bratus2014_2}.
We establish the following corollary.

\begin{corollary}
The limiting equilibrium frequency distribution in \eqref{c1:eq1.33} is
\begin{equation}
  \lim_{\mu\to+\infty}\bs{p}(\mu)
  =2^{-N}\bigl(C_0^N,C_1^N,\ldots,C_N^N\bigr),
  \label{c1:eq1.49}
\end{equation}
and the limiting mean fitness is
\begin{equation}
  \lim_{\mu\to+\infty}\bar{m}(\mu)
  =\bar{m}^*=2^{-N}\sum_{k=0}^{N}C_N^k m_k.
  \label{c1:eq1.50}
\end{equation}
\end{corollary}

\begin{proof}
From the normalisation condition $(\bs{x}(\mu),\bs{y}(\mu))=1$
it follows that
$\lim_{\mu\to+\infty}\bs{y}(\mu)=\bs{y}_\infty=2^N(1,0,\ldots,0)$.
Together with \eqref{c1:eq1.43}:
\[
  \lim_{\mu\to+\infty}\bar{m}'(\mu)=0.
\]
From \eqref{c1:eq1.41} one deduces $\bs{D}_N\bs{x}_\infty=0$, so
$\bs{x}_\infty$ is the eigenvector of $\bs{Q}_N$ for the zero
eigenvalue, giving $\bs{x}_\infty=(1,0,\ldots,0)\cdot 2^{-N}$.
Using property (c):
\[
  \bs{p}_\infty=\bs{V}^{-1}\bs{x}_\infty
  =2^{-N}\bs{V}\bs{x}_\infty
  =2^{-N}(C_N^0,C_N^1,\ldots,C_N^N)\in S_{N+1}.
  \qquad\blacksquare
\]
\end{proof}

\section{\texorpdfstring{$\varepsilon$}{epsilon}-Stabilisation and the Error Threshold}
\label{c1:section:1.8}

The results of \S\ref{c1:section:1.7} show that limiting stabilisation of
\eqref{c1:eq1.33} as $\mu\to+\infty$ always occurs.  On the other hand,
numerical experiments (Fig.~\ref{c1:fig1.11}) show that the stabilisation
of the leading eigenvalue is observable even for relatively small
values of $\mu$.  This prompts a mathematical description of the
stabilisation at \emph{finite} values of $\mu$.  To this end, we
introduce the following definition.

\begin{definition}
The leading eigenvalue $\bar{m}(\mu)$ of \eqref{c1:eq1.33} is said to
admit \emph{$\varepsilon$-stabilisation} if for every $\varepsilon>0$
there exist constants $\bar{m}^*_\varepsilon$ and $\mu^*_\varepsilon$
such that for all $\mu>\mu^*_\varepsilon$:
\[
  |\bar{m}(\mu)-\bar{m}^*_\varepsilon|<\varepsilon,\qquad
  |\bar{m}'(\mu)|<\varepsilon.
\]
\end{definition}

We say the system exhibits an \emph{error-catastrophe phenomenon}
if $\varepsilon$-stabilisation occurs for finite $\bar{m}^*_\varepsilon$
and $\mu^*_\varepsilon$.

Our goal is an approximate determination of the critical value
$\mu_\varepsilon$, beyond which the error threshold can be observed for
system \eqref{c1:eq1.33}.  Assume $m_0\geqslant m_1\geqslant\ldots\geqslant m_N$.
Consider the behaviour of the characteristic curve near
$\mu=0$ using the expansion in powers of $\mu$.

At $\mu=0$ the leading eigenvalue is $m_0$ with eigenvector
$\bs{p}(0)=(1,0,\ldots,0)$.  Computing $\bar{m}'(0)$ directly from
\eqref{c1:eq1.33} and applying standard matrix perturbation theory
\cite{c1:Kato1976}:
\begin{equation}
  \bar{m}'(0)=-N.
  \label{c1:eq1.51}
\end{equation}

Including second-order terms in \eqref{c1:eq1.46} and applying the
eigenvector perturbation technique \cite{c1:Kato1976} yields
\cite{c1:bratus2014_2}:
\begin{equation}
  \bar{m}''(0)=\frac{2N}{m_0-m_1}.
  \label{c1:eq1.52}
\end{equation}

The characteristic curve behaves for large $\mu$ as
$\lim_{\mu\to+\infty}\bar{m}(\mu)=\bar{m}^*$ (determined by
\eqref{c1:eq1.50}).  Writing $\bar{m}^*=1+\delta_N$ where
$\delta_N=2^{-N}\sum_{k=0}^N(m_k-1)\binom{N}{k}\sim O(N^{-2})$
is negligible for large $N$, the approximate critical mutation
parameter $\tilde{\mu}_\varepsilon$ is found by intersecting the
parabolic approximation
$f(\mu)=m_0+\bar{m}'(0)\mu+\frac{1}{2}\bar{m}''(0)\mu^2$
with the asymptote $\bar{m}^*$:
\[
  \tilde{\mu}_\varepsilon
  =\frac{m_0-m_1}{2}
   \Biggl(1-\sqrt{1-\frac{4(m_0-\bar{m}^*)}{(m_0-m_1)N}}\Biggr).
\]

For sufficiently large $N$:
\[
  \tilde{\mu}_\varepsilon
  =\frac{m_0-\bar{m}^*}{N}+O\!\Bigl(\frac{1}{N^2}\Bigr),
\]
and, accounting for the order of $\delta_N$:
\begin{equation}
  \tilde{\mu}_\varepsilon=\frac{m_0-1}{N}+O\!\Bigl(\frac{1}{N^2}\Bigr).
  \label{c1:eq1.53}
\end{equation}

Formula \eqref{c1:eq1.53} gives a guaranteed overestimate for the critical
mutation parameter $\mu^*_\varepsilon$, since the true value
$\bar{m}(\mu^*_\varepsilon)$ is unknown and lies in the interval
$(m_0,\bar{m}^*)$.

For a concrete example, take $N=30$, $m_0=20$,
$m_1=\ldots=m_{30}=1$: then $\tilde{\mu}_\varepsilon=0.633$,
in close agreement with the numerically computed value
(Fig.~\ref{c1:fig1.12}).

\begin{figure}[ht]
  \centering
  \includegraphics[width=1.0\textwidth]{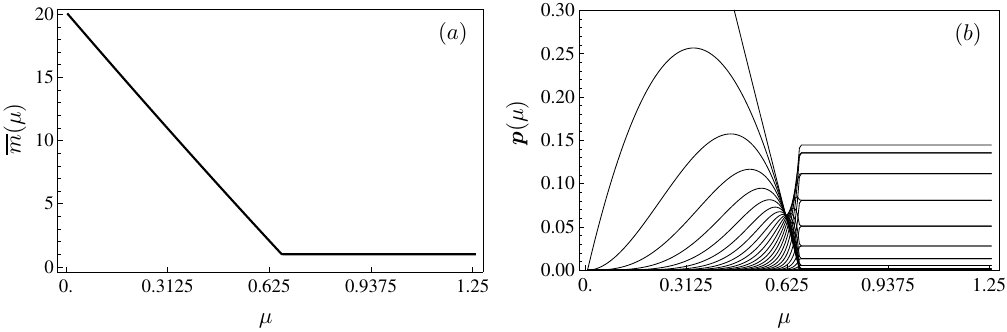}
  \caption{\small (a)~Leading eigenvalue of \eqref{c1:eq1.37} as a
           function of $\mu$.
           (b)~Coordinates of the eigenvector as a function of $\mu$.}
  \label{c1:fig1.12}
\end{figure}

As was shown, limiting stabilisation always exists.  However,
$\varepsilon$-stabilisation at finite $\mu$ does not always occur.
The formula \eqref{c1:eq1.53}, and also the formula
\[
  \mu^*_\varepsilon=\frac{m_0-m_1}{N}
\]
proposed in \cite{c1:Hofbauer1985},
do not always give the correct result.
Figure~\ref{c1:fig1.13} shows an example in which no stabilisation occurs
for $0\leqslant\mu\leqslant 0.3$, for the fitness landscape
$m_i=(30-i)\ln(1-s)$, $i=0,\ldots,30$, $N=30$, $s=0.1$.
The reasons for such behaviour of system \eqref{c1:eq1.33} constitute
an active stimulus for further
research \cite{c1:Nowak1989,c1:Eigen2002,c1:Takeuchi2007,c1:Schuster2011}.

\begin{figure}[ht]
  \centering
  \includegraphics[width=1.0\textwidth]{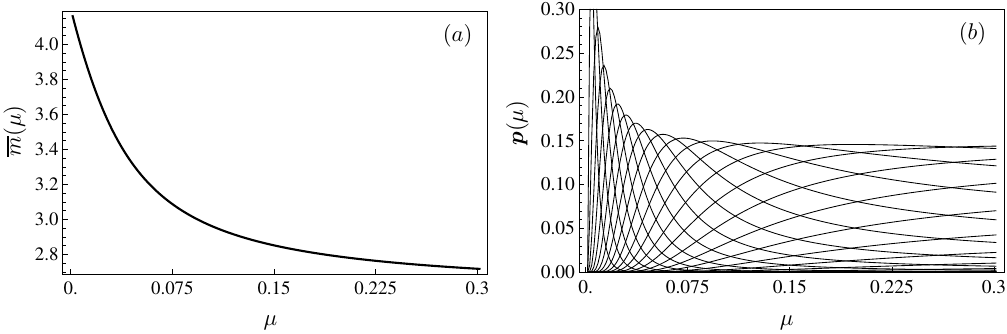}
  \caption{(a)~Leading eigenvalue for $0\leqslant\mu\leqslant 0.3$.
           (b)~Eigenvector coordinates for $0\leqslant\mu\leqslant 0.3$.}
  \label{c1:fig1.13}
\end{figure}

Finally, a fundamental question remains open: is this
$\varepsilon$-stabilisation phenomenon intrinsic to the mathematical
problem of finding the leading eigenvalue, or does it reflect some
deeper biological law governing living systems?

\section*{Summary}\addcontentsline{toc}{section}{Summary}

This chapter has developed the mathematical foundations of replicator
systems. The key results are:

\begin{enumerate}
	\item The replicator equation~\eqref{c1:eq1.5} arises from the general
	Kolmogorov growth equations by passing to relative frequencies,
	under a homogeneity assumption on the growth functions.
	
	\item Independent and autocatalytic replication both exhibit
	survival of a single species, with mean fitness non-decreasing in
	both cases (Fisher's theorem). Hypercyclic replication is qualitatively
	different: the system is permanent, admits a stable limit cycle for
	$n\geqslant 5$, and supports evolutionary variability via the
	row-dominance mechanism.
	
	\item Competing hypercycles obey once-for-ever selection: at most one
	survives, depending on initial conditions.
	
	\item The quasispecies models (Eigen~\eqref{c1:eq1.30} and
	Crow--Kimura~\eqref{c1:eq1.29}) describe replication with mutation.
	Under a primitivity condition, both have a unique globally stable
	equilibrium given by the dominant eigenvector of the respective
	matrix.
	
	\item On sequence space with permutation-invariant fitness, the
	problem reduces from dimension $2^N$ to $N+1$. The
	error-threshold phenomenon — the sharp transition in the quasispecies
	distribution at a critical mutation rate — is characterised
	mathematically by $\varepsilon$-stabilisation of the leading
	eigenvalue.
\end{enumerate}

%% file: chapter2.tex
\chapter{Geometry of the Fitness Surface and Trajectory Dynamics of Replicator Systems}
\label{ch:2}

\begin{chapterabstract}
This chapter studies the geometry of the fitness surface $\Sigma$ of replicator 
systems and its relationship to trajectory dynamics. Starting from the 
symmetric--antisymmetric decomposition ${\bf A} = {\bf B} + {\bf C}$ of the 
fitness landscape matrix, we derive a formula for the rate of change of mean 
fitness along trajectories and establish a necessary condition for its 
monotonicity. We show that in general, trajectories do not reach the maximum of 
$\Sigma$ even when a unique asymptotically stable equilibrium exists, and 
identify the precise algebraic conditions --- expressed in terms of ${\bf B}$ 
and ${\bf C}$ --- under which an equilibrium coincides with a local extremum of 
the fitness surface. A nontrivial class of matrices satisfying these conditions 
is provided by circulant matrices. We then establish the connection between local 
maxima of the fitness surface and evolutionarily stable states: an evolutionarily 
stable equilibrium always furnishes a local maximum, and under the identified 
conditions the converse holds as well. If the unique asymptotically stable 
equilibrium coincides with a local maximum, it is evolutionarily stable and 
realizes the global maximum of $\Sigma$; instability of the equilibrium forces 
the global maximum to the boundary of the simplex. The same geometric framework 
is extended to the general Lotka--Volterra system, where an analogue of mean 
fitness is identified and shown to share the same extremal properties. The 
results are illustrated through six examples including autocatalytic and 
hypercyclic replication, a parametric system exhibiting Andronov--Hopf 
bifurcation and heteroclinic cycles, and the Eigen quasispecies model.
\end{chapterabstract}

\bigskip\hrule\bigskip

The present chapter builds directly on the replicator framework introduced
in~Chapter~\ref{ch:1}. Recall that the dynamics of $n$ competing species,
represented by their relative abundances ${\bf u}(t) \in S_n$, is governed by
\begin{equation}\label{c2:eq1.5}
    \frac{du_{i}}{dt} = u_{i}\Big[\Big({\bf Au}\Big)_{i} - f({\bf u})\Big],
    \quad f({\bf u}) = \Big({\bf Au, u}\Big), \quad i = \overline{1,n},
\end{equation}
where $\Big({\bf Au}\Big)_i$ is the fitness of species $i$, $f({\bf u})$ is the
mean population fitness, and the matrix ${\bf A}$ defines the fitness landscape.
Three replication regimes were examined in~Chapter~\ref{ch:1}: independent,
autocatalytic, and hypercyclic. The last two appear
as running examples below. The Eigen quasispecies model,
also treated in~Chapter~\ref{ch:1}, is revisited in Section~\ref{c2:section:2.4}
in the context of fitness surface geometry.

Fisher's fundamental theorem of natural selection~\cite{c2:Fisher1930,c2:Crow2002} asserts that
mean fitness must not decrease along evolutionary trajectories. This holds when
${\bf A}$ is symmetric, but fails in general --- the hypercycle being the
paradigmatic case. The central question of this chapter is geometric: what is
the structure of the \textit{fitness surface} (or the \textit{fitness landscape})
$$
    \Sigma = \left\{z = f({\bf u}) :\,
    \sum\limits_{i,j=1}^{n} a_{ij}u_i u_j,\;\; {\bf u}(t) \in S_n \right\},
$$
and under what conditions do replicator trajectories converge to its maximum?
The analysis is organised around the decomposition ${\bf A} = {\bf B} + {\bf C}$,
where ${\bf B}$ is the symmetric part and ${\bf C}$ the antisymmetric part,
which separates the fitness surface geometry from the rotational component of the
flow.

\section{Trajectory Dynamics on the Fitness Surface}\label{c2:section:2.1}

Each trajectory $\gamma_{t}$ of system \eqref{c2:eq1.5} corresponds to a curve $\Gamma_{t}$ on the surface $\Sigma$. The relationship between the geometry of the fitness surface $\Sigma$ and the nature of the trajectory dynamics $\gamma_{t}$ of replicator systems is interesting in its own right, and also in connection with Fisher's fundamental theorem of natural selection~\cite{c2:Fisher1930}, which asserts that the mean fitness of a species does not decrease over the course of evolution.

The term ``theorem'' should not mislead the reader: Fisher's assertion is not tied to any mathematical proof for a specific biological model but rather has the character of an additional postulate supplementing Darwin's evolutionary triad.

Write the fitness landscape matrix ${\bf A}$ of the replicator system \eqref{c2:eq1.5} in the form
\begin{equation}
	{\bf A} = {\bf B} + {\bf C}, \quad {\bf B} = \frac{{\bf A} + {\bf A}^{T}}{2}, \quad {\bf C} = \frac{{\bf A} - {\bf A}^{T}}{2}.
	\label{c2:eq2.1}
\end{equation}

If ${\bf A}$ is symmetric, then ${\bf C} = 0$. On the other hand, since $\Big({\bf Cu, u}\Big) = 0$, the mean fitness of the system is completely determined by the matrix ${\bf B}$:
$$
f({\bf u}(t)) = \sum\limits_{i, j = 1}^{n}a_{ij}u_{i}u_{j} = \Big({\bf Au, u}\Big) = \Big({\bf Bu, u}\Big).
$$

The following formula holds:
\begin{equation}
	\dot{f}({\bf u}(t)) = 2\left(\sum\limits_{i = 1}^{n}\Big({\bf Bu}\Big)_{i}^{2}u_{i} - \sum\limits_{i = 1}^{n}\Big({\bf Bu}\Big)_{i}\Big({\bf Cu}\Big)_{i}u_{i} - f^{2}({\bf u})\right),
	\label{c2:eq2.2}
\end{equation}
\begin{equation*}
	\Big({\bf Au}\Big)_{i} = \sum\limits_{j = 1}^{n}a_{ij}u_{j}, \,\, \Big({\bf Bu}\Big)_{i} = \sum\limits_{j = 1}^{n}b_{ij}u_{j}, \,\, \Big({\bf Cu}\Big)_{i} = \sum\limits_{i = 1}^{n}\Big({\bf Cu}\Big)_{i}u_{i},
\end{equation*}
representing the rate of change of mean fitness along the system's trajectories.

When ${\bf A}$ is symmetric,
\begin{equation*}
	\dot{f}({\bf u}(t)) = 2\left(\sum\limits_{i = 1}^{n}\Big({\bf Au}\Big)_{i}^{2}u_{i} - \Bigg(\sum\limits_{i = 1}^{n}\Big({\bf Au}\Big)_{i}u_{i}\Bigg)^{2}\right) \geqslant 0,
\end{equation*}
since in this case the expression gives the variance of the quantity $\Big({\bf Au}\Big)_{i}$, where each value is realized with probability $u_{i}$~\cite{c2:Kimura1958}.

Rewrite \eqref{c2:eq2.2} in the form
\begin{equation*}
	\dot{f}({\bf u}(t)) = f({\bf u}) \left(\sum\limits_{i = 1}^{n}\Big({\bf Bu}\Big)_{i}v_{i} - f({\bf u})\right),
\end{equation*}
where
\begin{equation*}
	\begin{aligned}
		&v_i = \dfrac{\Big({\bf Au}\Big)_{i}u_i}{f({\bf u})}, \quad i=1,\,2,\, \ldots,\, n, \\
		&f({\bf u}) = \sum\limits_{i = 1}^{n} \Big({\bf Au}\Big)_{i}u_i = \Big({\bf Au, u}\Big) = \Big({\bf Bu, u}\Big).
	\end{aligned}
\end{equation*}
If ${\bf A}$ is non-negative, then $0 \leqslant v_i \leqslant 1$ for $i=1,\,2,\, \ldots,\, n$, and $\sum\limits_{i = 1}^{n} v_i = 1$.

This allows the quantities $v_i$ to be interpreted as a new tilted probability measure for the values $\Big({\bf Bu}\Big)_i$. The expression
\begin{equation*}
	\sum\limits_{i = 1}^{n}\Big({\bf Bu}\Big)_iv_i = \mathbb{E}_{{\bf v}(t)}^{\bf B}
\end{equation*}
then represents the expectation under that tilted measure. We denote by $\mathbb{E}_{{\bf u}(t)}^{\bf A}$ the expectation of $\Big({\bf Au}\Big)_i$ under the original measure, ${\bf u} \in S_n$.

A necessary condition for the mean fitness to be non-decreasing along all trajectories of the system is
\begin{equation*}
	\mathbb{E}_{{\bf v}(t)}^{\bf B} \geqslant \mathbb{E}_{{\bf u}(t)}^{\bf A}.
\end{equation*}
When this condition holds, it follows by LaSalle's theorem~\cite{c2:LaSalle1961,c2:Bratus2010} that the trajectories of the replicator system converge to the largest invariant subset of
\begin{equation*}
	K = \left\{\dot{f}({\bf u}) = 0, \quad {\bf u} \in S_n \right\}.
\end{equation*}
One checks directly that every rest point of the replicator system lies in this set.

Suppose that the replicator system has a limit cycle $\gamma_l \in \inside S_n$ with period $T$, and the above inequality holds. Then at every point of the limit cycle
\begin{equation*}
	\mathbb{E}_{{\bf v}(t)}^{\bf B} = \mathbb{E}_{{\bf u}(t)}^{\bf A}.
\end{equation*}
This equality allows the equation of the limit cycle $\gamma_l$ to be obtained in explicit form.

For a hypercycle, the cubic form in the components of ${\bf u} \in \gamma_l$ is indefinite in sign. This means that along the cycle, mean fitness may both increase and decrease, while its time-averaged value remains constant.

We consider the following three examples.

\medskip
\textbf{Example 1. Autocatalytic replication.}
Consider autocatalytic replication, described by the system
\begin{equation}  
\dot{u}_{i} = u_{i}(u_{i} - f(t)), \quad i = \overline{1, n}, \quad f(t) = \sum\limits_{i = 1}^{n}u_{i}^{2}, \quad {\bf u} \in S_{n}.
\label{c2:eq1.7}
\end{equation}
The maximum fitness value ($f = 1$) is attained at the vertices of the simplex $S_{n}$,
$$
e^{k} = (0,\, \ldots,\, 0,\, 1,\, 0,\, \ldots,\, 0), \quad k = \overline{1, n},
$$ 
where the unit occupies the $k$-th position, and the minimum is attained at the equilibrium $u_{i} = \dfrac{1}{n}$, $i = \overline{1, n}$.

The phase trajectories of the system \eqref{c2:eq1.7}, depending on the initial conditions, converge to one of the vertices, which acts as an attractor. Thus, the maximum fitness ultimately attained depends on which initial data are realized. This phenomenon is known as the \textit{adaptive fitness landscape} and was first studied by S.~Wright~\cite{c2:Maynard1982,c2:Wright1932}.
\hfill$\square$

\medskip
\textbf{Example 2. Hypercyclic replication.}

Let us now focus on replication defined by the matrix
$$
{\bf A} = 
\begin{pmatrix}
	0 & 0 & k_{1}\\
	k_{2} & 0 & 0\\
	0 & k_{3} & 0\\
\end{pmatrix}\!.
$$
The coordinates of the interior equilibrium are given by
$$
\bar{u}_{i} = \frac{\bar{m}}{k_{i + 1}}, \quad i = \overline{1, 3}, \quad k_{4} = k_{1}, \quad \bar{m} = \Bigg(\sum\limits_{i = 1}^{3}\frac{1}{k_{i}}\Bigg)^{-1}\!\!.
$$
The matrices ${\bf B}$ and ${\bf C}$ are computed directly from:
$$
{\bf B} = \frac{1}{2} 
\begin{pmatrix}
	0 & k_{2} & k_{1}\\
	k_{2} & 0 & k_{3}\\
	k_{1} & k_{3} & 0\\
\end{pmatrix}\!, \quad 
{\bf C} = \frac{1}{2} 
\begin{pmatrix}
	0 & -k_{2} & k_{1}\\
	k_{2} & 0 & -k_{3}\\
	-k_{1} & k_{3} & 0\\
\end{pmatrix}\!.
$$

The maximum of the fitness surface is attained at the point ${\bf u}_m$ with
\begin{align*}
    u_{1} &= \frac{k_{3}(k_{1} + k_{2}) - k_{3}^{2}}{\sigma}, \quad
    u_{2} = \frac{k_{1}(k_{2} + k_{3}) - k_{1}^{2}}{\sigma}, \\
    u_{3} &= \frac{k_{2}(k_{1} + k_{3}) - k_{2}^{2}}{\sigma},
\end{align*}
where $\sigma = 2(k_{1}k_{2} + k_{2}k_{3} + k_{1}k_{3}) - (k_{1}^{2} + k_{2}^{2} + k_{3}^{2})$.
For $k_{1} = 0.25$, $k_{2} = 0.3$, $k_{3} = 0.35$, the unique stable equilibrium is
$$
{\bf \bar{u}} = (0.3272,\, 0.2803,\, 0.3925),
$$
while the maximum of the fitness surface is attained at
$$
{\bf u}_{m} = (0.2692,\, 0.3846,\, 0.3462),
$$
with
$$
f({\bf \bar{u}}) < f({\bf u}_{m}).
$$
This inequality shows that the unique stable equilibrium does not provide the maximum of the surface $\Sigma$. Consequently, the curve $\Gamma_{t}$ does not reach the peak of $\Sigma$.
\hfill$\square$
\begin{figure}[ht]
	\begin{minipage}[ht]{0.48\linewidth}
		\center{\includegraphics[width=0.97\linewidth]{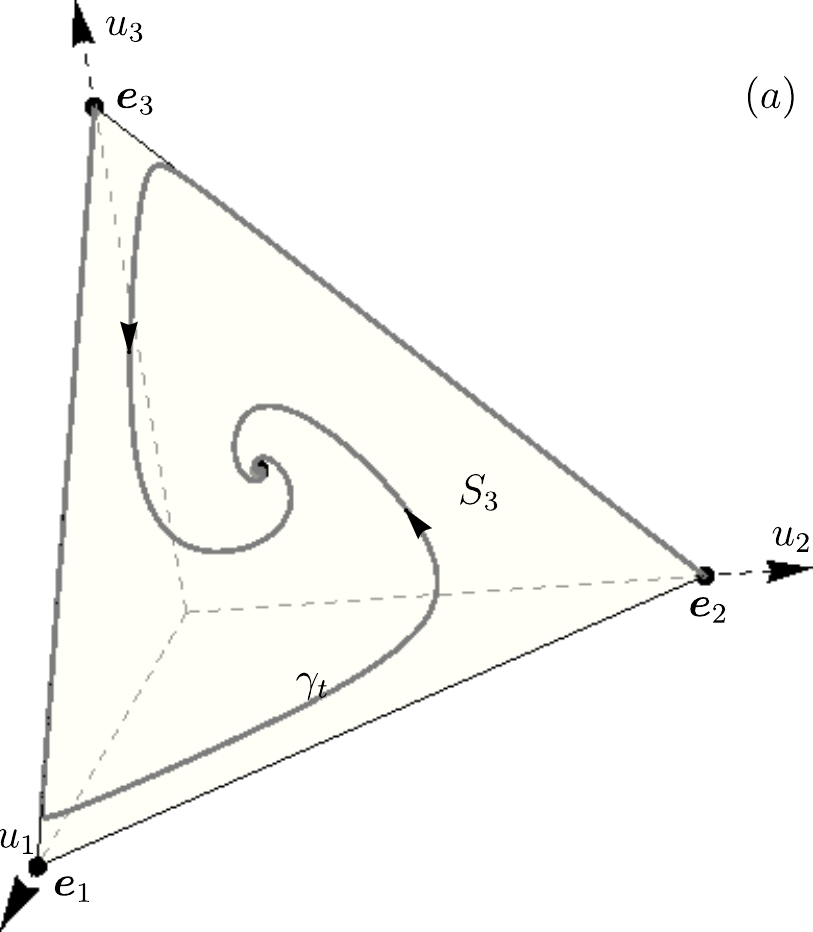}}
	\end{minipage}
	\hfill
	\begin{minipage}[ht]{0.48\linewidth}
		\center{\includegraphics[width=0.97\linewidth]{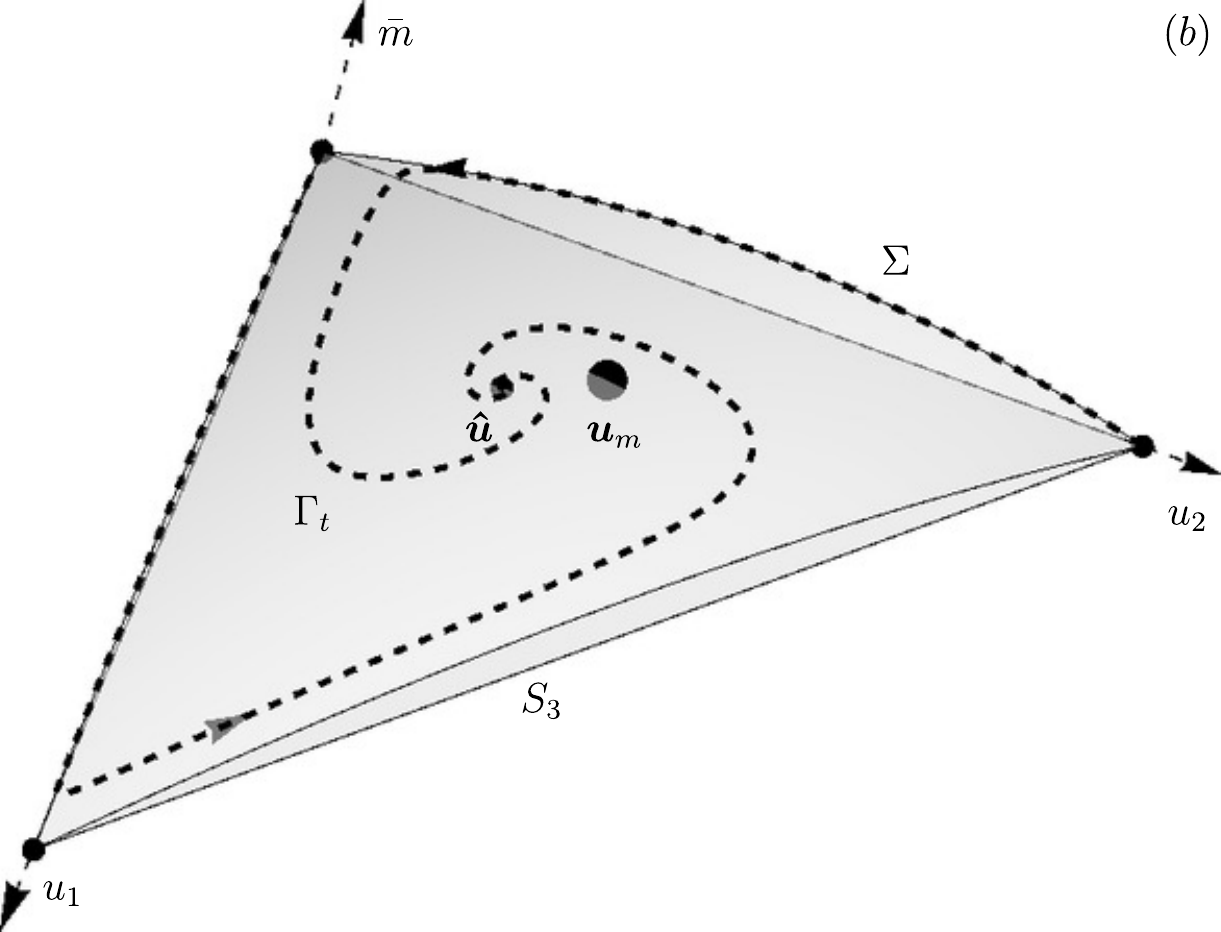}}
	\end{minipage}
\caption{Hypercyclic replication system of Example~2. 
(a) Phase trajectory $\gamma_t$ on the simplex $S_3$, converging to the 
stable equilibrium $\bar{\bf u}$. 
(b) The fitness surface $\Sigma$ over $S_3$: the curve $\Gamma_t$ is the 
lift of $\gamma_t$ onto $\Sigma$, with $\bar{\bf u}$ denoting the stable 
equilibrium and ${\bf u}_m$ the maximum of $\Sigma$. The two points are 
distinct: the trajectory does not reach the surface maximum.}
\end{figure}

\medskip
\textbf{Example 3. A fitness landscape with multiple local extrema.}

Consider the replicator system defined by the matrix
$$
{\bf A} = 
\begin{pmatrix}
	0 & 0 & -\mu & 1\\
	1 & 0 & 0 & -\mu\\
	-\mu & 1 & 0 & 0\\
	0 & -\mu & 1 & 0\\
\end{pmatrix}\!.
$$
For $|\mu| < 1$, the system is permanent, meaning its trajectories remain bounded away from the boundary of the simplex. At $\mu = 0$, a Poincar\'{e}--Andronov--Hopf bifurcation occurs, giving rise to a stable limit cycle for $\mu > 0$. For $|\mu| > 1$ the system ceases to be permanent~\cite{c2:Hofbauer2003,c2:Hofbauer1978}.

The matrix ${\bf B}$ in this case is
$$
{\bf B} = 
\begin{pmatrix}
	0 & 1 & -2\mu & 1\\
	1 & 0 & 1 & -2\mu\\
	-2\mu & 1 & 0 & 1\\
	1 & -2\mu & 1 & 0\\
\end{pmatrix}\!,
$$
with eigenvalues
$$
\lambda_{1} = 1 - \mu, \quad \lambda_{2, 3} = \mu, \quad \lambda_{4} = -(1 + \mu).
$$ 

The orthogonal transformation reducing ${\bf B}$ to diagonal form is given by
$$
{\bf U} = 
\begin{pmatrix}
	1/2 & -\sqrt{2}/2 & 0 & 1/2\\
	1/2 & 0 & -\sqrt{2}/2 & -1/2\\
	1/2 & \sqrt{2}/2 & 0 & 1/2\\
	1/2 & 0 & \sqrt{2}/2 & -1/2\\
\end{pmatrix}\!.
$$ 
In the new variables, the mean fitness is
$$
f({\bf w}) = (1 - \mu)w_{1}^{2} + \mu(w_{2}^{2} + w_{3}^{2}) - (1 + \mu)w_{4}^{2}.
$$
The constraint $\Big({\bf w, U^{T}I}\Big) = 1$ gives $w_{1} = \dfrac{1}{2}$. The simplex $S_{4}$ maps to the convex set
\begin{equation*}
    W_{4} = \left\{{\bf w}: 
    \begin{array}{l}
        -\dfrac{1}{4} - \dfrac{\sqrt{2}}{2}w_{2} + \dfrac{1}{2}w_{4} \geqslant 0,\\[6pt]
        \dfrac{1}{4} - \dfrac{\sqrt{2}}{2}w_{3} - \dfrac{1}{2}w_{4} \geqslant 0,\\[6pt]
        \dfrac{1}{4} + \dfrac{\sqrt{2}}{2}w_{2} - \dfrac{1}{2}w_{4} \geqslant 0
    \end{array}
    \right\}\!.
\end{equation*}

For $|\mu| < 1$, the maximum of the fitness surface is attained when $w_{4} = 0$. The set $W_{4}$ is then a square in the $(w_{2}, w_{3})$-plane with vertices at $p_{i} = \{\pm\sqrt{2}/4,\, \pm\sqrt{2}/4\}$, and at each of these four vertices the fitness landscape surface attains the same maximum value:
$$
f(p_{i}) = \frac{1}{4}(1 - \mu) + \frac{1}{4}\mu = 0.25.
$$ 

In the new coordinates, the interior equilibrium ${\bf \bar{u}} = (1/4, 1/4, 1/4, 1/4)$ corresponds to ${\bf \bar{w}} = (1/2, 0, 0, 0)$, and the fitness surface at this point equals
$$
f({\bf \bar{w}}) = \frac{1}{4} - \frac{\mu}{4}.
$$

For $-1 < \mu < 0$ this equilibrium is asymptotically stable, and the fitness surface attains its global maximum there since $f({\bf \bar{w}}) > f(p_{i})$. In general, the fitness surface has five peaks; the central one corresponds to the equilibrium. Orbits $\gamma_{t}$ converge to ${\bf \bar{w}}$, while the curve $\Gamma_{t}$ on $\Sigma$ approaches the central peak.

For $0 < \mu < 1$, an Andronov--Hopf bifurcation produces a limit cycle, and $\gamma_{t}$ converges to a closed orbit surrounding ${\bf \bar{w}}$. In this regime, mean fitness at the vertices $p_{i}$ exceeds mean fitness at the central peak. On the simplex $S_{4}$, the points $p_{i}$ correspond to $P_{1} = (0,\, 0,\, 1/2,\, 1/2)$, $P_{2} = (1/2,\, 0,\, 0,\, 1/2)$, $P_{3} = (1/2,\, 1/2,\, 0,\, 0)$, $P_{4} = (0,\, 1/2,\, 1/2,\, 0)$.

For $\mu > 1$, the points $P_{i}$ still furnish the maximum of $f({\bf u})$, but at the central point $f({\bf \bar{w}}) < 0$, so the global minimum of the fitness landscape is attained there. This case corresponds to a heteroclinic cycle.

Finally, for $\mu < -1$, the equilibrium ${\bf \bar{w}}$ is unstable and the maximum mean fitness $f({\bf w})$ is achieved at $w_{2} = w_{3} = 0$, corresponding to the points $Q_{1,2} = (1/2,\, 0,\, 0,\, \pm 1/2)$, where $f(Q_{1,2}) = -\dfrac{\mu}{2}$ and consequently $f(Q_{1,2}) > f({\bf \bar{w}})$. A heteroclinic cycle again arises in this case.

The system dynamics and fitness landscape are depicted schematically in Fig.~\ref{c2:fig2.2}.
\hfill$\square$
\begin{figure}[ht]
	\begin{minipage}[ht]{0.5\linewidth}
		\center{\includegraphics[width=1.0\linewidth]{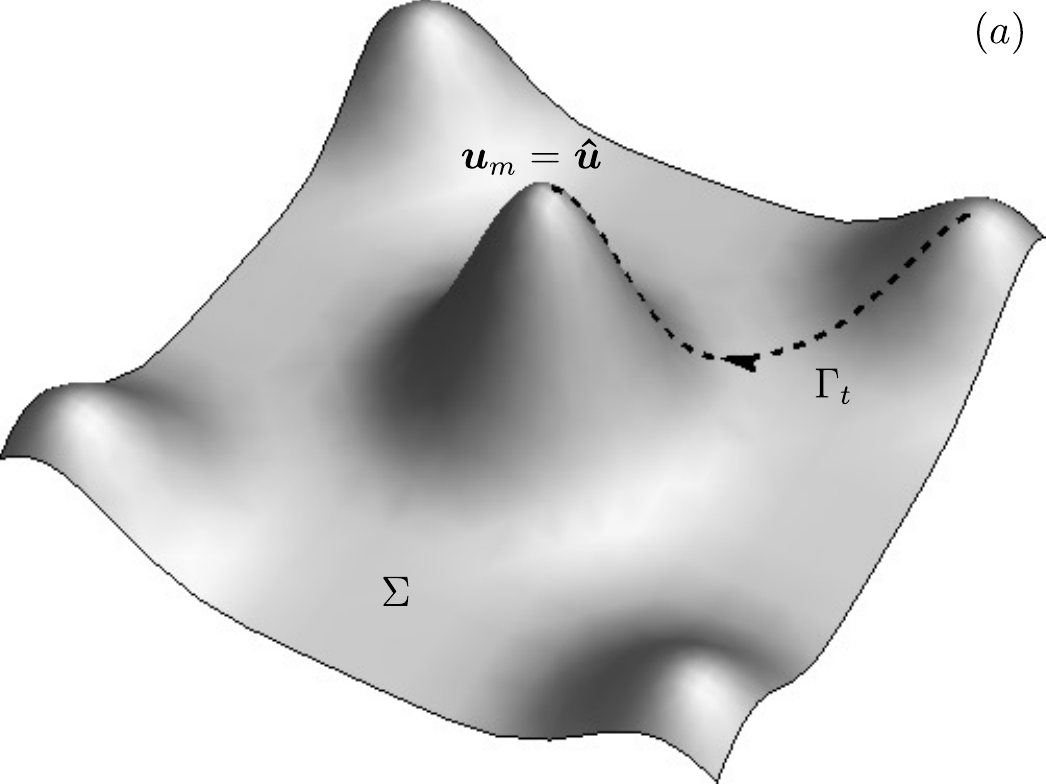}}
	\end{minipage}
	\hfill
	\begin{minipage}[ht]{0.5\linewidth}
		\center{\includegraphics[width=1.0\linewidth]{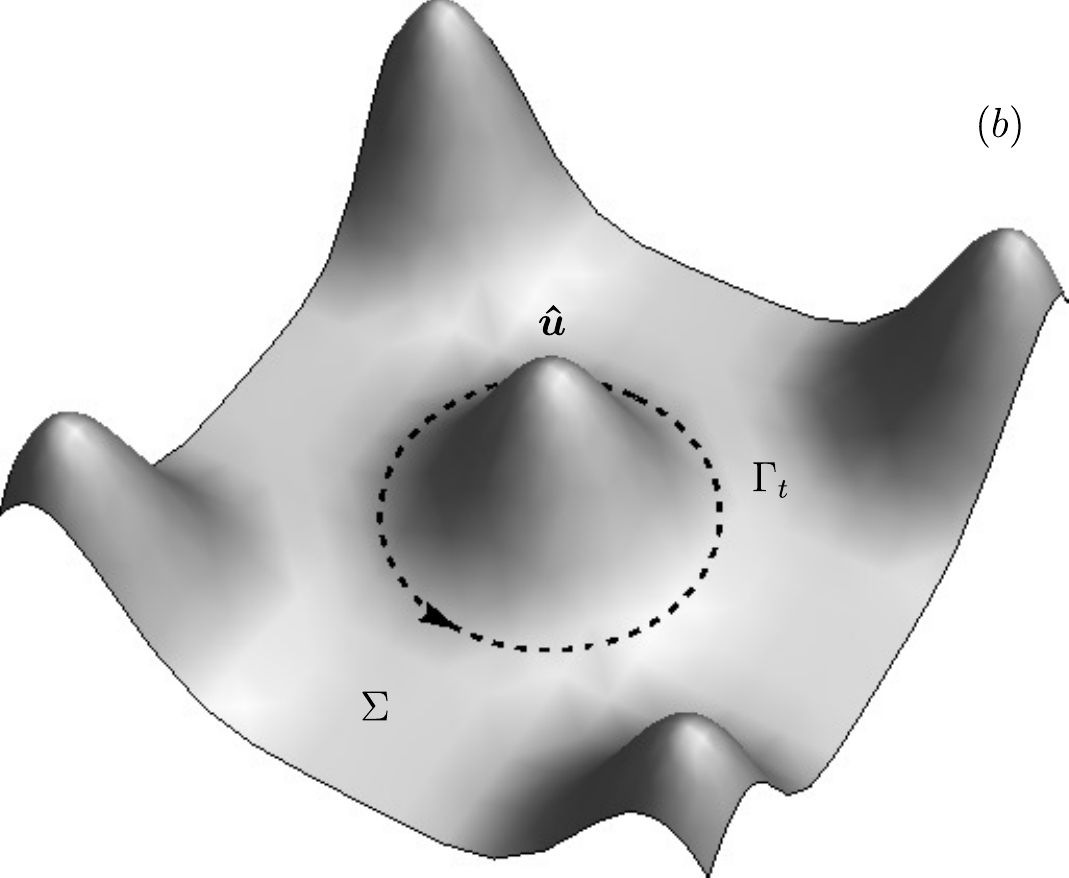}}
	\end{minipage}
	\caption{Dynamics of the system of Example~3. (a) The interior equilibrium coincides with the global maximum of the fitness landscape and is globally asymptotically stable. (b) The interior equilibrium is not the global maximum; solutions are attracted to a limit cycle. Note that the time-averaged mean fitness on the limit cycle coincides with the local maximum ${\bf \bar{u}}$.}
	\label{c2:fig2.2}
\end{figure}

These examples illustrate the diversity of mean fitness dynamics in replicator systems, ranging from monotone (adaptive) to periodic behavior.

\clearpage

\section{Necessary Conditions for a Local Extremum of the Fitness Surface}\label{c2:section:2.2}
Mathematically, the mean fitness function is a quadratic form defined by the symmetric matrix ${\bf B}$:
$$
f({\bf u}(t)) = \Big({\bf Au, u}\Big) = \Big({\bf Bu, u}\Big). 
$$ 
Therefore, there exists an orthogonal transformation ${\bf U}$ such that
$$
{\bf U^{T}BU} = {\bf \Lambda} = \diag(\lambda_{1},\, \lambda_{2},\, \ldots,\, \lambda_{n}),
$$
where $\lambda_{i}$ are the real eigenvalues of ${\bf B}$.

The change of variables ${\bf u = Uv}$ brings the quadratic form to its canonical (diagonal) form
$$
f({\bf v}) = \sum\limits_{i = 1}^{k}\lambda_{i}^{+}w_{i}^{2} + \sum\limits_{j = k + 1}^{n}\lambda_{j}^{-}w_{j}^{2},
$$
where $\lambda_{i}^{+}$ and $\lambda_{j}^{-}$ denote the positive and negative eigenvalues of ${\bf B}$, under the assumption that $\det{\bf B} \neq 0$.

Following standard geometric terminology, we call the fitness surface \textit{elliptic} if all $\lambda_{i} > 0$ or all $\lambda_{i} < 0$ (${i = \overline{1, n}}$), and \textit{hyperbolic} if there exist $i \neq j$ such that $\lambda_{i}\lambda_{j} < 0$.

In Example~1 (autocatalytic replication), the fitness surface is elliptic, whereas in Example~3 it may be either elliptic or hyperbolic depending on the parameter $\mu$.

The transformation ${\bf U}$ maps the simplex $S_{n}$ to the convex set $W_{n}$:
\begin{align*}
	W_{n} = \left\{{\bf v} \in \mathbb{R}^{n}: \Big({\bf v, U^{T}I}\Big) = 1,\,\, {\bf Uv} \geqslant 0, \right. \\
    \left. {\bf I} = (1,\, 1,\, \ldots,\, 1) \in \R^{n}\right\}.
\end{align*}

Finding the extrema of the fitness surface, therefore, reduces to the mathematical programming problem
$$
f({\bf v}) \to {\rm extr}, \quad {\bf v} \in W_{n}.
$$
A solution always exists, but it does not need to be unique~\cite{c2:Ioffe1974}.

An alternative approach to locating extrema of the fitness surface uses classical Lagrange multipliers. Consider the Lagrangian:
$$
\pounds({\bf u}, \mu) = \Big({\bf Bu, u}\Big) - \mu_{0}\Bigg(\Big({\bf u, I}\Big) - 1\Bigg) + \sum\limits_{j = 1}^{n}\mu_{j}u_{j},
$$ 
where the vector $\mu = (\mu_{0},\, \mu_{1},\, \ldots,\, \mu_{n})$, $\mu_{j} \geqslant 0$, $j = \overline{1, n}$, represents the Lagrange multipliers arising from the constraint ${\bf u} \in S_{n}$.

If ${\bf u} \in S_{n}$ is a local maximum of the quadratic form $\Big({\bf Bu, u}\Big)$, there exists a nonzero vector of Lagrange multipliers $\mu$ such that
$$
\frac{\partial\pounds}{\partial u_{j}}({\bf u}, \mu) \leqslant 0, \quad u_{j}\frac{\partial\pounds}{\partial u_{j}}({\bf u}, \mu) = 0, \quad j = \overline{1, n}.
$$

When ${\bf u} \in \inside S_{n}$, the necessary conditions for a local extremum take the simpler form
$$
\frac{\partial\pounds}{\partial u_{j}}({\bf u}, \mu_{0}) = 0, \quad j = \overline{1, n},
$$
which amounts to
\begin{equation}
	{\bf Bu} = \mu_{0}{\bf I}, \quad {\bf I} = (1,\, 1,\, \ldots,\, 1) \in \R^{n}.
	\label{c2:eq2.3}
\end{equation}

If ${\bf B}$ is non-singular, the constraint $\Big({\bf u, I}\Big) = 1$ gives $\mu_{0} = \Big({\bf B^{-1}I, I}\Big)$. If ${\bf B}$ is singular, any local extremum lies on the boundary of the simplex rather than in its interior.

Suppose ${\bf u} \in \inside S_{n}$ is a strict local maximum of the fitness surface. Then for sufficiently small $\varepsilon \neq 0$,
$$
\Big({\bf B}({\bf u} + \varepsilon{\bf w}), {\bf u} + \varepsilon{\bf w}\Big) - \Big({\bf Bu, u}\Big) < 0, 
$$
where ${\bf w} \in \mathbb{R}^{n}$ satisfies ${\bf u} + \varepsilon{\bf w} \in S_{n}$, i.e.,\ $\Big(({\bf u} + \varepsilon{\bf w}), {\bf I}\Big) = 1$. The latter condition gives $\Big({\bf w, I}\Big) = 0$. For any $\varepsilon \neq 0$ we then have
$$
2\varepsilon\Big({\bf Bu, w}\Big) + \varepsilon^{2}\Big({\bf Bw, w}\Big) < 0,
$$
which is possible only if
\begin{equation}
	\Big({\bf Bu, w}\Big) = 0 \quad \text{and} \quad \Big({\bf Bw, w}\Big) < 0.
	\label{c2:eq2.4}
\end{equation}

Consider a basis for the $(n-1)$-dimensional subspace of $\R^{n}$ orthogonal to ${\bf I} = (1,\, 1,\, \ldots,\, 1)$:
\begin{align*}
	r_{1} = (1,\, -1,\, 0,\, \ldots,\, 0), \quad r_{2} = (1,\, 0,\, -1,\, \ldots,\, 0), \quad \ldots, \\ r_{n - 1} = (1,\, 0,\, 0,\, \ldots,\, -1).
\end{align*} 

Any ${\bf w}$ may be written as
$$
{\bf w} = \sum\limits_{k = 1}^{n - 1}\xi_{k}r_{k}, \quad \xi_{k} \in \R.
$$
To verify conditions \eqref{c2:eq2.4} it therefore suffices to analyze the sign of the quadratic form
$$
\Big({\bf Bw, w}\Big) = \sum\limits_{i, j = 1}^{n - 1}\beta_{ij}\xi_{i}\xi_{j}, \quad \beta_{ij} = \Big({\bf Br_{i}, r_{j}}\Big)
$$
on the linear manifold
$$
\Big({\bf Bu, w}\Big) = \sum\limits_{i = 1}^{n - 1}\alpha_{i}\xi_{i} = 0, \quad \alpha_{i} = \Big({\bf Bu, r_{i}}\Big).
$$

\clearpage

\section{Coincidence of a Local Extremum of the Fitness Surface with an Equilibrium}
\label{c2:section:2.3}
A natural question arises: under what conditions does a stationary equilibrium coincide with an extremum of the fitness surface?

The preceding examples make clear that such a coincidence is not generic. In hypercyclic replication (Example~2), the maximum of the fitness surface does not coincide with the unique equilibrium, whereas in Example~3, coincidence occurs under certain conditions.

\begin{theorem}\label{c2:t2.1}
	Let ${\bf B}$ and ${\bf C}$ be the matrices in decomposition \eqref{c2:eq2.1}. An isolated equilibrium ${\bf \bar{u}} \in S_{n}$ of the system \eqref{c2:eq1.5} coincides with an extremal point of the fitness surface if and only if there exists a non-singular matrix ${\bf M} = \Big(\mu_{i,k}\Big)_{i, k = \overline{1, n}}$ such that
	\begin{equation}
		{\bf C = MB},
		\label{c2:eq2.5}
	\end{equation}
	together with
	\begin{equation}
		{\bf MI} = 0, \quad {\bf I} = (1,\, 1,\, \ldots,\, 1) \in \R^{n}.
		\label{c2:eq2.6}
	\end{equation}
\end{theorem}

\begin{proof}
	From decomposition \eqref{c2:eq2.1}, at any stationary equilibrium of the replicator system, the following equality must hold:
	\begin{equation}
		{\bf A\bar{u} = B\bar{u} + C\bar{u}} = f({\bf \bar{u}}){\bf I}, \quad {\bf I} = (1,\, 1,\, \ldots,\, 1) \in \R^{n},
		\label{c2:eq2.7}
	\end{equation}
	where $f({\bf \bar{u}})$ is the stationary fitness value.
	
	The solution of \eqref{c2:eq2.7} satisfies \eqref{c2:eq2.3} if and only if ${\bf C\bar{u}} = 0$ and $\mu_{0} = f({\bf \bar{u}})$.
	
	We prove the necessity of \eqref{c2:eq2.5} and \eqref{c2:eq2.6}. Taking the inner product of \eqref{c2:eq2.7} with ${\bf \bar{u}}$ and using $\Big({\bf C\bar{u}, \bar{u}}\Big) = 0$ gives $\Big({\bf B\bar{u}, \bar{u}}\Big) = f({\bf \bar{u}}) = \mu_{0}$.
	
	We now show that for a solution of
	\begin{equation}
		{\bf B\bar{u}} = f({\bf \bar{u}}){\bf I}
		\label{c2:eq2.8}
	\end{equation}
	to satisfy ${\bf C\bar{u}} = 0$, conditions \eqref{c2:eq2.5} and \eqref{c2:eq2.6} are both necessary and sufficient.
	
	If these conditions hold, then
	$$
	{\bf C\bar{u} = MB\bar{u}} = f({\bf \bar{u}}){\bf MI} = 0.
	$$
	
	Conversely, by the linear independence of the columns of ${\bf B}$, a representation \eqref{c2:eq2.5} exists for some matrix ${\bf M}$. If ${\bf \bar{u}}$ solves \eqref{c2:eq2.8}, then ${\bf C\bar{u}} = 0$ and consequently ${\bf MI} = 0$. 
\end{proof}  

\medskip
\textbf{Example 4. Hypercycle replication and extremum conditions.}

Returning to Example~2, we show that conditions \eqref{c2:eq2.5} and \eqref{c2:eq2.6} imply ${\bf \bar{u}_{max}} = {\bf \bar{u}}$.

We have
$$
c_{1} = \frac{k_{1}}{k_{3}}b_{1} - \frac{k_{2}}{k_{3}}b_{3}, \quad c_{2} = -\frac{k_{3}}{k_{1}}b_{1} + \frac{k_{2}}{k_{1}}b_{3}, \quad c_{3} = \frac{k_{3}}{k_{2}}b_{1} - \frac{k_{1}}{k_{2}}b_{2}.
$$

Conditions \eqref{c2:eq2.5} and \eqref{c2:eq2.6} are satisfied if and only if $k_{1} = k_{2} = k_{3} = k$. Taking $k = 1$, the eigenvalues of ${\bf B}$ are $\lambda_{1} = 1$, $\lambda_{2} = \lambda_{3} = 1/2$. The orthogonal change of the variables ${\bf u = Uw}$ that brings the mean fitness to the canonical form $w_{1}^{2} - \dfrac{1}{2}(w_{2}^{2} + w_{3}^{2})$ is given by
$$
{\bf U} = 
\begin{pmatrix}
	\sqrt{3}/3 & 0 & 2/\sqrt{6}\\
	\sqrt{3}/3 & -\sqrt{2}/2 & -1/\sqrt{6}\\
	\sqrt{3}/3 & \sqrt{2}/2 & -1/\sqrt{6}\\
\end{pmatrix}.
$$ 

From $\Big({\bf u, I}\Big) = \Big({\bf w, U^{T}I}\Big) = 1$ it follows that $w_{1} = \sqrt{3}/3$, and the simplex $S_{3}$ maps to the convex set
\begin{equation*}
    W_{3} = \left\{{\bf w}: 
    \begin{array}{l}
        w_{1} = \dfrac{\sqrt{3}}{3},\\[6pt]
        \dfrac{\sqrt{2}}{2}w_{2} - \dfrac{1}{\sqrt{6}}w_{3} + \dfrac{1}{3} \geqslant 0,\\[6pt]
        -\dfrac{\sqrt{2}}{2}w_{2} - \dfrac{1}{\sqrt{6}}w_{3} + \dfrac{1}{3} \geqslant 0,\\[6pt]
        w_{3} \geqslant -\dfrac{1}{\sqrt{6}}
    \end{array}
    \right\}\!.
\end{equation*}

In the new coordinates, the fitness surface is
$$
f({\bf w}) = \frac{1}{3} - \frac{1}{2}(w_{2}^{2} + w_{3}^{2}),
$$
a paraboloid with apex at $w_{2} = w_{3} = 0$. In this special case, the equilibrium coincides with the maximum of the fitness surface.

A nontrivial class of matrices satisfying all conditions of Theorem~\ref{c2:t2.1} is provided by the so-called circulant matrices.
$$
{\bf A} = 
\begin{pmatrix}
	a_{1} & a_{2} & \ldots & a_{n - 1} & a_{n}\\
	a_{n} & a_{1} & \ldots & a_{n - 2} & a_{n - 1}\\
	\ldots & \ldots & \ldots & \ldots & \ldots\\
	a_{2} & a_{3} & \ldots & a_{n} & a_{1}\\
\end{pmatrix}\!.
$$ 

Indeed, equation ${\bf B\bar{u}} = f({\bf \bar{u}}){\bf I}$ then has the solution ${\bf \bar{u}}_{1} = {\bf \bar{u}}_{2} = \ldots = {\bf \bar{u}}_{n} = \dfrac{1}{n}$, $f({\bf \bar{u}}) = \dfrac{1}{n}\sum\limits_{i = 1}^{n}a_{i}$, and one directly verifies that this solution lies in the kernel of ${\bf C}$. Thus, the interior rest point of a replicator system with a circulant matrix is always a local extremum of the fitness surface.
\hfill$\square$\hfill$\square$

It is of interest to find sufficient conditions for the mean fitness to be monotone along trajectories of a general replicator system.
\begin{theorem}\label{c2:t2.2}
	Suppose ${\bf B}$ is non-singular and the matrix ${\bf Q = CB^{-1}}$ is positive semi-definite, i.e.\ $\Big({\bf Q\xi, \xi}\Big) \geqslant 0$ for all $\xi \in \R^{n}$. Then the mean fitness of the replicator system \eqref{c2:eq1.5} is a monotone function, and the system has no closed trajectories in $\inside S_{n}$.
\end{theorem} 

\begin{proof}
	Consider formula \eqref{c2:eq2.2} for the rate of change of mean fitness along trajectories. By the stated hypothesis,
	\begin{equation}
		\label{c2:eq2.9}
		\sum\limits_{i = 1}^{n}\Big({\bf Bu}\Big)_{i}\Big({\bf Cu}\Big)_{i} u_i = \sum\limits_{\substack{i,j = 1 \\ i \neq j}}^{n}q_{ij}\xi_{i}(t)\xi_{j}(t)u_{i}(t),
	\end{equation}
	where $q_{ij}$ are the entries of ${\bf Q}$ and $\xi_{i}(t) = \Big({\bf Bu}\Big)_{i}$, $i = \overline{1, n}$.
	
	Since ${\bf u} \in \inside S_{n}$, this expression is non-negative, so $\dot{f}({\bf u}(t)) > 0$.
	
	If the system had a closed trajectory, we would need
	$$
	\int\limits_{t}^{T + t}\dot{f}(\tau)\,d\tau = 0,
	$$  
	which is impossible when the form \eqref{c2:eq2.9} is positive semi-definite. 
\end{proof} 

\clearpage

\section{Evolutionarily Stable State as a Maximum of the Fitness Surface}\label{c2:section:2.4}
Another notable property of local maxima of the fitness surface $\Sigma$ is their connection to the concept of an evolutionarily stable state (ESS), introduced in~\cite{c2:Maynard1982} in the context of analyzing replicator dynamics from a game-theoretic perspective. The payoff matrix is taken to be the matrix ${\bf A}$ of the replicator system \eqref{c2:eq1.5}. Each entry $a_{ij}$ represents the payoff to species $i$ when interacting with species $j$, and this payoff is assumed to affect the species' replication rate. Only pure strategies are used: each species employs a fixed strategy (the row $a_{i1},\, a_{i2},\, \ldots,\, a_{in}$) and offspring inherit the parental strategy. The entries $a_{ij}$ are thus the payoffs when strategy $i$ is played against strategy $j$. The fitness of species $i$ against the entire population is
$$
f_{i}({\bf u}) = \sum\limits_{j = 1}^{n}a_{ij}u_{j} = \Big({\bf Au}\Big)_{i}, \quad {\bf u} \in S_{n},
$$    
and the mean population payoff is
$$
\sum\limits_{i = 1}^{n}\Big({\bf Au}\Big)_{i}u_{i} = \Big({\bf Au, u}\Big) = f({\bf u}),
$$ 
which is the mean fitness.

The advantage of strategy $i$ is
$$
\Big({\bf Au}\Big)_{i} - \Big({\bf Au, u}\Big) = \Big({\bf Au}\Big)_{i} - f.
$$

The growth rate of species $i$ is assumed proportional to this advantage:
$$
\dot{u}_{i} = u_{i}\Bigg(\Big({\bf Au}\Big)_{i} - f\Bigg), \quad i = \overline{1, n}, \quad {\bf u} \in S_{n},
$$
which is precisely the general replicator system \eqref{c2:eq1.5}.

This game-theoretic framework allows the methods of~\cite{c2:Maynard1982} to be applied to replicator systems, in particular the concepts of Nash equilibrium and evolutionarily stable state.

We recall the following definitions~\cite{c2:Hofbauer2003}.
\begin{definition}
	A state ${\bf \bar{u}} \in S_{n}$ is a \emph{symmetric Nash equilibrium} of the replicator system \eqref{c2:eq1.5} if
	\begin{equation}
		\label{c2:eq2.10}
		\Big({\bf u, A\bar{u}}\Big) \leqslant \Big({\bf \bar{u}, A\bar{u}}\Big)
	\end{equation}
	for all ${\bf u} \in S_{n}$.
\end{definition} 

\begin{definition}
	A state ${\bf \bar{u}} \in S_{n}$ is \emph{evolutionarily stable} if
	\begin{equation}
		\label{c2:eq2.11}
		\Big({\bf \bar{u}, Au}\Big) > \Big({\bf u, Au}\Big)
	\end{equation}
	for all ${\bf u} \neq {\bf \bar{u}}$ in some neighborhood of ${\bf \bar{u}}$ in $S_{n}$.
\end{definition}  

We are interested in the relationship between states ${\bf \bar{u}} \in S_{n}$ satisfying \eqref{c2:eq2.10} or \eqref{c2:eq2.11} and the extremal properties of the fitness surface $\Sigma$.

\begin{theorem}\label{c2:t2.3}
	If an equilibrium ${\bf \bar{u}} \in S_{n}$ is evolutionarily stable, then it is a strict local maximum of the fitness surface $\Sigma$.
\end{theorem}

\begin{proof}
	Let ${\bf \bar{u}}$ be evolutionarily stable and consider ${\bf u} = {\bf \bar{u}} + \varepsilon {\bf w} \in S_{n}$. Since ${\bf u} \in S_{n}$, we have $\Big({\bf w, I}\Big) = 0$. Condition \eqref{c2:eq2.11} yields
	$$
	\varepsilon\Big({\bf w, A\bar{u}}\Big) + \varepsilon^{2}\Big({\bf w, Aw}\Big) < 0. 
	$$
	
	Since ${\bf \bar{u}}$ is an equilibrium, ${\bf A\bar{u}} = f({\bf \bar{u}}){\bf I}$, so the first term vanishes and
	$$
	\Big({\bf w, Aw}\Big) < 0, \quad {\bf w} \neq 0. \quad 
	$$
\end{proof}

The converse holds only under the additional conditions \eqref{c2:eq2.5} and \eqref{c2:eq2.6}.
\begin{theorem}\label{c2:t2.4}
	Suppose conditions \eqref{c2:eq2.5} and \eqref{c2:eq2.6} of Theorem~\ref{c2:t2.1} hold and ${\bf \bar{u}} \in \inside S_{n}$ is a strict local maximum of the fitness surface $\Sigma$. Then ${\bf \bar{u}}$ is evolutionarily stable.
\end{theorem}

\begin{proof}
	For any ${\bf u} \in S_{n}$,
	$$
	\Big({\bf \bar{u}, Au}\Big) = \Big({\bf \bar{u}, Bu}\Big) + \Big({\bf \bar{u}, Cu}\Big) = \Big({\bf \bar{u}, Bu}\Big) - \Big({\bf C\bar{u}, u}\Big) = \Big({\bf B\bar{u}, u}\Big).
	$$
	
	On the other hand,
	$$
	f({\bf \bar{u}}) = \Big({\bf B\bar{u}, u}\Big) = \Big({\bf A\bar{u}, \bar{u}}\Big).
	$$
	
	Since ${\bf \bar{u}} \in \inside S_{n}$ is a strict local maximum,
	$$
	\Big({\bf A\bar{u}, \bar{u}}\Big) > \Big({\bf u, Au}\Big) \quad \forall\, {\bf u} \in S_{n}.
	$$
	
	Therefore $\Big({\bf \bar{u}, Au}\Big) > \Big({\bf u, Au}\Big)$. 
\end{proof}

Regarding Nash equilibria: if conditions of Theorem~\ref{c2:t2.1} hold and the equilibrium ${\bf \bar{u}} \in \inside S_{n}$ is a local extremum of the fitness surface, then it is a Nash equilibrium. Indeed, condition \eqref{c2:eq2.3} holds and ${\bf A\bar{u} = B\bar{u}} = f({\bf \bar{u}}){\bf I}$, so \eqref{c2:eq2.10} is satisfied with equality for all ${\bf u} \in S_{n}$.

\begin{conseq}\label{c2:con2.1}
	If the equilibrium ${\bf \bar{u}} \in \inside S_{n}$ is a strict local maximum of the fitness surface for a replicator system with a circulant matrix, then ${\bf \bar{u}}$ is a globally stable equilibrium.
\end{conseq}

\begin{proof}
	As shown above, conditions \eqref{c2:eq2.5} and \eqref{c2:eq2.6} are always satisfied for circulant systems. It is known~\cite{c2:Hofbauer2003} that an evolutionarily stable state ${\bf \bar{u}} \in \inside S_{n}$ is a globally stable equilibrium. 
\end{proof}

\begin{theorem}\label{c2:t2.5}
	Suppose the unique equilibrium ${\bf \bar{u}} \in \inside S_{n}$ of system \eqref{c2:eq1.5} coincides with a local maximum of the fitness surface $\Sigma$. If this equilibrium is asymptotically stable, it is evolutionarily stable and furnishes the global maximum of $\Sigma$ on $S_{n}$. If ${\bf \bar{u}}$ is unstable, the global maximum of the fitness surface is attained on the boundary $\bound S_{n}$.
\end{theorem}

\begin{proof}
	Suppose the asymptotically stable equilibrium coincides with a local maximum. Consider the Lyapunov function
	$$
	V({\bf u}) = \sum\limits_{i = 1}^{n}\Big[(u_{i} - \bar{u}_{i}) - \bar{u}_{i}\ln{\frac{u_{i}}{\bar{u}_{i}}}\Big] > 0, \quad {\bf u} \neq {\bf \bar{u}}, \quad V({\bf \bar{u}}) = 0.
	$$
	Then
	\begin{align*}
		\dot{V}({\bf u}) = \sum\limits_{i = 1}^{n}\dot{u}_{i}\Big(1 - \frac{\bar{u}_{i}}{u_{i}}\Big) = \sum\limits_{i = 1}^{n}\Big[\Big({\bf Au}\Big)_{i} - \Big({\bf Au, u}\Big)\Big](u_{i} - \bar{u}_{i}) = \\ = \Big({\bf Au, u - \bar{u}}\Big) < 0, \quad {\bf u} \neq {\bf \bar{u}}.
	\end{align*}
	This inequality implies that ${\bf \bar{u}} \in \inside S_{n}$ is evolutionarily stable and hence globally stable in $\inside S_{n}$ by Corollary~\ref{c2:con2.1}.
	
	Set ${\bf u} = {\bf \bar{u} + \xi}$. Since ${\bf \bar{u}, u} \in S_{n}$,
	$$
	\Big({\bf \xi, I}\Big) = \Big({\bf u, I}\Big) - \Big({\bf \bar{u}, I}\Big) = 0, \quad {\bf I} = (1, 1, \ldots, 1) \in \R^{n}.
	$$
    Using ${\bf A\bar{u}} = {\bf I}\bar{f}$, we obtain
	$$
	\Big({\bf Au, u - \bar{u}}\Big) = \Bigg({\bf A}\Big({\bf u - \bar{u}}\Big), {\bf u - \bar{u}}\Bigg) = \Big({\bf A\xi, \xi}\Big).
	$$
	
	Global stability therefore gives
	$$
	\Big({\bf A\xi, \xi}\Big) = \Big({\bf B\xi, \xi}\Big) < 0
	$$ for all ${\bf \xi} = {\bf u - \bar{u}}$ with ${\bf \bar{u}, u} \in \inside S_{n}$. This shows that the fitness surface attains its global maximum in $\inside S_{n}$, and by continuity and boundedness of $f({\bf u})$ this property extends to all of $S_{n}$.
	
	If the unique equilibrium ${\bf \bar{u}} \in \inside S_{n}$ is unstable, then
	$$
	\Big({\bf \bar{u}, Au}\Big) < \Big({\bf Au, u}\Big)
	$$
	in some neighborhood $U_{{\bf \bar{u}}}$. Since
	$$
	\Big({\bf Au, \bar{u}}\Big) \leqslant \max\limits_{{\bf u} \in S_{n}}\Big({\bf Au, u}\Big), \quad {\bf u} \in U_{{\bf \bar{u}}},
	$$
	and by the conditions of the theorem 
    $${\bf C} = {\bf MB},\quad {\bf MI} = 0, \quad{\bf B\bar{u}} = f({\bf \bar{u}}){\bf I},$$ we have
	\begin{align*}
		\Big({\bf \bar{u}, Au}\Big) = \Big({\bf \bar{u}, (B + C)u}\Big) = \Big({\bf \bar{u}, Bu}\Big) + \Big({\bf \bar{u}, Cu}\Big) = \\ = \Big({\bf B\bar{u}, u}\Big) - \Big({\bf C\bar{u}, u}\Big) = f({\bf \bar{u}})\Big({\bf I, u}\Big) - \Big({\bf MB\bar{u}, u}\Big) = f({\bf \bar{u}}).
	\end{align*}
	Therefore,
	$$
	f({\bf \bar{u}}) < \max\limits_{{\bf u} \in S_{n}}\Big({\bf Au, u}\Big).
	$$
	
	The maximum of the fitness surface is thus not attained at ${\bf \bar{u}} \in \inside S_{n}$. By Theorem~\ref{c2:t2.3}, it cannot be attained at any interior local maximum, so the global maximum of $\Sigma$ is attained on $\bound S_{n}$. 
\end{proof}

\begin{conseq}\label{c2:conseq2.2} \cite{c2:Drozhzhin2021} 
	If the equilibrium of a replicator system with a circulant matrix is asymptotically stable, it is evolutionarily stable and furnishes the global maximum of the fitness surface $\Sigma$ on $S_{n}$. If it is unstable, the global maximum is attained on the boundary of the simplex.
\end{conseq}

\medskip
\textbf{Example 5. A circulant replicator system.}

Consider the matrices ${\bf A}$ and ${\bf B}$:

\begin{equation*}
	{\bf A} = 
	\begin{pmatrix}
		0 & a_{1} & a_{2} & a_{3}\\
		a_{3} & 0 & a_{1} & a_{2}\\
		a_{2} & a_{3} & 0 & a_{1}\\
		a_{1} & a_{2} & a_{3} & 0\\
	\end{pmatrix}\!,
	\quad
	{\bf B} = 
	\begin{pmatrix}
		0 & \alpha & \beta & \alpha\\
		\alpha & 0 & \alpha & \beta\\
		\beta & \alpha & 0 & \alpha\\
		\alpha & \beta & \alpha & 0\\
	\end{pmatrix}\!, 
\end{equation*} 
$$	\alpha = \frac{a_{1} + a_{3}}{2}, \quad \beta = a_{2}.
$$
The eigenvalues of ${\bf B}$ are
$$
\lambda_{1} = 2\alpha + \beta, \quad \lambda_{2, 3} = -\beta, \quad \lambda_{4} = \beta - 2\alpha.
$$
The first eigenvalue corresponds to the eigenvector ${\bf u}^{1} = (1,\, 1,\, 1,\, 1)$, which is orthogonal to the simplex $S_{4}$ and is thus excluded from consideration.

For $\beta > 0$ and $\beta < 2\alpha$, the equilibrium ${\bf \bar{u}} = (1/4,\, 1/4,\, 1/4,\, 1/4)$ is asymptotically stable, and the maximum of the fitness surface is
$$
f({\bf \bar{u}}) = \frac{2\alpha + \beta}{4}.
$$

For $\beta < 0$, the maximum of $\Sigma$ is attained at the boundary points of the simplex $S_{4}$:\, $p_{1} = (1/2,\, 1/2,\, 0,\, 0)$, $p_{2} = (0,\, 1/2,\, 1/2,\, 0)$, $p_{3} = (1/2,\, 0,\, 0,\, 1/2)$, $p_{4} = (0,\, 0,\, 1/2,\, 1/2)$:
$$
f(p_{i}) = \frac{\alpha}{2} > \frac{2\alpha + \beta}{4}, \quad \beta < 0.
$$
\hfill$\square$

\bigskip
We now turn to Eigen's quasispecies model. In this case, the fitness surface is a linear function defined on the convex set $S_{l}$:
$$
{\bf \bar{w}}({\bf p}(t)) = \sum\limits_{i = 1}^{n}w_{i}p_{i}(t), \quad {\bf p}(t) \in S_{l}.
$$
Consequently, both the maximum and minimum of ${\bf \bar{w}}(t)$ are attained at the vertices of $S_{l}$.

Assuming $w_{1} > w_{2} \geqslant \ldots \geqslant w_{l}$,
$$
\max_{{\bf p} \in S_{l}}\sum\limits_{i = 1}^{l}w_{i}p_{i} = w_{1}, \quad \min_{{\bf p} \in S_{l}}\sum\limits_{i = 1}^{l}w_{i}p_{i} = w_{l}.
$$
Here, ${\bf p}^{1} = (1,\, 0,\, \ldots,\, 0)$ and ${\bf p}^{l} = (0,\, \ldots,\, 0,\, 1)$ are the vertices of the simplex at which the maximum and minimum of the fitness surface are respectively attained.

From the Eigen model \cite{c2:Eigen1971}, one derives the following expression for the rate of change of the fitness surface along system trajectories:
\begin{equation}
	{\bf \dot{\bar{w}}}({\bf p}(t)) = \sum\limits_{i = 1}^{l}w_{i}\dot{p}_{i}(t) = \Big({\bf QWp}(t), {\bf w}\Big) - {\bf \bar{w}}^{2}(t).
	\label{c2:eq2.12}
\end{equation}
From this expression, we can determine the sign of $\dot{\bar{w}}$ along trajectories emanating from the vertices ${\bf p}^{1}$ and ${\bf p}^{l}$ of $S_{l}$:
\begin{align*}
    {\bf \dot{\bar{w}}}({\bf p}^{1}) 
        &= w_{1}\sum\limits_{j = 1}^{l}q_{1j}w_{j} - w_{1}^{2} 
        < w_{1}^{2}\sum\limits_{j = 1}^{l}q_{1j} - w_{1}^{2} = 0,\\[6pt]
    {\bf \dot{\bar{w}}}({\bf p}^{l}) 
        &= w_{l}\sum\limits_{j = 1}^{l}q_{lj}w_{j} - w_{l}^{2} 
        > w_{l}^{2}\sum\limits_{j = 1}^{l}q_{lj} - w_{l}^{2} = 0.
\end{align*}

By continuity, these inequalities remain in force for trajectories starting at points ${\bf p} \in \inside S_{l}$ in a neighborhood of ${\bf p}^{1}$ or ${\bf p}^{l}$. Consequently, along trajectories $\gamma_{t}$ of the quasispecies system, mean fitness may either increase or decrease. This does not, however, mean that the equilibrium ${\bf \bar{p}} \in \inside S_{l}$ fails to be extremal for trajectories $\gamma_{t}$ as $t \to +\infty$. In particular, from formula \eqref{c2:eq2.12} it follows that
\begin{align*}
	\lim\limits_{t \to +\infty}{\bf \dot{\bar{w}}}({\bf p}(t)) = \Big({\bf QW\bar{p}, w}\Big) - {\bf \bar{w}}^{2} = 0, \\
    {\bf \bar{w}} = \lim\limits_{t \to +\infty}{\bf \bar{w}(p}(t)),\quad {\bf \bar{p}} = \lim\limits_{t \to +\infty}{\bf p}(t).
\end{align*} 
Direct computation shows that
$$
\lim\limits_{t \to +\infty}\frac{d^s}{dt^{s}}{\bf \bar{w}(p}(t)) = 0, \quad s = 2,\, 3,\, \ldots,
$$
so ${\bf \bar{w}(p}(t)) \to {\bf \bar{w}}$ exponentially as $t \to +\infty$.

\clearpage

\section{The Lotka--Volterra System}\label{c2:section:2.5}
Consider the general Lotka--Volterra system in $\R_{+}^{n}$:
\begin{align}
	\frac{du_{i}}{dt} = u_{i}\Big(r_{i} - \Big({\bf Au}\Big)_{i}\Big), \quad i = \overline{1, n}, \\
    \Big({\bf Au}\Big)_{i} = \sum\limits_{j = 1}^{n}a_{ij}u_{j}, \quad {\bf r} = (r_{1},\, r_{2},\, \ldots,\, r_{n}).
	\label{c2:eq2.13}
\end{align}
The natural question is what function plays the role of mean fitness for this system. This function was proposed in~\cite{c2:Svirezhev1978}:
$$
F({\bf u}) = \sum\limits_{i = 1}^{n}r_{i}u_{i} - \frac{1}{2}\Big({\bf Au, u}\Big).
$$
A direct calculation gives
\begin{equation}
	\dot{F}({\bf u}(t)) = \sum\limits_{i = 1}^{n}u_{i}\Big(r_{i} - \Big({\bf Bu}\Big)_{i}\Big)^{2} + \sum\limits_{i = 1}^{n}\Big({\bf Cu}\Big)_{i}\Big(\Big({\bf Bu}\Big)_{i} - r_{i}\Big)u_{i},
	\label{c2:eq2.14}
\end{equation}
using decomposition \eqref{c2:eq2.1} of ${\bf A}$.

Formula \eqref{c2:eq2.14} is the analogue of \eqref{c2:eq2.2}. In particular, when ${\bf A}$ is symmetric, $\dot{F}({\bf u}(t)) \geqslant 0$, just as for replicator systems.

\begin{theorem}\label{c2:t2.6}
	Let ${\bf \bar{u}} \in \inside \mathbb{R}_+^n$ be the unique stable equilibrium of system \eqref{c2:eq2.13}. This equilibrium furnishes the global maximum of the surface $z=F({\bf u})$ on $\mathbb{R}_+^n$ if and only if the matrix ${\bf C}$ in decomposition \eqref{c2:eq2.1} admits the representation
	\begin{equation}\label{c2:eq2.15}
		{\bf C=MB}, \quad |{\bf M}| \neq 0, \quad {\bf Mr=0}.
	\end{equation}
\end{theorem}
\begin{proof}
	Assuming \eqref{c2:eq2.15} and ${\bf \bar{u}} \in \mathbb{R}_+^n$, we have 
    $${\bf A\bar{u}} = {\bf r} \text{ with } |{\bf A}| = 0.$$
	Since ${\bf \bar{u}} \in \inside \R_{+}^{n}$ is a stable equilibrium, the matrix ${\bf B}$ in \eqref{c2:eq2.1} is positive definite~\cite{c2:Svirezhev1978}.
	
	Condition \eqref{c2:eq2.15} is sufficient to ensure
	\begin{equation}
		{\bf B\bar{u}} = {\bf r}.
		\label{c2:eq2.16}
	\end{equation}
	The necessary condition for ${\bf \bar{u}}$ to be a critical point of $z = F({\bf u})$ is satisfied:
	$$
	\frac{\partial F({\bf u})}{\partial u_{k}}\Bigg|_{{\bf u} = {\bf \bar{u}}} = r_{k} - \Big({\bf B\bar{u}}\Big)_{k} = 0, \quad k = \overline{1, n}.
	$$
	The sufficient condition for a maximum is the negative definiteness of the Hessian $\Big(z_{kj}\Big)_{k, j = 1}^{n}$:
	$$
	z_{kj} = \frac{\partial^{2}F({\bf \bar{u}})}{\partial u_{k}\partial u_{j}} = -b_{kj}, \quad k, j = \overline{1, n}.
	$$
	
	This holds since ${\bf B}$ is positive definite. The surface $z = F({\bf u})$ is therefore strictly concave on $\R_{+}^{n}$, and the equilibrium ${\bf \bar{u}} \in \inside \R_{+}^{n}$ furnishes its global maximum.
	The necessity of \eqref{c2:eq2.15} is proved by the same argument as the analogous statement of Theorem~\ref{c2:t2.1} for replicator systems. 
\end{proof}

\begin{conseq}\label{c2:con2.4}
	Suppose condition \eqref{c2:eq2.15} holds. If the unique equilibrium ${\bf \bar{u}} \in \inside \R_{+}^{n}$ is unstable, the maximum of the surface $z = F({\bf u})$ is not attained at the equilibrium.
\end{conseq}

\medskip
\textbf{Example 6. The Lotka--Volterra with the circulant matrix.}

Consider the Lotka--Volterra system in $\R_{+}^{3}$ with the circulant matrix
$$
{\bf A} = 
\begin{pmatrix}
	1 & \alpha & \beta\\
	\beta & 1 & \alpha\\
	\alpha & \beta & 1\\
\end{pmatrix}\!,
\quad
{\bf r} = (1, 1, 1).
$$
In this case,
$$
{\bf B} = 
\begin{pmatrix}
	1 & \mu & \mu\\
	\mu & 1 & \mu\\
	\mu & \mu & 1\\
\end{pmatrix}\!,
\quad
\mu = \frac{\alpha + \beta}{2}.
$$

It is known~\cite{c2:Hofbauer2003} that for $\mu < 1$ the equilibrium ${\bf \bar{u}} = \left(\dfrac{1}{1 + 2\mu},\, \dfrac{1}{1 + 2\mu},\, \dfrac{1}{1 + 2\mu}\right)$ is stable and loses stability when $\mu > 1$. Condition \eqref{c2:eq2.15} is satisfied, since ${\bf C\bar{u}} = 0$.
The maximum of the surface $z = u_{1} + u_{2} + u_{3} - \dfrac{1}{2}(u_{1}^{2} + u_{2}^{2} + u_{3}^{2}) + 2\mu (u_{1}u_{2} + u_{1}u_{3} + u_{2}u_{3})$ is attained at the equilibrium ${\bf \bar{u}}$:
$$
F_{\max} = \frac{3}{2}\frac{1}{(1 + 2\mu)}.
$$ 
The maximum of $F$ on the boundary faces $u_{1} = 0$, $u_{2} = 0$, or $u_{3} = 0$ equals
$$
F_{\bound} = \frac{1}{1 + \mu}.
$$

The inequality $F_{\max} > F_{\bound}$ holds for $\mu < 1$, i.e.\ precisely when the equilibrium is stable. For $\mu > 1$ the equilibrium becomes unstable and the maximum of $F({\bf u})$ is attained at the boundary points
$$
\Big(\frac{1}{1 + \mu},\, \frac{1}{1 + \mu},\, 0\Big), \quad \Big(0,\, \frac{1}{1 + \mu},\, \frac{1}{1 + \mu}\Big),\quad \Big(\frac{1}{1 + \mu},\, 0,\, \frac{1}{1 + \mu}\Big).
$$

A heteroclinic orbit then emerges, tending toward the limit set consisting of the saddle separatrices lying in the coordinate planes and successively passing through the points $(1, 0, 0)$, $(0, 1, 0)$, and $(0, 0, 1)$.
\hfill$\square$

\clearpage

\section{Conclusion}\label{c2:section:2.6}
In general, trajectories of replicator systems do not reach the maximum of the fitness surface, even when a unique asymptotically stable equilibrium exists (Example~2). Coincidence is possible only in special cases where the dynamics of the replicator system are aligned in a specific way with the geometry of the fitness surface. Additional conditions are required for this to occur. A nontrivial class of matrices satisfying these conditions is furnished by circulant and symmetric matrices.

There is a clear connection between evolutionarily stable equilibria and states furnishing local and global maxima of the fitness surface. An evolutionarily stable equilibrium furnishes a local maximum of the fitness surface. If the unique asymptotically stable equilibrium is a local maximum of the fitness surface, then it is evolutionarily stable and realizes the global maximum. If instead the unique equilibrium is unstable, the global maximum is attained on the boundary. The fitness surface of the general Lotka--Volterra system shares these properties. When the fitness surface has multiple local maxima, closed and heteroclinic trajectories may appear, arising when the unique equilibrium of the system is unstable (Example~3).


%% file: chapter3.tex
\chapter{Mathematical Model of Evolution of Non-Degenerate Replicator Systems}
\label{ch:3}

\begin{chapterabstract}
We propose and analyse a mathematical model of evolutionary adaptation for
non-degenerate (permanent) replicator systems, in which the fitness landscape
matrix ${\bf A}(\tau)$ evolves on a slow timescale $\tau$ --- the
\textit{evolutionary time} --- while the species dynamics unfold on a fast
timescale. Under a two-timescale separation justified by Tikhonov's theorem,
the adaptation problem reduces to maximising the mean fitness $\bar{f}(\tau)$
at steady state over a convex admissible set $\mathcal{M}$ of fitness landscape
matrices. We derive a fitness variation formula and establish necessary and
sufficient conditions for a fitness maximum (Theorem~\ref{c3:theorem:3.1}), showing that the
optimisation reduces at each step to a linear programming problem. The algorithm
is applied to four canonical replicator systems: the hypercycle, the
bi-hypercycle, the anthill system, and the RNA molecule network. In all cases
the evolutionary process follows a universal three-phase pattern: an initial
phase of fitness growth without equilibrium shift, during which purely altruistic
replication gives way to mixed altruistic--selfish behaviour; a second phase of
dominant species emergence; and a stabilisation phase analogous to the error
catastrophe threshold in quasispecies models. A key consequence is that all
evolved systems acquire resistance to parasitic species that would destroy the
original system. We further prove that without non-degeneracy constraints the
process leads to sequential species annihilation, with a provable spectral lower
bound on how much fitness increases by dimension reduction.
\end{chapterabstract}

\bigskip\hrule\bigskip

\section*{Introduction}\addcontentsline{toc}{section}{Introduction}

This chapter is an edited English version of Chapter~3 of the monograph
\textit{Mathematical Models of Evolution and Replicator Systems
Dynamics}~\cite{c3:BratusBook2022}. It builds on the replicator framework of
Chapter~\ref{ch:1} and on the geometry of the fitness surface studied in
Chapter~\ref{ch:2}, but can be read independently of either: the necessary
background is summarised below.

\subsection*{Replicator systems: a brief summary}

A \textit{replicator} is any entity capable of self-reproduction, heritable
variation, and competition for resources~\cite{c3:Eigen1971,c3:Eigen1979}. The
dynamics of $n$ competing species, represented by their relative frequencies
${\bf u}(t) = (u_1(t), \ldots, u_n(t)) \in S_n$, where
$$
	S_n = \left\{ {\bf u} \geqslant 0 : \sum_{i=1}^n u_i = 1 \right\}
$$
is the standard simplex, is governed by the \textit{replicator equation}
\begin{equation}
	\frac{du_i}{dt} = u_i\Big[\Big({\bf Au}\Big)_i - f({\bf u})\Big],
	\quad f({\bf u}) = \Big({\bf Au}, {\bf u}\Big), \quad i = \overline{1,n}.
	\label{c3:eq1.5}
\end{equation}
Here $\Big({\bf Au}\Big)_i = \sum_{j=1}^n a_{ij} u_j$ is the \textit{fitness}
of species $i$, $f({\bf u})$ is the \textit{mean fitness} of the population,
and the matrix ${\bf A} = (a_{ij})$ defines the \textit{fitness landscape}.
System~\eqref{c3:eq1.5} is studied on $S_n$; the scalar product
$(\cdot,\cdot)$ denotes the standard inner product in $\mathbb{R}^n$.

A replicator system is called \textit{non-degenerate} (or \textit{permanent}) if
there exists a compact set $K \subset \inside S_n$ such that every trajectory
starting in $\inside S_n$ eventually enters and remains in $K$. Permanence
guarantees long-run coexistence of all species; its characterisation for specific
system types is treated in~Chapter~\ref{ch:1}.

Three canonical replication regimes were analysed in~Chapter~\ref{ch:1}: independent
(linear) replication with $a_{ij} = k_i \delta_{ij}$; autocatalytic replication;
and the \textit{hypercycle}
\begin{equation}
	\frac{du_i}{dt} = u_i(u_{i-1} - f({\bf u})), \quad u_0 = u_n,
	\quad f({\bf u}) = \sum_{i=1}^n u_i u_{i-1},
	\label{c3:eq1.8}
\end{equation}
which is permanent for all $n \geqslant 2$ and possesses a unique interior
equilibrium. The geometry of the mean fitness surface
$\Sigma = \{ z = f({\bf u}) : {\bf u} \in S_n \}$ and its relationship to
trajectory dynamics and evolutionarily stable states were studied
in~Chapter~\ref{ch:2}; in particular, it was shown there that Fisher's
fundamental theorem of natural selection holds for system~\eqref{c3:eq1.5}
only when ${\bf A}$ is symmetric.

Two further special systems introduced in~Chapter~\ref{ch:1} are used as
illustrations in the present chapter:
\begin{itemize}
\item the \textit{``anthill'' system}, in which species $0, \ldots, n-1$ form a
hypercycle and are additionally catalysed by a dominant ``queen'' species $n$:
\begin{equation}
	\begin{aligned}
		&\dot{u}_i = u_i\Big(\alpha u_n + k_i u_{i-1} - f(t)\Big),
		\quad i = \overline{0, n-1}, \\
		&\dot{u}_n = u_n\!\left(\sum_{i=0}^{n-1}\beta_i u_i - f(t)\right), \quad
		{\bf u} \in S_{n+1},
	\end{aligned}
	\label{c3:eq1.21}
\end{equation}
with $\alpha, \beta_i, k_i > 0$ and $u_{-1} = u_{n-1}$;
\item the \textit{RNA network}~\cite{c3:Vaidya2012}, describing six RNA macromolecules
split into two groups: species $4$--$5$--$6$ forming a hypercycle, and species
$1$--$2$--$3$ exhibiting both hypercyclic and autocatalytic replication:
\begin{equation}
	\begin{aligned}
		&\dot{u}_1 = u_1(r_1 u_1 + k_1 u_4 - f), \quad
		 \dot{u}_2 = u_2(r_2 u_2 + k_2 u_5 - f), \quad
		 \dot{u}_3 = u_3(r_3 u_3 + k_3 u_6 - f), \\
		&\dot{u}_4 = u_4(k_4 u_3 + \bar{k}_4 u_5 - f), \quad
		 \dot{u}_5 = u_5(k_5 u_1 + \bar{k}_5 u_6 - f), \quad
		 \dot{u}_6 = u_6(k_6 u_2 + \bar{k}_6 u_4 - f),
	\end{aligned}
	\label{c3:eq1.25}
\end{equation}
with ${\bf u} \in S_6$ and $r_i, k_i, \bar{k}_i > 0$.
\end{itemize}

\subsection*{Contribution of the present chapter}

The central question addressed here is: \textit{how does the fitness landscape
itself evolve?} The standard assumption in replicator theory --- that the
landscape matrix ${\bf A}$ is fixed --- is relaxed. We propose a mathematical
model of \textit{evolutionary adaptation} in which ${\bf A}$ changes on a slow
timescale $\tau$ (the \textit{evolutionary time}), much slower than the fast
timescale of the active replicator dynamics. Under this separation of scales, the
adaptation problem reduces, via Tikhonov's theorem~\cite{c3:Tihonov1948}, to
maximising the mean fitness $\bar{f}(\tau)$ at steady state over an admissible
convex set $\mathcal{M}$ of fitness landscape matrices.

The chapter derives necessary and sufficient conditions for a fitness maximum in
terms of a fitness variation formula (Theorem~\ref{c3:theorem:3.1}), shows that the
resulting optimisation reduces to a sequence of linear programming problems
(Section~\ref{c3:section:3.3}), and applies the algorithm to four replicator
systems: the hypercycle, the bi-hypercycle, the anthill system, and the RNA
network (Sections~\ref{c3:section:3.4}--\ref{c3:section:3.7}). The main findings,
summarised in Section~\ref{c3:section:3.8}, are: (i) evolutionary adaptation
universally increases mean fitness and induces a structural transition from
altruistic to mixed altruistic--selfish replication; (ii) the evolved systems
acquire resistance to parasitic species; (iii) without non-degeneracy
constraints, the process leads to sequential species annihilation, with a
provable lower bound on how much fitness can increase by dimension reduction.

These results extend and complement the earlier paper~\cite{c3:Drozhzhin2021}, which
established fitness optimisation and permanence for a related class of systems.

\section{Model of Evolutionary Adaptation}\label{c3:section:3.1}
The ideas of S.~Wright~\cite{c3:Wright1930} and Fisher's fundamental theorem of natural selection --- \textit{``The rate of increase in fitness of any organism at any time is equal to its genetic variance in fitness at the time''}~\cite{c3:Fisher1930} --- initiated the application of extremal principles in the theory of biological evolution. Fisher did not tie his theorem to any particular replicator system, nor did he ever give a precise meaning to the notion of ``genetic variance in fitness''. Accordingly, despite the established use of the word ``theorem'' in the scientific literature, Fisher's statement lacks a rigorous mathematical proof and is better regarded as an additional postulate supplementing the general theory of evolution. Nevertheless, Fisher's fundamental theorem is widely employed in theoretical biology as well as in numerous studies in mathematical biology, including~\cite{c3:Crow2002,c3:Dieckmann2004,c3:Ewens2004,c3:Grafen2003}. For replicator models of the form~\eqref{c3:eq1.5}, identifying the mean fitness with the function $f({\bf u})$, it was shown in~Chapter~\ref{ch:2} that Fisher's theorem holds only when the interaction matrix is symmetric, that is, in the case of a diploid population. The notion of ``genetic variance in fitness'' is commonly identified with the mean fitness function. Many authors, however, consider these two concepts to be non-equivalent in the general case~\cite{c3:Ao2005}.

In theoretical biology, the mean fitness landscape is often depicted as a stationary geometric object consisting of peaks and valleys, while the evolutionary trajectory of a species is represented as a path that, despite occasional descents, ultimately climbs through saddle points toward one of the peaks~\cite{c3:Poelwijk2007} (see Fig.~\ref{c3:fig3.1}). From the mathematical standpoint, Fisher's theorem postulates the existence of a Lyapunov-type function that increases monotonically along the system trajectories during evolutionary adaptation. Such a function, however, can be constructed only in special cases, for instance when the replicator system possesses a unique stable equilibrium. The use of the fitness function in the theory of biological organisation has a long history~\cite{c3:Fontana1994}. As noted above, the function $f({\bf u})$ defining the mean fitness of the replicator system~\eqref{c3:eq1.5} satisfies Fisher's theorem only when the matrix ${\bf A}$ is symmetric. A detailed account of the history of fitness maximisation in population genetics can be found in~\cite{c3:Grodwohl2017}. The question of whether general laws of biology can be expressed in the language of the exact sciences (mathematics, physics, chemistry) lies at the heart of many important investigations~\cite{c3:Burger2000,c3:Ewens2004,c3:Gavrilets2004,c3:Kimura1983,c3:Rice2004,c3:Wright1932}. These studies typically assume that the fitness landscape remains fixed throughout the system's temporal evolution until the limiting state is reached.

\begin{figure}[ht]
	\centering
	\includegraphics[width=0.65\textwidth]{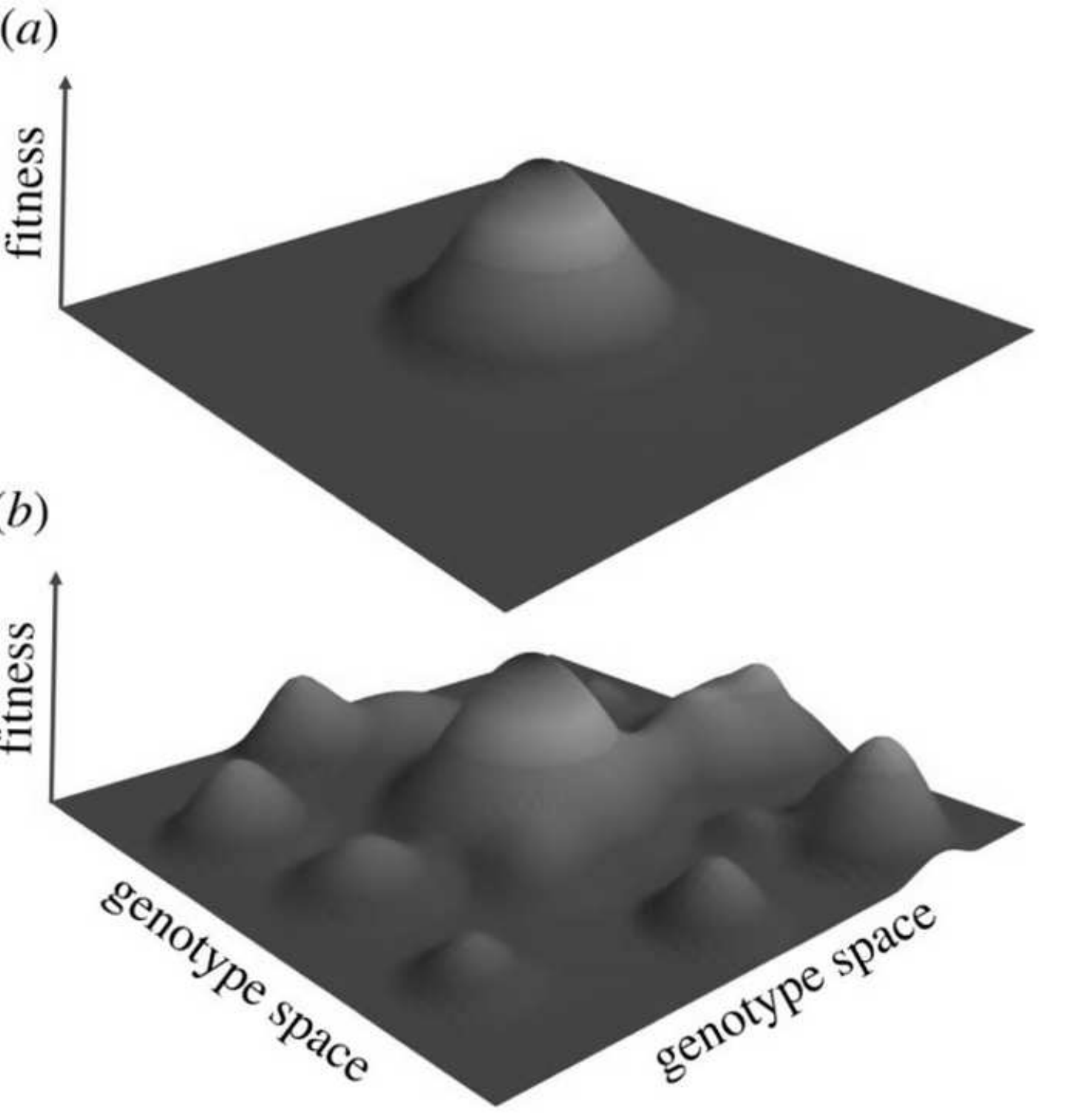}
	\caption{Schematic representation of the mean fitness landscape~\cite{c3:Pesce2016}.}
	\label{c3:fig3.1}
\end{figure}

The situation changes fundamentally when one allows the fitness landscape itself to undergo adaptive modification. This raises the natural question of the mechanism by which such a process might proceed. \textit{The central hypothesis of the proposed mathematical model is that the timescale on which the fitness landscape adapts may be many times slower than the timescale of the active dynamics of the system up to its approach to the stationary state.} In other words, we assume that the evolution of the system's fitness landscape unfolds on a distinct timescale that does not coincide with the timescale of the active dynamics.

This assumption implies that the process of evolutionary adaptation can be described not only through the dynamical equations but also through steady-state equations, all of whose elements depend on a parameter $\tau$, referred to as the \textit{evolutionary parameter} or \textit{evolutionary time}. It is then necessary to characterise the admissible set of fitness landscapes.

For a general replicator system of the form \eqref{c3:eq1.5}, the admissible set of fitness landscapes that the function $f({\bf u}(t))$ may realise is described by the set of non-degenerate matrices ${\bf A}(\tau)$ whose elements depend smoothly on the evolutionary parameter $\tau \geqslant 0$ and whose spherical norm remains bounded for all $\tau \geqslant 0$:
\begin{equation}
	\mathcal{M} = \left\{{\bf A}(\tau) = (a_{ij})_{i, j = 1}^{n}:\,\, \sum
	\limits_{i, j = 1}^{n}a_{ij}^{2}(\tau) \leqslant M = \const > 0\right\}\!.
	\label{c3:eq3.1}
\end{equation}    
The set \eqref{c3:eq3.1} of matrices ${\bf A}(\tau)$, describing the admissible fitness landscapes of the replicator system under variation of the evolutionary time, is convex.

Indeed, let ${\bf A}_1(\tau),\, {\bf A}_2(\tau) \in \mathcal{M}$. Consider the matrix
\begin{equation*}
	{\bf A}_{\lambda}(\tau) = \lambda{\bf A}_1(\tau)+(1-\lambda){\bf A}_2(\tau), \quad 0 < \lambda < 1.
\end{equation*}
We show that ${\bf A}_{\lambda}(\tau) \in \mathcal{M}$:
\begin{gather*}
	\sum\limits_{i,j=1}^n \left(\lambda a_{ij}^{(1)}+(1-\lambda) a_{ij}^{(2)}\right)^2 =
	\lambda^2 \sum\limits_{i,j=1}^n \left(a_{ij}^{(1)}\right)^2
	+ (1-\lambda)^2 \sum\limits_{i,j=1}^n \left(a_{ij}^{(2)}\right)^2
	+ 2(1-\lambda)\lambda \sum\limits_{i,j=1}^n a_{ij}^{(1)} a_{ij}^{(2)} \leqslant \\
	\leqslant \left(\lambda^2 + (1-\lambda)^2\right) M
	+ 2(1-\lambda)\lambda
	\left(\sum\limits_{i,j=1}^n \left(a_{ij}^{(1)}\right)^2\right)^{1/2}
	\cdot
	\left(\sum\limits_{i,j=1}^n \left(a_{ij}^{(2)}\right)^2\right)^{1/2}\leqslant \\
	= \left(\lambda^2 + 2\lambda(1-\lambda) + (1-\lambda)^2\right) M = M.
\end{gather*}

The stationary equilibrium (not necessarily stable) satisfies the linear system
\begin{equation}
	\begin{aligned}
		&{\bf A}(\tau){\bf \bar{w}}(\tau) = \bar{f}(\tau){\bf I}, \quad {\bf I} = (1,\, 1,\, \ldots,\, 1) \in \mathbb{R}^{n}, \\ &\bar{f}({\bf \bar{w}}) = \Big({\bf A}(\tau){\bf \bar{w}}(\tau), {\bf \bar{w}}(\tau)\Big), \quad {\bf \bar{w}} \in \inside S_{n},
	\end{aligned}
	\label{c3:eq3.2}
\end{equation} 
which governs the evolution of the equilibrium position as a function of the evolutionary time. 

From the mathematical standpoint, the hypothesis of a slowly varying fitness landscape means that the matrix elements $a_{ij}$ can be written as
$$
	a_{ij} = a_{ij}(\varepsilon t) = a_{ij}(\tau), \quad i, j = \overline{1, n},
$$
where $\varepsilon$ is a sufficiently small parameter.

The general equation governing the evolutionary dynamics of the replicator system takes the form
\begin{equation}
	\begin{aligned}
		&\displaystyle\frac{du_{i}(t,\tau)}{dt} = u_{i}(t,\tau)\left[\Big({\bf A}(\tau){\bf u}(t,\tau)\Big)_{i} - f({\bf u}(t,\tau))\right], \\
		&\displaystyle\frac{da_{ij}(\tau)}{d\tau} = v_{ij}, \quad |v_{ij}| \leqslant \delta, \quad i, j = \overline{1, n}, \\
		&{\bf u}(t,\tau) \in S_{n}.
	\end{aligned}
	\label{c3:eq3.3}
\end{equation}
Here $\delta$ is a positive number.

If for each fixed value of $\tau \geqslant 0$ the system \eqref{c3:eq3.3} remains non-degenerate (permanent), then the time-averaged values of the functions $u_i(t, \tau)$, $i = 1, 2, \ldots, n$ and $f({\bf u}(t, \tau))$ coincide with the equilibrium of system \eqref{c3:eq3.3}:
\begin{equation}
	\bar{u}_{i}(\tau) = \lim_{t \to +\infty}\frac{1}{t}\int_{0}^{t}u_{i}(t, \tau)\,dt, \qquad \bar{f}({\bf \bar{u}}(\tau)) = \lim_{t \to +\infty}\frac{1}{t}\int_{0}^{t}f({\bf u}(t, \tau))\,dt.
	\label{c3:eq3.4}
\end{equation}

Introducing the slow evolutionary time $\tau = \varepsilon t$, $0 \leqslant t \leqslant T$,
system~\eqref{c3:eq3.3} becomes
\begin{align*}
&\varepsilon\frac{du_i(\tau/\varepsilon, \tau)}{d\tau} =
u_i(\tau/\varepsilon, \tau)
\left[\Big({\bf A}(\tau){\bf u}(\tau/\varepsilon, \tau)\Big)_i
- f({\bf u}(\tau/\varepsilon, \tau))\right],
\\
&\frac{da_{ij}(\tau)}{d\tau} = v_{ij}, \quad i,j = \overline{1,n}.
\end{align*}
Letting $\varepsilon \to 0$ and using~\eqref{c3:eq3.4}, we obtain (in accordance with Tikhonov's theorem~\cite{c3:Tihonov1948}):
$$
	\begin{aligned}
		&{\bf A}(\tau){\bf \bar{u}}(\tau) = \bar{f}({\bf \bar{u}}(\tau)){\bf I}, \\
		&\frac{d{\bf A}(\tau)}{d\tau} = {\bf V}, \quad {\bf V} = (v_{ij})_{i, j = 1}^{n}.
	\end{aligned}
$$
The matrix ${\bf V}$ describes the rates of change of the fitness landscape elements.

The problem of finding the evolutionary change of the system over the admissible set of fitness landscapes \eqref{c3:eq3.1} reduces to finding the maximum value of the mean fitness $\bar{f}(\tau)$ over the set of solutions to equation \eqref{c3:eq3.2}. {\it In other words, the evolutionary adaptation problem for a replicator system reduces to applying Fisher's fundamental theorem of natural selection at the stationary equilibrium under variation of the evolutionary time.}

These conditions allow to reconstruct the active dynamics that may have occurred at any given stage of evolutionary adaptation. Suppose, for instance, that the fitness maximisation problem has been solved up to some evolutionary time $\tau^{*}$. This means that the fitness landscape matrix ${\bf A}(\tau^{*}) \in \mathcal{M}$ is known. The corresponding active dynamics at this stage is governed bys
$$
	\frac{du_{i}(t)}{dt} = u_{i}(t)\left[\Big({\bf A}(\tau^{*}){\bf u}(t)\Big)_{i} - \Big({\bf A}(\tau^{*}){\bf u}(t), {\bf u}(t)\Big)\right],
	\quad u_{i}(0) = u_{i}^{0}, \quad i = \overline{1, n}, \quad {\bf u}(0) \in S_{n}.
$$

The proposed method of adaptive fitness landscape modification will be applied below to the four systems introduced above: the hypercycle~\eqref{c3:eq1.8}, the bi-hypercycle (Section~\ref{c3:section:3.5}), the anthill system~\eqref{c3:eq1.21}, and the RNA molecule network~\eqref{c3:eq1.25} (see~Chapter~\ref{ch:1} for detailed analysis of each).

\clearpage

\section{Fitness Variance. Necessary and Sufficient Conditions for a Fitness Maximum}\label{c3:section:3.2}

Suppose that for some value of the evolutionary parameter $\tau \geqslant 0$, the vector ${\bf \bar{u}}(\tau) \in \inside S_{n}$ and the fitness landscape matrix ${\bf A}(\tau)$ satisfy equation \eqref{c3:eq3.2}. We introduce the following assumption.

\begin{assumpt}\label{c3:assumpt3.1}
	The elements of the non-degenerate matrix ${\bf A}(\tau)$ and the components of the vector ${\bf \bar{u}}(\tau)$ are twice continuously differentiable functions of the evolutionary parameter $\tau$.
\end{assumpt}

Under Assumption~\ref{c3:assumpt3.1}, for sufficiently small increments $\Delta \tau > 0$ the following Taylor expansions hold:
\begin{equation}
	\begin{aligned}	
		&\displaystyle{\bf A}(\tau + h) = {\bf A}(\tau) + \delta{\bf A}(\tau)\Delta\tau + \dfrac{1}{2}\delta^{2}{\bf A}(\tau)\Delta\tau^{2} + o(\Delta\tau^{2}), \\		
		&\displaystyle{\bf \bar{u}}(\tau + h) = {\bf \bar{u}}(\tau) + \delta {\bf \bar{u}}(\tau)\Delta\tau + \dfrac{1}{2}\delta^{2}{\bf \bar{u}}(\tau)\Delta\tau^{2} + o(\Delta\tau^{2}), \\		
		&\displaystyle\bar{f}(\tau + h) = \bar{f}(\tau) + \delta\bar{f}(\tau)\Delta\tau + \dfrac{1}{2}\delta^{2}\bar{f}(\tau)\Delta\tau^{2} + o(\Delta\tau^{2}).
	\end{aligned}
	\label{c3:eq3.5}
\end{equation}
Here $\delta{\bf A}(\tau)$ and $\delta^{2}{\bf A}(\tau)$ are matrices with entries $a_{ij}^{'}(\tau)$ and $a_{ij}^{''}(\tau)$, $i, j = \overline{1, n}$ respectively; $\delta{\bf \bar{u}}(\tau)$ and $\delta^{2}{\bf \bar{u}}(\tau)$ are vectors with components $u_{i}^{'}(\tau)$ and $u_{i}^{''}(\tau)$, $i = \overline{1, n}$; and $\delta^{k}\bar{f}(\tau) = \dfrac{d^{k}}{d\tau^{k}}\bar{f}(\tau)$ for $k = 1, 2$.

Consider the following linear equation for the vector ${\bf \bar{v}}(\tau) \in \mathbb{R}^{n}$:
\begin{equation}
	{\bf A}^{T}(\tau){\bf \bar{v}}(\tau) = {\bf I}.
	\label{c3:eq3.6}
\end{equation}
Taking the scalar product of equation \eqref{c3:eq3.2} with ${\bf \bar{v}}$ yields
$$
	\Big({\bf \bar{v}}(\tau), {\bf I}\Big) = \frac{1}{\bar{f}} \neq 0, \quad {\bf I} = (1,\, 1,\, \ldots,\, 1) \in \mathbb{R}^{n}.
$$

The following result holds.

\begin{theorem}\label{c3:theorem:3.1}	Let Assumption~\ref{c3:assumpt3.1} be satisfied. Then
	\begin{equation}
		\frac{d\bar{f}}{d\tau} = \delta\bar{f}(\tau) = \Big(\delta{\bf A}(\tau){\bf \bar{u}, \bar{v}}\Big)\bar{f}(\tau).
		\label{c3:eq3.7}
	\end{equation}
\end{theorem}

\begin{proof}
Substituting expansion \eqref{c3:eq3.5} into equation \eqref{c3:eq3.2} and collecting terms linear in $\Delta\tau$, gives
\begin{equation}
	\delta{\bf A}(\tau){\bf \bar{u}}(\tau) + {\bf A}(\tau)\delta{\bf \bar{u}}(\tau) = \delta \bar{f}(\tau){\bf I}, \quad {\bf \bar{u}}(\tau) \in S_{n}.
	\label{c3:eq3.8}
\end{equation}
Taking the scalar product with the adjoint vector ${\bf \bar{v}}$:
$$
	\Big(\delta{\bf A}(\tau){\bf \bar{u}}(\tau), {\bf \bar{v}}(\tau)\Big) + \Big({\bf A}(\tau)\delta{\bf \bar{u}}(\tau), {\bf \bar{v}}(\tau)\Big) = \delta\bar{f}(\tau)\Big({\bf \bar{v}}(\tau), {\bf I}\Big).
$$
Since ${\bf \bar{u}} \in S_{n}$, we have $\Big({\bf \bar{u}}, {\bf I}\Big) = 1$, hence $\Big(\delta{\bf \bar{u}}, {\bf I}\Big) = 0$. Using the adjoint equation \eqref{c3:eq3.6}:
$$
	\Big({\bf A}(\tau)\delta{\bf \bar{u}}(\tau), {\bf \bar{v}}(\tau)\Big) = \Big(\delta{\bf \bar{u}}(\tau), {\bf A}^{T}(\tau){\bf \bar{v}}(\tau)\Big) = \Big(\delta{\bf \bar{u}}(\tau), {\bf I}\Big) = 0. 
$$
\end{proof}

\begin{conseq}
A necessary condition for $\bar{f}$ to attain a local maximum on the admissible set \eqref{c3:eq3.1} at a point $\tau^{*}$ for which ${\bf \bar{u}}(\tau^{*}) \in \inside S_{n}$ is
\begin{equation}
	\bar{u}_{i}(\tau^{*})\bar{v}_{j}(\tau^{*}) = \mu a_{ij}(\tau^{*}), \quad \mu = \const.
	\label{c3:eq3.9}
\end{equation}
\end{conseq}

\begin{proof}
Let $F(\tau) = \bar{f}({\bf \bar{u}}(\tau))$ attain a local maximum at $\tau^{*} > 0$ over the admissible set \eqref{c3:eq3.1}, so that $F(\tau^* + \Delta\tau) - F(\tau^*) \leqslant 0$ for all sufficiently small $\Delta\tau$. From \eqref{c3:eq3.7}:
\begin{equation}\label{c3:eq3.10}
	F(\tau^*)\sum\limits_{i,j=1}^n a_{ij}^{'} (\tau^*)\bar{u}_i(\tau^*)\bar{v}_j(\tau^*) \Delta \tau \leqslant 0.
\end{equation}
Since the fitness landscape matrix belongs to  $\mathcal{M}$:
\begin{equation}\label{c3:eq3.11}
	\sum\limits_{i,j=1}^n a_{ij}^{'} (\tau^*)a_{ij}(\tau^*) \Delta \tau \leqslant 0.
\end{equation}
Since \eqref{c3:eq3.10} and \eqref{c3:eq3.11} must hold for all sufficiently small $\Delta\tau$, it follows that
\begin{equation*}
	\begin{aligned}	
		&\sum\limits_{i,j=1}^n a_{ij}^{'}(\tau^*)\bar{u}_i(\tau^*)\bar{v}_j(\tau^*) = 0, \\		
		&\sum\limits_{i,j=1}^n a_{ij}^{'}(\tau^*)a_{ij}(\tau^*) = 0,
	\end{aligned}
\end{equation*}
from which condition \eqref{c3:eq3.9} follows. 
\end{proof}

To derive sufficient conditions for a maximum, we require the representation of $\delta^2\bar{f}(\tau)$ on the set of variations satisfying $\delta\bar{f}(\tau)=0$. Substituting the expansions \eqref{c3:eq3.5} into \eqref{c3:eq3.2} and collecting terms involving $\delta{\bf \bar{u}}(\tau)$:
\begin{equation}
	\Big(\delta^{2}{\bf A}\Big){\bf \bar{u}} + 2\Big(\delta{\bf A}\Big)\Big(\delta{\bf \bar{u}}\Big) + {\bf A}\Big(\delta^{2}{\bf \bar{u}}\Big) = \Big(\delta^{2}\bar{f}\Big){\bf I}.
\label{c3:eq3.12}
\end{equation}

When $\delta\bar{f}(\tau) = 0$, equation \eqref{c3:eq3.8} gives $\delta{\bf \bar{u}} = -{\bf A}^{-1}\Big(\delta{\bf A}\Big){\bf \bar{u}}$. Taking the scalar product of \eqref{c3:eq3.12} with ${\bf \bar{v}}$ and using $\Big(\delta^{2}{\bf \bar{u}, I}\Big) = 0$:
$$
	\delta^{2}\bar{f}\Big({\bf \bar{v}, I}\Big) = \Bigg(\Big(\delta^{2}{\bf A}\Big) {\bf \bar{u}, \bar{v}}\Bigg) - 2\Bigg(\Big(\delta{\bf A}\Big) {\bf A}^{-1}\Big(\delta{\bf A}\Big) {\bf \bar{u}, \bar{v}}\Bigg).
$$ 

Let $b_{ij}(\tau)$, $i, j=\overline{1,n}$ denote the entries of the matrix $B(\tau)=2(\delta A)A^{-1}(\delta A)$. The sufficient conditions for a maximum then take the form
\begin{equation*}
	\sum\limits_{i,j=1}^n \left(a_{ij}''(\tau) - b_{ij}(\tau)\right) \bar{u}_i(\tau)\bar{v}_j(\tau) \leqslant 0
\end{equation*}
subject to
\begin{equation*}
	\sum\limits_{i,j=1}^n a_{ij}'(\tau) \bar{u}_i(\tau)\bar{v}_j(\tau) = 0.
\end{equation*}

\clearpage

\section{Numerical Method for the Evolutionary Adaptation Process}\label{c3:section:3.3}
The fitness variation formulae derived above reduce the problem of maximising the mean fitness $\bar{f}(\tau)$ to a sequence of linear programming problems.

It is important to emphasise that the proposed adaptation method is valid only for non-degenerate replicator systems. Accordingly, the numerical procedure must verify the condition ${\bf \bar{u}}(\tau) > 0$ for all $\tau \geqslant 0$. This means that whenever one or more components of ${\bf \bar{u}}(\tau)$ approach the boundary of the simplex $S_{n}$, one must ensure that ${\bf \bar{u}}(\tau) + \delta{\bf \bar{u}}(\tau)\Delta\tau \in \inside S_{n}$. An explicit formula for $\delta{\bf \bar{u}}(\tau)$ is therefore required. From \eqref{c3:eq3.8}:
$$
	\delta{\bf \bar{u}}(\tau) = \Big(\delta\bar{f}(\tau)\Big){\bf A}^{-1}(\tau){\bf I} - {\bf A}^{-1}(\tau)\Big(\delta{\bf A}(\tau)\Big){\bf \bar{u}}(\tau).
$$ 
Using the identities
\begin{equation*}
	\begin{aligned}
		&{\bf \bar{u}}(\tau) = \bar{f}(\tau){\bf A}^{-1}(\tau){\bf I}, \quad \delta\bar{f}(\tau) = \Bigg(\Big(\delta{\bf A}(\tau)\Big){\bf \bar{u}}(\tau), {\bf \bar{v}}(\tau)\Bigg)\bar{f}(\tau), \\
		&{\bf \bar{v}}(\tau) = \Big({\bf A}^{-1}\Big)^{T}{\bf I},
	\end{aligned}
\end{equation*}
we obtain
\begin{equation}
	\delta{\bf \bar{u}}(\tau) = \Bigg(\Big(\delta{\bf A}(\tau)\Big){\bf \bar{u}}(\tau), \Big({\bf A}^{-1}\Big)^{T}{\bf I}\Bigg){\bf \bar{u}}(\tau) - {\bf A}^{-1}(\tau)\Big(\delta{\bf A}(\tau)\Big){\bf \bar{u}}(\tau).
	\label{c3:eq3.13}
\end{equation}
A direct check confirms that $\Big(\delta{\bf \bar{u}}(\tau), {\bf I}\Big) = 0$.

If, for instance, a component $\bar{u}_{i}(\tau) = \delta_{i}$ with $\delta_{i} \to 0$ is  sufficiently small, then the admissible variation $\delta\bar{u}_{i}(\tau)$ must satisfy $\delta_{i} + \delta\bar{u}_{i}(\tau) > 0$, which requires
$$
	\delta_{i}\Bigg(1 + \Bigg(\Big(\delta{\bf A}(\tau)\Big){\bf \bar{u}}(\tau), \Big({\bf A}^{-1}\Big)^{T}{\bf I}\Bigg)\Bigg) > \Bigg({\bf A}^{-1}(\tau)\Big(\delta{\bf A}(\tau)\Big){\bf \bar{u}}(\tau)\Bigg)_{i},
$$
where the subscript $i$ denotes the $i$-th component of the vector ${\bf A}^{-1}(\tau)\Big(\delta{\bf A}(\tau)\Big){\bf \bar{u}}(\tau)$.

Let $\varepsilon > 0$ be a fixed step size for the adaptation process in evolutionary time, so that $\Delta\tau = \varepsilon > 0$. We adopt the hypothesis of sufficiently slow variation of the fitness landscape matrix, namely that for all $\tau \geqslant 0$:
$$
	\begin{aligned}
		&\displaystyle|\delta a_{ij}(\tau)| = |a_{ij}^{'}(\tau)| \leqslant \alpha\varepsilon, \\
		&\displaystyle|\delta^{2} a_{ij}(\tau)| = |a_{ij}^{''}(\tau)| \leqslant \beta\varepsilon^{2}, \quad i, j = \overline{1, n},\\
	\end{aligned}
$$
where $\alpha$ and $\beta$ are positive constants. Under these bounds, each expansion in \eqref{c3:eq3.5} holds with error $O(\varepsilon^{3})$ when retaining only linear terms in $\Delta\tau = \varepsilon$, and with error $O(\varepsilon^{5})$ when retaining quadratic terms.

Expression \eqref{c3:eq3.7} for the fitness variation can be written as
$$
	\delta\bar{f}(\tau) = \bar{f}(\tau)\sum\limits_{i, j = 1}^{n}a_{ij}^{'}(\tau)\bar{u}_{i}(\tau)\bar{v}_{j}(\tau).
$$
The vectors ${\bf \bar{u}}(\tau)$ and ${\bf \bar{v}}(\tau)$ are determined from equations \eqref{c3:eq3.2} and \eqref{c3:eq3.6} respectively. The constraint \eqref{c3:eq3.1} requires that inequality \eqref{c3:eq3.10} hold for all $\tau \geqslant 0$.

Each step of the evolutionary process of length $\Delta\tau = \varepsilon$ thus reduces to the linear programming problem of choosing the matrix elements of ${\bf A} \in \mathcal{M}$ so as to maximise
$$
	\delta\bar{f}(\tau) = \sum\limits_{i, j = 1}^{n}a_{ij}^{'}(\tau)\bar{u}_{i}(\tau)\bar{v}_{j}(\tau) \to \max,
$$
subject to the constraints
$$
	\sum\limits_{i, j = 1}^{n}a_{ij}^{'}(\tau)a_{ij}(\tau) \leqslant 0, \quad |a_{ij}^{'}(\tau)| \leqslant \alpha\varepsilon, \quad i, j = \overline{1, n}.
$$
When one or more components of ${\bf \bar{u}}(\tau)$ are close to degeneracy, the additional constraints from formula \eqref{c3:eq3.13} must be incorporated.

Convexity of $\mathcal{M}$ and $S_n$ does not in general guarantee convexity of
$\bar{f}(\tau)$ or the attainment of a global maximum. Nevertheless, numerical
experiments show that the equilibrium ${\bf \bar{u}}(\tau) \in \inside S_n$
of~\eqref{c3:eq3.2} remains unchanged over a large number of iterations.

We now show that if the equilibrium ${\bf u}^*$ of~\eqref{c3:eq3.2} remains fixed
with ${\bf u}^* \in \inside S_n$ throughout $0 \leqslant \tau \leqslant \tau^*$,
then the functional
\begin{equation*}
	F({\bf A}) = \Big({\bf A}(\tau){\bf u}^*,\, {\bf u}^*\Big) = \bar{f}(\tau),
	\quad {\bf A} \in \mathcal{M}
\end{equation*}
attains its global maximum at some ${\bf A}^* \in \mathcal{M}$.

Linearity and convexity of $F({\bf A})$ are clear. Boundedness follows from
\begin{equation*}
	|F({\bf A})| \leqslant \|{\bf A}(\tau)\| \cdot \|{\bf u}^*\|^2.
\end{equation*}
Since $\|{\bf A}(\tau)\| \leqslant M$ and ${\bf u}^* \in \inside S_n$:
\begin{equation*}
	\|{\bf u}^*\|^2 = \sum\limits_{i=1}^n (u_i^*)^2
	\leqslant \left(\sum\limits_{i=1}^n u_i^*\right)^2 = 1,
\end{equation*}
hence $|F({\bf A})| \leqslant M$, and the global maximum is attained at some
${\bf A}^*(\tau) \in \mathcal{M}$.

The algorithm is applied to four replicator systems in
Sections~\ref{c3:section:3.4}--\ref{c3:section:3.7}.

\clearpage

\section{Hypercycle Evolution}\label{c3:section:3.4}
We consider the evolutionary adaptation of the hypercycle system, whose initial dynamics is governed by
\begin{equation}
	\frac{du_{i}}{dt} = u_{i}(t)\Big(u_{i - 1}(t) - f(t)\Big), \quad i = \overline{1, n},
	\label{c3:eq3.14}
\end{equation}
where 
$$
	u_{0} = u_{n},\quad \sum\limits_{i = 1}^{n}u_{i}(t) = 1, \quad f({\bf u}) = \sum\limits_{i = 1}^{n}u_{i}(t)u_{i - 1}(t),
$$ 
over the set $\mathcal{M}$ defined by \eqref{c3:eq3.1} with $M = n^{2}$.

In the numerical experiments, the step size $\varepsilon$ governing the rate of evolutionary change in the fitness landscape takes values from the set $\{1 \cdot 10^{-3},\, 5 \cdot 10^{-4},\, 1 \cdot 10^{-4}\}$. For a third-order hypercycle, the corresponding changes in the norm of ${\bf A}$ are $5 \cdot 10^{-3},\, 1 \cdot 10^{-3}$, and $5 \cdot 10^{-5}$, respectively.

Fig.~\ref{c3:fig3.2a} shows the evolution of the equilibrium coordinates of system \eqref{c3:eq3.14} in evolutionary time $\tau$ for $n = 9$. During the first $125$ steps, the coordinates remain practically unchanged ($u_{i} = 1/9$, $i = \overline{1, 9}$). Beginning from step $125$, the fixed-point components undergo splitting.

Fig.~\ref{c3:fig3.2b} shows the evolution of the mean fitness of system \eqref{c3:eq3.14} for $n = 9$. As seen from the graph, the fitness is a monotonically increasing function of the evolutionary time $\tau$.  

\begin{figure}[h!]
	\centering
	\includegraphics[width=0.8\linewidth]{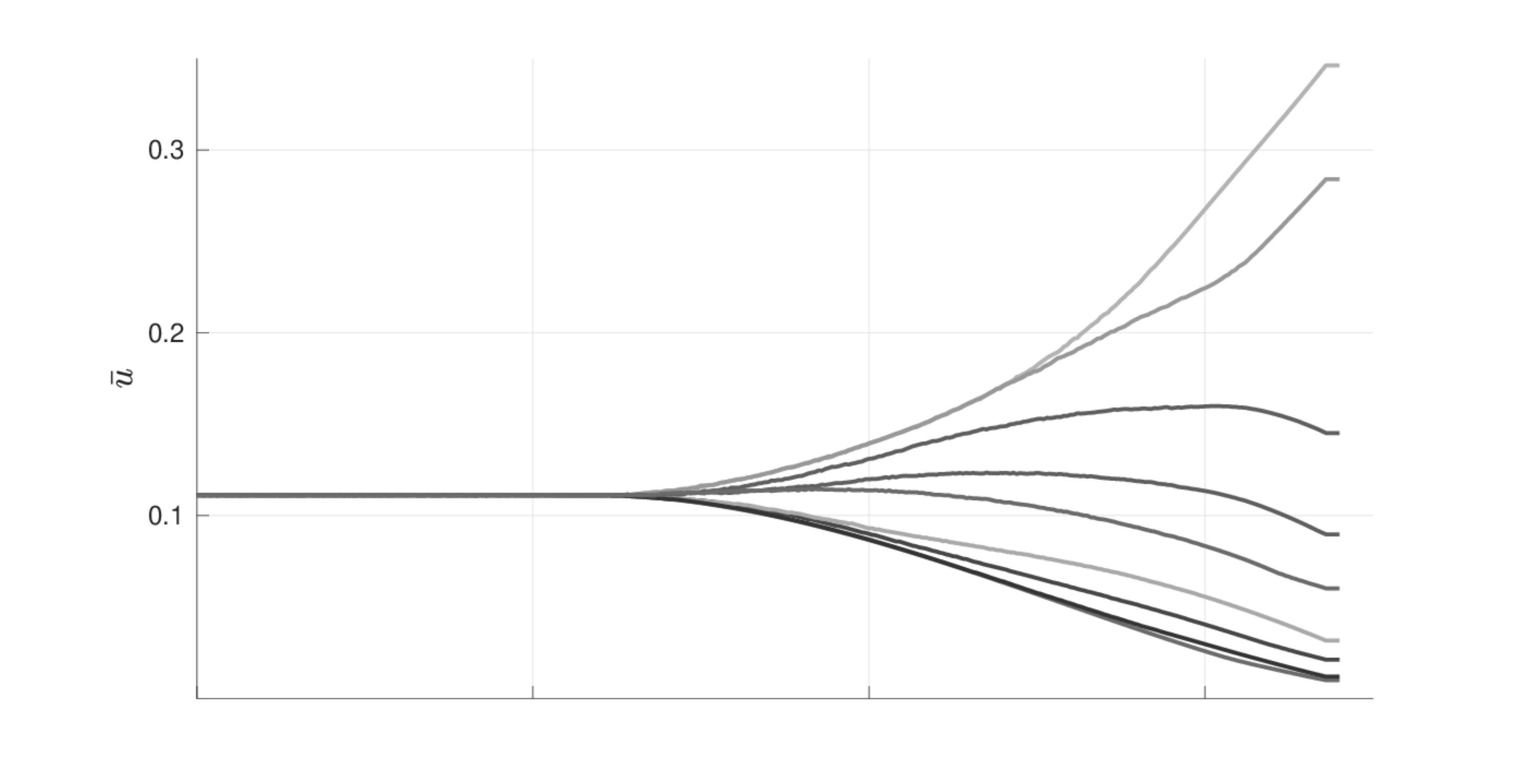}
	\caption{Equilibrium coordinates ${\bf \bar{u}}$ of system~\eqref{c3:eq3.14} as functions
		of evolutionary time $\tau$ for $n = 9$. Components remain equal ($u_i = 1/9$)
		for the first 125 steps, then split.}
	\label{c3:fig3.2a}
\end{figure}

\begin{figure}[h!]
	\centering
	\includegraphics[width=0.8\linewidth]{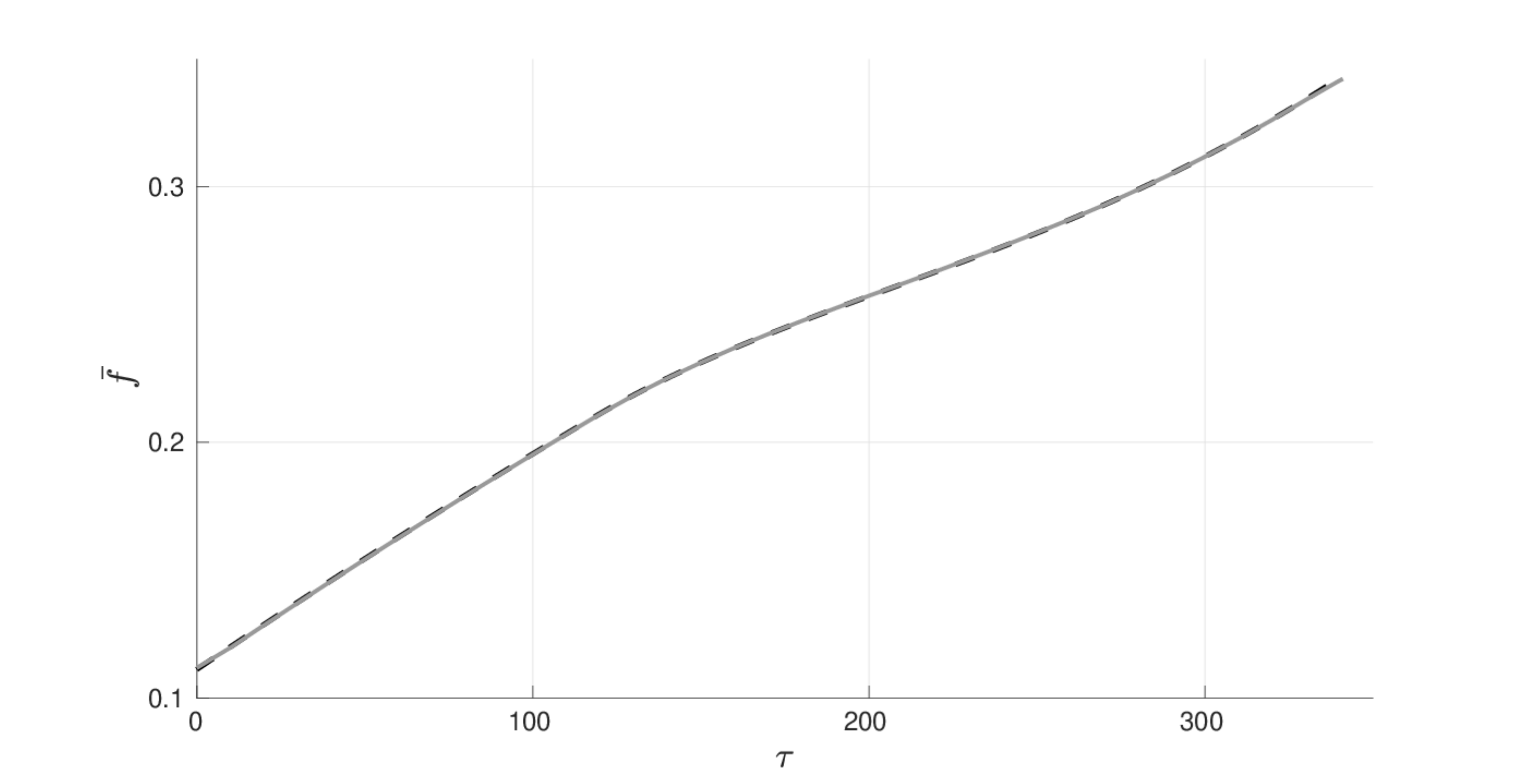}
	\caption{Mean fitness $\bar{f}(\tau)$ of system~\eqref{c3:eq3.14} for $n = 9$.
		The fitness increases monotonically throughout the evolutionary process.}
	\label{c3:fig3.2b}
\end{figure}

The effect of the evolutionary process on the species interaction graph is readily visualised. Fig.~\ref{c3:fig3.3a} shows the interaction graph of the original hypercycle \eqref{c3:eq3.14} for $n = 5$, while Fig.~\ref{c3:fig3.3b} shows the corresponding graph of the system obtained after $350$ evolutionary iterations. As seen from the evolved graph, the original hypercycle --- built on purely altruistic replication, in which each species catalyses only its neighbour in a closed cycle --- undergoes significant qualitative changes. One notable outcome is the emergence of autocatalytic replication within the system.

\begin{figure}[H]
	\centering
	\includegraphics[width=0.6\linewidth]{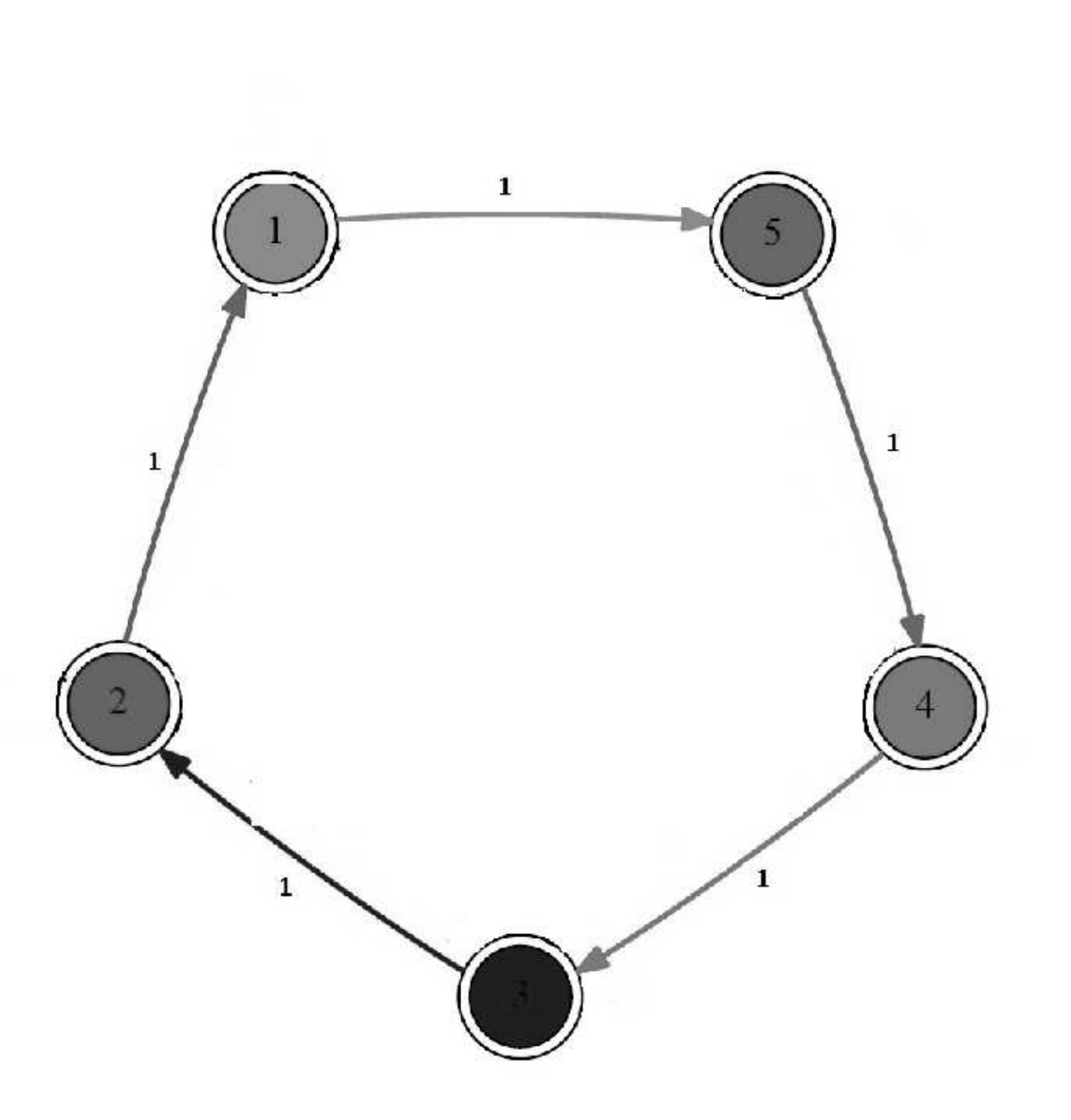}
	\caption{Interaction graph of the original fifth-order hypercycle~\eqref{c3:eq3.14}.
		Each species catalyses only its cyclic neighbour --- purely altruistic replication,
		no autocatalytic links.}
	\label{c3:fig3.3a}
\end{figure}

\begin{figure}[H]
	\centering
	\includegraphics[width=0.7\linewidth]{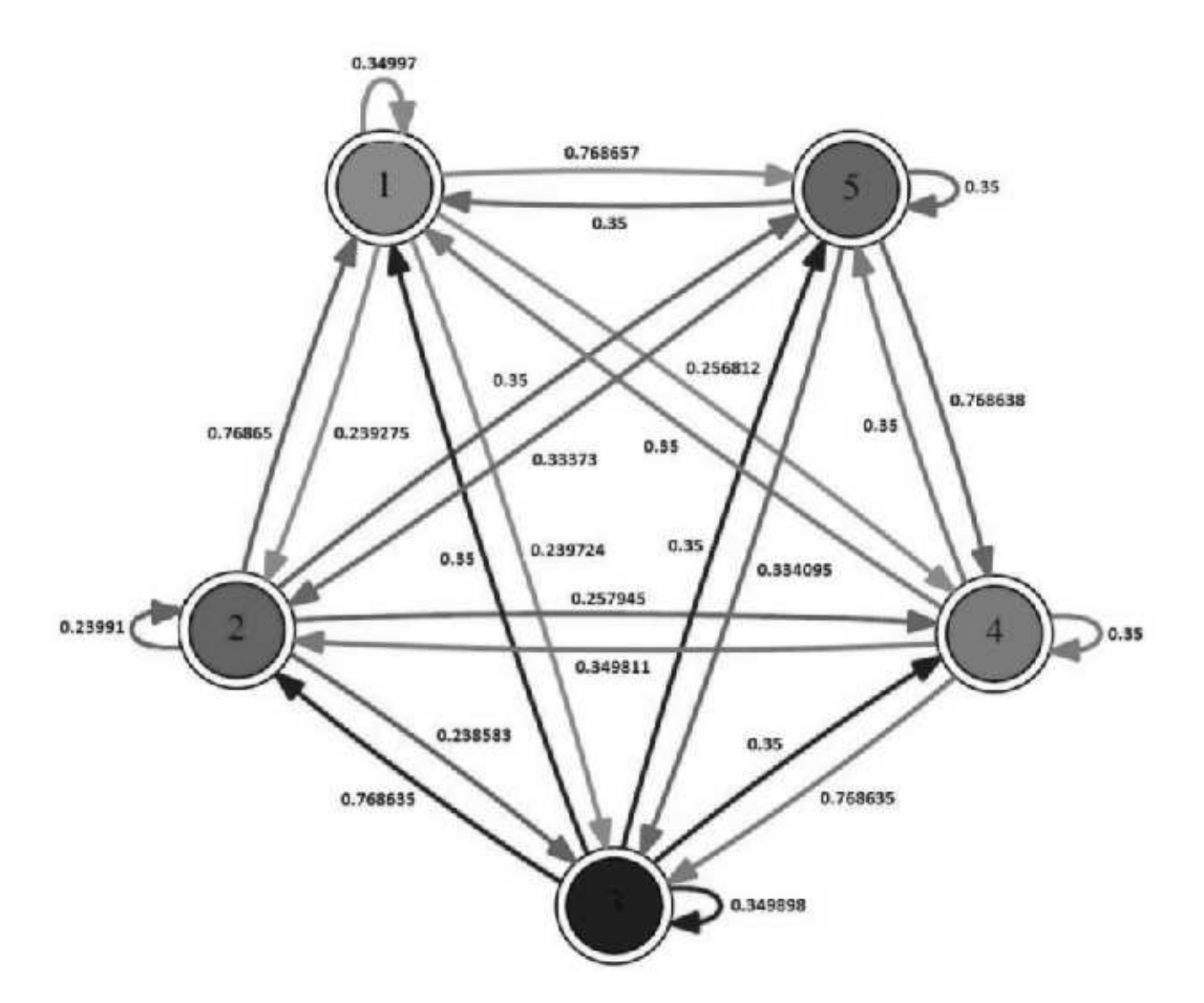}
	\caption{Interaction graph of the evolved system at iteration $350$.
		Reverse hypercyclic connections have appeared alongside the original ones,
		and each species has acquired autocatalytic self-replication links.}
	\label{c3:fig3.3b}
\end{figure}

In terms of hypercycle classification~\cite{c3:Eigen1982}, a hypercycle of connectivity degree six emerges. The increase in mean fitness is accompanied not only by greater complexity in the system's behaviour, but also by the acquisition of resistance to parasitic species.

We now examine the interaction of a fifth-order hypercycle with a parasite in more detail. System \eqref{c3:eq3.14} becomes:
\begin{equation}
	\begin{aligned}	
		&\displaystyle\dot{u}_{i} = u_{i}(u_{i - 1} - f({\bf u})), \quad i = \overline{1, 5}, \\
		&\displaystyle\dot{u}_{6} = u_{6}(1.7u_{5} - f({\bf u})),\\
	\end{aligned}
	\label{c3:eq3.15}
\end{equation}
$$
	f({\bf u}) = \sum\limits_{i = 1}^{5}u_{i}u_{i - 1} + 1.7u_{5}u_{6},
	\quad \sum\limits_{i = 1}^{6}u_{i} = 1, \quad u_{0} = u_{5}.
$$

In the classical setting, the original hypercycle collapses and only the parasite survives. However, after $200$ evolutionary steps, the evolved hypercycle becomes resistant to the parasite.

Fig.~\ref{c3:fig3.4} shows the active dynamics of the components $u_{i}$, $i = \overline{1, 5}$ of system \eqref{c3:eq3.15}: the original hypercycle collapses.

Fig.~\ref{c3:fig3.4}\,(b) shows the species frequency dynamics when the parasite (dashed line) is introduced to the evolved hypercycle: the evolved hypercycle survives while the parasite collapses.

\clearpage
\begin{figure}[H]
	\begin{minipage}[ht]{1.0\linewidth}
		\center{\includegraphics[width=0.92\linewidth]{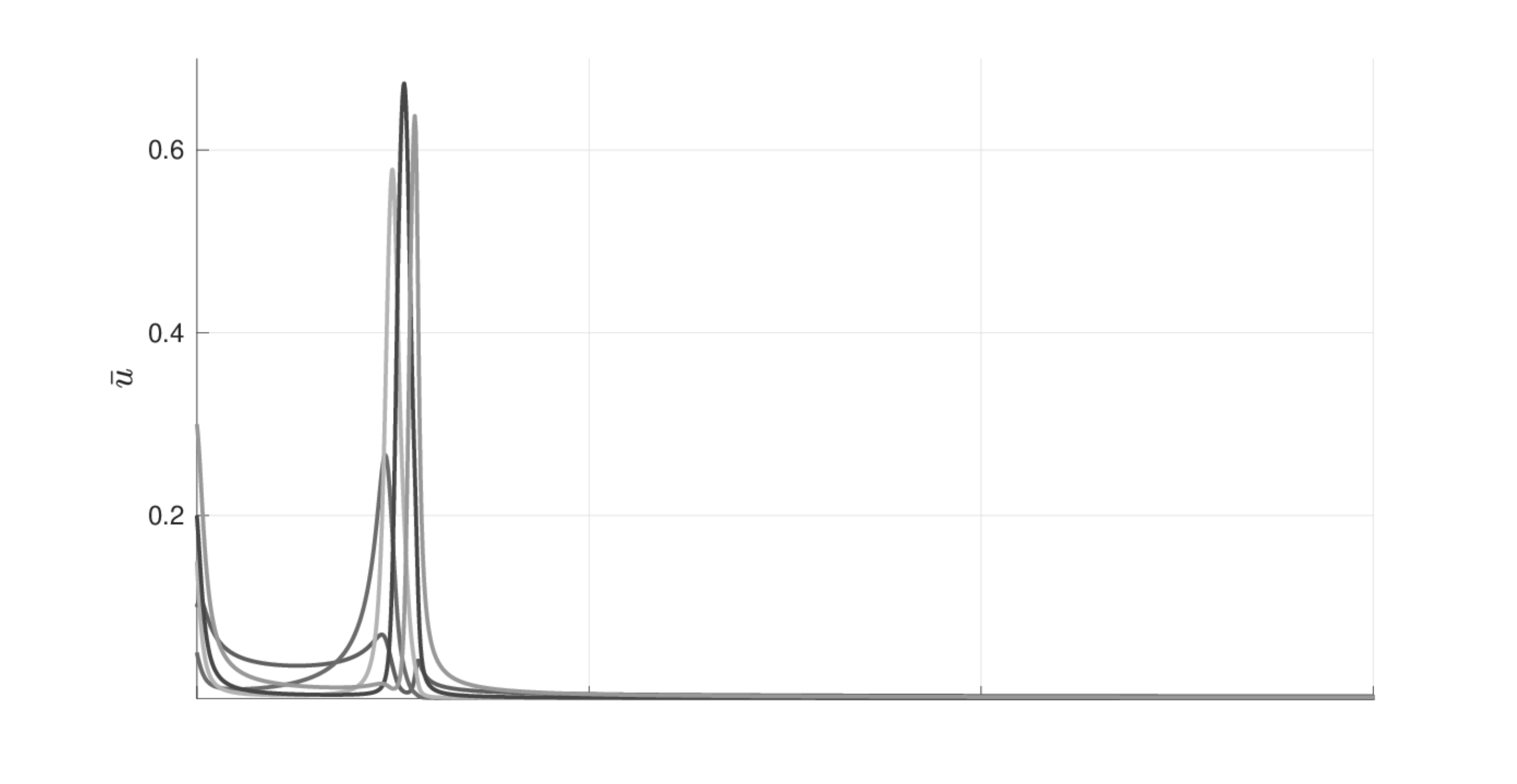} (a)}
	\end{minipage}
	\vfill
	\begin{minipage}[ht]{1.0\linewidth}
		\center{\includegraphics[width=0.92\linewidth]{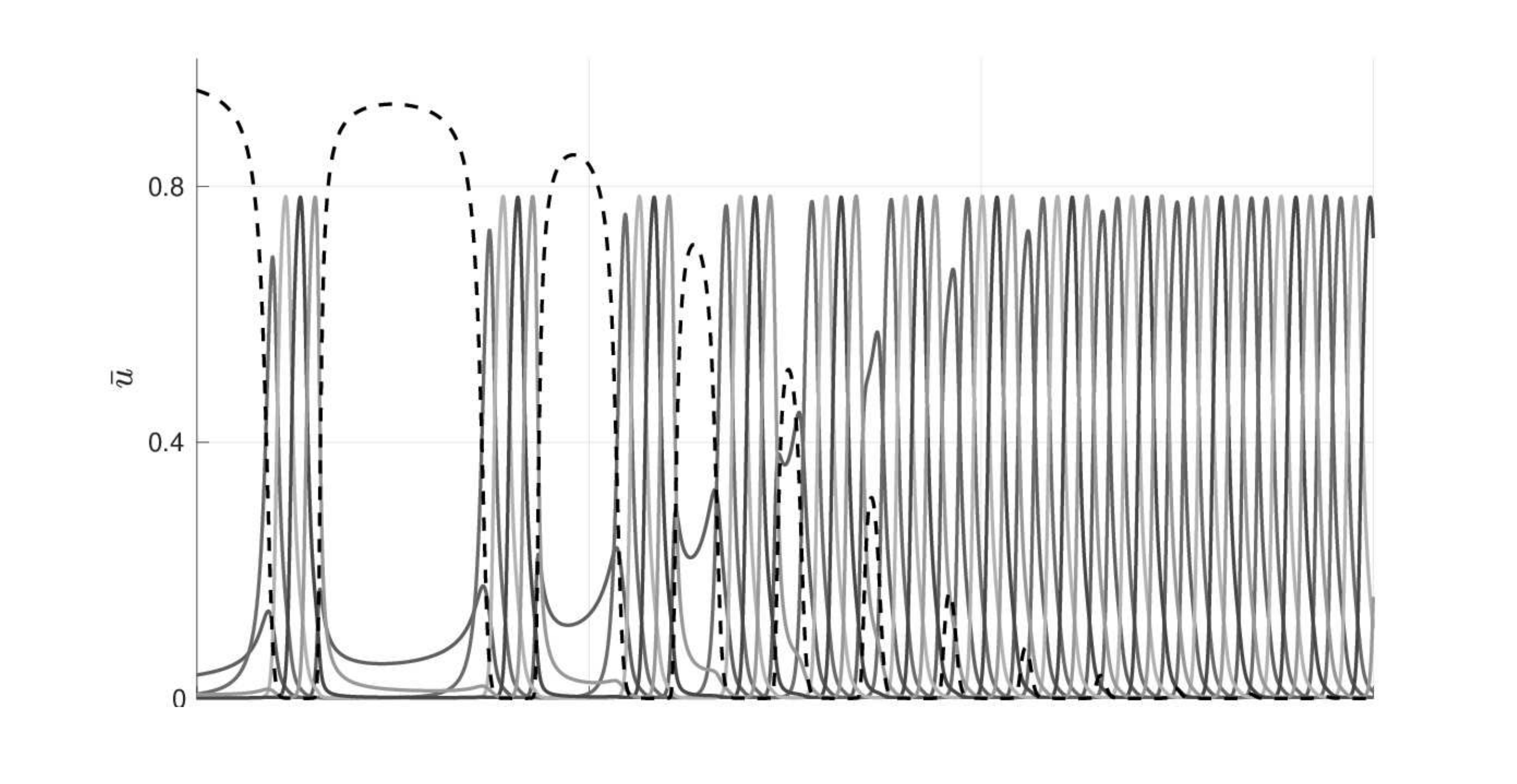} (b)}
	\end{minipage}
	\caption{(a) Frequency dynamics of the original fifth-order hypercycle interacting with a parasite. (b) Frequency dynamics when the parasite is introduced to the evolved fifth-order hypercycle at iteration $200$.}
	\label{c3:fig3.4}
\end{figure}     

The number of evolutionary steps required to produce a parasite-resistant hypercycle depends on the replication coefficients of the parasitic species. For example, if in system \eqref{c3:eq3.15} the coefficient is increased from $1.7$ to $1.8$, two hundred iterations no longer suffice (Fig.~\ref{c3:fig3.5}\,(a)); $250$ iterations are required in this case (Fig.~\ref{c3:fig3.5}\,(b)).

\begin{figure}[H]
	\begin{minipage}[ht]{1.0\linewidth}
		\center{\includegraphics[width=0.92\linewidth]{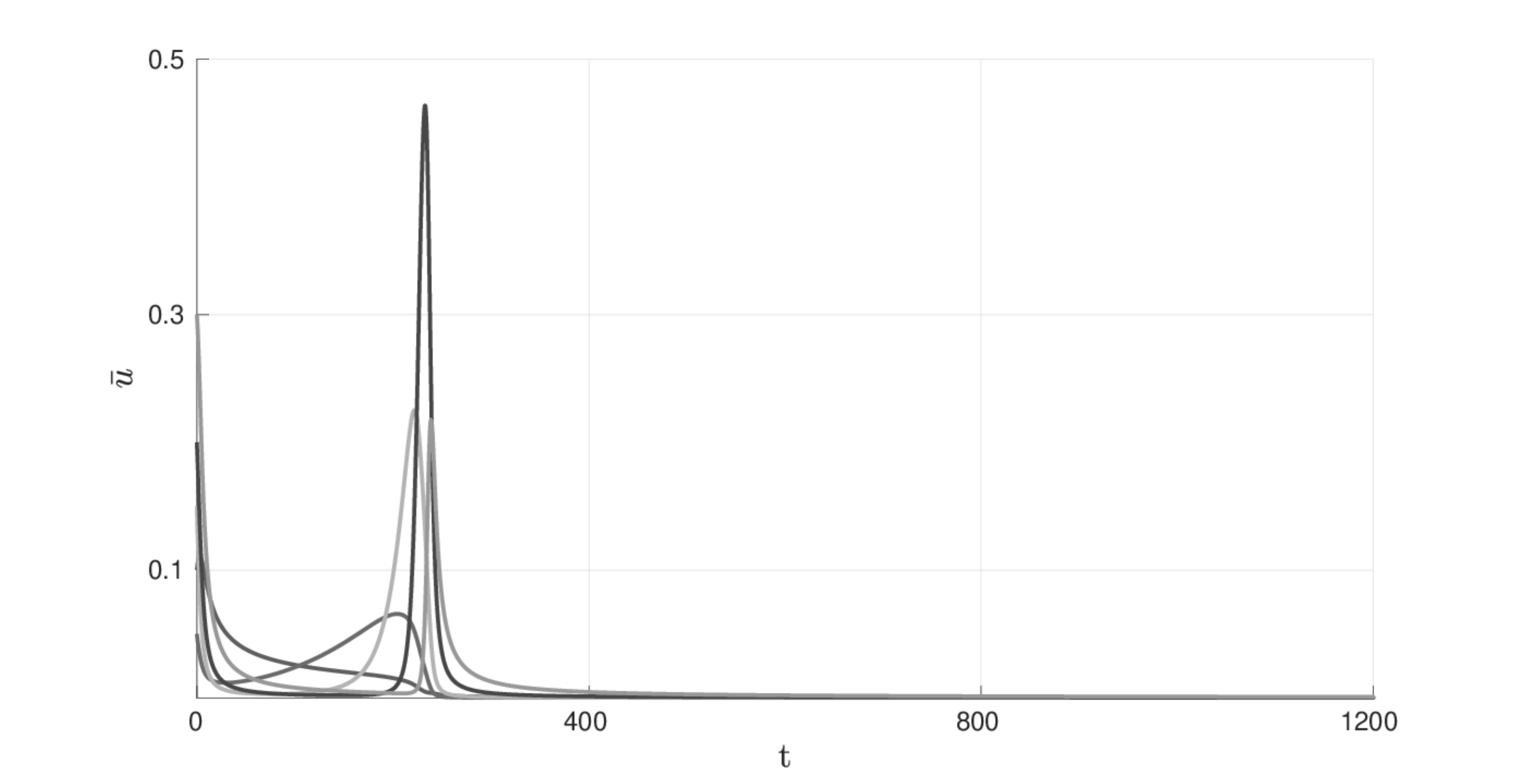} (a)}
	\end{minipage}
	\vfill
	\begin{minipage}[ht]{1.0\linewidth}
		\center{\includegraphics[width=0.92\linewidth]{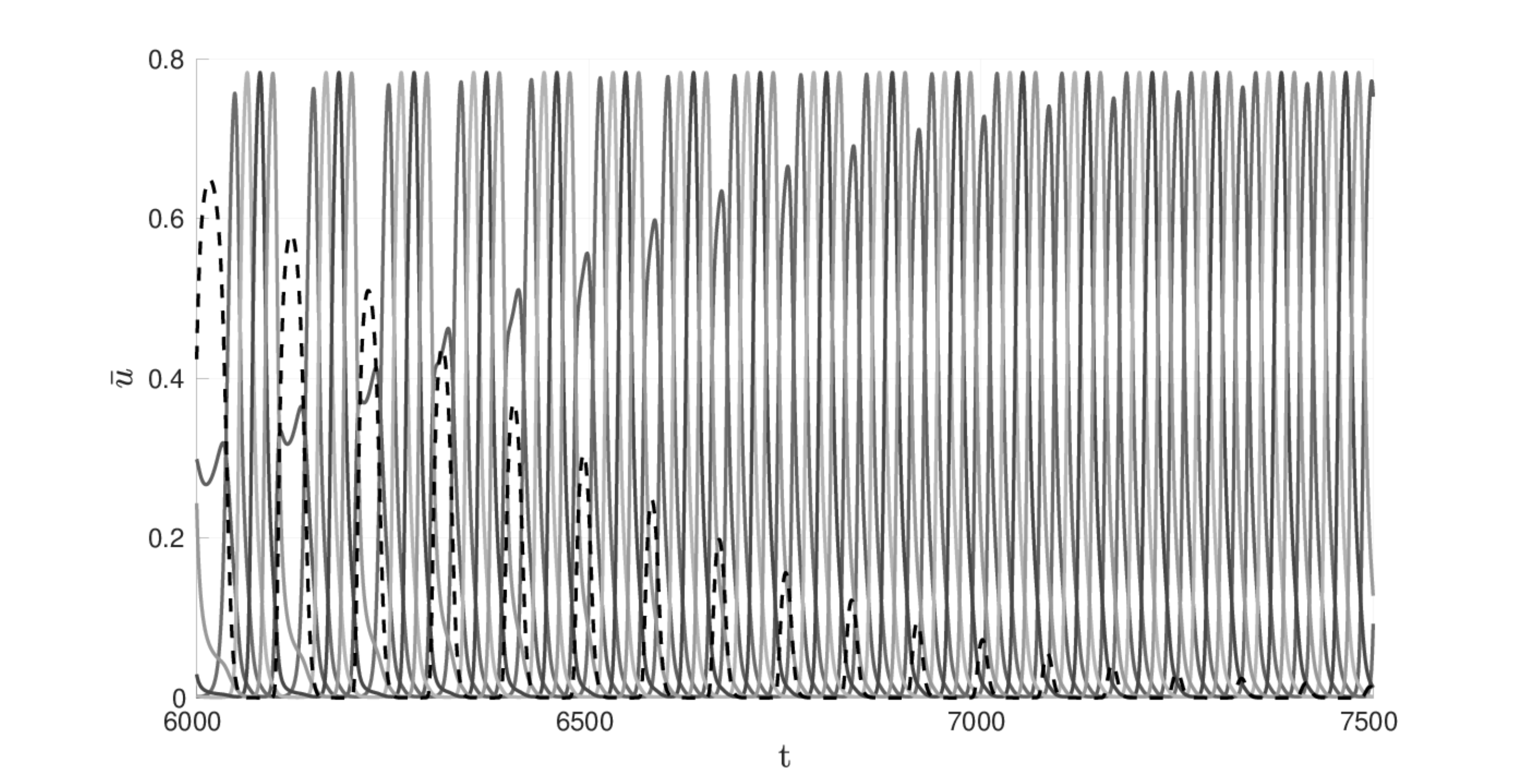} (b)}
	\end{minipage}
	\caption{(a) Frequency dynamics of the evolved fifth-order hypercycle (obtained at iteration $200$) interacting with a parasite. (b) Frequency dynamics when the parasite interacts with the evolved hypercycle at iteration $250$.}
	\label{c3:fig3.5}
\end{figure} 

\bigskip
Consider now the interaction of a ninth-order hypercycle with two parasites:
\begin{equation}
	\begin{aligned}
		&\displaystyle\dot{u}_{i} = u_{i}(u_{i - 1} - f({\bf u})), \quad i = \overline{1, 9}, \\
		&\displaystyle\dot{u}_{j} = u_{j}(0.95(u_{8} + u_{9}) - f({\bf u})), \quad j = 10,\, 11,\\
	\end{aligned}
	\label{c3:eq3.16}
\end{equation}
$$
	f({\bf u}) = \sum\limits_{i = 1}^{9}u_{i}u_{i - 1} + 0.95(u_{10} + u_{11})(u_{8} + u_{9}), \quad \sum\limits_{i = 1}^{11}u_{i} = 1, \quad u_{0} = u_{9}.
$$ 

\begin{figure}[H]
	\begin{minipage}[ht]{1.0\linewidth}
		\center{\includegraphics[width=0.8\linewidth]{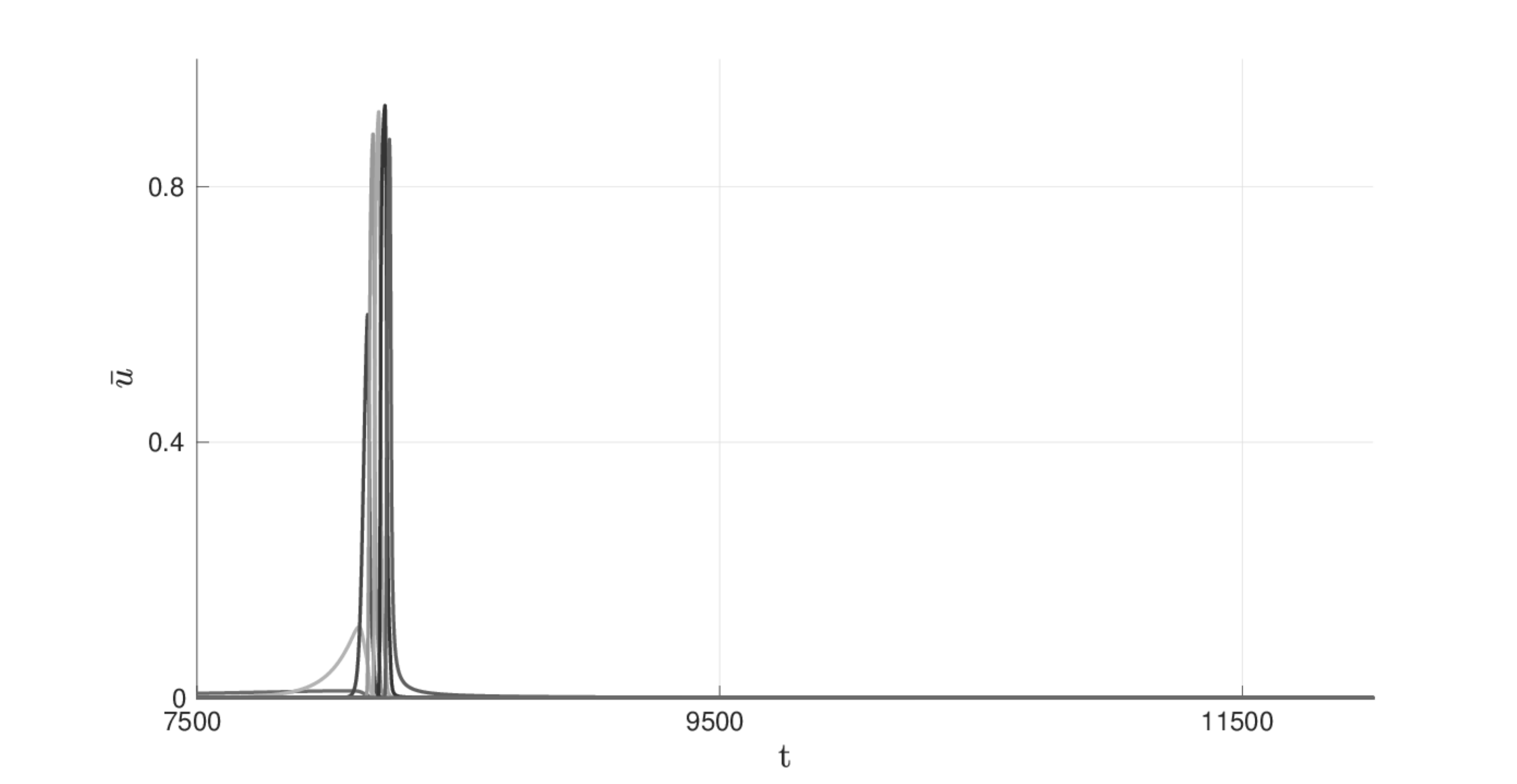} (a)}
	\end{minipage}
	\vfill
	\begin{minipage}[ht]{1.0\linewidth}
		\center{\includegraphics[width=0.8\linewidth]{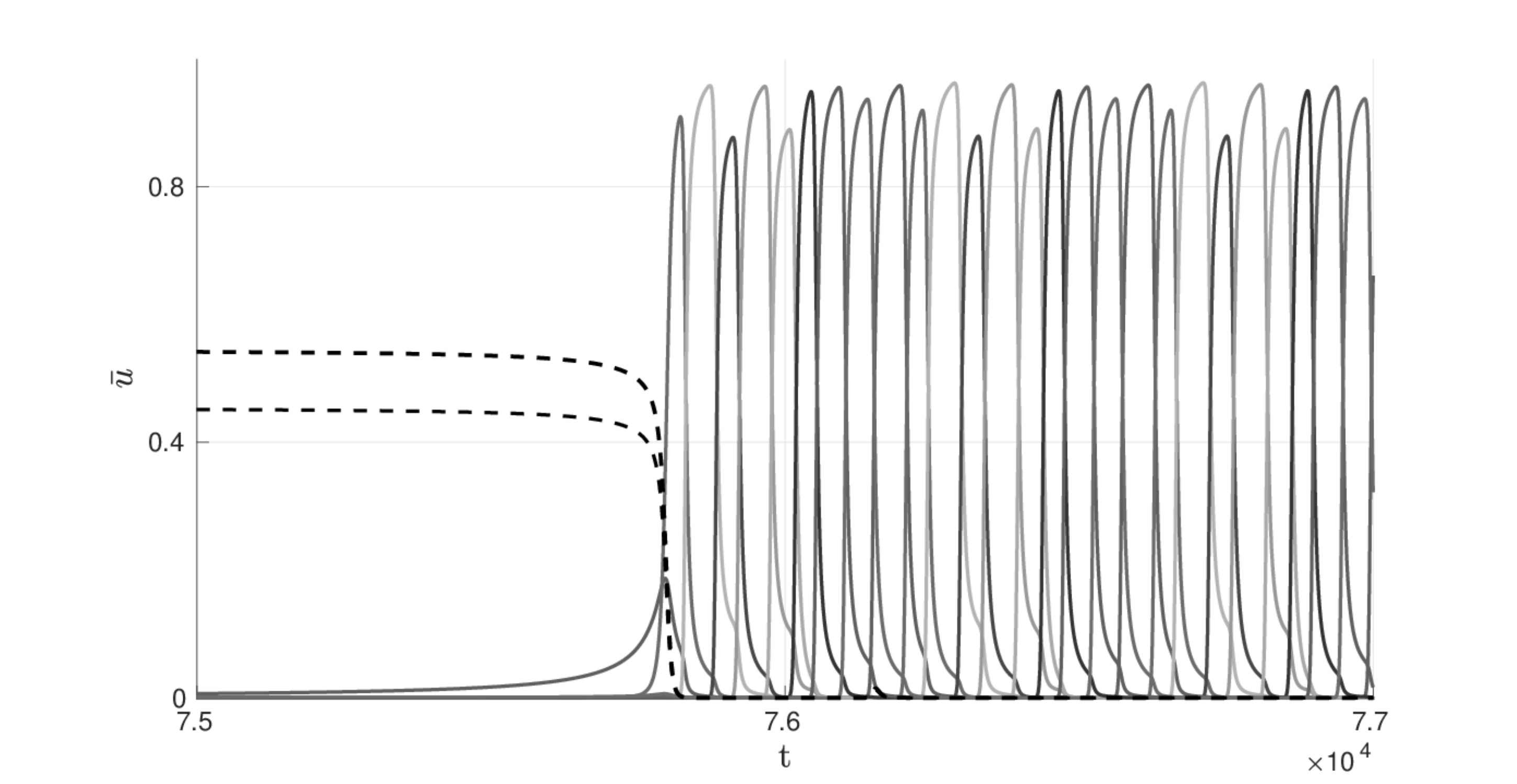} (b)}
	\end{minipage}
	\caption{(a) Frequency dynamics of the original ninth-order hypercycle interacting with two parasites. (b) Frequency dynamics when the two parasites are introduced to the evolved ninth-order hypercycle obtained at iteration $200$.}
	\label{c3:fig3.6}
\end{figure}

Fig.~\ref{c3:fig3.6}\,(a) shows the active dynamics of $u_{i}$, $i = \overline{1, 9}$ of system \eqref{c3:eq3.16}: the original hypercycle collapses. Fig.~\ref{c3:fig3.6}\,(b) shows the frequency dynamics when the two parasites (dashed lines) are introduced to the evolved hypercycle obtained at iteration $200$: the evolved hypercycle survives while the parasites collapse.

Thus, fitness landscape evolution substantially increases the system's resistance even when multiple parasites are present.

The evolutionary adaptation process has a natural bound: beyond a certain threshold the fitness landscape matrix ${\bf A}$ stabilises and its elements cease to change. If a penalty function $g = \prod\limits_{i = 1}^{n}u_{i}^{u_{i}}$ (related to the system entropy via $\ln{g} = \sum\limits_{i = 1}^{n}u_{i}\ln{u_{i}}$) is added to the constraints \eqref{c3:eq3.1}, the problem reduces to maximising
\begin{equation}
	F({\bf u}) = f({\bf u}) - \prod\limits_{i = 1}^{n}u_{i}^{u_{i}}, \quad \sum\limits_{i = 1}^{n}u_{i} = 1,
	\label{c3:eq3.17}
\end{equation}
over the set \eqref{c3:eq3.1}.

\begin{figure}[ht]
	\begin{minipage}[ht]{0.48\linewidth}
		\center{\includegraphics[width=\linewidth]{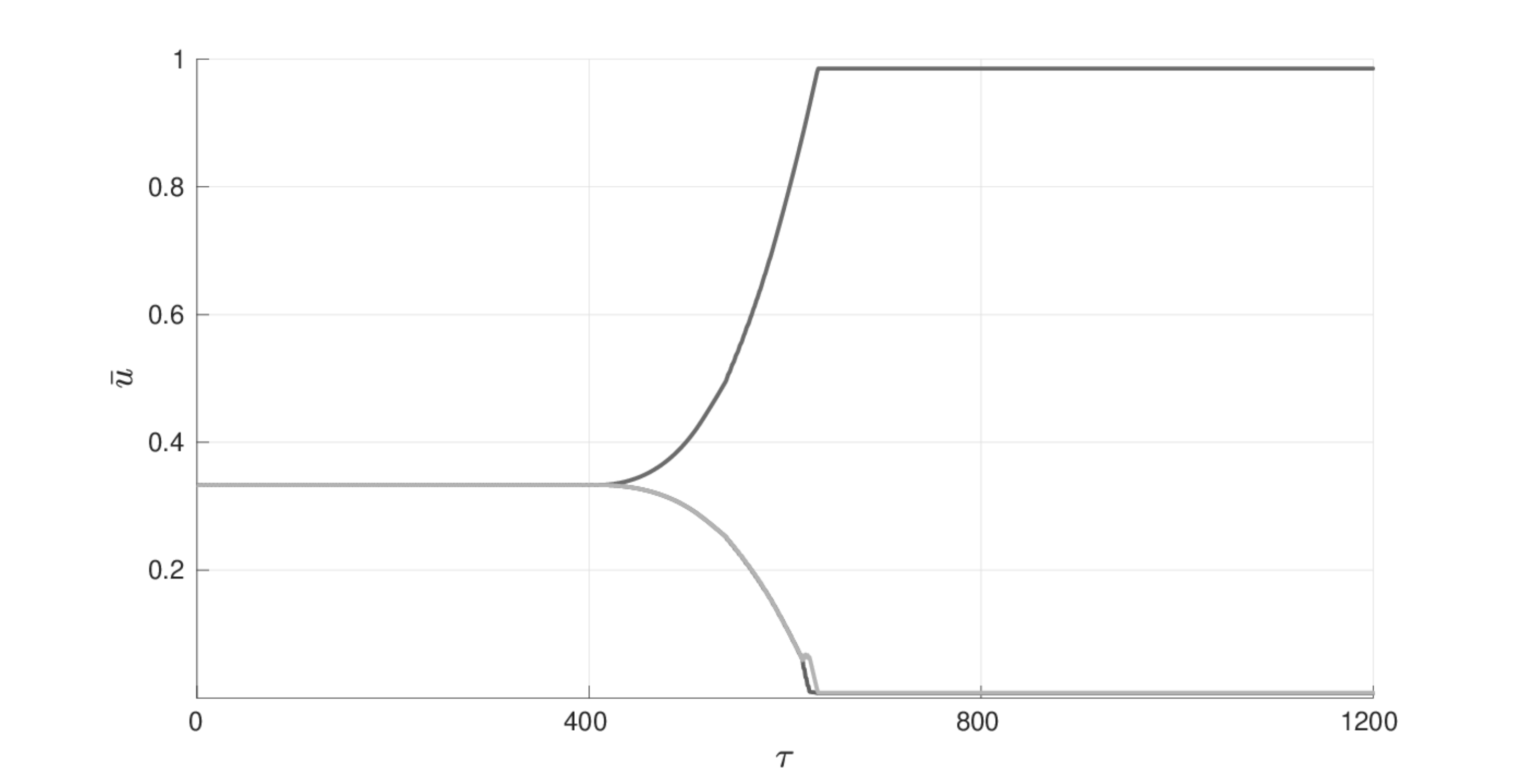} (a)}
	\end{minipage}
	\hfill
	\begin{minipage}[ht]{0.48\linewidth}
		\center{\includegraphics[width=\linewidth]{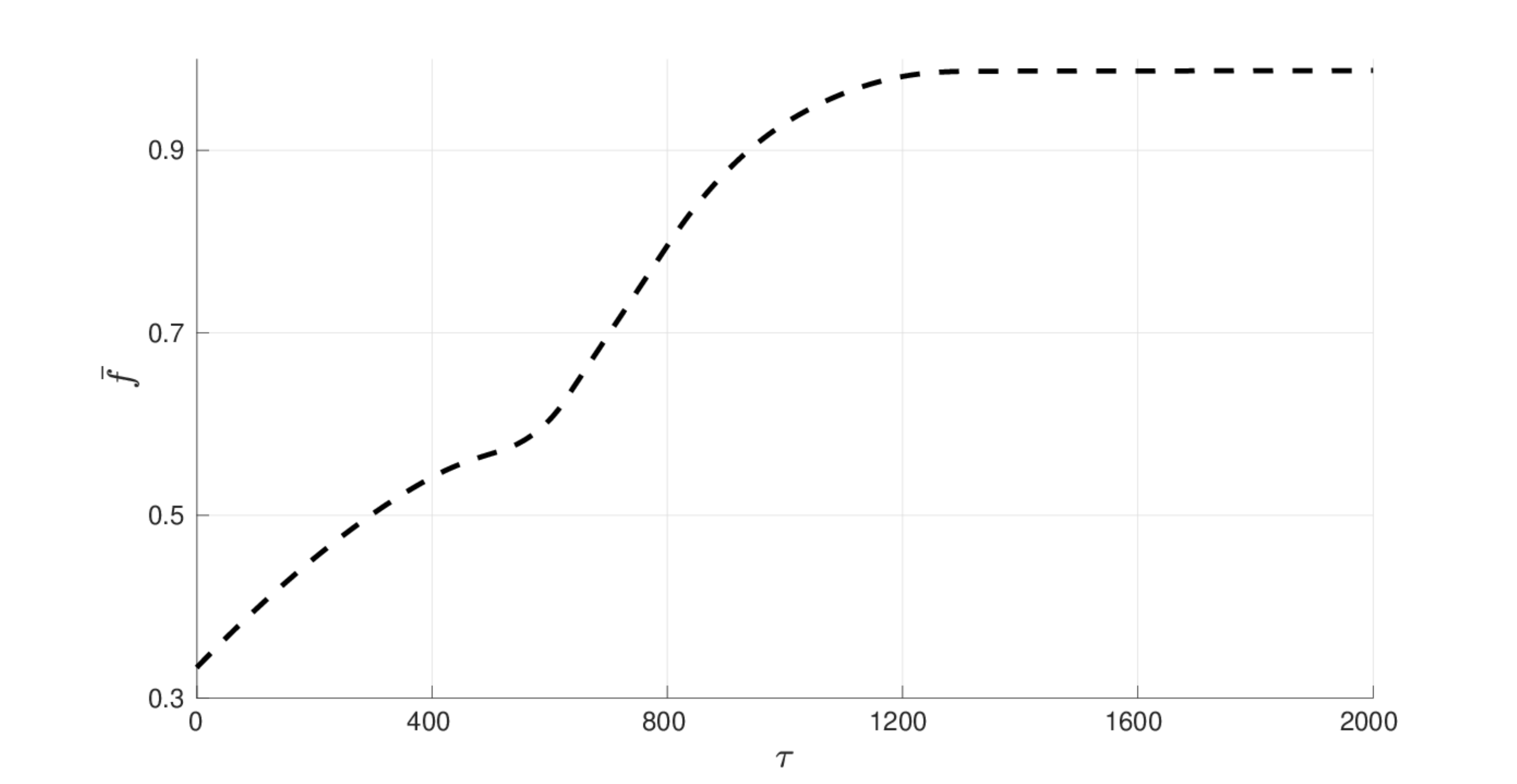} (b)}
	\end{minipage}
	\caption{Stabilisation for $n = 3$: (a) equilibrium coordinates and (b) mean fitness
		of hypercycle system~\eqref{c3:eq3.14} as functions of evolutionary time $\tau$.}
	\label{c3:fig3.7ab}
\end{figure}

\begin{figure}[ht]
	\begin{minipage}[ht]{0.48\linewidth}
		\center{\includegraphics[width=\linewidth]{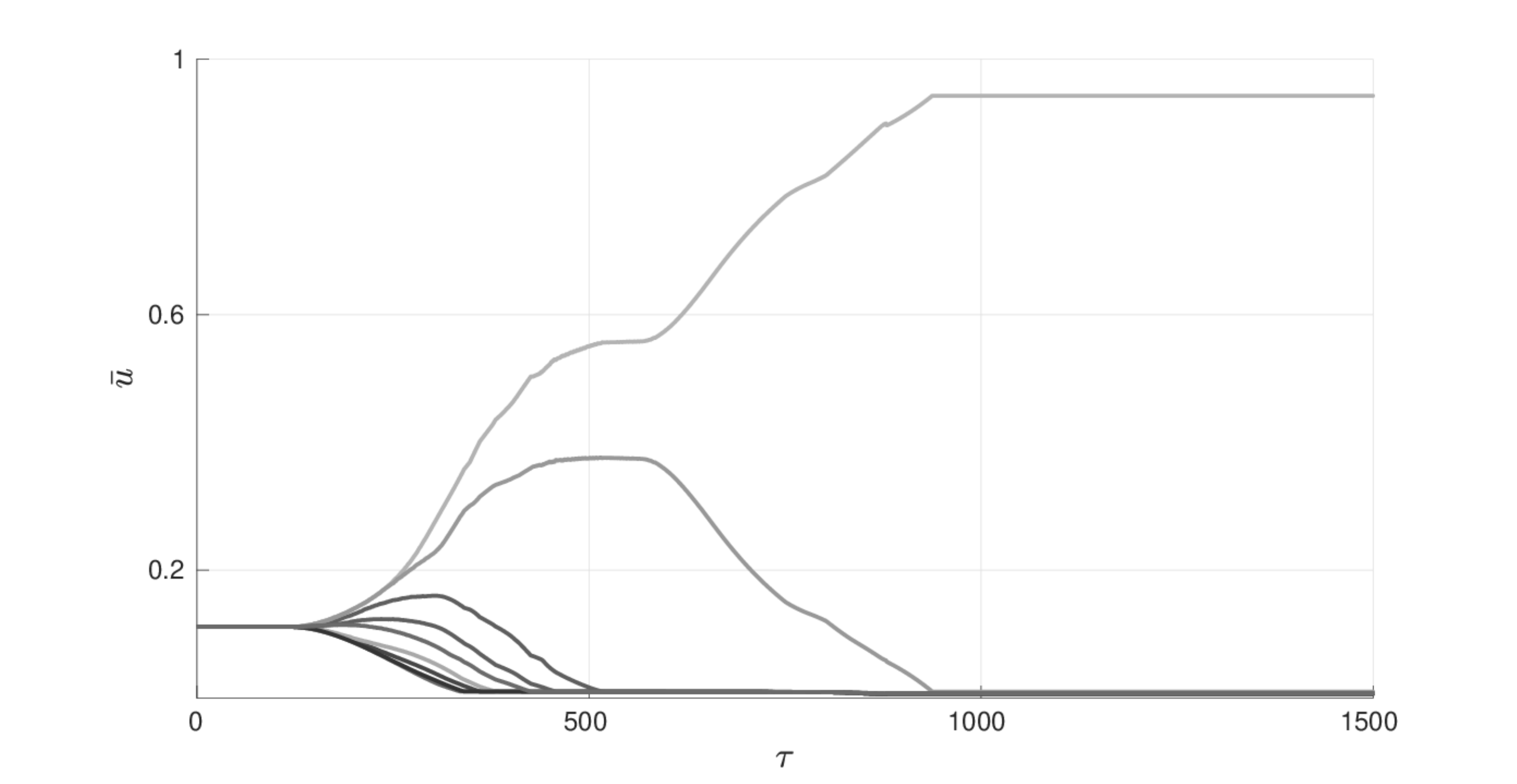} (c)}
	\end{minipage}
	\hfill
	\begin{minipage}[ht]{0.48\linewidth}
		\center{\includegraphics[width=\linewidth]{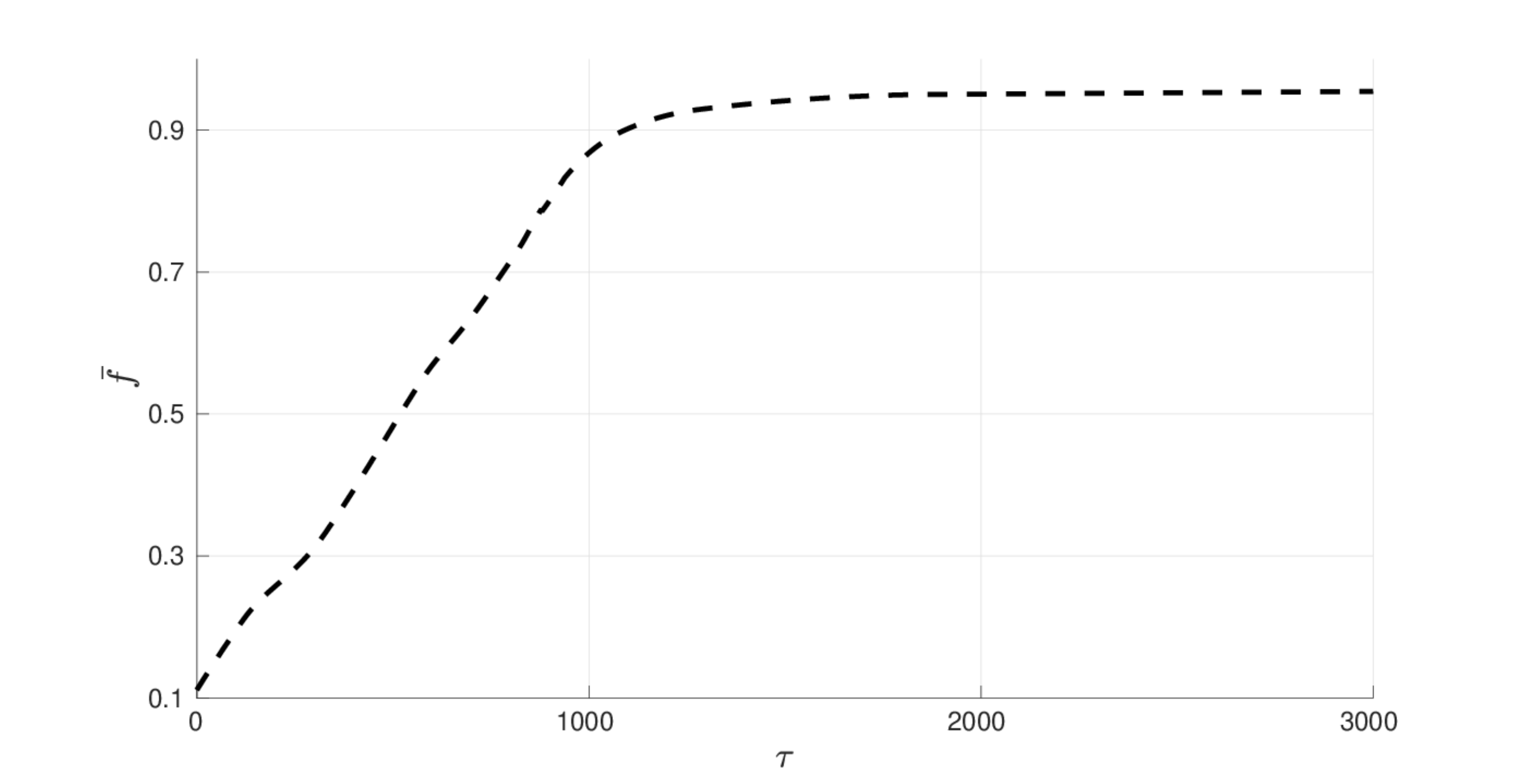} (d)}
	\end{minipage}
	\caption{Stabilisation for $n = 9$: (c) equilibrium coordinates and (d) mean fitness
		of hypercycle system~\eqref{c3:eq3.14} as functions of evolutionary time $\tau$.}
	\label{c3:fig3.7cd}
\end{figure}

Figs.~\ref{c3:fig3.7ab},  \ref{c3:fig3.7cd} show the evolution of the equilibrium coordinates and the mean fitness under the functional \eqref{c3:eq3.17}.

\clearpage

\section{Bi-Hypercycle Evolution}\label{c3:section:3.5}
Consider the system referred to hereafter as the \textit{bi-hypercycle}:
\begin{equation}
	\begin{aligned}	
		&\displaystyle\dot{u}_{i}(t) = u_{i}(t)\Big(a_{i}a_{i - 1}u_{i - 1}(t)u_{i - 2}(t) - f(t)\Big), \quad i = \overline{1, n}, \\
		&\displaystyle f(t) = \sum\limits_{i = 1}^{n}a_{i}a_{i - 1}u_{i}(t)u_{i - 1}(t)u_{i - 2}(t),\\
	\end{aligned}
	\label{c3:eq3.18}
\end{equation}
where $a_{i} > 0$, $a_{0} = a_{n}$, $u_{0}(t) = u_{n}(t)$, $u_{-1}(t) = u_{n - 1}(t)$.

In this type of hypercycle, the catalysis of species $i$ is carried out by the two preceding species $i - 1$ and $i - 2$. As in the simple hypercycle, the frequencies $u_{i}$ belong to the simplex $S_{n}$.

It is proved in~\cite{c3:Safro2013} that the bi-hypercycle is non-degenerate (permanent) for odd values $n \geqslant 5$. For odd degree $n \geqslant 5$, the bi-hypercycle has a unique equilibrium, which is stable for $n = 5$ and unstable for $n > 5$. Numerical experiments indicate that the qualitative behaviour of the phase trajectories of the bi-hypercycle is analogous to that of the simple hypercycle; however, a proof of the existence of a stable limit cycle in this case remains unknown to the authors.

To write system \eqref{c3:eq3.18} in matrix form, introduce
$$
	{\bf A} = 
	\begin{pmatrix}
		0 & 0 & \ldots & 0 & a_{1}\\
		a_{2} & 0 & \ldots & 0 & 0\\
		0 & a_{3} & \ldots & 0 & 0\\
		\ldots & \ldots & \ldots & \ldots & \ldots\\
		0 & 0 & \ldots & a_{n} & 0\\
	\end{pmatrix},
$$ 
${\bf U}(t) = \diag(u_{1}(t),\, u_{2}(t),\, \ldots,\, u_{n}(t))$, ${\bf u}(t) = (u_{1}(t),\, u_{2}(t),\, \ldots,\, u_{n}(t))$ and ${\bf I} = (1,\, 1,\, \ldots,\, 1)$. System \eqref{c3:eq3.18} then becomes
\begin{equation}
	\begin{aligned}	
		&\dot{u}_{i}(t) = u_{i}(t)\Bigg(\Big(\Big({\bf AU}(t)\Big)^{2}{\bf I}\Big)_{i} - f(t)\Bigg), \quad i = \overline{1, n}, \\
		&f(t) = \Bigg(\Big({\bf AU}(t)\Big)^{2}{\bf I},\, {\bf u}(t)\Bigg), \quad \Big({\bf U}(t){\bf I}, {\bf I}\Big) = 1.\\
	\end{aligned}
	\label{c3:eq3.19}
\end{equation}

The steady-state equilibrium at evolutionary time $\tau$ satisfies
\begin{equation}\label{c3:eq3.20}
	({\bf A}(\tau){\bf U}(\tau))^2{\bf I}=\bar{f}(\tau){\bf I},
\end{equation}
where
\begin{equation}\label{c3:eq3.21}
	\bar{f}(\tau) = \left(({\bf A}(\tau){\bf U}(\tau))^2{\bf I}, {\bf I}\right).
\end{equation}

To derive the fitness variance formula, introduce the adjoint equation to \eqref{c3:eq3.20} for the matrix $V(\tau) = \diag (v_1(\tau),\, v_2(\tau), \, \ldots, \, v_n(\tau))$:
\begin{equation}\label{c3:eq3.22}
	\Big({\bf A}^T(\tau)V(\tau){\bf A}^T(\tau)\Big)V(\tau){\bf I}={\bf I}.
\end{equation}

From \eqref{c3:eq3.22} it follows that
\begin{equation*}
	2{\bf A}(\tau){\bf U}(\tau)\left(\delta {\bf A}(\tau){\bf U}(\tau)+{\bf A}(\tau) \delta {\bf U}(\tau)\right){\bf I} = \delta \bar{f}(\tau){\bf I}.
\end{equation*}
Taking the scalar product with the adjoint vector $v = V(\tau){\bf I}$, where $V(\tau)$ solves \eqref{c3:eq3.22}:
\begin{align*}
	\Big({\bf A}(\tau){\bf U}(\tau){\bf A}(\tau)\,\delta{\bf U}(\tau){\bf I},\,
	V(\tau){\bf I}\Big)
	&= \Big(\delta{\bf U}(\tau){\bf I},\,
	\big({\bf A}^T(\tau){\bf U}(\tau){\bf A}^T(\tau)\big)V(\tau){\bf I}\Big) \\
	&= \Big(\delta{\bf U}(\tau){\bf I},\, {\bf I}\Big)
	= \sum_{i=1}^n \delta u_i = 0,
\end{align*}
since $\sum\limits_{i=1}^n \bar{u}_i(\tau)=1$ implies $\sum\limits_{i=1}^n \delta \bar{u}_i=0$. Consequently,
\begin{equation}\label{c3:eq3.23}
	\delta\bar{f}(\tau) = \dfrac{2\left(\delta {\bf A}(\tau){\bf U}(\tau){\bf I}, {\bf A}^TV(\tau){\bf I}\right)}{\sum\limits_{i=1}^n v_i }.
\end{equation}

Formula \eqref{c3:eq3.23} yields the necessary extremality condition for the fitness by the same argument as in Section~\ref{c3:section:3.2}; the iteration scheme carries over directly.

As in the simple hypercycle, numerical computations show that for
$0 \leqslant \tau \leqslant \tau^*$ the equilibrium
${\bf U}^* = \diag(u_1^*,\, \ldots,\, u_n^*)$ of~\eqref{c3:eq3.20} remains fixed.
The functional
\begin{equation*}
	F({\bf A}) = \Big(\big({\bf A}(\tau){\bf U}^*\big)^2{\bf I},\, {\bf I}\Big)
\end{equation*}
is continuous and bounded on the convex set $\mathcal{M}$. We show its convexity:
\begin{multline*}
	\dfrac{1}{2} F({\bf A}_{\lambda})
	= \dfrac{1}{2} F\big(\lambda {\bf A}_1(\tau) + (1-\lambda) {\bf A}_2(\tau)\big)
	\leqslant \lambda^2 F({\bf A}_1(\tau)) + (1-\lambda)^2 F({\bf A}_2(\tau)) \\
	\leqslant \lambda F({\bf A}_1(\tau)) + (1-\lambda) F({\bf A}_2(\tau)),
	\quad 0 < \lambda < 1.
\end{multline*}
Consequently, $F({\bf A})$ attains its maximum at some ${\bf A}^* \in \mathcal{M}$.

Consider now the fitness landscape evolution of the bi-hypercycle system \eqref{c3:eq3.19} with $n = 5$ and initial matrix
\begin{equation}
	{\bf A} = 
	\begin{pmatrix}
		0 & 0 & 0 & 0 & 1\\
		1 & 0 & 0 & 0 & 0\\
		0 & 1 & 0 & 0 & 0\\
		0 & 0 & 1 & 0 & 0\\
		0 & 0 & 0 & 1 & 0\\
	\end{pmatrix}.
	\label{c3:eq3.24}
\end{equation}

\begin{figure}[H]
	\begin{minipage}[ht]{1.0\linewidth}
		\center{\includegraphics[width=0.8\linewidth]{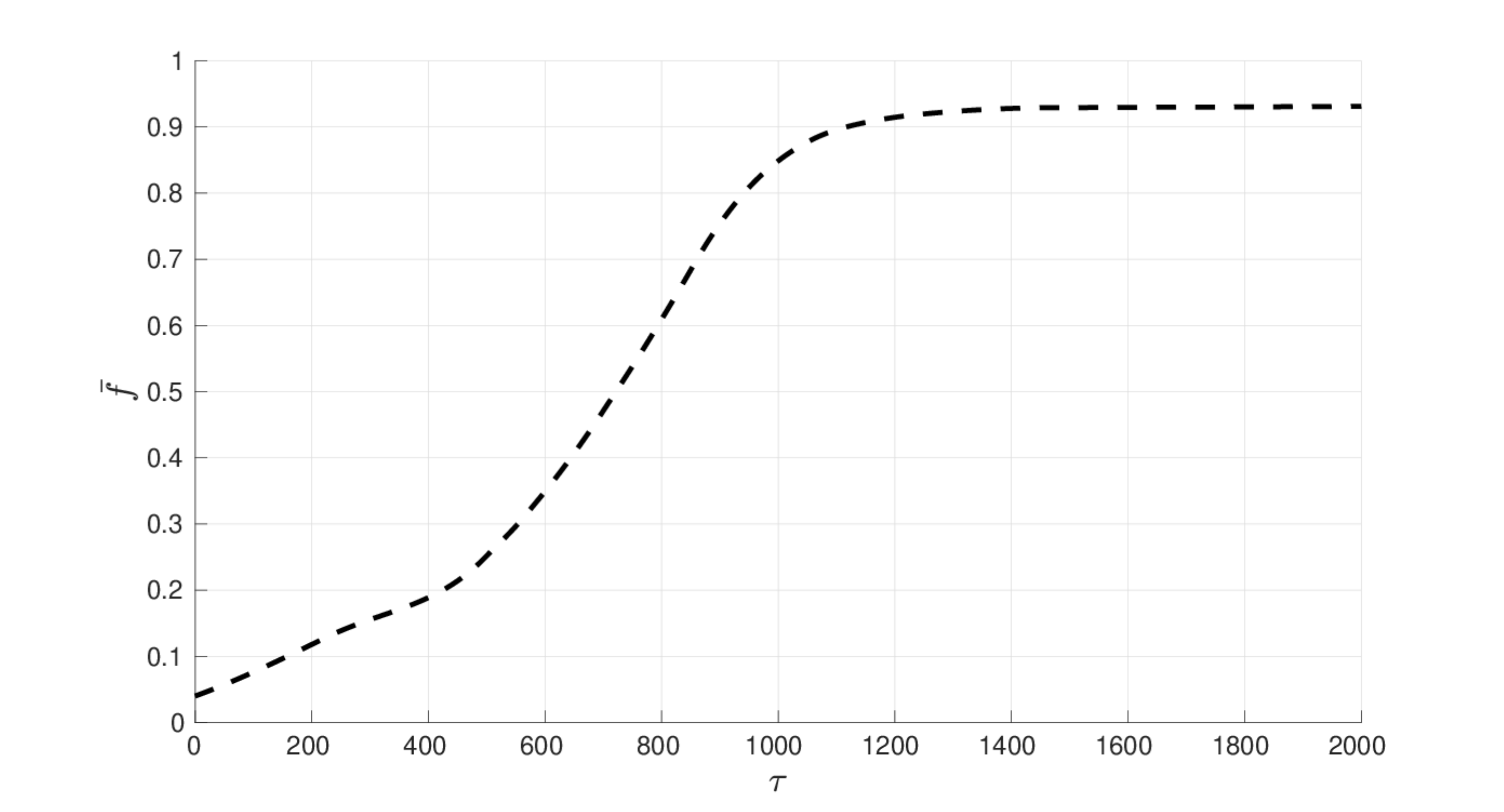} (a)}
	\end{minipage}
	\hfill
	\begin{minipage}[ht]{1,0\linewidth}
		\center{\includegraphics[width=0.8\linewidth]{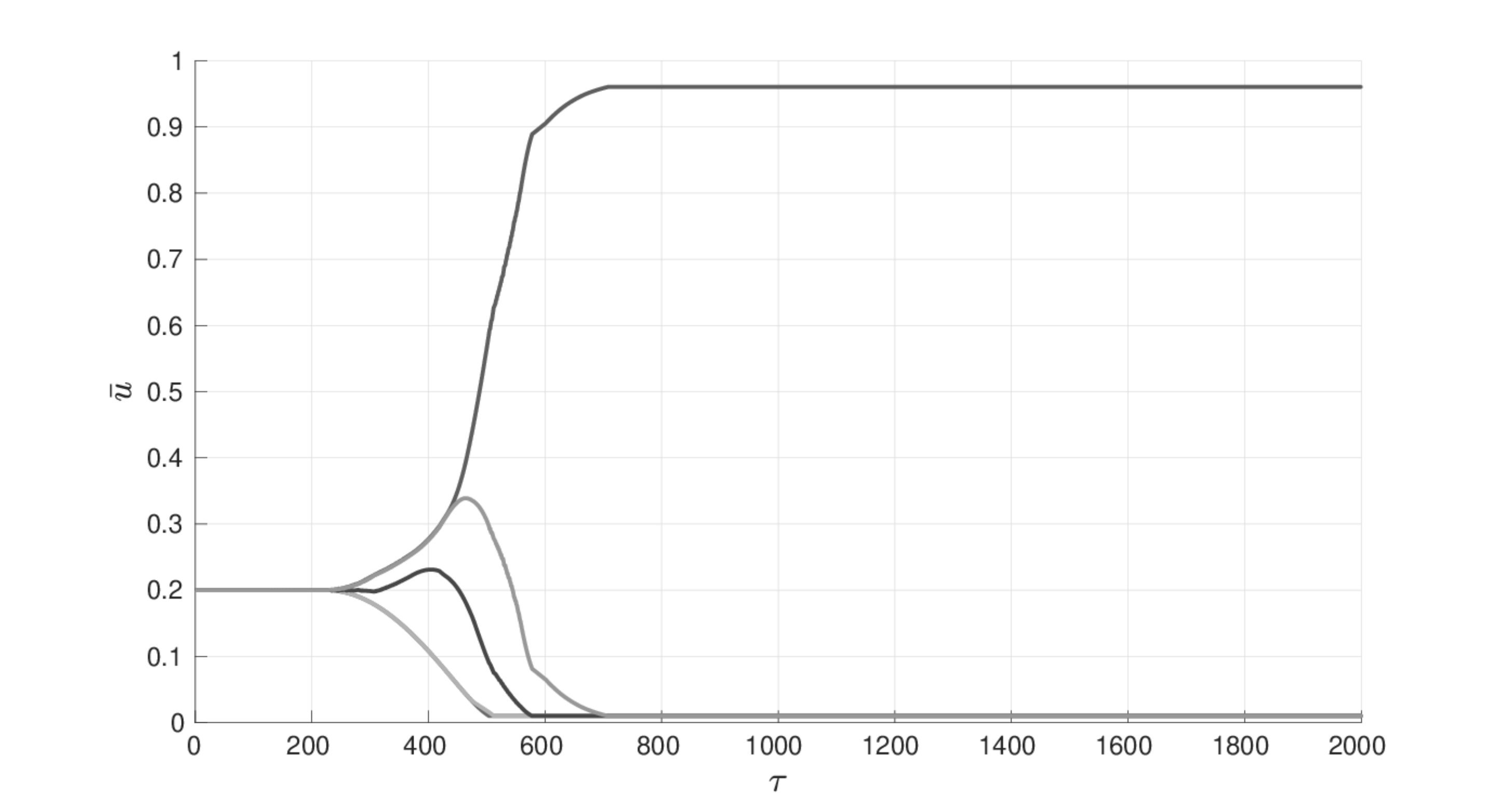} (b)}
	\end{minipage}
	\caption{Evolutionary dynamics of (a) the mean fitness and (b) the equilibrium coordinates of system \eqref{c3:eq3.19} with matrix ${\bf A}$ \eqref{c3:eq3.24}.}
	\label{c3:fig3.8}
\end{figure}

Fig.~\ref{c3:fig3.8}\,(a) shows the evolution of the mean fitness function $\bar{f}$. After a certain number of iterations, $\bar{f}$ converges to its maximum value. Fig.~\ref{c3:fig3.8}\,(b) shows the evolution of the equilibrium coordinates: they remain unchanged for approximately $300$ iterations, after which they split along separate trajectories, one of which approaches $1$ while the rest approach the boundary of the simplex.

\begin{figure}[H]
	\begin{minipage}[ht]{1.0\linewidth}
		\center{\includegraphics[width=0.75\linewidth]{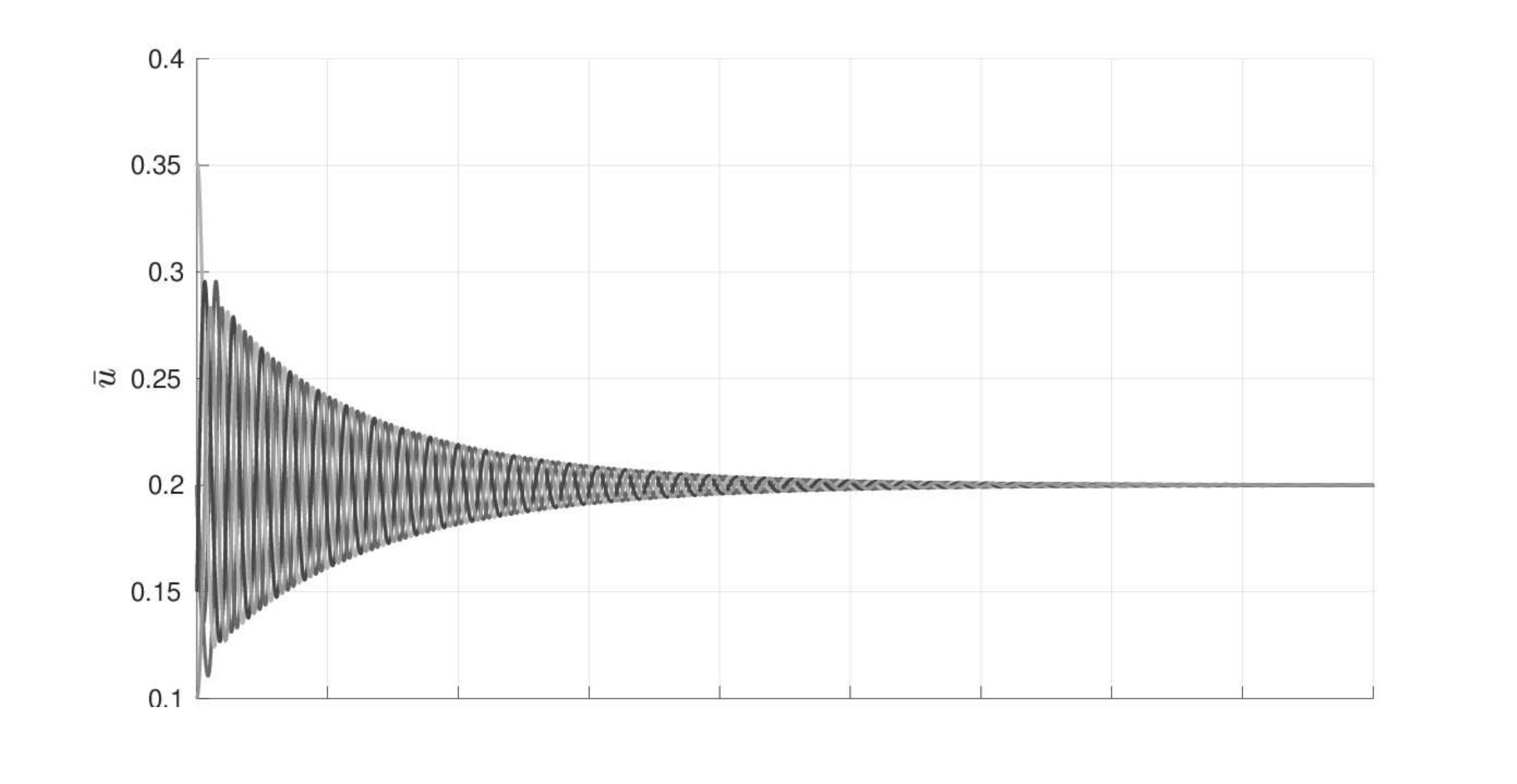} (a)}
	\end{minipage}
	\hfill
	\begin{minipage}[ht]{1.0\linewidth}
		\center{\includegraphics[width=0.75\linewidth]{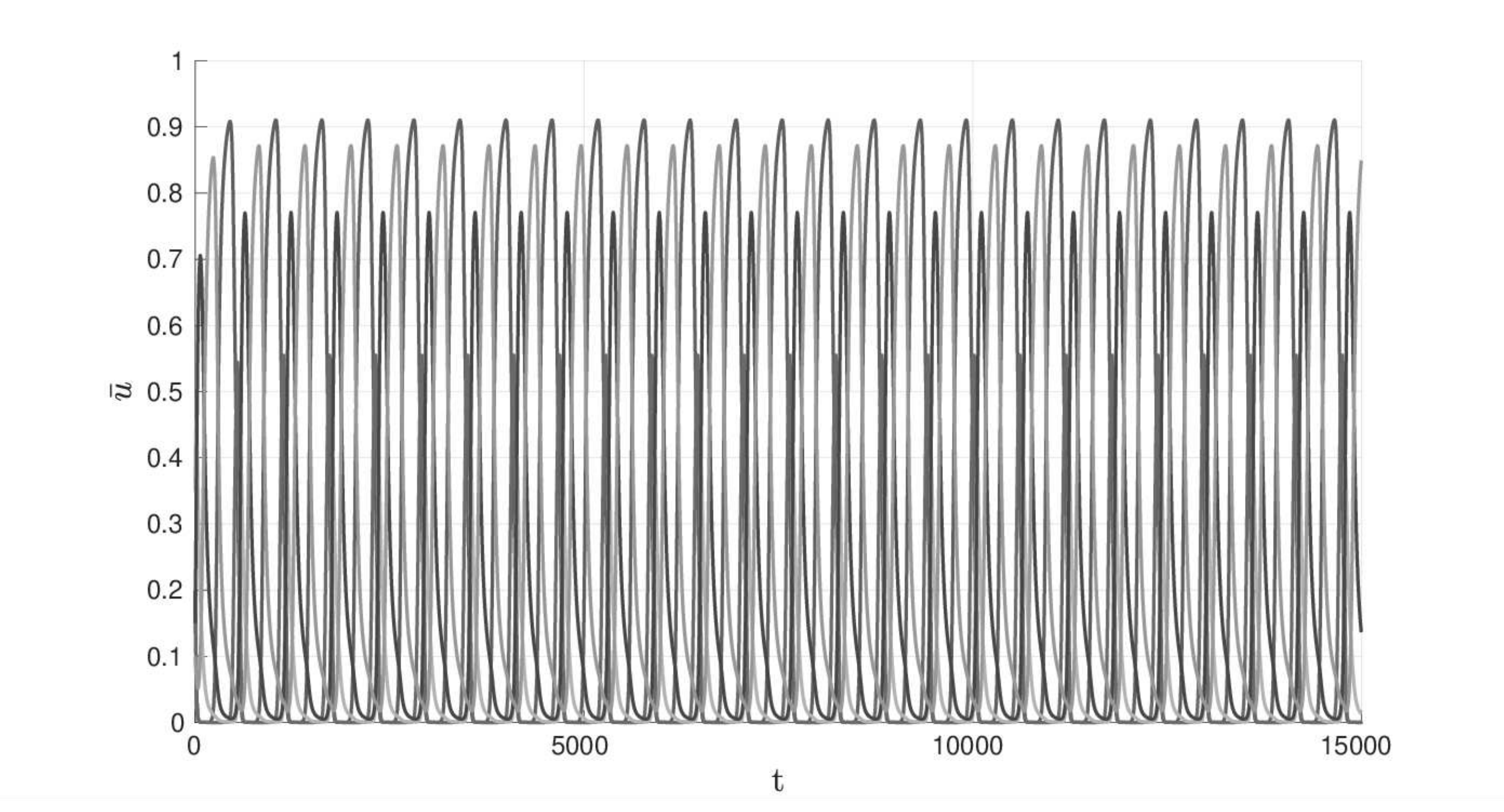} (b)}
	\end{minipage}
	\caption{Active dynamics of system \eqref{c3:eq3.19} with matrix ${\bf A}$ \eqref{c3:eq3.24} at (a) iteration $200$ and (b) iteration $450$.}
	\label{c3:fig3.9}
\end{figure}

Fig.~\ref{c3:fig3.9}\,(a) shows the active dynamics at iteration $200$: phase trajectories converge to the stable equilibrium $\bar{u}_{i} = 0.2$, $i = \overline{1, 5}$. At iteration $450$ (Fig.~\ref{c3:fig3.9}\,(b)), no stable equilibrium exists and the trajectories exhibit cyclic behaviour.

\begin{figure}[H]
	\begin{minipage}[ht]{0.48\linewidth}
		\center{\includegraphics[width=0.9\linewidth]{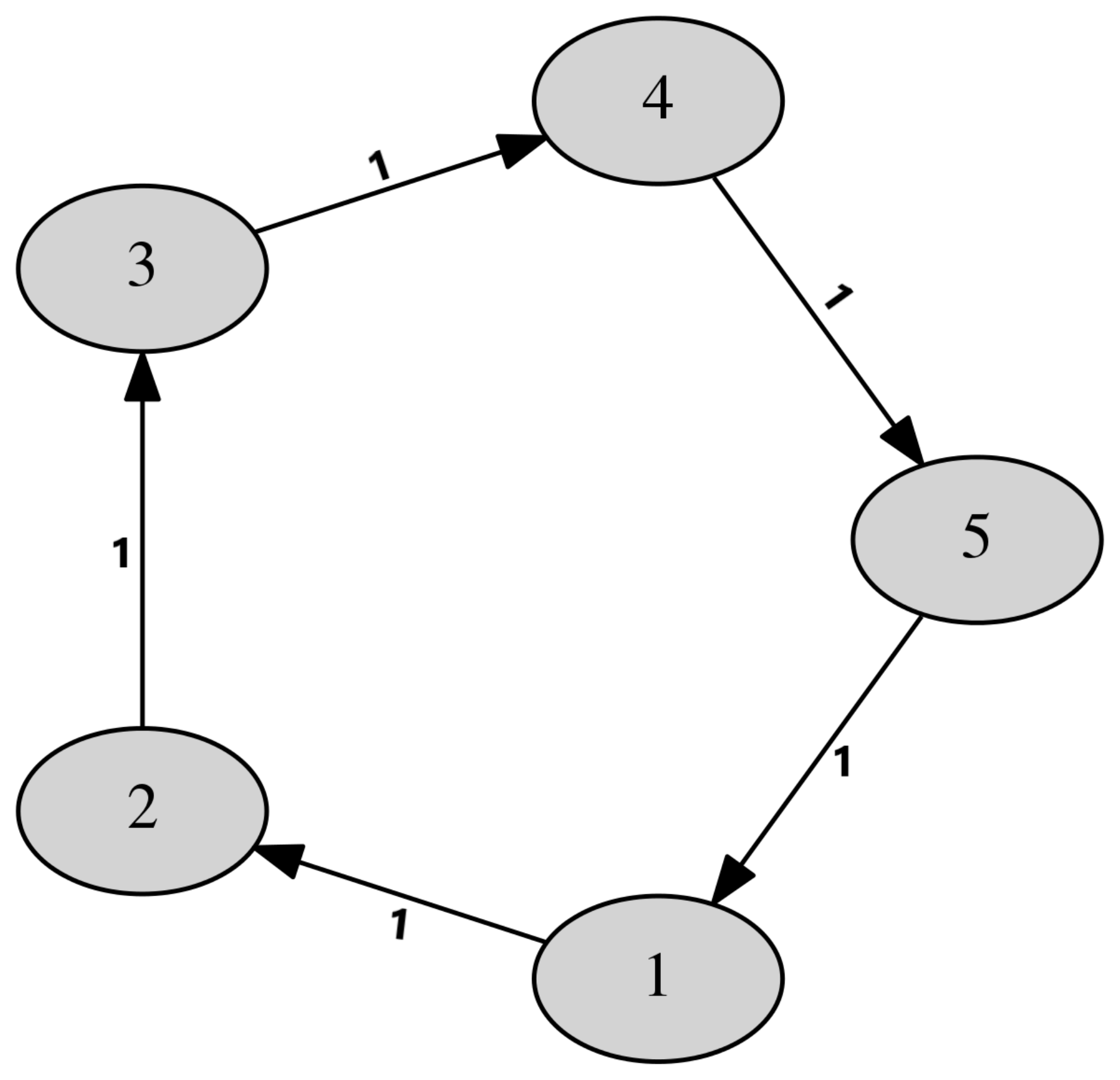}}
		\center{(a)}
	\end{minipage}
	\hfill
	\begin{minipage}[ht]{0.48\linewidth}
		\center{\includegraphics[width=0.9\linewidth]{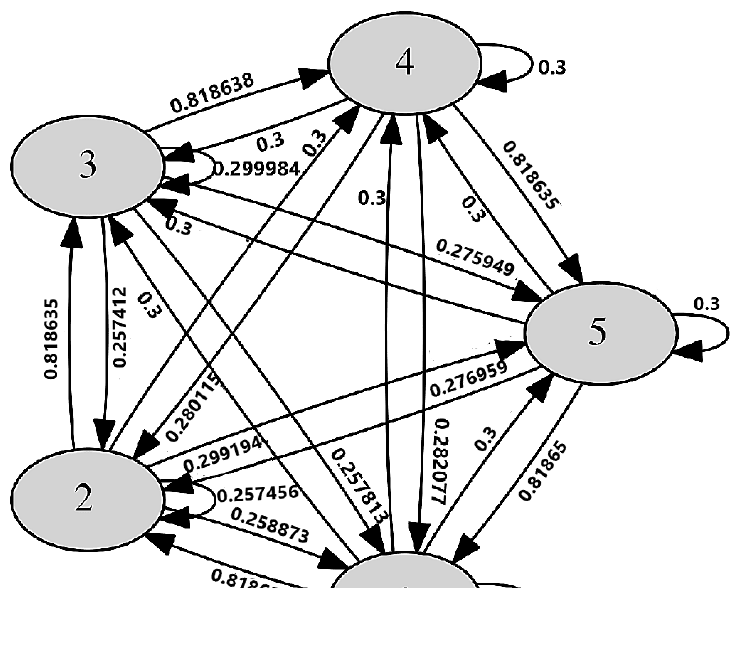}}
		\center{(b)}
	\end{minipage}
	\caption{Interaction graphs for (a) the original fifth-order bi-hypercycle
		and (b) the evolved system at iteration $300$.}
	\label{c3:fig3.10}
\end{figure}

Fig.~\ref{c3:fig3.10} shows the interaction graphs for the original and the evolved (iteration $300$) bi-hypercycle.

\bigskip
We now study the evolutionary adaptation of the {\bf fifth-order bi-hypercycle in the presence of two parasitic species.} Augmenting system \eqref{c3:eq3.24} with
\begin{equation}
	\begin{aligned}
		&\displaystyle\dfrac{du_{6}}{dt} = u_{6}(1.1u_{5}u_{4} - f({\bf u})), \\
		&\displaystyle\dfrac{du_{7}}{dt} = u_{7}(1.1u_{6}u_{5} - f({\bf u})), \\
		&\displaystyle f({\bf u}) = \sum\limits_{i = 1}^{5}u_{i}u_{i - 1}u_{i - 2} + 1.1(u_{4}u_{5}u_{6} + u_{7}u_{6}u_{5}),
	\end{aligned}
	\label{c3:eq3.25}
\end{equation}
where $u_{0} = u_{5}$, $u_{-1} = u_{4}$, the interaction matrix of the extended system is

\medskip
\begin{equation}
	{\bf \widetilde{A}} = 
	\begin{pmatrix}
		0 & 0 & 0 & 0 & 1 & 0 & 0\\
		1 & 0 & 0 & 0 & 0 & 0 & 0\\
		0 & 1 & 0 & 0 & 0 & 0 & 0\\
		0 & 0 & 1 & 0 & 0 & 0 & 0\\
		0 & 0 & 0 & 1 & 0 & 0 & 0\\
		0 & 0 & 0 & 1.1 & 0 & 0 & 0\\
		0 & 0 & 0 & 0 & 1.1 & 0 & 0\\
	\end{pmatrix}\!.
	\label{c3:eq3.26}
\end{equation}

\begin{figure}[ht]
	\begin{minipage}[ht]{0.5\linewidth}
		\center{\includegraphics[width=1.1\linewidth]{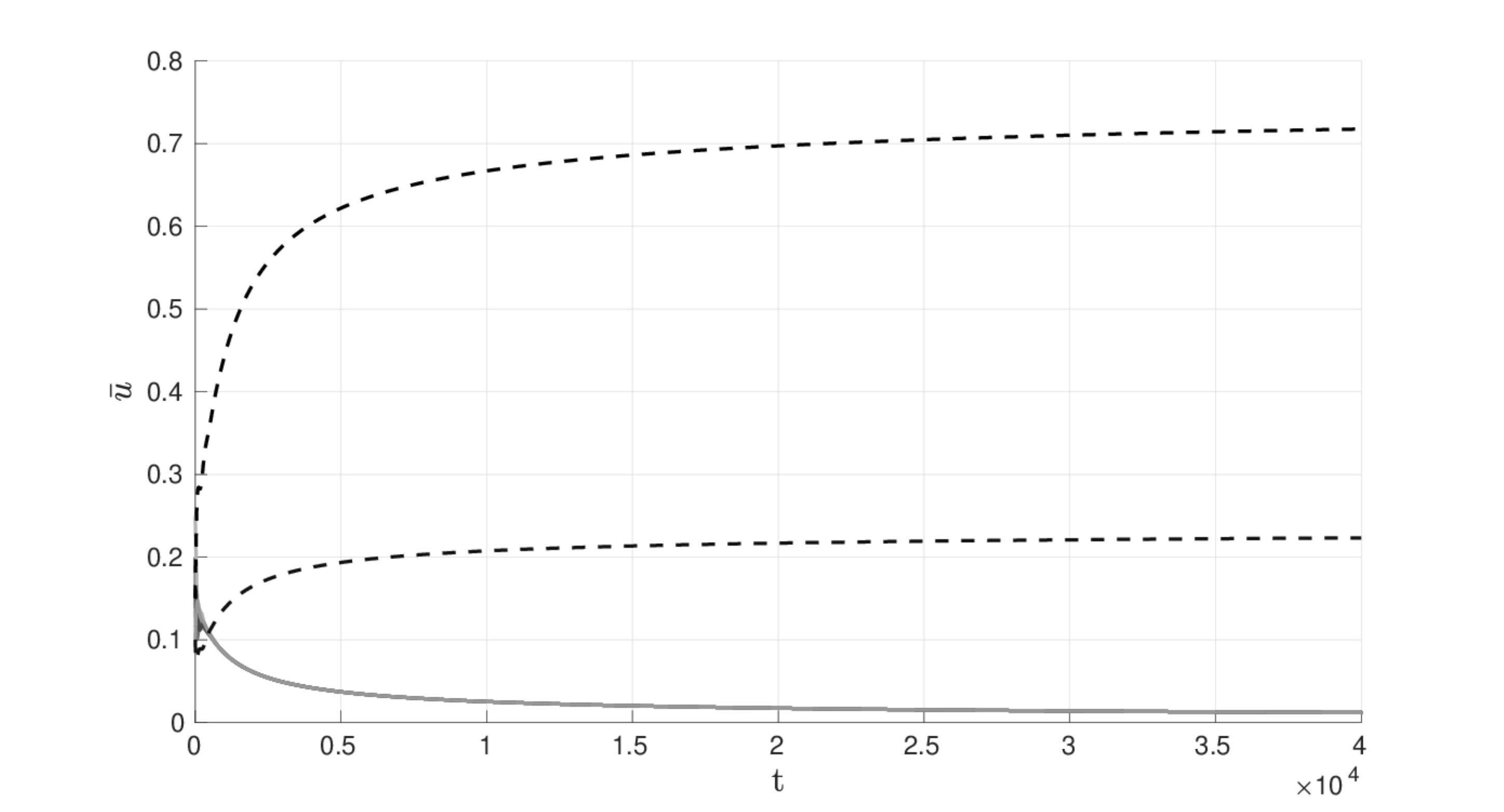} (a)}
	\end{minipage}
	\hfill
	\begin{minipage}[ht]{0.5\linewidth}
		\center{\includegraphics[width=1.1\linewidth]{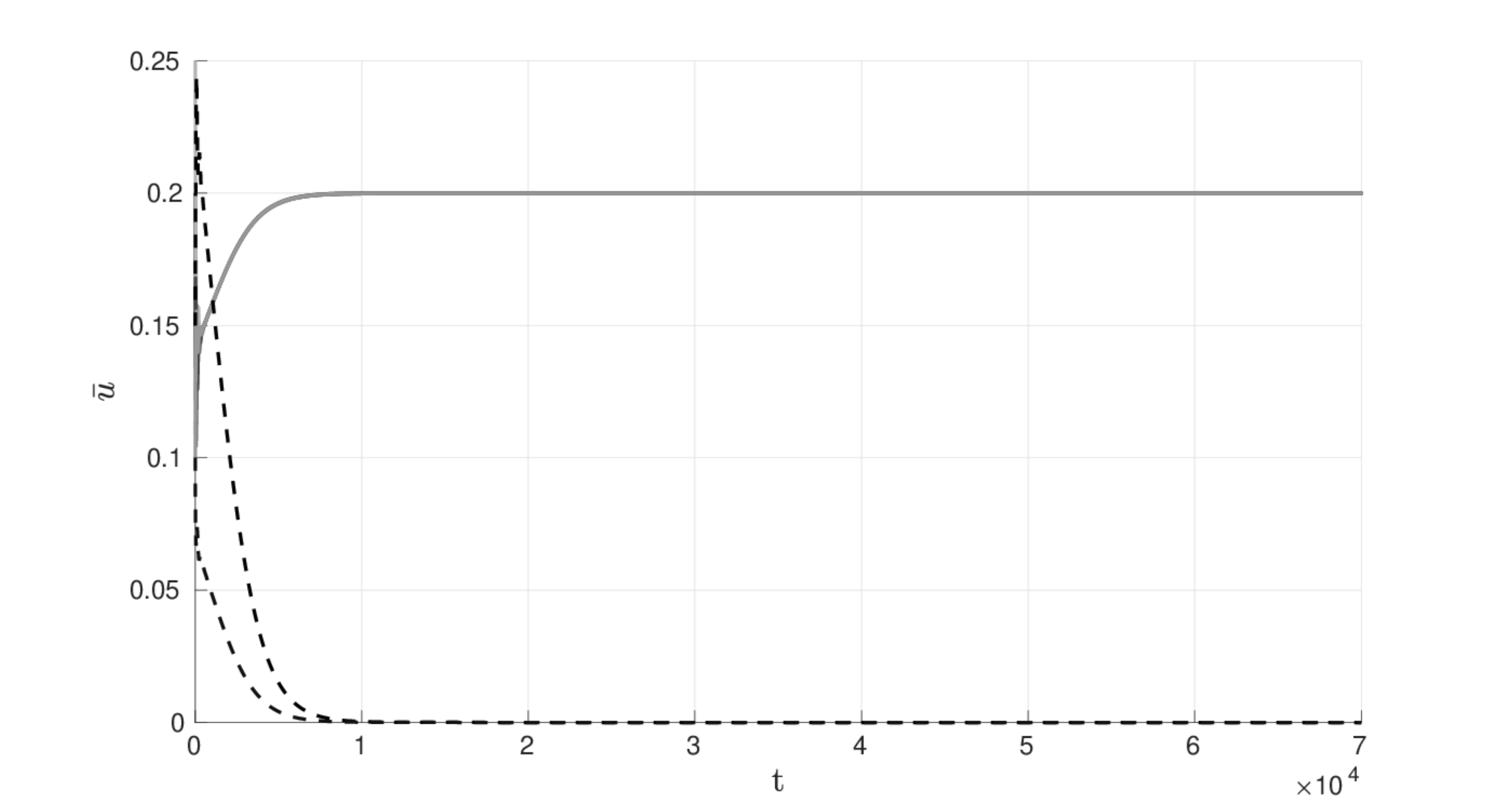} (b)}
	\end{minipage}
	\vfill
	\begin{minipage}[ht]{0.5\linewidth}
		\center{\includegraphics[width=1.1\linewidth]{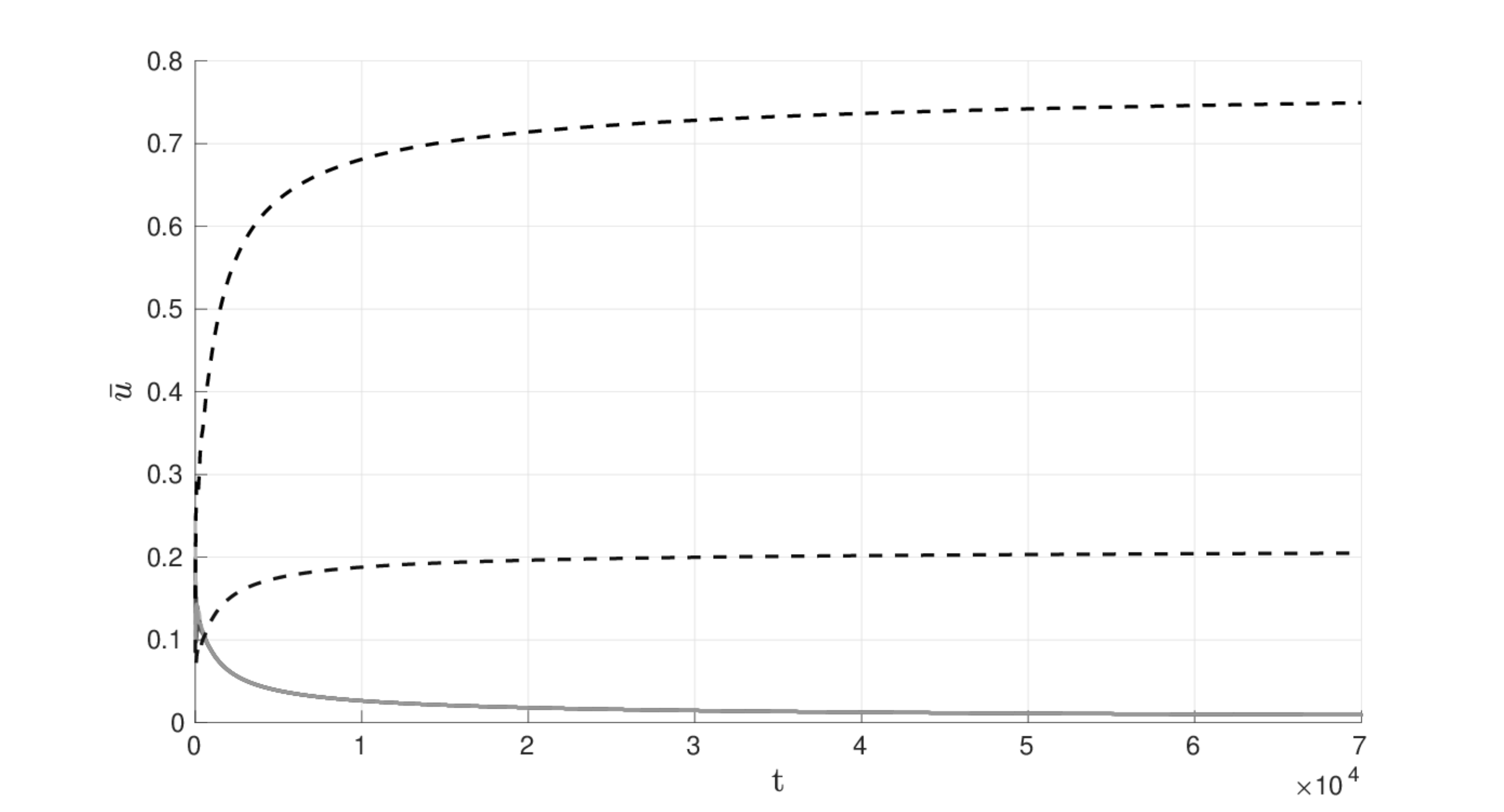} (c)}
	\end{minipage}
	\hfill
	\begin{minipage}[ht]{0.5\linewidth}
		\center{\includegraphics[width=1.1\linewidth]{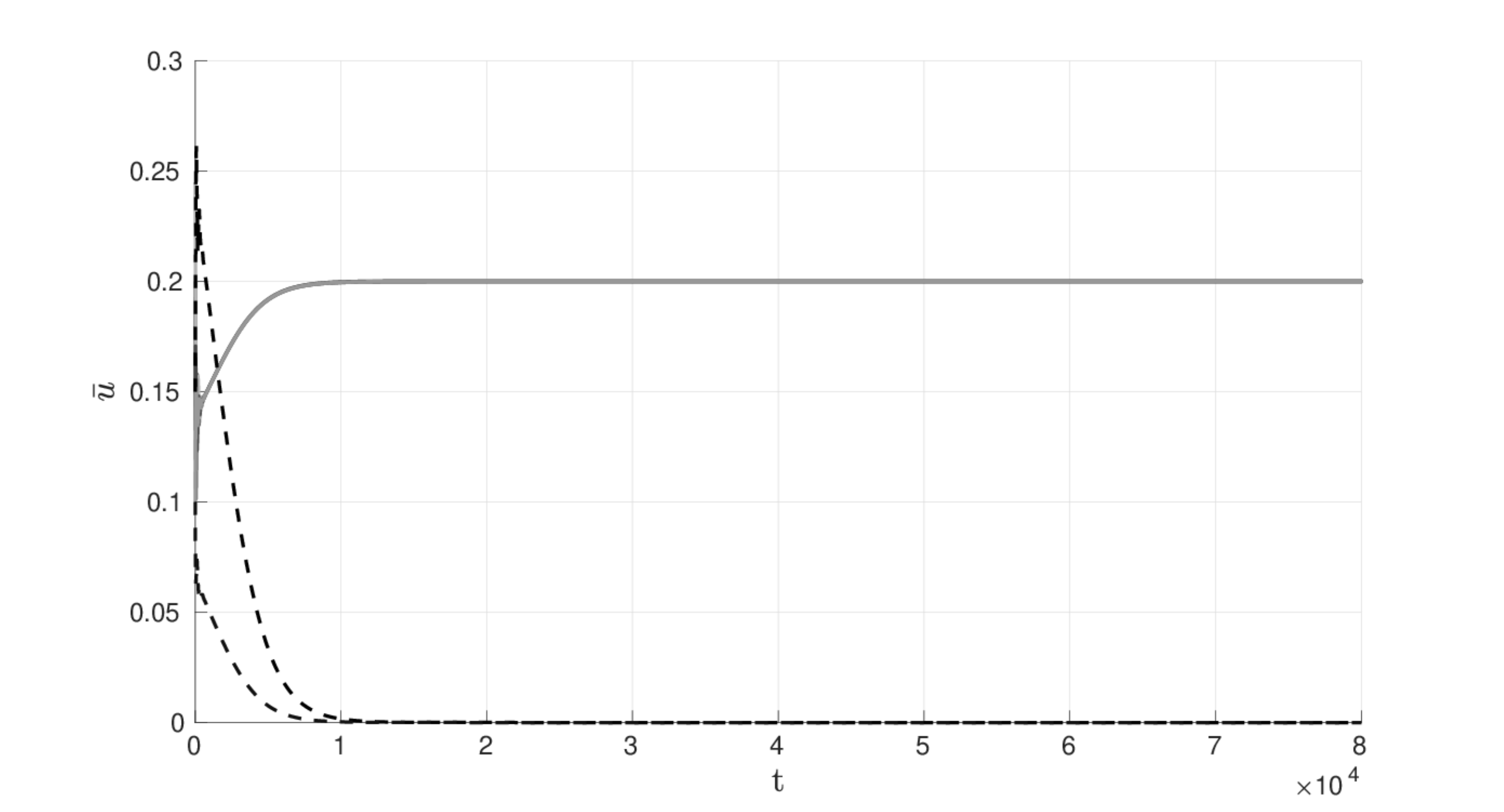} (d)}
	\end{minipage}
	\caption{(a) Frequency dynamics when the original fifth-order bi-hypercycle interacts with two parasites. (b) Frequency dynamics when the evolved fifth-order bi-hypercycle (iteration $30$) interacts with two parasites. (c) Frequency dynamics when the evolved bi-hypercycle (iteration $30$) interacts with two parasites after increasing their replication coefficients. (d) Frequency dynamics when the evolved bi-hypercycle (iteration $55$) interacts with two parasites with increased replication coefficients. Parasite frequencies are shown as dashed lines.}
	\label{c3:fig3.11}
\end{figure}

As established earlier, for simple hypercycles there exists an iteration beyond which the evolved system becomes resistant to parasites. The same property holds for bi-hypercycle systems.

Fig.~\ref{c3:fig3.11}\,(a): the original bi-hypercycle collapses when interacting with two parasites. Fig.~\ref{c3:fig3.11}\,(b): after $30$ evolutionary steps, the evolved bi-hypercycle is resistant to both parasites. If the parasite replication coefficients are increased from $1.1$ to $1.2$, $30$ iterations no longer suffice (Fig.~\ref{c3:fig3.11}\,(c)); at least $55$ iterations are required (Fig.~\ref{c3:fig3.11}\,(d)).

\medskip
The interaction matrix of the extended system obtained from \eqref{c3:eq3.26} after $30$ evolutionary adaptation steps is:
{\small
\begin{equation}
	{\bf \widetilde{A}}_{30} = 
	\begin{pmatrix}
		0.03 & 0.03 & 0.03 & 0.03 & 0.998096 & 0 & 0\\
		0.999063 & 0.03 & 0.03 & 0.03 & 0.03 & 0 & 0\\
		0.03 & 0.998857 & 0.03 & 0.03 & 0.03 & 0 & 0\\
		0.03 & 0.03 & 0.997774 & 0.03 & 0.03 & 0 & 0\\
		0.03 & 0.03 & 0.03 & 0.997564 & 0.03 & 0 & 0\\
		0 & 0 & 0 & 1.1 & 0 & 0 & 0\\
		0 & 0 & 0 & 0 & 1.1 & 0 & 0\\
	\end{pmatrix}\!.
	\label{c3:eq3.27}
\end{equation}
}

\clearpage

\section{``Anthill'' Replicator System}\label{c3:section:3.6}
We illustrate the evolutionary adaptation algorithm applied to the ``anthill'' system~\eqref{c3:eq1.21}, introduced in~Chapter~\ref{ch:1}. The interaction graph of this system --- in which species $0,\ldots,n-1$ form a hypercycle each additionally catalysed by the dominant ``queen'' species $n$ --- is described in~Chapter~\ref{ch:1}.

Setting $\alpha = 0.1$, $\beta_{i} = 0.8$, $k_{i} = 1$, $n = 5$ in system \eqref{c3:eq1.21}, the fitness landscape matrix takes the form
\begin{equation}
	{\bf A} = 
	\begin{pmatrix}
		0 & 0 & 0 & 1 & 0.1\\
		1 & 0 & 0 & 0 & 0.1\\
		0 & 1 & 0 & 0 & 0.1\\
		0 & 0 & 1 & 0 & 0.1\\
		0.8 & 0.8 & 0.8 & 0.8 & 0\\
	\end{pmatrix}\!.
	\label{c3:eq3.28}
\end{equation}

The evolutionary adaptation process, which maximises the mean fitness, leads to a decrease in the parameters $\beta_{i}$ governing the catalysis of the dominant species (the ``queen'') (Fig.~\ref{c3:fig3.12}). The equilibrium coordinates corresponding to the dominant species converge to zero at approximately iteration $310$.

Figs.~\ref{c3:fig3.13ab}, \ref{c3:fig3.13cd} show the phase trajectories of system \eqref{c3:eq1.21} with matrix ${\bf A}$ \eqref{c3:eq3.28} at iterations $50$, $175$, $250$, and $400$ of the mean fitness maximisation process. The frequency component corresponding to the dominant species is shown as a dashed line. Notably, the influence of the ``queen'' decreases with each iteration while the influence of the remaining species grows.

\begin{figure}[ht]
	\begin{minipage}[ht]{0.48\linewidth}
		\center{\includegraphics[width=\linewidth]{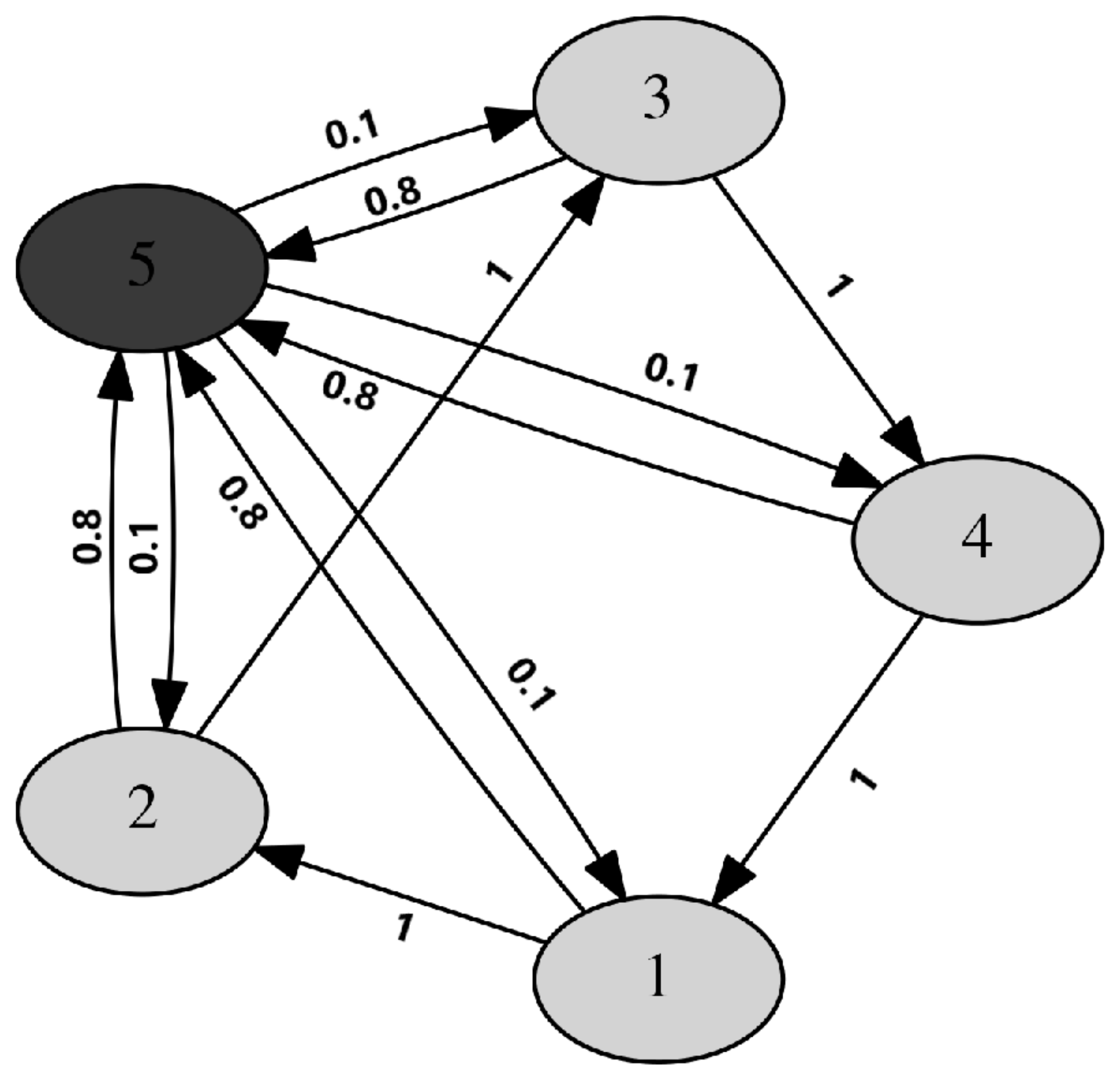}}
		\center{(a)}
	\end{minipage}
	\hfill
	\begin{minipage}[ht]{0.48\linewidth}
		\center{\includegraphics[width=\linewidth]{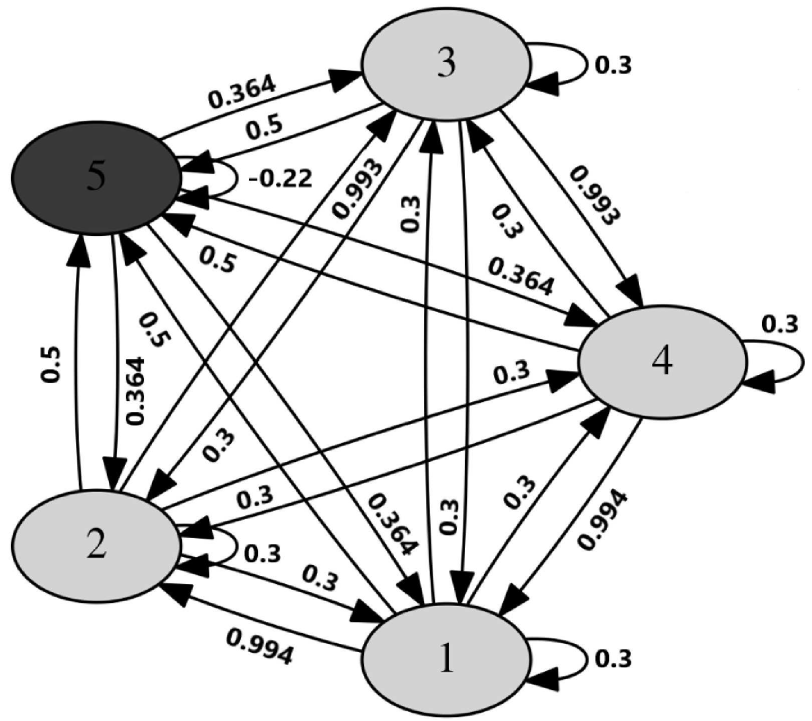}}
		\center{(b)}
	\end{minipage}
	\caption{Interaction graphs for the ``anthill'' system~\eqref{c3:eq1.21}:
		(a) original, with the dominant species catalysing all others;
		(b) evolved system at iteration $300$, where the queen's influence
		has diminished significantly.}
	\label{c3:fig3.12}
\end{figure}

\begin{figure}[ht]
	\begin{minipage}[ht]{0.48\linewidth}
		\center{\includegraphics[width=\linewidth]{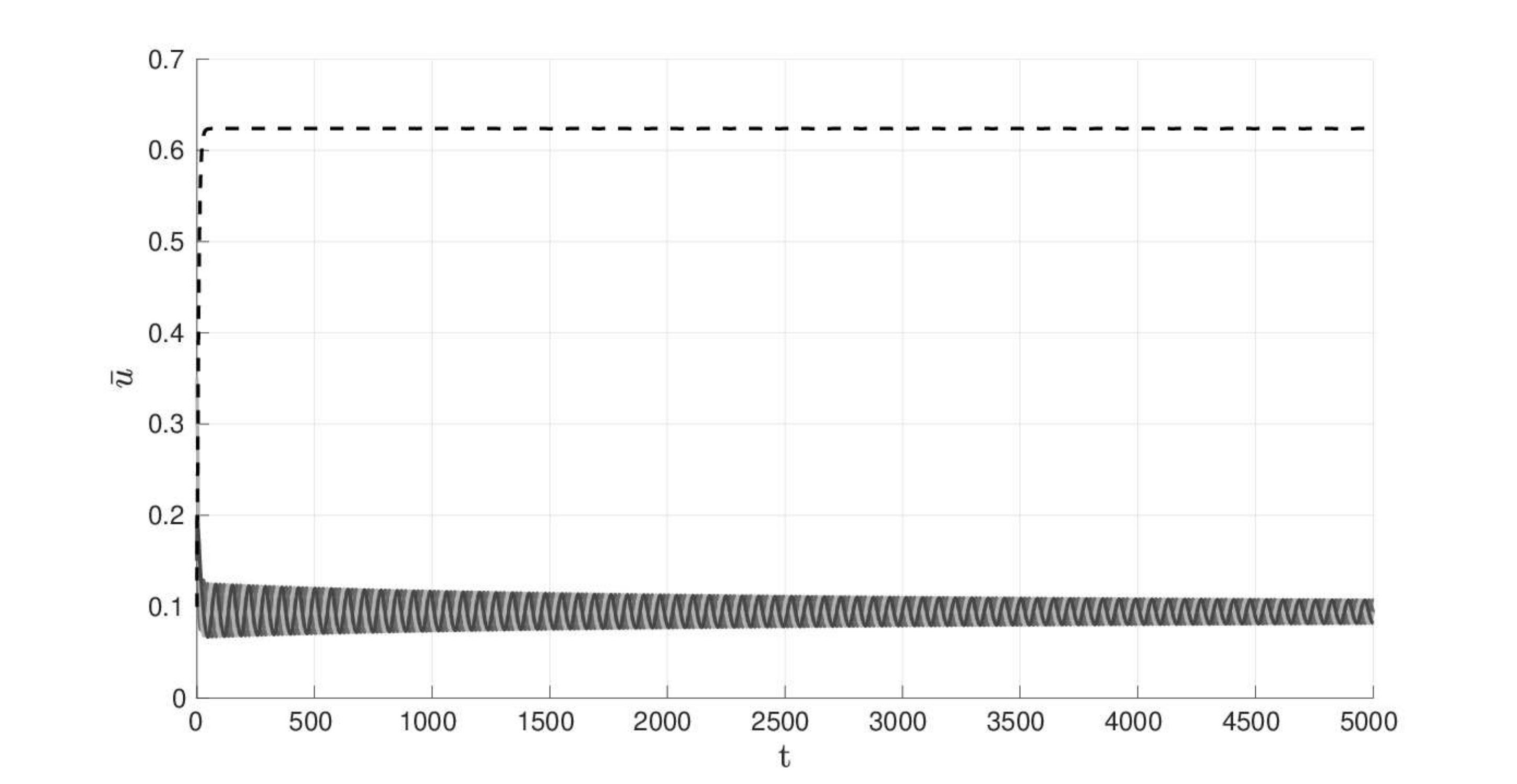}}
		\center{(a)}
	\end{minipage}
	\hfill
	\begin{minipage}[ht]{0.48\linewidth}
		\center{\includegraphics[width=\linewidth]{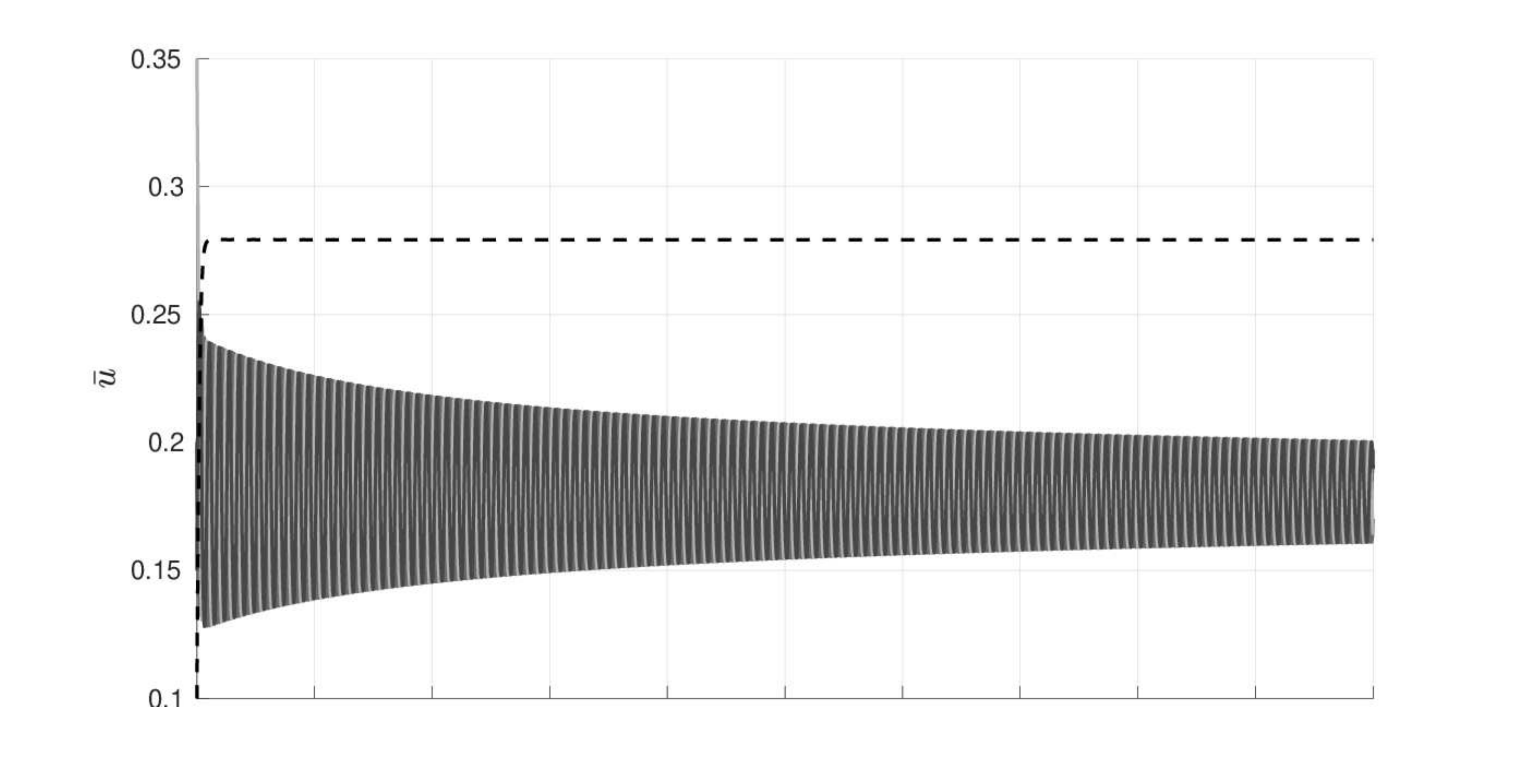}}
		\center{(b)}
	\end{minipage}
	\caption{Phase trajectories of system~\eqref{c3:eq1.21} with matrix ${\bf A}$~\eqref{c3:eq3.28}
		at evolutionary steps (a) $N = 50$ and (b) $N = 175$.
		The dominant species (dashed line) retains significant influence.}
	\label{c3:fig3.13ab}
\end{figure}

\begin{figure}[ht]
	\begin{minipage}[ht]{0.48\linewidth}
		\center{\includegraphics[width=\linewidth]{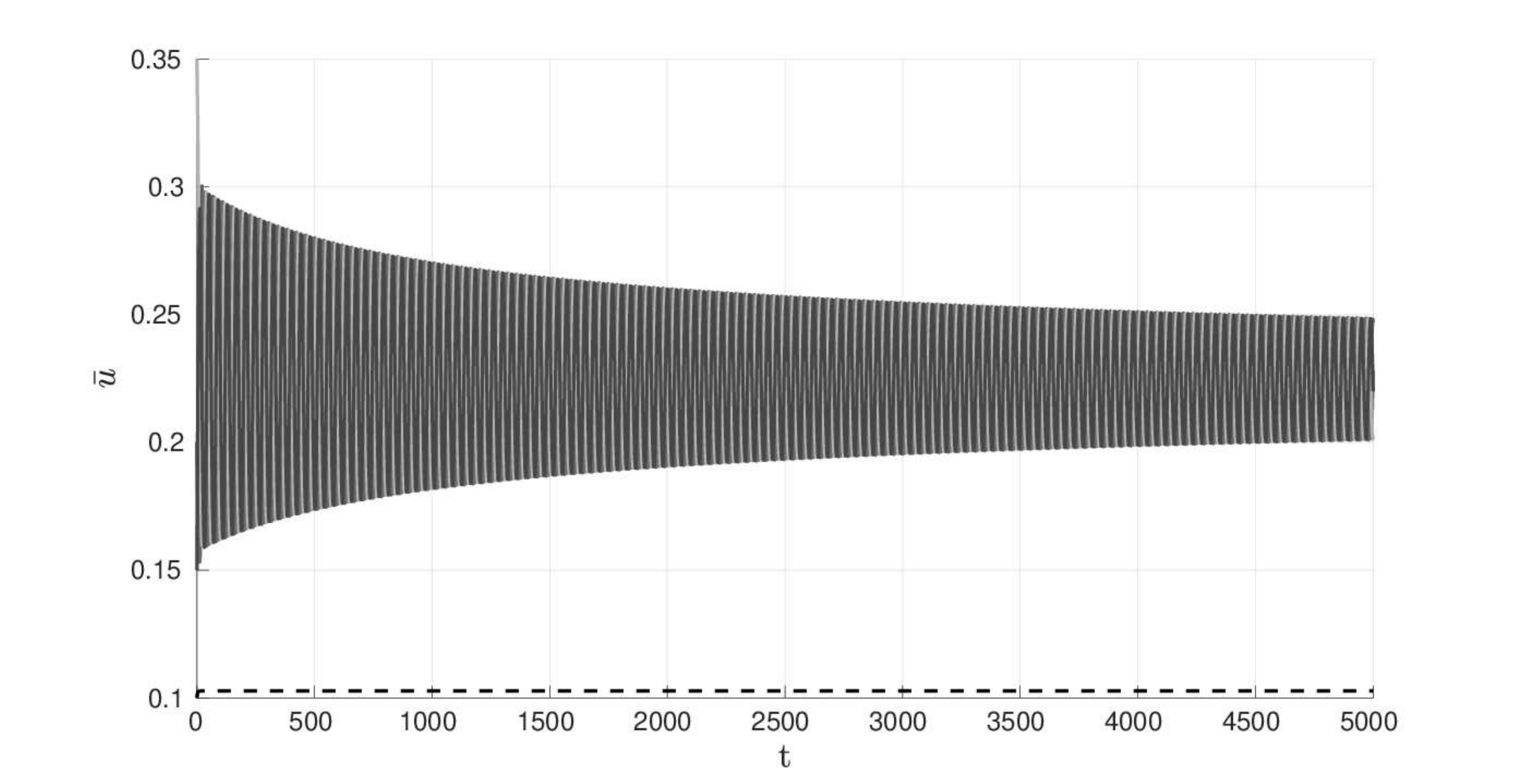}}
		\center{(c)}
	\end{minipage}
	\hfill
	\begin{minipage}[ht]{0.48\linewidth}
		\center{\includegraphics[width=\linewidth]{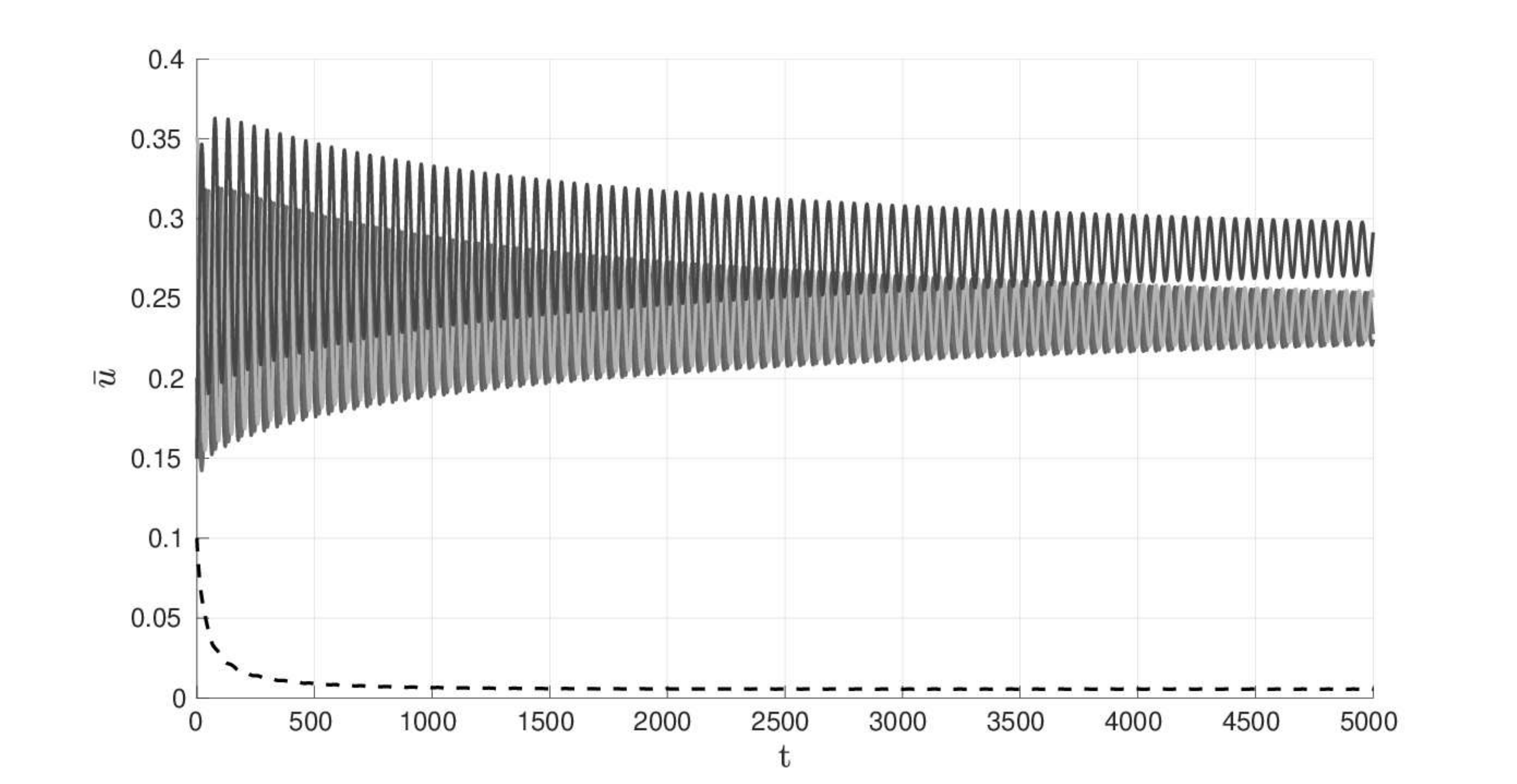}}
		\center{(d)}
	\end{minipage}
	\caption{Phase trajectories of system~\eqref{c3:eq1.21} with matrix ${\bf A}$~\eqref{c3:eq3.28}
		at evolutionary steps (c) $N = 250$ and (d) $N = 400$.
		The queen's frequency (dashed) approaches zero as the remaining species dominate.}
	\label{c3:fig3.13cd}
\end{figure}
\clearpage

\section{Evolution of the RNA Molecule Network}\label{c3:section:3.7}
Consider the RNA catalysis system~\eqref{c3:eq1.25}, introduced in~Chapter~\ref{ch:1}. The interaction graph of this system is described there.

The system describes the interaction of six distinct RNA macromolecules and was proposed on the basis of \textit{in vitro} experiments~\cite{c3:Vaidya2012}. The macromolecules fall into two groups: species $4$--$5$--$6$ form a hypercycle, while species $1$--$2$--$3$ participate in both hypercyclic and autocatalytic replication, thus exhibiting a dual nature.

Setting the parameters in \eqref{c3:eq1.25} as
$$
	\begin{aligned}
		&r_{1} = r_{2} = r_{3} = -0.3, \\
		&k_{1} = k_{2} = k_{3} = 0.4, \\
		&k_{4} = k_{5} = k_{6} = 0.1, \\
		&\bar{k}_{4} = \bar{k}_{5} = \bar{k}_{6} = 0.05,\\
	\end{aligned}
$$
the fitness landscape matrix at evolutionary step $250$ is
\begin{equation}
	{\bf A} = 
	\begin{pmatrix}
		-0.3 & 0 & 0 & 0.4 & 0 & 0\\
		0 & -0.3 & 0 & 0 & 0.4 & 0\\
		0 & 0 & -0.3 & 0 & 0 & 0.4\\
		0 & 0 & 0.1 & 0 & 0.05 & 0\\
		0.1 & 0 & 0 & 0 & 0 & 0.05\\
		0 & 0.1 & 0 & 0.05 & 0 & 0\\
	\end{pmatrix}\!.
	\label{c3:eq3.29}
\end{equation}

The interaction graph for the original system \eqref{c3:eq1.25} with matrix ${\bf A}$ \eqref{c3:eq3.29} is shown in Fig.~\ref{c3:fig3.14}\,(a); Fig.~\ref{c3:fig3.14}\,(b) shows the corresponding graph of the evolved system at iteration $200$.

\begin{figure}[H]
	\begin{minipage}[ht]{0.48\linewidth}
		\center{\includegraphics[width=\linewidth]{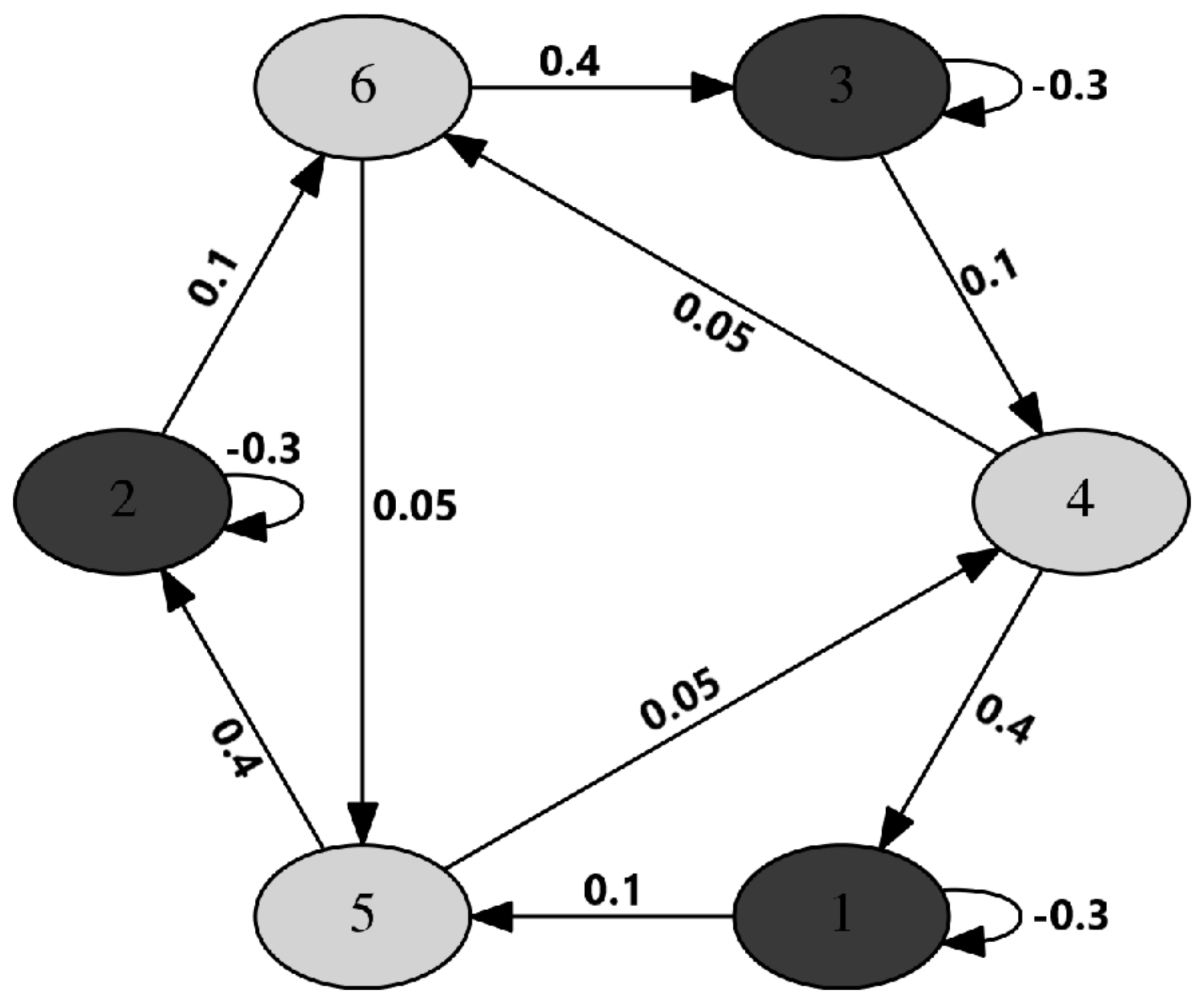}}
		\center{(a)}
	\end{minipage}
	\hfill
	\begin{minipage}[ht]{0.48\linewidth}
		\center{\includegraphics[width=\linewidth]{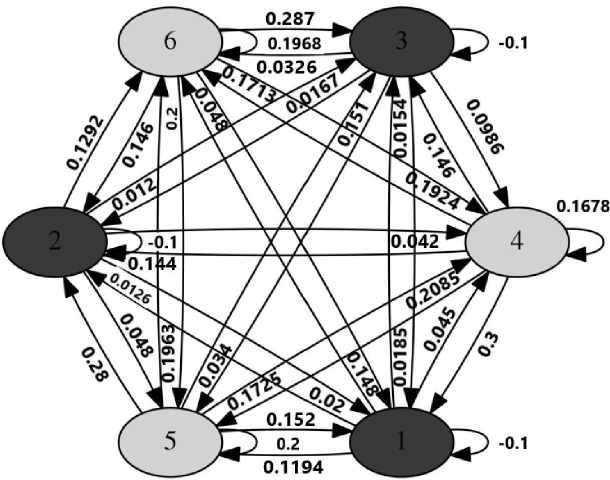}}
		\center{(b)}
	\end{minipage}
	\caption{Interaction graphs for the RNA network~\eqref{c3:eq1.25} with
		matrix~\eqref{c3:eq3.29}: (a) original system, with species $1$--$3$ exhibiting
		dual autocatalytic and hypercyclic links; (b) evolved system at step $200$.}
	\label{c3:fig3.14}
\end{figure}

Fig.~\ref{c3:fig3.15} shows the phase trajectories of system \eqref{c3:eq1.25} prior to the evolutionary process and at steps $125$, $175$, and $200$. Solid lines indicate the relative frequencies of the first group of macromolecules and dashed lines those of the second group.

\begin{figure}[ht]
	\begin{minipage}[ht]{0.5\linewidth}
		\center{\includegraphics[width=1.05\linewidth]{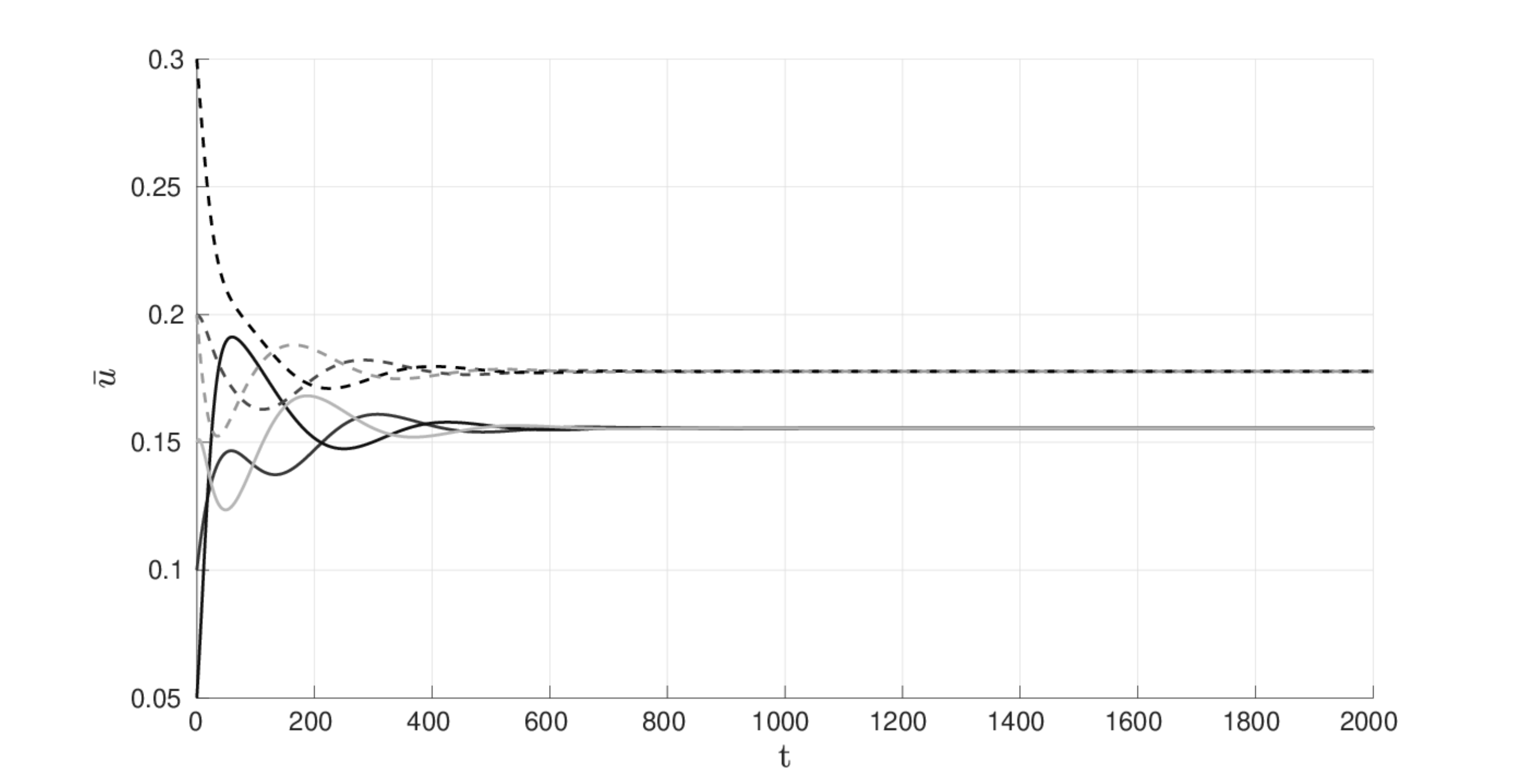} (a)}
	\end{minipage}
	\hfill
	\begin{minipage}[ht]{0.5\linewidth}
		\center{\includegraphics[width=1.05\linewidth]{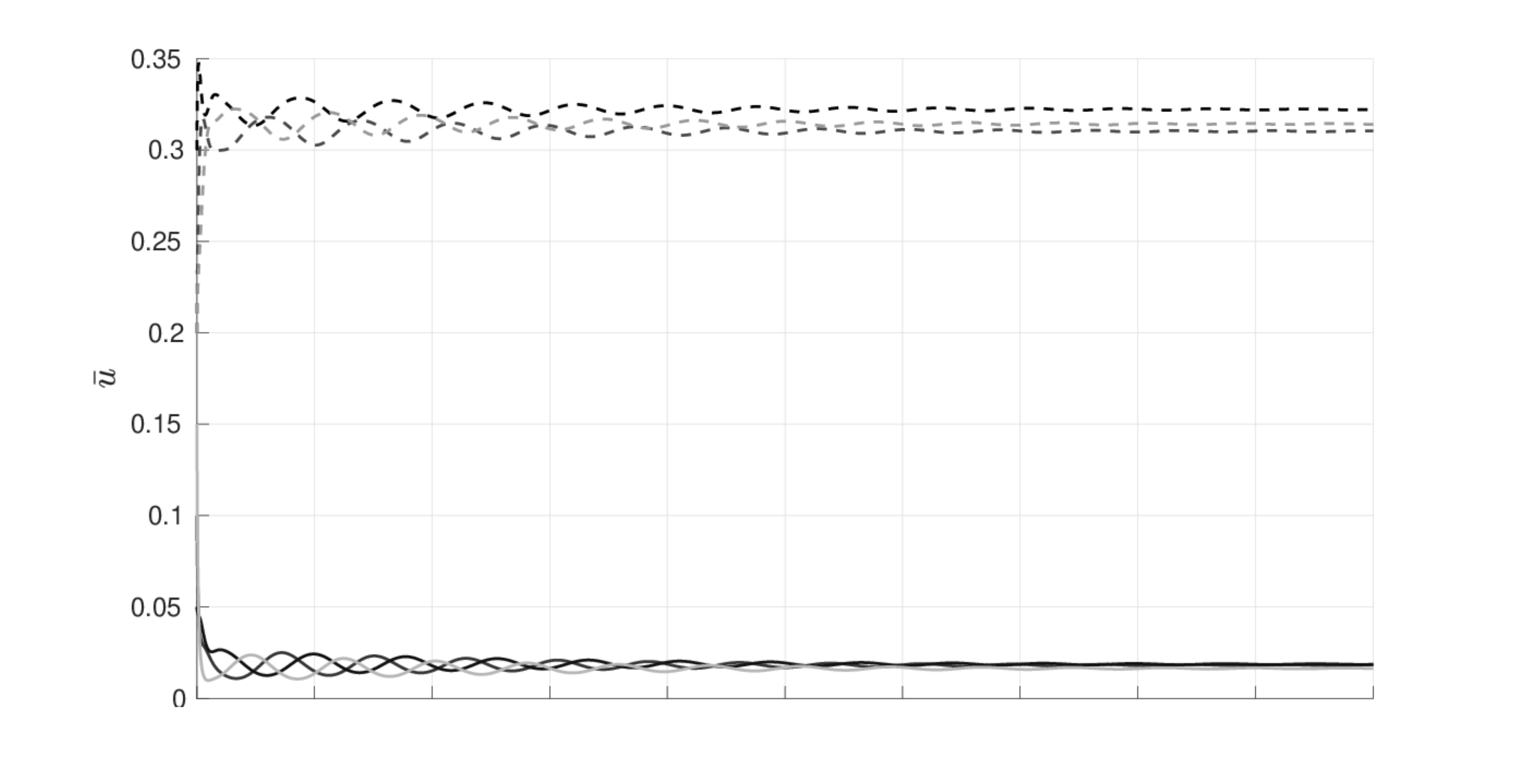} (b)}
	\end{minipage}
	\vfill
	\begin{minipage}[ht]{0.5\linewidth}
		\center{\includegraphics[width=1.05\linewidth]{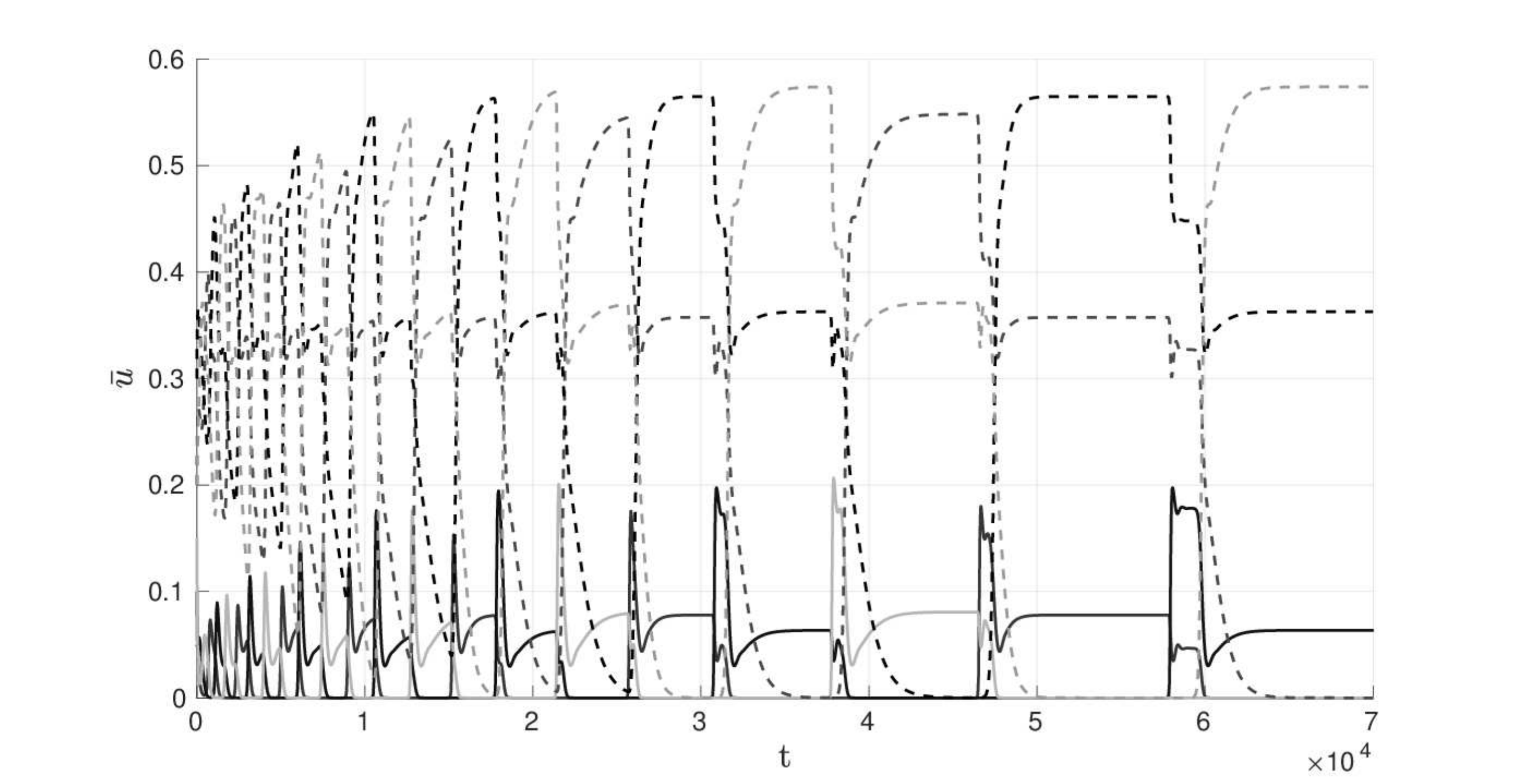} (c)}
	\end{minipage}
	\hfill
	\begin{minipage}[ht]{0.5\linewidth}
		\center{\includegraphics[width=1.05\linewidth]{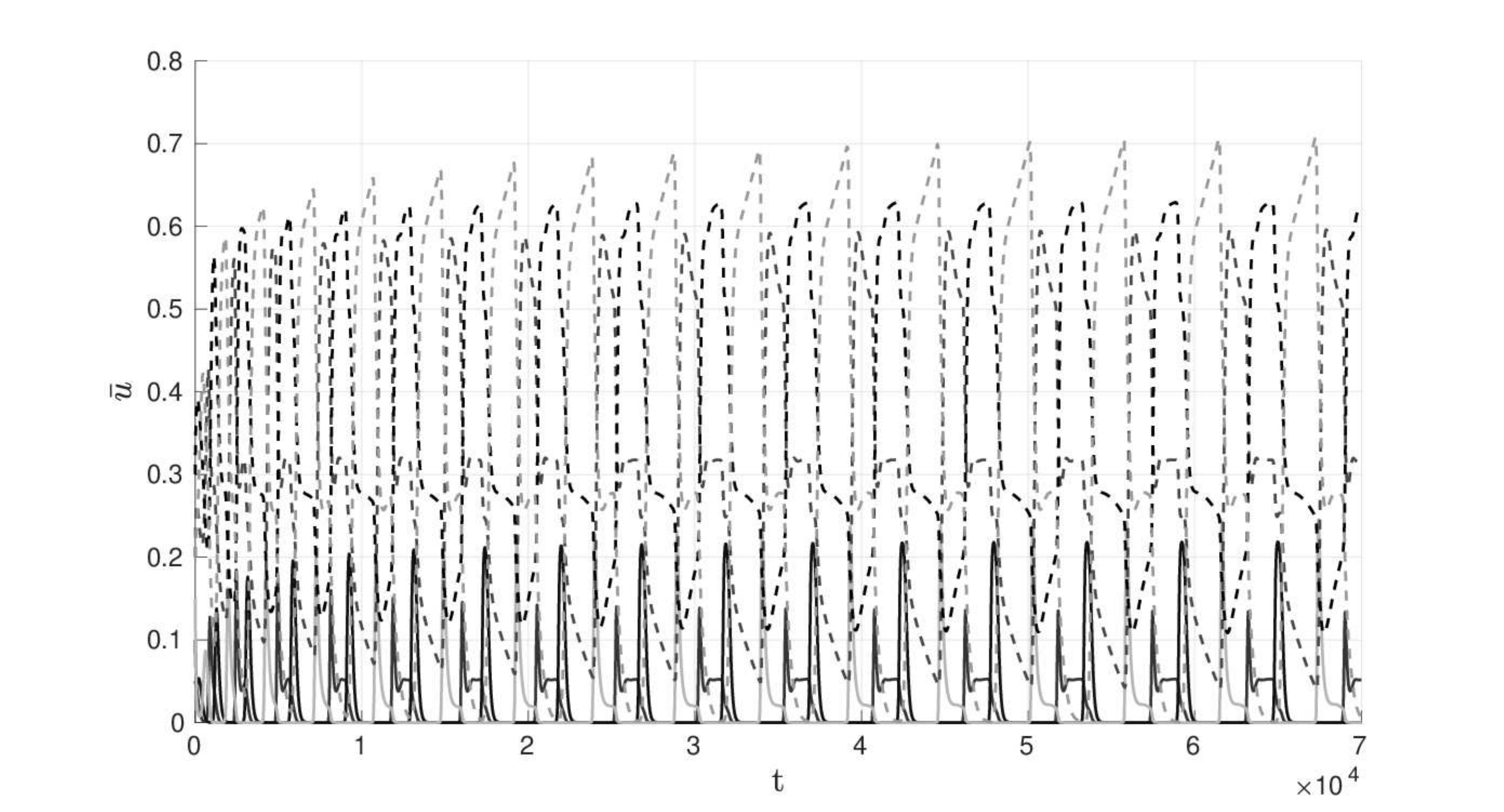} (d)}
	\end{minipage}
	\caption{Phase trajectories of system \eqref{c3:eq1.25} with matrix ${\bf A}$ \eqref{c3:eq3.29} at evolutionary step $N$: (a) $N = 0$, (b) $N = 125$, (c) $N = 175$, (d) $N = 200$.}
	\label{c3:fig3.15}
\end{figure}

\clearpage

\section{Overview of Replicator System Evolution Examples}\label{c3:section:3.8}
The numerical results show that the evolutionary increase in mean fitness (fitness landscape maximisation) of both the hypercycle and the bi-hypercycle follows a common pattern. 

\textit{Phase 1.} Over a considerable number of evolutionary parameter steps, the system equilibrium remains unchanged. During this phase, the mean fitness increases and the species interaction graph undergoes a fundamental reorganisation: in addition to the original unidirectional hypercyclic connections, reverse hypercyclic connections appear, and each species acquires bidirectional links to all other species. Moreover, the character of replication changes: the original purely altruistic mode --- in which each species catalyses only its cyclic neighbour --- gives way to autocatalytic replication, representing the emergence of selfish behaviour in which each species can replicate itself. All species thus acquire the capacity for both altruistic and selfish replication.

\textit{Phase 2.} The fixed-point coordinates of the system split: one coordinate begins to increase while the remaining equilibrium components decrease. The monotone increase in fitness continues, but the autocatalytic (selfish) coefficient of the species whose equilibrium coordinate is growing begins to increase, while those of the remaining species do not. Numerical experiments indicate that the index of the dominant species is a random variable, most likely dependent on rounding errors in the computation.

\textit{Phase 3.} Stabilisation occurs: both the mean fitness and the interaction graph effectively cease to change. The evolutionary adaptation process terminates. At this stage, the constraints on the signs of the equilibrium coordinates --- the non-degeneracy (permanence) conditions --- play an essential role in the numerical computation. As noted earlier, this phenomenon is analogous in nature to the ``error catastrophe'' threshold in the Crow--Kimura quasispecies model and signifies that, for a given resource bound \eqref{c3:eq3.1} on the matrix ${\bf A}$, the evolutionary fitness increase has exhausted its possibilities.

Numerical experiments show that the larger the resource $M > 0$ in \eqref{c3:eq3.1}, the later stabilisation occurs and the higher the limiting fitness value.

A key feature of the proposed evolutionary adaptation process is that the system acquires resistance (robustness) to parasitic species by which it would otherwise have been destroyed before the adaptation process began. The larger the resource $M$, the greater the system's potential resistance to parasitic species.

A natural question arises concerning the relationship between the proposed approach and the concept of evolutionarily stable state, widely used in the theory of evolutionary dynamics of replicator systems via game-theoretic methods.

Observe that an evolutionarily stable equilibrium is necessarily asymptotically stable, and moreover globally stable. Many replicator systems, however, possess no asymptotically stable equilibria (e.g., the hypercycle for $n \geqslant 5$ and the bi-hypercycle for $n \geqslant 6$). Despite this, the proposed evolutionary adaptation process is capable of increasing the mean fitness several-fold without altering the equilibrium coordinates.

On the other hand, biological systems are characterised precisely by the property of stable non-equilibrium states that enable effective responses to environmental change. In contrast to the selfish Nash equilibrium, the proposed evolutionary adaptation process employs the balancing approach of~\cite{c3:Berge1957}, reconciling altruistic and selfish tendencies to achieve maximum collective benefit.

For low-dimensional hypercycles ($n = 2, 3$), the equilibrium is asymptotically stable and evolutionarily stable. Numerical experiments confirm that at each step of the proposed process the equilibrium retains its evolutionarily stable state property in these cases.

In other words, the proposed process does not conflict with the game-theoretic approach and the notion of evolutionarily stable state when a unique asymptotically stable equilibrium exists.

Throughout the evolutionary adaptation process, the geometry of the fitness surface changes substantially, while the fitness landscape retains its hyperbolic character in the sense defined above.

Fig.~\ref{c3:fig3.16}\,(a) shows the fitness values at the simplex vertices
$$
	(1,\, 0,\, 0,\, 0,\, 0), \quad (0,\, 1,\, 0,\, 0,\, 0), \quad (0,\, 0,\, 1,\, 0,\, 0),
$$
$$
	(0,\, 0,\, 0,\, 1,\, 0), \quad (0,\, 0,\, 0,\, 0,\, 1),
$$
and at the simplex centre $\left(\dfrac{1}{5}, \dfrac{1}{5}, \dfrac{1}{5}, \dfrac{1}{5}, \dfrac{1}{5}\right)$.

\begin{figure}[H]
	\begin{minipage}[ht]{0.48\linewidth}
		\center{\includegraphics[width=\linewidth]{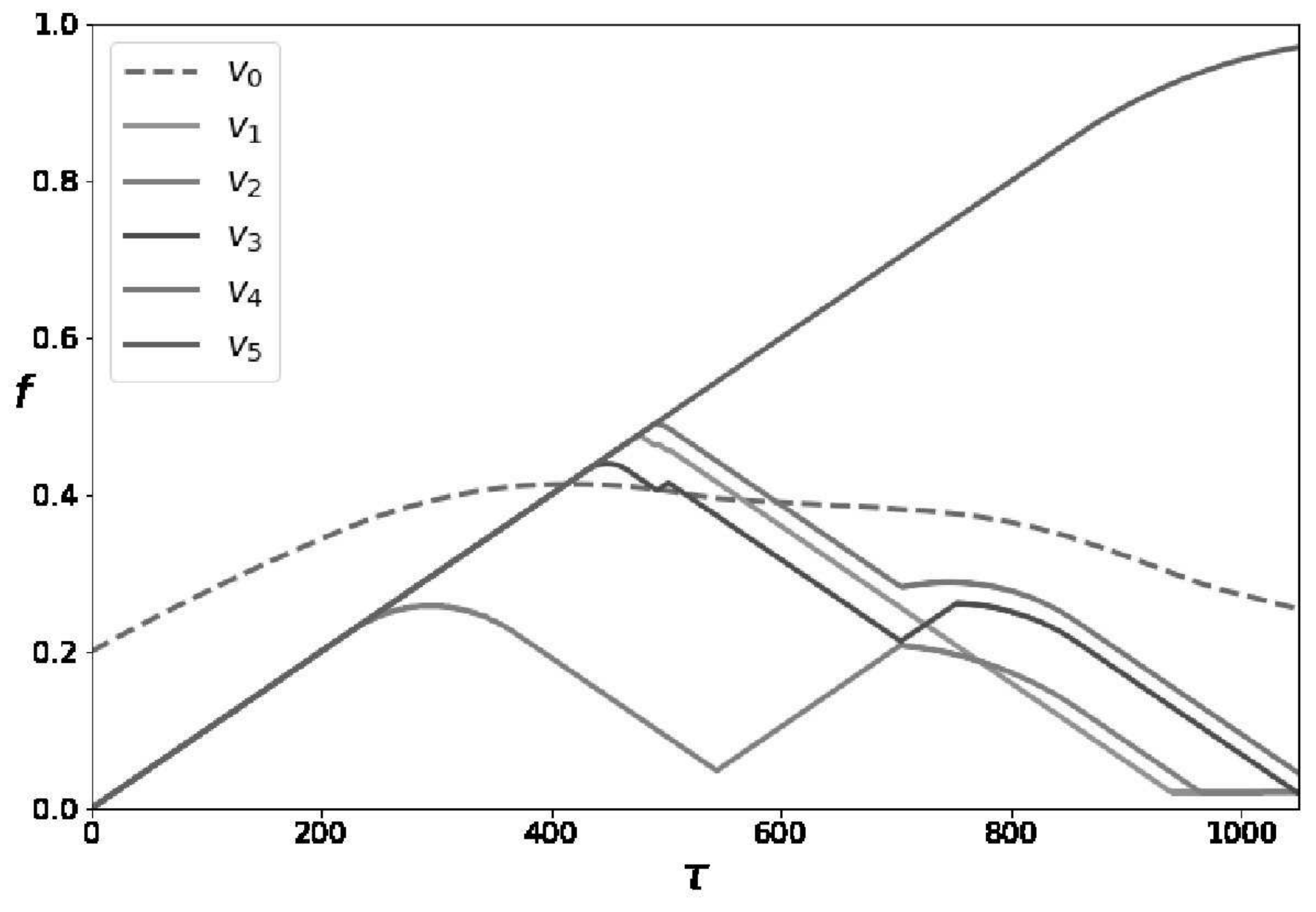}}
		\center{(a)}
	\end{minipage}
	\hfill
	\begin{minipage}[ht]{0.48\linewidth}
		\center{\includegraphics[width=\linewidth]{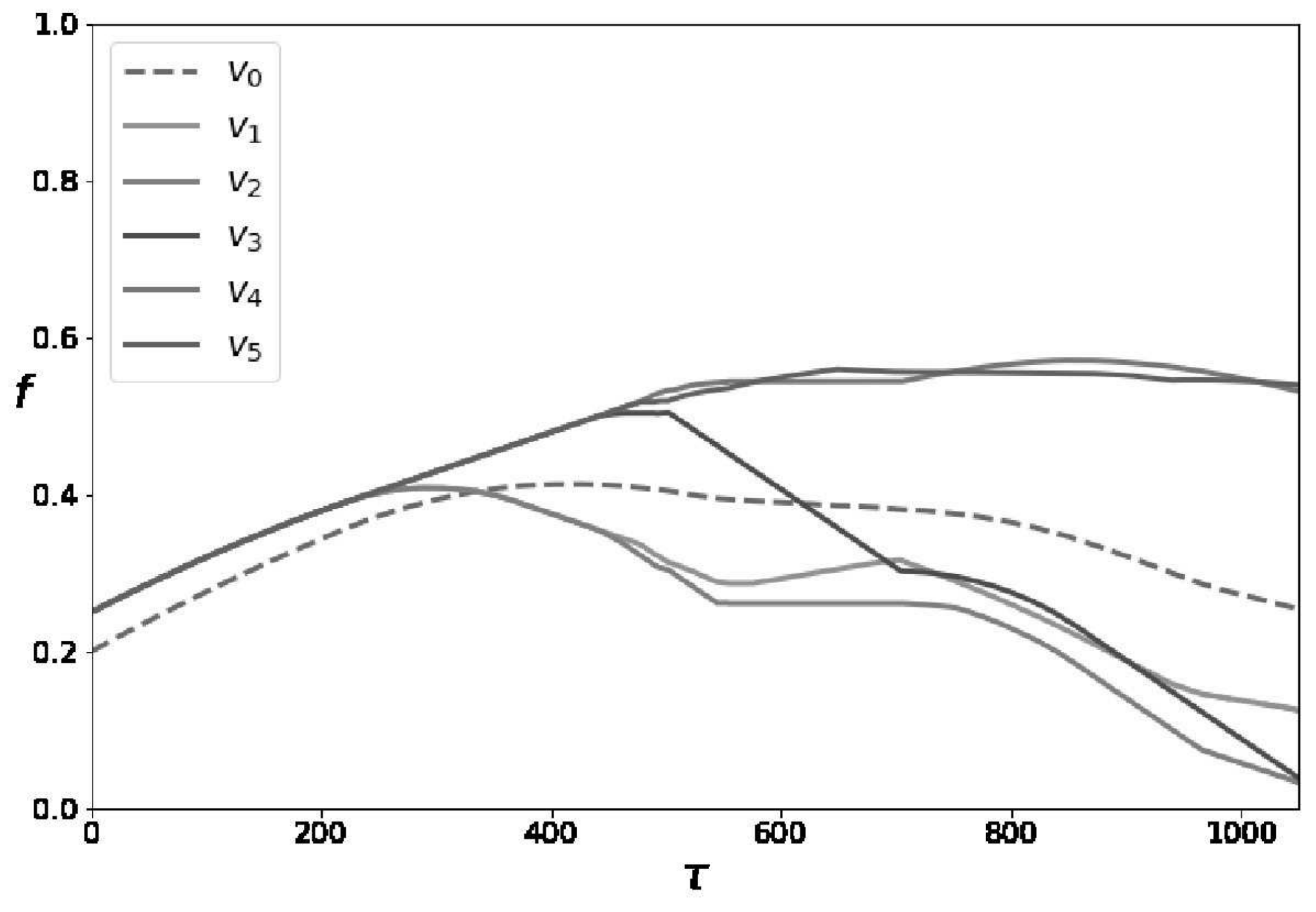}}
		\center{(b)}
	\end{minipage}
	\caption{Fitness values as functions of evolutionary time $\tau$ for the
		fifth-order hypercycle: (a) at the simplex vertices and centre;
		(b) at the midpoints of the simplex faces and centre (dashed).}
	\label{c3:fig3.16}
\end{figure}

Fig.~\ref{c3:fig3.16}\,(b) shows the fitness values at the midpoints of the simplex faces
\begin{multline*}
	\Bigg(\frac{1}{2},\, \frac{1}{2},\, 0,\, 0,\, 0\Bigg), \quad \Bigg(0,\, \frac{1}{2},\, \frac{1}{2},\, 0,\, 0\Bigg), \quad \Bigg(0,\, 0,\, \frac{1}{2},\, \frac{1}{2},\, 0\Bigg), \\ \Bigg(0,\, 0,\, 0,\, \frac{1}{2},\, \frac{1}{2}\Bigg), \quad \Bigg(\frac{1}{2},\, 0,\, 0,\, 0,\, \frac{1}{2}\Bigg),
\end{multline*}
and at the centre, as a function of the evolutionary parameter $\tau$ for the fifth-order hypercycle. The evolution of the centre is indicated by a dashed line.

Analysis of the phase trajectory behaviour reveals a transition from a stable limit cycle at the beginning of the process to a heteroclinic cycle at its end. A detailed study of the changes in fitness surface geometry during the maximisation process requires separate investigation and lies beyond the scope of this work.

Using the hypercycle classification of~\cite{c3:Eigen1982}, the hypercycles produced by the evolutionary adaptation process attain the maximum degree of functional connectivity.

A further natural question concerns the system behaviour when non-degeneracy constraints are removed. Numerical experiments show that in this case a sequential process of species annihilation occurs, ultimately leaving a single surviving species. This observation admits a mathematical justification.

Let ${\bf A}$ and ${\bf B}$ be matrices of dimensions $n \times n$ and $(n - 1) \times (n - 1)$ respectively, satisfying
$$
	\sum\limits_{i, j = 1}^{n}a_{ij}^{2} = \sum\limits_{i, j = 1}^{n - 1}b_{ij}^{2} = M, \quad M = \const > 0.
$$ 
Denoting the eigenvalues of ${\bf A}$ and ${\bf B}$ by $\lambda_{i}({\bf A})$, $i = \overline{1, n}$ and $\lambda_{j}({\bf B})$, $j = \overline{1, n - 1}$ respectively, we claim that
\begin{equation}
	\max\limits_{({\bf u},{\bf I}_{n}) = 1}\Big({\bf Au}, {\bf u}\Big) \geqslant \frac{M}{n^{2}}, \quad \max\limits_{({\bf v},{\bf I}_{n - 1}) = 1}\Big({\bf Bv}, {\bf v}\Big) \geqslant \frac{M}{(n - 1)^{2}}.
	\label{c3:eq3.30}
\end{equation}
Note that without loss of generality ${\bf A}$ and ${\bf B}$ may be taken to be symmetric, since the skew-symmetric part of a matrix does not contribute to the fitness.

Let $\lambda_{1}({\bf A})$ and $\lambda_{1}({\bf B})$ denote the eigenvalues of ${\bf A}$ and ${\bf B}$ largest in absolute value. Then
\begin{equation}
	\lambda_{1}({\bf A}) \geqslant \frac{M}{n}, \quad \lambda_{1}({\bf B}) \geqslant \frac{M}{n - 1}.
	\label{c3:eq3.31}
\end{equation} 
Let ${\bf \bar{u}}_{A}$ be the eigenvector of ${\bf A}$ corresponding to $\lambda_{1}({\bf A})$, normalised so that $({\bf \bar{u}}_{A}, {\bf I}_{n}) = 1$, and ${\bf \bar{u}}_{B}$ the eigenvector of ${\bf B}$ corresponding to $\lambda_{1}({\bf B})$, normalised by $({\bf \bar{u}}_{B}, {\bf I}_{n - 1}) = 1$. Then
\begin{equation*}
	\begin{aligned}
		&\max\limits_{({\bf u}_{A},{\bf I}_{n}) = 1}\Big({\bf A}{\bf u}_{A}, {\bf u}_{A}\Big) = \lambda_{1}(\bf A)||{\bf \bar{u}}_{A}||^{2}, \\ 
		&\max\limits_{({\bf u}_{B},{\bf I}_{n - 1}) = 1}\Big({\bf B}{\bf u}_{B}, {\bf u}_{B}\Big) = \lambda_{1}(\bf B)||{\bf \bar{u}}_{B}||^{2}.
	\end{aligned}
\end{equation*} 
Since
\begin{equation}
	\begin{aligned}
		&\displaystyle 1 = ({\bf u}_{A}, {\bf I}_{n}) \leqslant \Big(\sum\limits_{i = 1}^{n}({\bf u}_{A})_{i}^{2}\Big)^{1/n}\sqrt{n} = ||{\bf u}_{A}||\sqrt{n}, \\
		&\displaystyle 1 = ({\bf u}_{B}, {\bf I}_{n - 1}) \leqslant ||{\bf u}_{B}||\sqrt{n - 1}, \\
	\end{aligned}
	\label{c3:eq3.32}
\end{equation} 
inequality \eqref{c3:eq3.30} follows from \eqref{c3:eq3.31} and \eqref{c3:eq3.32}.

Inequality \eqref{c3:eq3.30} shows that the maximum fitness can only increase as the system dimension decreases, provided the spherical norm of the interaction matrix is fixed. This provides a mathematical explanation for the observed process of species annihilation when non-degeneracy constraints are removed.

\section*{Conclusion}\addcontentsline{toc}{section}{Conclusion}

We have proposed and analysed a mathematical model of evolutionary adaptation for
non-degenerate replicator systems. The key structural result (Theorem~\ref{c3:theorem:3.1})
expresses the rate of change of mean fitness in terms of a fitness variation
formula, and shows that the adaptation problem reduces, at each step of the
evolutionary time, to a linear programming problem over the admissible set
$\mathcal{M}$ of fitness landscape matrices. The algorithm was applied to four
canonical replicator systems, yielding a consistent picture across all examples.

The evolutionary adaptation process unfolds in three phases: an initial phase in
which the equilibrium is unchanged while the interaction graph undergoes a
fundamental restructuring from purely altruistic to mixed altruistic--selfish
replication; a second phase in which one dominant species emerges; and a
stabilisation phase in which fitness and graph topology effectively cease to
change. Throughout, the non-degeneracy constraints play a critical role:
they prevent species annihilation and are responsible for the error-catastrophe-like
threshold that terminates fitness growth.

A notable consequence is that all evolved systems acquire resistance to parasitic
species that would have destroyed the original system. The evolved hypercycles
attain maximum functional connectivity in the classification of~\cite{c3:Eigen1982},
and the proposed process does not conflict with the game-theoretic notion of
evolutionarily stable state when a unique asymptotically stable equilibrium
exists.

The present work extends the fitness optimisation results of~\cite{c3:Drozhzhin2021}
and provides a constructive algorithm for evolutionary adaptation that can be
applied to more general permanent replicator systems. The connection to the
geometry of the fitness surface studied in~Chapter~\ref{ch:2} --- in particular
the relationship between fitness maxima, evolutionarily stable states, and the
structure of the phase portrait --- suggests several directions for further
investigation, including the geometry of the fitness surface during the
adaptation process and adaptation under stochastic perturbations.


%% file: chapter4.tex
\chapter{Evolution of the Replicator System with New Species Adaptation at Random Time Moments}
\label{ch:4}

\begin{chapterabstract}
We extend the mathematical model of evolutionary adaptation for permanent
replicator systems developed in~Chapter~\ref{ch:3} to the setting in which
new species enter the system at random time moments governed by a Poisson
process. Four structural assumptions are imposed at each incorporation event,
relating the fitness matrix and equilibrium frequencies before and after
the arrival of a new species. Under these assumptions the entire extended
fitness matrix is determined by a single scalar parameter $0 < \alpha_k < 1$,
whose admissible range is characterised explicitly in terms of the minimum
and maximum equilibrium frequencies of the incumbent species. We prove that
whenever $\alpha_k$ falls below the upper bound
$1/(u_{\max}^k + 1)$, the evolutionary process continues to maximise mean
stationary fitness without loss of permanence. Numerical experiments on
hypercyclic systems of order $n = 3$ and $n = 5$ confirm these findings:
following each incorporation event, mean fitness resumes its monotone increase
and the newly added species settle into oscillation zones whose amplitudes are
strictly narrower than those of the original system. The problem of
establishing necessary and sufficient conditions for new-species integration
in the general case remains open.
\end{chapterabstract}

\bigskip\hrule\bigskip

\section*{Introduction}\addcontentsline{toc}{section}{Introduction}
This chapter is an edited English version of Chapter~4 of the monograph
\textit{Mathematical Models of Evolution and Replicator Systems
Dynamics}~\cite{c4:BratusBook2022}. It extends the evolutionary adaptation
model of Chapter~\ref{ch:3} and relies on the replicator framework of
Chapters~\ref{ch:1} and~\ref{ch:2}; the present chapter can nevertheless be
read independently, given the self-contained background below.

\subsection*{Replicator systems: a brief summary}

A \textit{replicator} is any entity capable of self-reproduction, heritable
variation, and competition for resources~\cite{c4:Eigen1971,c4:Eigen1979}. The
dynamics of $n$ competing species, represented by their relative frequencies
${\bf u}(t) = (u_1(t), \ldots, u_n(t)) \in S_n$, where
\[
    S_n = \left\{ {\bf u} \geqslant 0 : \sum_{i=1}^n u_i = 1 \right\}
\]
is the standard simplex, is governed by the \textit{replicator equation}
\begin{equation}
    \frac{du_i}{dt} = u_i\Big[\Big({\bf Au}\Big)_i - f({\bf u})\Big],
    \quad f({\bf u}) = \Big({\bf Au}, {\bf u}\Big), \quad i = \overline{1,n}.
    \label{c4:eq1.5}
\end{equation}
Here $\Big({\bf Au}\Big)_i = \sum_{j=1}^n a_{ij} u_j$ is the \textit{fitness}
of species $i$, $f({\bf u})$ is the \textit{mean fitness} of the population,
and ${\bf A} = (a_{ij})$ is the \textit{fitness landscape matrix}. The scalar
product $(\cdot,\cdot)$ denotes the standard inner product in $\mathbb{R}^n$.

A replicator system is called \textit{non-degenerate} (or \textit{permanent}) if
there exists a compact set $K \subset \inside S_n$ such that every trajectory
starting in $\inside S_n$ eventually enters and remains in $K$. Permanence
guarantees the long-run coexistence of all species; the hypercycle is the
paradigmatic permanent system (Chapter~\ref{ch:1}).

Three canonical replication regimes are analysed in~Chapter~\ref{ch:1}: independent
replication ($a_{ij} = k_i\delta_{ij}$), autocatalytic replication, and the
\textit{hypercycle}
\begin{equation}
    \frac{du_i}{dt} = u_i(u_{i-1} - f({\bf u})), \quad u_0 = u_n,
    \quad f({\bf u}) = \sum_{i=1}^n u_i u_{i-1},
    \label{c4:eq1.8}
\end{equation}
which is permanent for all $n \geqslant 2$ and admits a unique interior
equilibrium at $\bar{u}_i = k_{i+1}^{-1}/\sum_{j} k_{j+1}^{-1}$.
The geometry of the mean fitness surface and its relationship to evolutionarily
stable states were studied in~Chapter~\ref{ch:2}.

Variability --- the capacity for heritable change --- is one of Darwin's four
fundamental evolutionary postulates. As shown in~Chapter~\ref{ch:1}, the
hypercycle already encodes a rudimentary form of variability: if one species
is replaced by another with superior catalytic properties, selection retains
only one of the competing catalytic branches. This observation motivates the
central question of the present chapter.

\subsection*{The evolutionary adaptation framework and the present contribution}

Chapter~\ref{ch:3} developed a mathematical model of \textit{evolutionary
adaptation} for permanent replicator systems. The key hypothesis is a separation
of timescales: the active dynamics of system~\eqref{c4:eq1.5} operates on a fast
time $t$, while the fitness landscape matrix ${\bf A}$ evolves on a slow
evolutionary time $\tau = \varepsilon t$, $\varepsilon \ll 1$. In this limit,
justified by Tikhonov's theorem~\cite{c4:Tihonov1948}, the system is always near
a steady state ${\bf \bar{u}}(\tau)$, and the adaptation problem reduces to
maximising the mean fitness $\bar{f}(\tau) = f({\bf \bar{u}}(\tau))$ over an
admissible convex set $\mathcal{M}$ of fitness matrices. Each step of this
maximisation reduces to a linear programming problem (Chapter~\ref{ch:3}).
Numerical experiments in~Chapter~\ref{ch:3} demonstrate that the process
follows a universal three-phase pattern and that evolved systems acquire
resistance to parasitic species.

All prior work in this series assumed the set of competing species to be fixed
throughout the evolutionary process. The present chapter removes that assumption.
New species arise at random time moments, modelled by a Poisson process, and
must either integrate into the existing permanent system or be rejected. The
central result is a set of sufficient conditions --- expressed entirely in
terms of the incumbent equilibrium frequencies --- under which integration
preserves permanence and the mean fitness maximisation process continues
uninterrupted.

The chapter is organised as follows. Section~\ref{c4:section:4.1} states the four
structural assumptions, derives the single-parameter parametrisation of the
extended fitness matrix, and characterises the admissible range of $\alpha_k$.
Section~\ref{c4:section:4.2} presents numerical experiments on hypercyclic systems
of dimension $n = 3$ and $n = 5$, illustrating the dynamics of mean fitness
and species interaction graphs before and after each incorporation event.

\section{Sufficient Conditions for New Species Adoption}
\label{c4:section:4.1}

It was shown in~Chapter~\ref{ch:1} that the hypercycle possesses the capacity for
change: if one of its species is replaced by another with superior catalytic
properties, selection acts on the catalytic branches and retains only one.
Variability of this kind is one of Darwin's four fundamental evolutionary
postulates. It is therefore natural to study the evolutionary adaptation process
of~Chapter~\ref{ch:3} in the setting where new species arrive at random time
moments and attempt to integrate into an existing permanent replicator system.
As emphasised in~\cite{c4:Markov2014,c4:Kunin2011}, such dynamics, driven by the
episodic arrival of novel types rather than by gradual drift within a fixed species
pool, are central to understanding major evolutionary transitions.

Numerical experiments show that integration is far from automatic. In some cases
the newcomer acts as a parasite and destroys the system; in others, the system
eliminates it. The conditions below are sufficient to guarantee successful
integration.

Suppose a new species is added to the $n$-dimensional system~\eqref{c4:eq1.5} at
the random time $\tau_k = k\Delta\tau$. Let
${\bf A}^k \in \mathbb{R}^{(n+1)\times(n+1)}$ denote the fitness matrix
arising at the $k$-th incorporation event, and let
${\bf \bar{u}}_k = (u_1^k, u_2^k, \ldots, u_{n+1}^k)$ be the corresponding
equilibrium frequency vector.

We impose the following four assumptions at each incorporation event.

\begin{assumpt}[Fitness continuity]\label{c4:ass:1}
The mean fitness at equilibrium on the $k$-th iteration equals the value it
would have attained without the new species:
\begin{equation}
    \bar{f}(\tau_k) = \bar{f}_k.
    \label{c4:eq4.1}
\end{equation}
\end{assumpt}

\begin{assumpt}[Matrix inheritance]\label{c4:ass:2}
The first $n$ rows and $n$ columns of ${\bf A}^k$ coincide with the
corresponding entries of ${\bf A}^{k-1}$:
\begin{equation}
    a_{ij}^k = a_{ij}^{k-1}, \quad i, j = \overline{1,n}.
    \label{c4:eq4.2}
\end{equation}
\end{assumpt}

\begin{assumpt}[Proportional equilibrium shift]\label{c4:ass:3}
The first $n$ components of the equilibrium vector satisfy the proportionality
relations between iterations $k$ and $k-1$:
\begin{equation}
    \frac{u_i^k}{u_j^k} = \frac{u_i^{k-1}}{u_j^{k-1}}, \quad
    i, j = \overline{1,n}.
    \label{c4:eq4.3}
\end{equation}
\end{assumpt}

\begin{assumpt}[Equal incumbent contributions]\label{c4:ass:4}
The contribution of each incumbent species to the fitness of the new species
is equal:
\begin{equation}
    a_{n+1,\,j}^k\, u_j^k = a_{n+1,\,i}^k\, u_i^k, \quad
    i, j = \overline{1,n}.
    \label{c4:eq4.4}
\end{equation}
\end{assumpt}

Assumption~\ref{c4:ass:3} immediately implies the existence of a parameter
$0 < \alpha_k < 1$ such that
\begin{equation}
    u_{n+1}^k = 1 - \alpha_k, \qquad
    u_i^k = \alpha_k\, u_i^{k-1}, \quad i = \overline{1,n}.
    \label{c4:eq4.5}
\end{equation}
This is verified directly:
\[
    \sum_{i=1}^{n+1} u_i^k
    = \alpha_k \sum_{i=1}^{n} u_i^{k-1} + (1 - \alpha_k) = 1.
\]
Using Assumption~\ref{c4:ass:1} together with~\eqref{c4:eq4.5}, the mean fitness at
the new equilibrium satisfies
\[
    \bar{f}_k = \sum_{j=1}^{n+1} a_{ij}^k u_j^k
    = \alpha_k \bar{f}_{k-1} + (1-\alpha_k)\, a_{i,\,n+1}^k,
\]
from which all entries of the new column of ${\bf A}^k$ (except
$a_{n+1,\,n+1}^k$) are uniquely determined:
\begin{equation}
    a_{i,\,n+1}^k = \frac{\bar{f}_k - \alpha_k \bar{f}_{k-1}}{1 - \alpha_k},
    \quad i = \overline{1,n}.
    \label{c4:eq4.6}
\end{equation}
The entries of the new row follow from Assumption~\ref{c4:ass:4}:
\begin{equation}
    \left\{
    \begin{aligned}
        &a_{n+1,\,j}^k = \dfrac{\bar{f}_k}{\alpha_k\, u_j^{k-1}(n+1)},
            \quad j = \overline{1,n}, \\[6pt]
        &a_{n+1,\,n+1}^k = \dfrac{\bar{f}_k}{(1-\alpha_k)(n+1)}.
    \end{aligned}
    \right.
    \label{c4:eq4.7}
\end{equation}
Thus, under Assumptions~\ref{c4:ass:1}--\ref{c4:ass:4}, the extended system at
iteration $k$ is fully determined by specifying a single parameter
$0 < \alpha_k < 1$.

The admissible range of $\alpha_k$ is characterised as follows. Set
\[
    u_{\min}^k = \min(u_1^{k-1}, \ldots, u_n^{k-1}), \qquad
    u_{\max}^k = \max(u_1^{k-1}, \ldots, u_n^{k-1}).
\]

\begin{proposition}\label{c4:prop:alpha}
If $\alpha_k > 1/(u_{\min}^k + 1)$, the replicator system obtained by
adjoining the $(n+1)$-th species is no longer permanent. Conversely, if
\begin{equation}
    \alpha_k < \frac{1}{u_{\max}^k + 1},
    \label{c4:eq:alpha_bound}
\end{equation}
then the frequency of the new species at the new equilibrium exceeds that of
every incumbent species, and the fitness-maximisation process continues without
loss of permanence of the extended system.
\end{proposition}

The bound~\eqref{c4:eq:alpha_bound} is established numerically in
Section~\ref{c4:section:4.2}; its analytic proof in full generality remains an open
problem.

In the numerical algorithm, the time moments of incorporation are generated by a
Poisson process. At each such moment, $\alpha_k$ is drawn by first sampling
$\beta_k \sim \mathcal{U}(0,1)$ and then setting
\begin{equation}
    \alpha_k = \beta_k \cdot \frac{1}{u_{\max}^k + 1}.
    \label{c4:eq4.8}
\end{equation}
This construction guarantees that~\eqref{c4:eq:alpha_bound} is satisfied at every
incorporation event.

\section{Numerical Results}\label{c4:section:4.2}

\subsection*{Hypercyclic system of order $n = 5$}

We illustrate the algorithm on a hypercyclic replicator system of order $n = 5$.
Two incorporation events were generated by the Poisson process, occurring at
iterations $682$ and $1240$ of the evolutionary adaptation process, with
parameters $\beta_{682} = 0.829$ and $\beta_{1240} = 0.557$.

Figure~\ref{c4:fig4.1} shows the dynamics of (a) mean fitness and (b) the
equilibrium fixed point. After each incorporation event the mean fitness
resumes its monotone increase; the evolutionary optimisation trajectory suffers
no disruption.

\begin{figure}[H]
    \begin{minipage}[ht]{1.0\linewidth}
        \center{\includegraphics[width=0.8\linewidth]{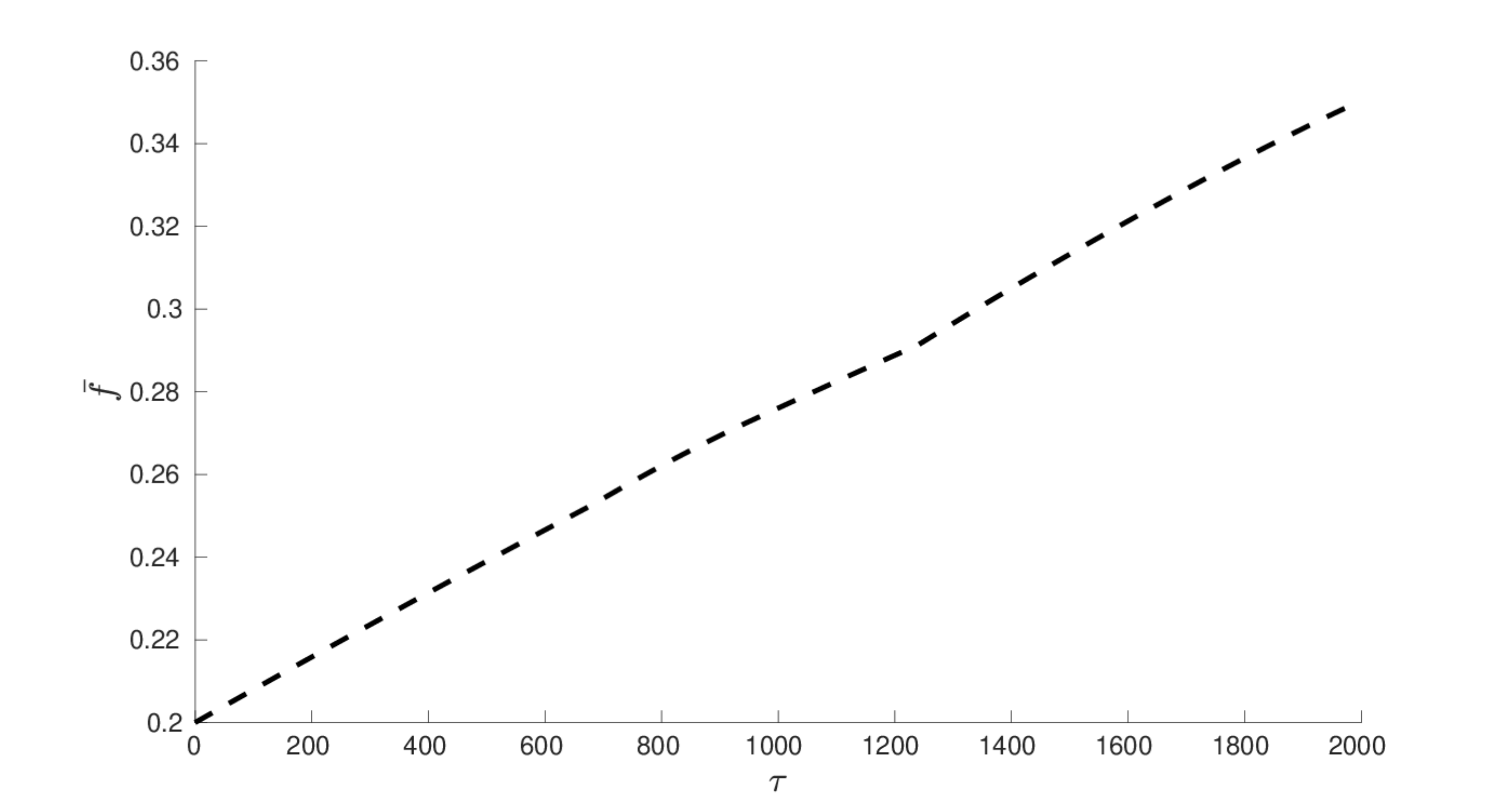} (a)}
    \end{minipage}
    \vfill
    \begin{minipage}[ht]{1.0\linewidth}
        \center{\includegraphics[width=0.8\linewidth]{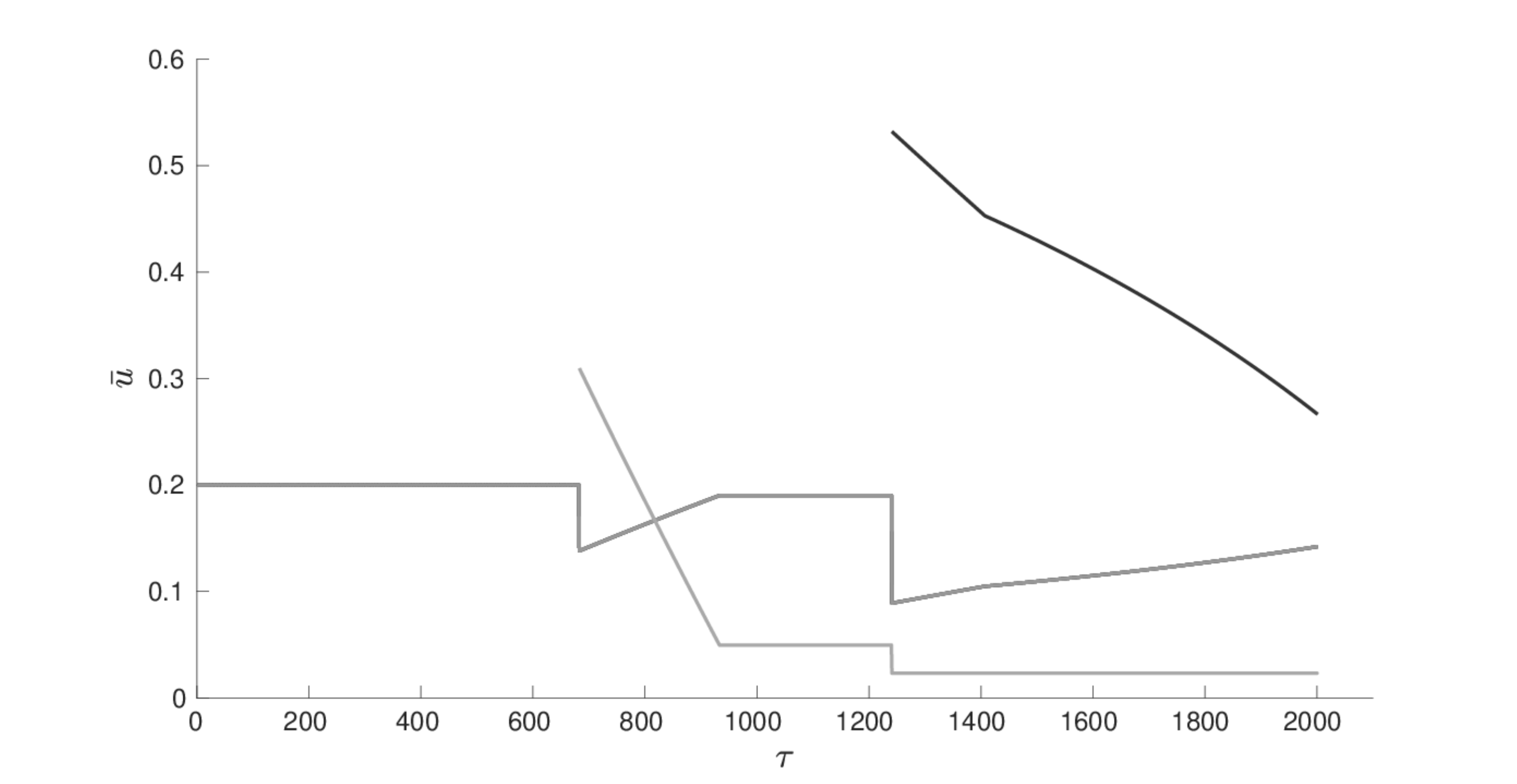} (b)}
    \end{minipage}
    \caption{Evolutionary dynamics of (a) mean fitness and (b) the equilibrium
    fixed point of a hypercyclic replicator system of dimension $n = 5$,
    following the addition of two new species. Vertical dashed lines mark the
    incorporation events at iterations $682$ and $1240$.}
    \label{c4:fig4.1}
\end{figure}

Figures~\ref{c4:fig4.2}--\ref{c4:fig4.4} show the active dynamics at iterations
$681$ (immediately before the first incorporation), $1239$ (immediately before
the second), and $1735$ (well after both events).

\begin{figure}[ht]
    \centering
    \includegraphics[width=1.0\textwidth]{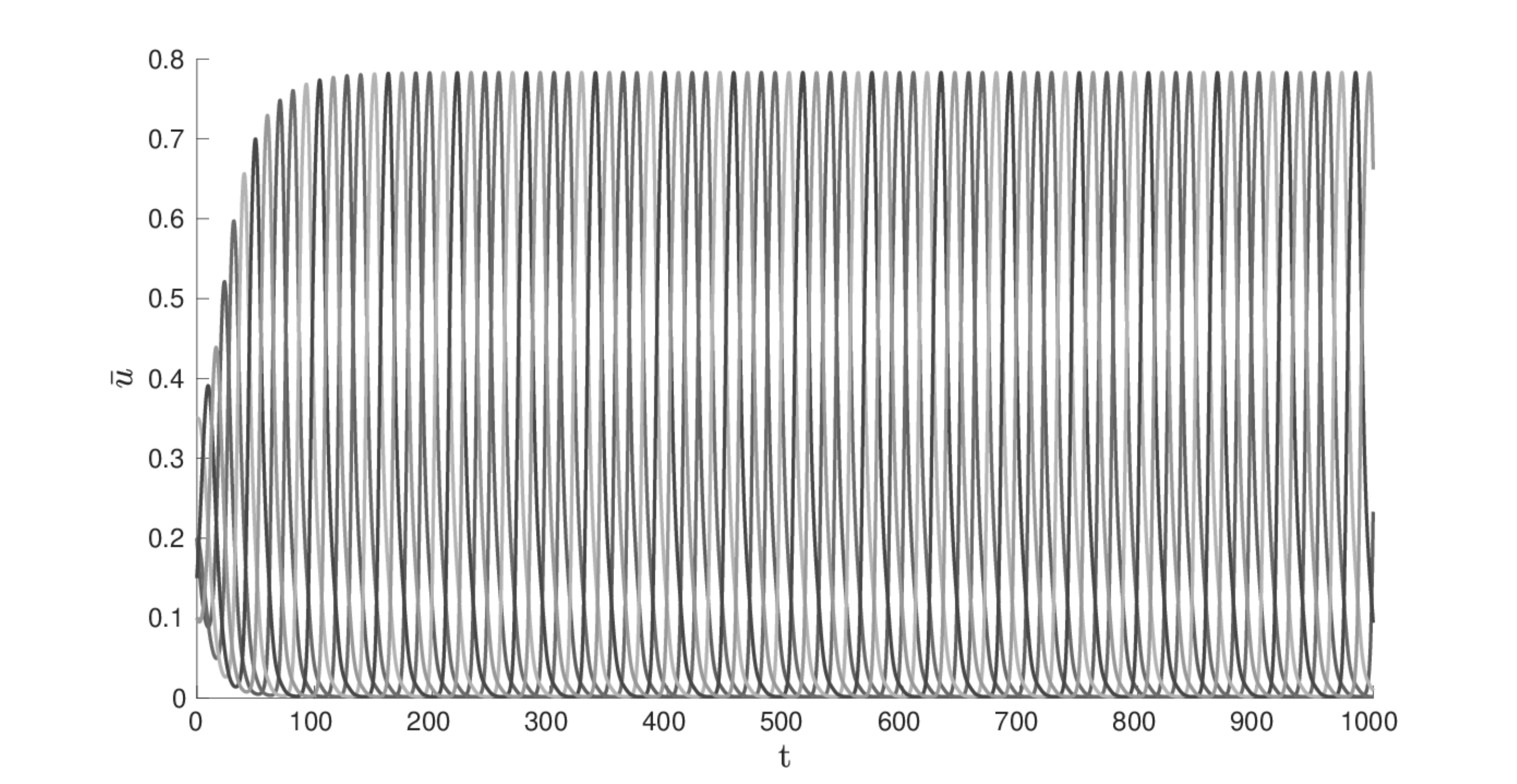}
    \caption{Active dynamics of the system at iteration $681$, dimension
    $n = 5$, immediately before the first new species is added.}
    \label{c4:fig4.2}
\end{figure}

\begin{figure}[ht]
    \centering
    \includegraphics[width=1.0\textwidth]{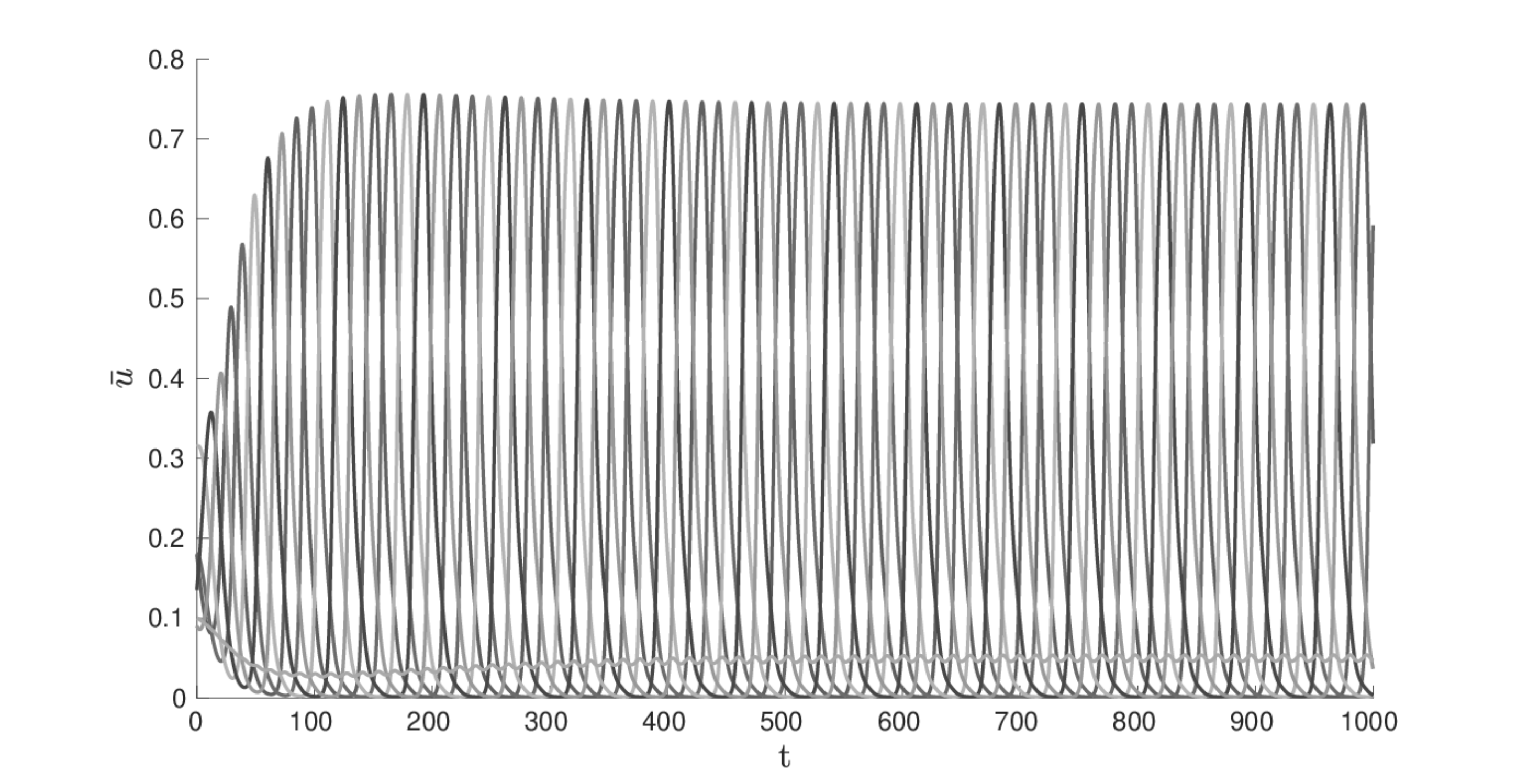}
    \caption{Active dynamics of the system at iteration $1239$, dimension
    $n = 5$, immediately before the second new species is added.
    The newcomer is indicated by the wavy line at the bottom.}
    \label{c4:fig4.3}
\end{figure}

\begin{figure}[ht]
    \centering
    \includegraphics[width=1.0\textwidth]{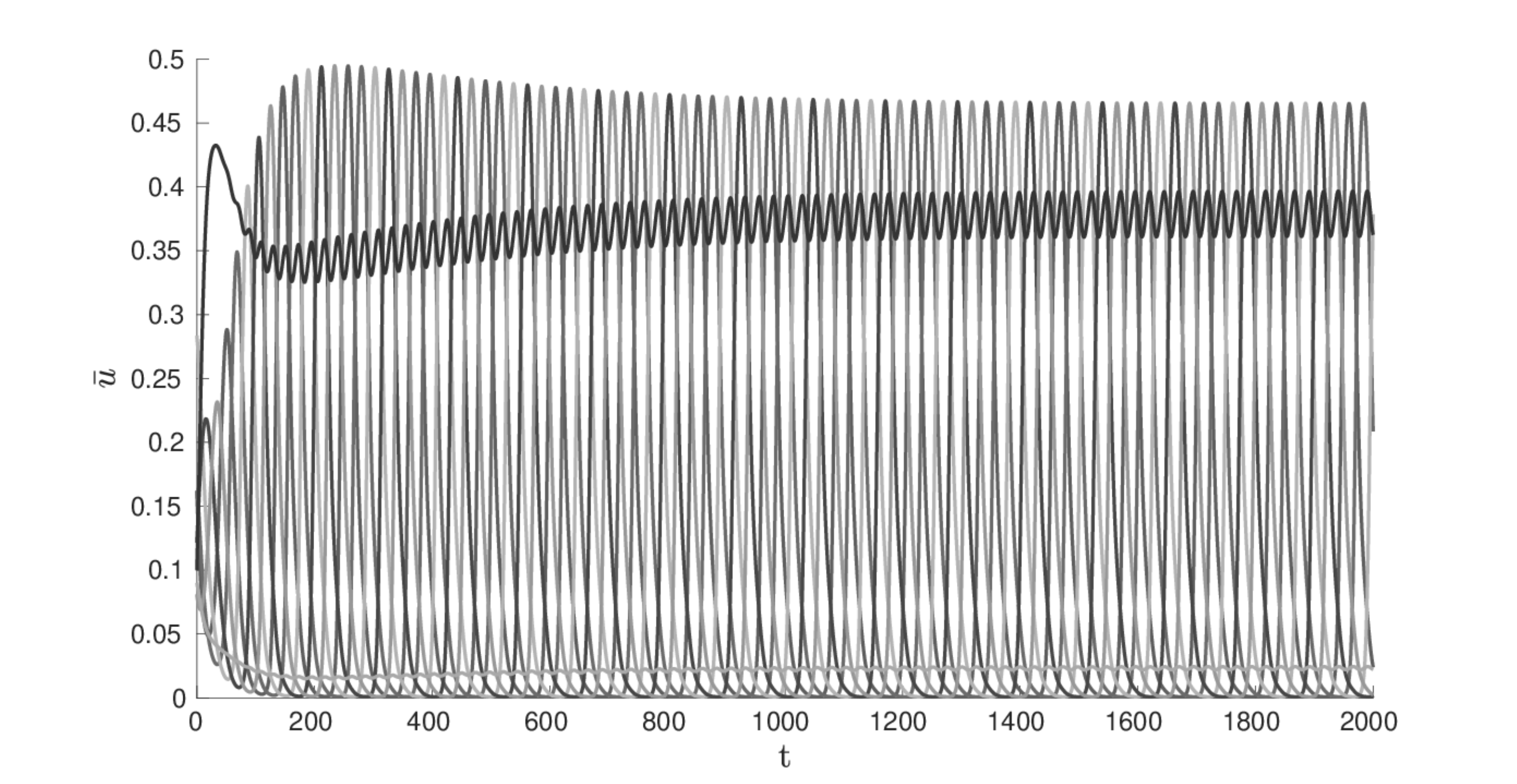}
    \caption{Active dynamics of the system at iteration $1735$, dimension
    $n = 5$. The most recently incorporated species is shown by the dark line.}
    \label{c4:fig4.4}
\end{figure}

\FloatBarrier
\subsection*{Hypercyclic system of order $n = 3$}

Figures~\ref{c4:fig4.5} and~\ref{c4:fig4.6} display the species interaction graphs
for a system that starts with $n = 3$ species, into which two new species were
added at iterations $68$ and $529$. In each figure, panel~(a) shows the
interaction graph immediately before the incorporation event and panel~(b)
shows it immediately after.

\begin{figure}[ht]
    \begin{minipage}[ht]{0.5\linewidth}
        \center{\includegraphics[width=1.0\linewidth]{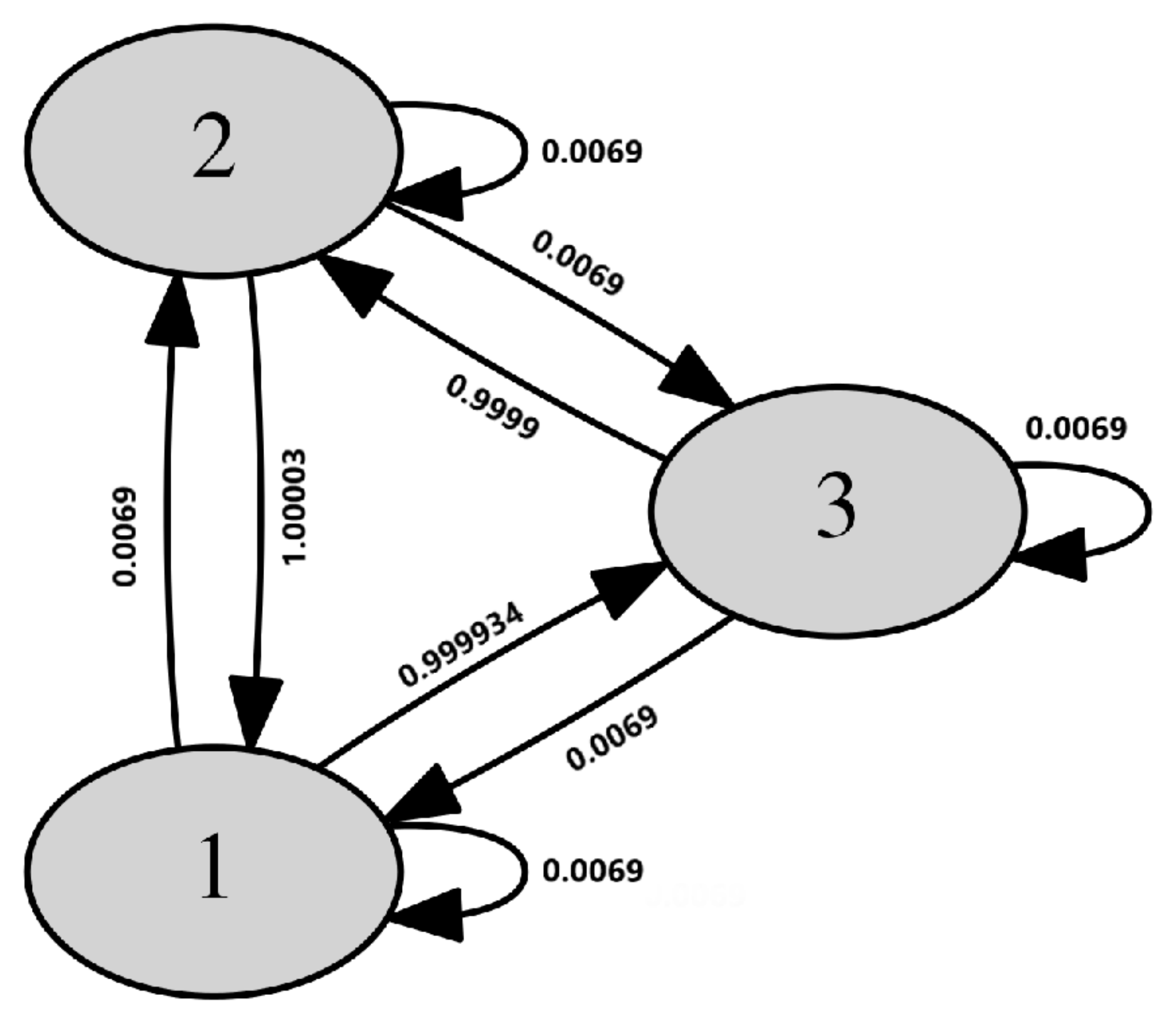} (a)}
    \end{minipage}
    \hfill
    \begin{minipage}[ht]{0.5\linewidth}
        \center{\includegraphics[width=1.0\linewidth]{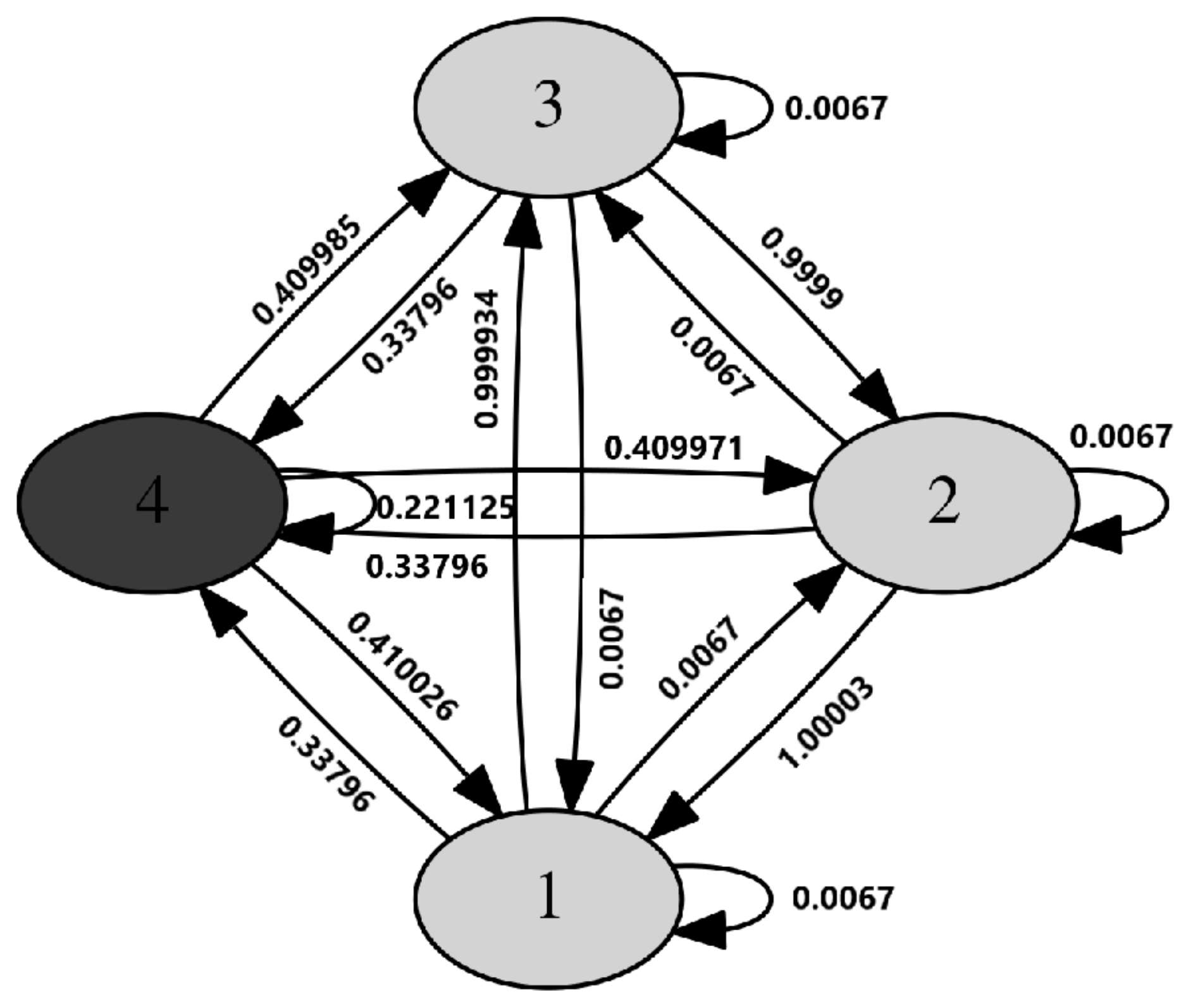} (b)}
    \end{minipage}
    \caption{Species interaction graphs at iterations $67$~(a) and $68$~(b),
    dimension $n = 3$.}
    \label{c4:fig4.5}
\end{figure}

\begin{figure}[H]
    \begin{minipage}[ht]{0.5\linewidth}
        \center{\includegraphics[width=1.0\linewidth]{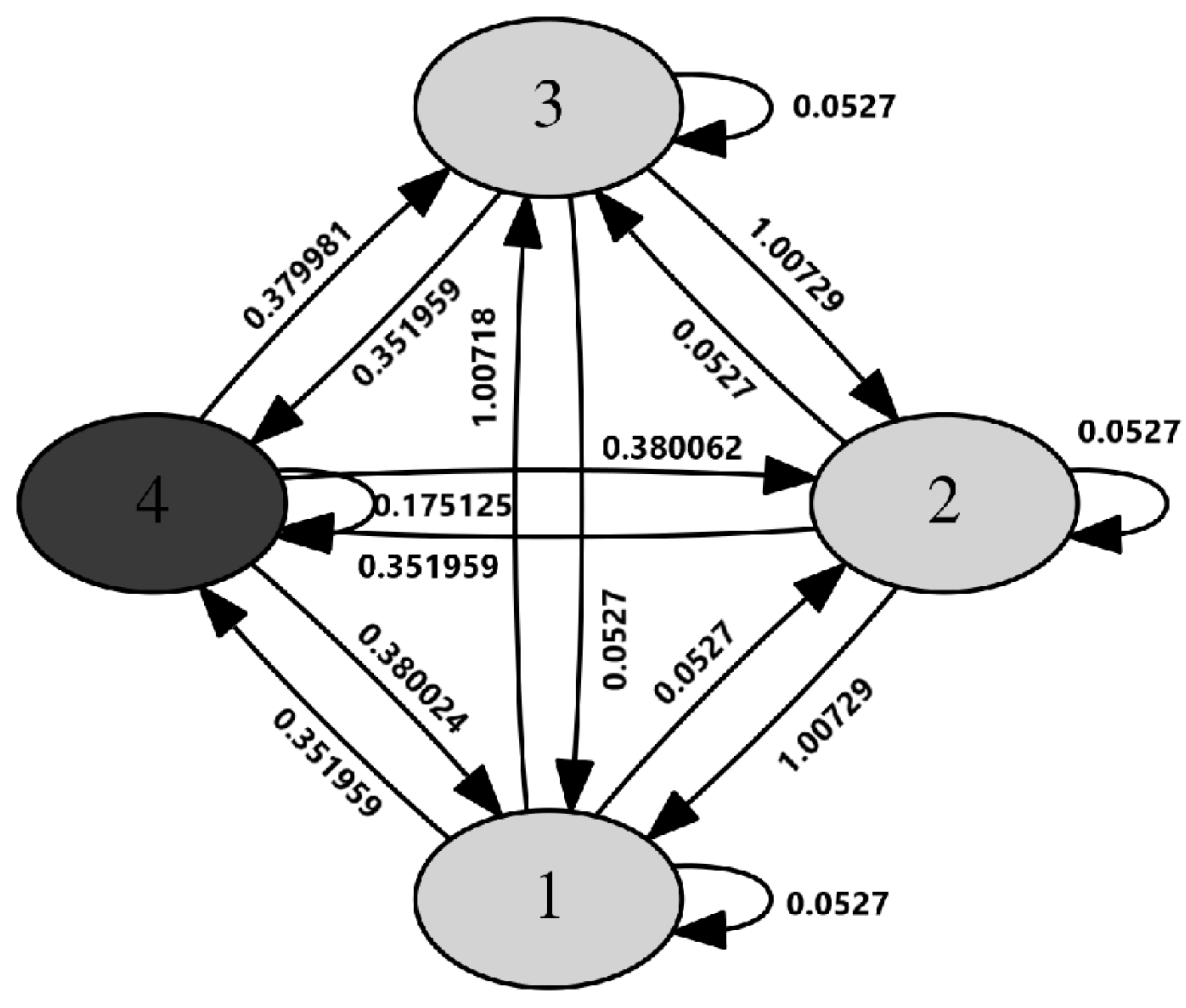} (a)}
    \end{minipage}
    \hfill
    \begin{minipage}[ht]{0.5\linewidth}
        \center{\includegraphics[width=1.0\linewidth]{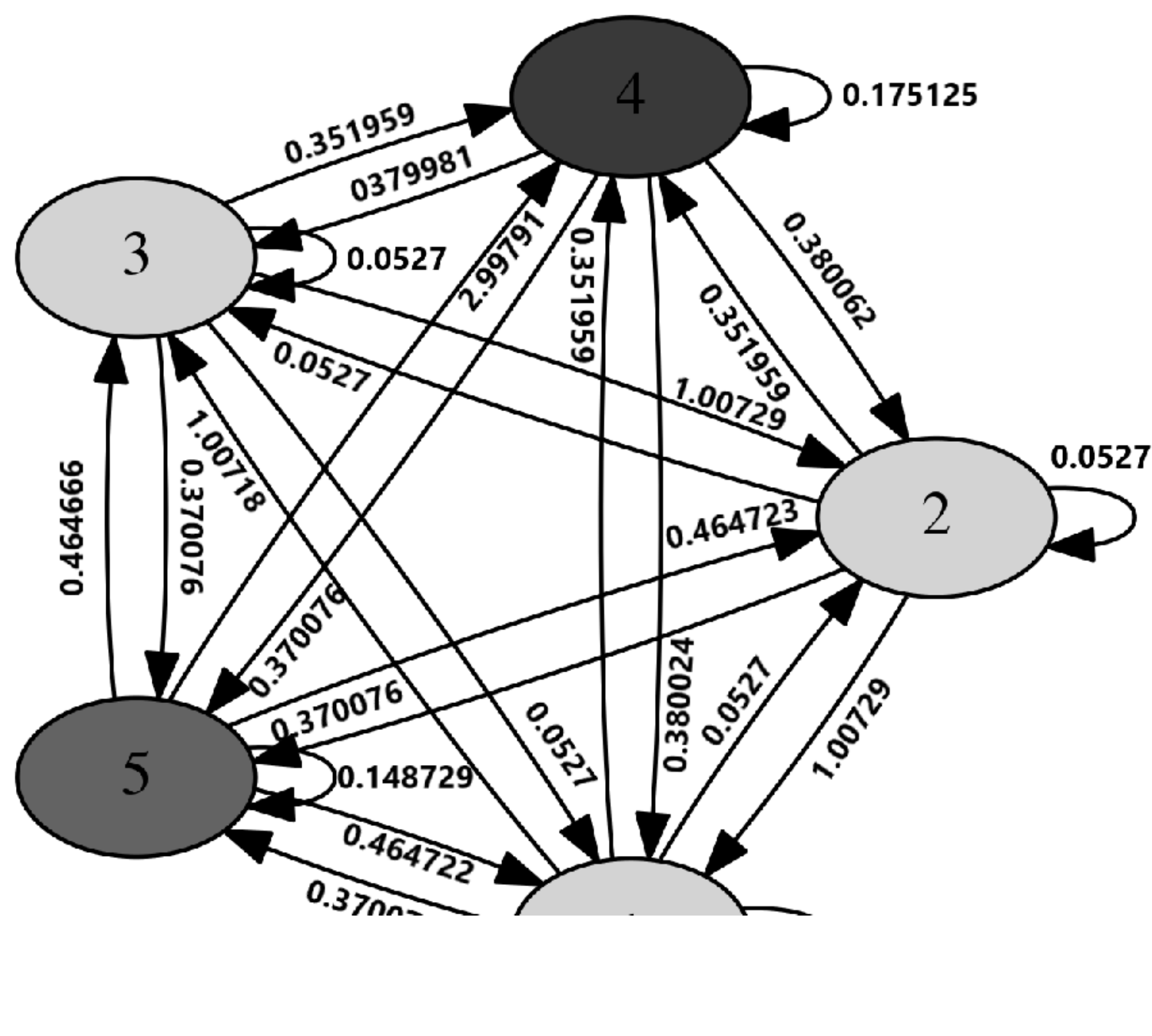} (b)}
    \end{minipage}
    \caption{Species interaction graphs at iterations $528$~(a) and $529$~(b),
    dimension $n = 3$.}
    \label{c4:fig4.6}
\end{figure}

\FloatBarrier
\subsection*{Summary of numerical findings}

The experiments across both system dimensions lead to three consistent
observations. First, each incorporation event may cause an abrupt jump in mean
fitness, after which the monotone increase resumes. Second, the frequencies of
newly added species settle into oscillation zones whose amplitudes are
noticeably narrower than those of the incumbent species in the original system;
this localisation becomes more pronounced with successive incorporation events.
Third, the species interaction graphs retain their structural connectivity after
each event, consistent with the preservation of permanence guaranteed by
condition~\eqref{c4:eq:alpha_bound}.

\clearpage
\section*{Conclusion}\addcontentsline{toc}{section}{Conclusion}

We have extended the evolutionary adaptation framework of~Chapter~\ref{ch:3}
to the case in which new species enter a permanent replicator system at random
time moments. Under four structural assumptions on the fitness matrix and
equilibrium frequencies, the problem of incorporating a new species reduces to
the choice of a single scalar parameter $\alpha_k \in (0, 1)$. The sufficient
condition $\alpha_k < 1/(u_{\max}^k + 1)$ guarantees that permanence is
preserved and mean fitness continues to increase after each incorporation event.

The conditions derived here identify a comparatively restricted class of
perturbations. Finding necessary and sufficient conditions for new-species
integration in the general case --- without the proportionality constraint of
Assumption~\ref{c4:ass:3} or the equal-contribution constraint of
Assumption~\ref{c4:ass:4} --- is the natural next step. A further direction is the
connection to the resistance-to-parasites phenomenon documented
in~Chapter~\ref{ch:3}: the present results suggest that systems evolved under
the adaptation algorithm of Chapter~3 may be more hospitable to new species
than non-evolved systems, but this conjecture has not been proved.
